\documentclass[11pt]{article}

\usepackage{amsthm}
\usepackage{amsmath}
\usepackage{amssymb}
\usepackage{bm} 
\usepackage{bbm}
\usepackage{mathrsfs}
\usepackage{mathtools}
\usepackage{dsfont}
\usepackage{float}
\usepackage{bm}
\usepackage{natbib}
\setcitestyle{authoryear,open={(},close={)}}

\usepackage{graphicx}
\usepackage{adjustbox}
\usepackage{booktabs}
\usepackage{multirow}
\usepackage{caption}
\usepackage[dvipsnames]{xcolor}
\usepackage{comment}
\usepackage{nicefrac}
\usepackage{enumitem}
\usepackage{apptools}
\AtAppendix{\counterwithin{lemma}{section}}
\AtAppendix{\counterwithin{corollary}{section}}
\AtAppendix{\counterwithin{proposition}{section}}
\AtAppendix{\counterwithin{theorem}{section}}
\AtAppendix{\counterwithin{figure}{section}}
\AtAppendix{\counterwithin{table}{section}}
\AtAppendix{\counterwithin{remark}{section}}

\theoremstyle{definition}
\newtheorem{theorem}{Theorem} 
 
\newtheorem{lemma}{Lemma} 
\newtheorem{proposition}{Proposition} 
\newtheorem{remark}{Remark} 
\newtheorem{corollary}{Corollary}

\usepackage{xspace}
\usepackage[nolist,smaller]{acronym}  
\PassOptionsToPackage{hyphens}{url}\usepackage[colorlinks=true, allcolors=blue]{hyperref}

\newcommand{\acro}[1]{\textsc{#1}\xspace }

\newcommand{\MHSE}{\acro{\smaller MHSE}}%Minimum Hyvarinen Score Estimate 

\newcommand{\KDE}{\acro{\smaller KDE}}
\newcommand{\Hscore}{$\mathcal{H}$-score}
\newcommand{\Hposterior}{$\mathcal{H}$-posterior}
\newcommand{\HBayes}{$\mathcal{H}$-Bayes}
\newcommand{\NLP}{\acro{\smaller NLP}}
\newcommand{\LP}{\acro{\smaller LP}}
\newcommand{\DPMM}{\acro{\smaller DPMM}}

\newcommand{\TGFB}{\acro{\smaller TGF-$\beta$}}
\newcommand{\DLD}{\acro{\smaller DLD}}
\newcommand{\SMIC}{\acro{\smaller SMIC}} %score matching information criteria

\newcommand{\hyvarinen}{Hyv\"arinen }

\definecolor{TukeyGreen}{RGB}{49, 163, 84}

\newcommand{\R}{\textit{R}}
\newcommand{\stan}{\textit{stan}}
\newcommand{\dirichletprocess}{\textit{dirichletprocess} }

\newcommand{\jack}[1]{{{\color{red}Jack: [#1]}}}

\newcommand{\new}[1]{{{\color{black}#1}}}

\def\*#1{\bm{#1}} %\def\*#1{#1}
\DeclareMathOperator*{\argmax}{arg\,max} % Jan Hlavacek
\DeclareMathOperator*{\argmin}{arg\,min}

\def\@fnsymbol#1{\ensuremath{\ifcase#1\or *\or \dagger\or \ddagger\or
   \mathsection\or \mathparagraph\or \|\or **\or \dagger\dagger
   \or \ddagger\ddagger \else\@ctrerr\fi}}
\newcommand{\ssymbol}[1]{^{\@fnsymbol{#1}}}

\let\OLDthebibliography\thebibliography
\renewcommand\thebibliography[1]{
  \OLDthebibliography{#1}
  \setlength{\parskip}{0pt}
  \setlength{\itemsep}{3pt plus 0.3ex}
}

\def\*#1{\bm{#1}} %\def\*#1{#1}
\usepackage[left=2.3cm, right=2.3cm, top=3cm]{geometry}
\usepackage{authblk}
\title{\bf  {General Bayesian Loss Function Selection and the use of Improper Models}}
\author[1,2]{Jack Jewson}
\author[1,2]{David Rossell}
\affil[1]{Department of Business and Economics, Universitat Pompeu Fabra, Barcelona, Spain}
\affil[2]{Data Science Center, Barcelona Graduate School of Economics, Spain}
\affil[ ]{\textit {\textcolor{blue}{jack.jewson@upf.edu, david.rossell@upf.edu}}}%, piotr.zwiernik@upf.edu}}}
\date{March 2022}

\begin{document}

%\doparttoc % Tell to minitoc to generate a toc for the parts
%\faketableofcontents % Run a fake tableofcontents command for the partocs

%\part{} % Start the document part
%\parttoc % Insert the document TOC

%\bibliographystyle{natbib}

\def\spacingset#1{\renewcommand{\baselinestretch}%
{#1}\small\normalsize} \spacingset{1}

\setcounter{Maxaffil}{0}
\renewcommand\Affilfont{\itshape\small}

\spacingset{1.42} % DON'T change the spacing!

\maketitle
\begin{abstract}
Statisticians often face the choice between using probability models or a paradigm defined by minimising a loss function.  Both approaches are useful and, if the loss can be re-cast into a proper probability model, there are many tools to decide which model or loss is more appropriate for the observed data, in the sense of explaining the data's nature. However, when the loss leads to an improper model, there are no principled ways to guide this choice. We address this task by combining the \hyvarinen score, which naturally targets infinitesimal relative probabilities, and general Bayesian updating, which provides a unifying framework for inference on losses and models. Specifically we propose the \Hscore,  a general Bayesian selection criterion  and prove that it consistently selects the (possibly improper) model closest to the data-generating truth in Fisher's divergence.  We also prove that an associated \Hposterior{} consistently learns optimal hyper-parameters featuring in loss functions, including a challenging tempering parameter in generalised Bayesian inference.  As salient examples, we consider robust regression and non-parametric density estimation where popular loss functions define improper models for the data and hence cannot be dealt with using standard model selection tools. These examples illustrate advantages in robustness-efficiency trade-offs and enable Bayesian inference for kernel density estimation, opening a new avenue for Bayesian non-parametrics. 

\end{abstract}

\noindent % 

\spacingset{1.45} % DON'T change the spacing!

\section{Introduction}

%\david{The intro denotes $x_{1:n}$ but in Section 2 we use $y_{1:n}$. Also, I suggest minimising notation and formulae in the intro. For example we can remove the marginal likelihood expression at no cost, same for Bayes factors, the set $M_1,\ldots,M_K$ of considered models.}

%\david{ I think it would be good to introduce the robust regression example early on, e.g. somewhere in the 2nd paragraph of the intro you could refer to Figure 1.} \jack{Just a reference to, or actually move the figure?}

%\david{Sorry for being a pain. Generally best not to use colors in figures unless indispensable (they cost money). In Fig 1 we can use black/grey and solid/dashed lines to distinguish the 4 cases.} \jack{Done}

A common task in Statistics is selecting which amongst a set of models is most appropriate for an observed data set $y$. 
Tools to address this problem include a variety of penalised likelihood, shrinkage prior, and Bayesian model selection methods. 
%A standard Bayesian solution is to construct a prior over the models and use the marginal likelihood
%\begin{equation}
%    \mathcal{M}_k(x_{1:n}) = \int f_k(x_{1:n};\theta_k)\pi_k(\theta_k)d\theta_k,\label{Equ:MarginalLikelihood}
%\end{equation}
%where $f_k(\cdot;\theta_k)$ and $\pi_k(\theta_k)$ are the likelihood function and parameter prior associated with $M_k$, 
%to produce posterior model probabilities. %When a uniform prior is taken over the space of models the ratio of the marginal likelihoods (a.k.a. the Bayes factor) provides a criteria for selecting models. 
%A considerable advantage of these model selection methods is that integrating over parameter values ensures consistent model selection amongst nested models.
Under suitable conditions, these approaches consistently select the model closest to the data-generating truth in Kullback-Leibler divergence (for example, see \cite{rossell:2021} and references therein for a recent discussion). 
%Alternatively, one may use shrinkage priors or a variety of likelihood penalties.
However, many data analysis methods are not defined in terms of probability models but as minimising a given loss function, 
%We provide an analog to the Bayesian \textit{model} selection procedure above that allows comparing \textit{loss functions}. 
%Data analysis methods that minimise loss functions are popular alternatives to statistical models, 
for example to gain robustness or flexibility. 
It is then no longer clear how to use the data to guide the choice of the most appropriate loss function, or its associated hyper-parameters. 
A key observation is that while the likelihood $f_k(y;\theta_k)$ of a model $k$ with parameters $\theta_k$ always defines a loss function $\ell_k(y;\theta_k) = -\log f_k(y;\theta_k)$ \citep{good1952rational}, the converse is not true. The exponential of an arbitrary loss $f_k(y;\theta_l) = \exp\{-\ell_k(y;\theta_k)\}$ may not integrate to a finite constant and therefore, defines an improper model on $y$. For example, this occurs in robust regression with Tukey's loss (Figure \ref{Fig:TukeysLoss}) and in kernel density estimation. 
\new{There is also a growing literature following \cite{basu1998robust} that applies the Tsallis score to define a robust loss function to fit any given probability model. Such robustification depends on a hyper-parameter that governs robustness-efficiency trade-offs and often leads to an improper model similar to Figure \ref{Fig:TukeysLoss} (see Figure \ref{Fig:TsallisTukeysLoss}, and \cite{yonekura2021adaptation} for a recent pre-print building on our improper model interpretation to address the hyper-parameter selection).}
%\jack{In fact since the first version of this paper appeared, \cite{yonekura2021adaptation} provided methodology analogous to that developed in this paper to select this hyperparameter.}
In these scenarios traditional model selection tools are not applicable to choose the more appropriate loss. Neither are methods to evaluate predictive performance such as cross-validation, since they require specifying a loss or criterion to evaluate performance in the first place, and do not attain consistent model selection even in simple settings \citep{shao1997asymptotic}.
%and in any case require the selection of the loss function with which to cross-validate. 
%\jack{Thinking about space, is his relevant if we can't simulate?}
%\david{I think so. When thinking about hard normalising constants ABC is one of the first things that comes to mind, so one might naively think that it also works for infinite norm constants.}
Methods to tackle intractable but finite normalisation constants, such as approximate Bayesian computation \citep{beaumont2002approximate,robert2016approximate}, also do not apply since they require simulating from a proper model.

We propose methodology
%Instead, we wish 
to evaluate how well each given loss $\ell_k$ captures the unknown data-generating distribution $g(y)$.
\new{The main idea is viewing $\ell_k$} as defining a (possibly improper) model $\exp \{ -\ell_k(y,\theta_k) \}$, and then measuring how well it approximates $g(y)$ via Fisher's divergence. 
As we shall see, Fisher's divergence and its related \hyvarinen score \citep{hyvarinen2005estimation} do not depend on normalising constants, \color{black} and in fact they allow for such constant to be infinite, \color{black} hence giving a strategy to compare improper models. 
Note that one could conceivably define the likelihood in ways other than $\exp \{ -\ell_k(y,\theta_k) \}$. However, defining losses as negative log-likelihoods provides the only smooth, local, proper scoring rule \citep{bernardo1979expected}, and is also the only transformation that leads to consistent parameter estimation for a certain general class of likelihoods \citep{bissiri:2012}. Further, it seems reasonable that the loss should be additive over independent pieces of information, and that the likelihood of an improper model should factorize under such independence, and for both properties to hold one must take the exponent of the negative loss.  
%We take advantage of the homogeneity property of Fisher's divergence and the computational convenience of the \hyvarinen score \citep{hyvarinen2005estimation} to compare such models, regardless of their normalising constant. 
Our framework consistently selects the best model  in Fisher's divergence, and in particular the (proper) data-generating model $g(y)$ if it is under consideration. We also show how, after a model is chosen, one can learn important hyper-parameters such as the likelihood tempering in generalised Bayes and PAC-Bayes, robustness-efficiency trade-offs in regression and the level of smoothing in kernel density estimation. 
%We prove consistency to the data generating \textit{models} when it is under consideration and demonstrate the excellent performance of our method for optimising the robustness-efficiency tradeoff associated with robust regression and for non-parametric density estimation. 
\new{For clarity and space we focus on continuous real-valued $y=(y_1,\ldots,y_n)$ with full support on $\mathbb{R}^n$. Our ideas can be extended to other settings such as discrete or positive $y$, e.g. following \cite{hyvarinen2007some}. However, these are slightly less interesting in our context. Improper models cannot occur for discrete $y$ with finite support, and one may log-transform a positive outcome and subsequently apply our methodology.}

%However, increasingly modern inference is being conducted with loss functions instead of models.
The use of probability models versus algorithms is one of the most fundamental, long-standing debates in Statistics.  
%The two main arguments against models are their lack of flexibility and their lack of robustness.  
In an influential piece, \cite{breiman2001statistical} argued that models are not realistic enough to represent reality in a useful manner, nor flexible enough to predict accurately complex real-world phenomena. 
Despite advances in flexible and non-parametric models, this view remains
%This issue arguably gained importance 
in the current era where predictive machine learning proliferates, and shows ample potential to tackle large and/or complex data. % (e.g.   areas  such  as  machine  translation,  facial  recognition  and  demand  forecasting).
%Another  classical argument from the robust Statistics literature suggests to abandon a model, along with any probabilistic interpretation and mathematical optimally it may possess, to ensure that results are not overly affected by unusual or possibly contaminated data.
%
%The rise of such methods poses the following problem for applied statisticians: probability models versus algorithms for data analysis. 
Their  limitations  notwithstanding,  probability  models remain  a  fundamental  tool  for  research. Paraphrasing \cite{efron2020prediction}:  ``Abandoning mathematical models comes close to abandoning the historic scientific goal of understanding nature.''  
We agree with this view that there are many situations where models facilitate describing the phenomenon under study. We seek to bridge these two views by noting that loss functions define improper models, which also lead to natural interpretations in terms of relative probabilities, and proposing a strategy to learn which loss gives a better description of the process underlying the data.

Our strategy is to view $f_k$ as expressing relative (as opposed to absolute) probabilities, for example $f_k(y_0, \theta_k)/f_k(y_1, \theta_k)$ describes how much more likely is it to observe $y$ near $y_0$ than near $y_1$. A convenient manner to describe such ratios is by comparing the gradient of $\log f_k(y, \theta_k)$ to the gradient of the log data-generating density $\log g(y)$. This can be achieved by minimising Fisher's divergence 
%(see e.g. \cite{otto2000generalization}, Equation (8); \cite{johnson2004information}, Definition 1.13; \cite{dasgupta2008asymptotic}, Definition 2.5) \jack{norm or shall we just go for squared?}
\begin{equation}
    D_F(g||f_k) := \frac{1}{2}\int ||\nabla_y \log g(y) - \nabla_y \log f_k(y;\theta_k)||_2^2 g(y) dy,\label{Equ:FishersDivergence}
\end{equation} 
where $\nabla_y$ is the gradient operator. 
%The convenience of this approach comes from the observation that 
Under certain minimal tail conditions, minimising Fisher's divergence is equivalent to minimising the \hyvarinen score \citep{hyvarinen2005estimation}. The latter has been used for models with intractable, but finite normalising constants \citep{hyvarinen2005estimation} and more recently to 
\new{define posterior distributions based on scoring rules \citep{giummole2019objective} and to}
conduct Bayesian model selection using improper priors \citep{dawid2015bayesian, shao2019bayesian}.  We consider for the first time its use to select between possibly improper models, and learn their associated hyper-parameters.
%As well as selecting between a discrete set of distinct loss functions and models we 
%Further, we consider the estimation of hyperparameters indexing loss functions such as the `cut-off' point for Tukey's loss and a kernel's bandwidth parameter. 
%\david{I deleted a mention to Masuda et al, as we repeated the exact same sentence in the last paragraph of Sec 2.}
%We remark that \cite{matsuda2019information} proposed a framework to choose between models with intractable, but finite, normalising constants. Although not considered by the authors, their framework could potentially be used for improper models. These authors focused on predictive criteria which we show do not lead to consistent model selection, whereas we focus on structural learning where one seeks guarantees on recovering the loss that best approximates the data-generating $g(y)$. %We compare against this method in our experiments.

\new{A possible alternative Fisher's divergence proposed by  \cite{lyu2009interpretation} is to use linear operators to define a generalized Fisher divergence. The operators do not require a finite normalization constant, i.e. they can in principle be applied to improper models.
Although interesting, the specific proposals in \cite{lyu2009interpretation} are a conditional mean operator for latent variable models and a marginalization operator that requires proper conditionals, neither of which seems directly applicable to our setting.
In fact, it is important to distinguish our framework from settings where, by combining proper conditional models, one defines an improper joint model.
For example, intrinsic auto-regressive models in spatial Statistics have proper conditionals and an improper joint. Such models can be fit using a pseudo-likelihood \citep{besag1975statistical} or the marginalization operator of \cite{lyu2009interpretation}, for instance.
In our framework, neither the conditionals nor the joint of $f_k(y;\theta_k)$ need be proper, e.g. Tukey's loss example in Figure \ref{Fig:TukeysLoss}.
%\cite{besag1995conditional}

%An appealing particular case of \cite{lyu2009interpretation}, when $y$ is multivariate, is to consider a marginalization operator based on the full conditionals of $f_k$ (a similar idea leads to the pseudo-likelihood framework of \citep{besag1975statistical}).
%Provided one defines a notion of conditional densities associated to an improper $f_k$, the marginalizaton operator provides an alternative to \eqref{Equ:FishersDivergence} in multivariate settings. 
%We did not consider this here however, as our focus is on univariate outcomes.
%, and a conditional expectation operator for latent variable models. 
%Although such extensions can be interesting, we focus on \eqref{Equ:FishersDivergence} due to its giving a solution that is easily interpretable in terms of relative probabilities, a concept which seems natural for improper models. 

It is also important to distinguish our framework with approaches designed for models with intractable, but finite, normalization constants (i.e. proper models).
Popular strategies include contrastive divergence \citep{hinton2002training},
minimum velocity learning (unpublished work by Movellan in
2007, see \cite{wang2020wasserstein}) and
contrastive noise estimation \citep{gutmann2010noise}.
These methods define certain Monte Carlo dynamics to transition from observed samples from the data-generating $g(y)$ into samples from $f_k(y; \theta_k)$.
Informally, if $g(y)$ is close to $f_k(y; \theta_k)$ then such dynamics have a small gradient, defining a divergence between these two distributions. These methods do not require evaluating the normalization constant. However, the notion of sampling from $f_k(y; \theta_k)$ requires it to be a proper probability model, and hence these divergences do not apply to improper models. 
Another interesting example for intractable normalization constant are the kernel Stein discrepancy posteriors of  \cite{matsubara2021robust}.
However, Stein discrepancies are based on defining expectations, and hence also require $f_k$ to be proper. A further issue of kernel discrepancies is that they  do not lead to coherent updating of beliefs, i.e. the posterior $\pi(\theta_k \mid y_1,y_2)$ obtained after observing $(y_1,y_2)$ does not match the posterior based on observing $y_2$ and using $\pi(\theta_k \mid y_1)$ as the prior.
} 

The paper proceeds as follows. Section \ref{Sec:Foundations} reviews recent developments in Bayesian updating with loss functions, discusses our motivating examples and some failures of standard methodology. 
Section \ref{Sec:Methodology} explains how we interpret the inference provided by an improper model in terms of relative probabilities, and their relation to Fisher's divergence and the \Hscore{}.
It also outlines our methodology: the definition of an \Hposterior{}, a $\sqrt{n}$-consistency result to learn parameters and hyper-parameters, and the definition of the integrated \Hscore{} and \HBayes{} factors as a criterion to choose among possibly improper models.
Section \ref{Sec:HscoreModelConsistency} gives consistency rates for \HBayes{} factors, including important non-standard cases where optimal hyper-parameters lie at the boundary, as can happen when considering nested models.
Section \ref{Sec:RobustRegression} applies our procedure to robust regression. Section \ref{Sec:KDE} produces a Bayesian implementation of kernel density estimation, which cannot be tackled by standard Bayesian methods, since kernel densities define an improper model for the observed data. All proofs and some additional technical results are in the supplementary material. 
Code to reproduce the examples of Sections \ref{Sec:RobustRegression} and \ref{Sec:KDE} can be found at %\href{https://github.com/jejewson/HyvarinenImproperModels}{https://github.com/jejewson/HyvarinenImproperModels}.
\url{https://github.com/jejewson/HyvarinenImproperModels}.
%\david{Indicate where you posted the code. JRSS-B requires GitHub or similar. Also, in Sec A.2.1 indicate where the code is located at, additionally to saying it's supplementary material. Also, in Sec A.4.1 please indicate where the code is located (current XXXXXX).}

\section{Problem formulation}{\label{Sec:Foundations}}
%\section{Foundations}{\label{Sec:Foundations}}

%\david{Section 2 felt way too long, it takes a long time before getting to the "meat" of our paper. I made some cuts but more may be needed, we should focus on the essentials.}

%\david{Having both $\*y$ and $y_{1:n}$ is confusing. Actually, can we just use $y=(y_1,\ldots,y_n)$ throughout, this would be consistent with its use in the intro\jack{done - now $y_i \sim g$ my only concern is do we need a way to distinguish between one observation ad a vector of}. Also, can we only use $g$ or $G$? (no problem if easier to use both, but if not needed let's simplify notation) \jack{should now be $g$ throughout}}
%\jack{$y$ is always a vector (even in the definition of Fisher's divergence), either $y_i$ for one observation or use something like $z$ in the conditions of the throorem}
%\david{Here we use $l_k$, before it was $\ell_k$. This also happens in other sections.}\jack{I think i prefer $\ell_k$}

We define the problem and notation and then provide the necessary foundations by reviewing general Bayesian updating, providing two motivating examples, identifying the shortcomings of standard approaches, and finally introducing Fisher's divergence and the \hyvarinen score.

%\subsection{Notation}

%\jack{In the appendix notation we talk about having $\omega_l \in \mathbb{R}$ make sure I am consistent, or is it just part of $\kappa$? where we discuss the weight we could say that this is part of hyperparameter $\kappa$ - I mention this in 3.2}

Denote by $y = (y_1,…,y_n)$ an observed continuous outcome, where $y_i \in \mathbb{R}$ are independent draws from an unknown data-generating distribution with density $g(\cdot)$. One is given a set of $M$ probability models and $L$ loss functions which, in general, may or may not include $g$. As usual each model $k = 1, \ldots, M$ is associated to a density $f_k(y; \theta_k, \kappa_k)$, where $\theta_k$ are parameters of interest and $\kappa_k$ are hyperparameters, 
%$(\theta_k, \kappa_k) \in \Theta_k \times \Phi_k \subseteq \mathbb{R}^{p_k} \times \mathbb{R}^{s_k}$, 
%with $\Theta_k \subseteq \mathbb{R}^{p_k}$ and $\Phi_k \subseteq \mathbb{R}^{s_k}$  
%resulting in $d_k = p_k + s_k$ as the total number of parameters plus hyper-parameters in $k$. 
$(\theta_k, \kappa_k) \in \Theta_k \times \Phi_k \subseteq \mathbb{R}^{d_k}$.  
Any such density defines a loss $\ell_k(y; \theta_k, \kappa_k) = -\log f_k(y; \theta_k, \kappa_k)$. Similarly, denote by $\ell_k(y; \theta_k, \kappa_k)$ for $k = M+1, \ldots, M+L$ the given loss functions. For $k>M$, we refer to $f_k(y; \theta_k, \kappa_k) = \exp\{ - \ell_k(y; \theta_k, \kappa_k)\}$ as the (possibly improper) density associated to $\ell_k$. In general such $f_k$ need not integrate to a finite number with respect to $y$, i.e. $f_k$ may define an improper model on $y$. 
%If it does, then $f_k$ defines a proper probability model (after suitable normalisation), otherwise it simply defines a positive measure. 
Our goal is to choose which among $f_1,\ldots, f_{L+M}$ provides the best representation of $g$, in a sense made precise below.%, or certain aspects associated to $G$ such as measures of central tendency.

\subsection{General Bayesian updating}

In the frequentist paradigm it is natural to infer parameters by minimising loss functions, a classical example being $M$-estimation \citep{huber1981robust}.
Loss functions are also used in the PAC-Bayes paradigm, where one considers the posterior distribution on the parameters
%\david{I adjusted the eq below, I think it wasn't quite right}
\begin{equation}
    \pi_k(\theta_k|y, \kappa_k) \propto \pi_k(\theta_k \mid \kappa_k)\exp\{-\ell_k(y;\theta_k,\kappa_k)\}
    %Jack's version
    %\pi_{\kappa_k}(\theta_k|\*y) \propto \pi_k(\theta_k)\exp\{-\ell_k(\*y;\theta_k,\kappa_k)\}
    \label{Equ:GeneralBayesRule}
\end{equation}
where $\pi_k(\theta_k \mid \kappa_k)$ 
%where $\pi(\theta_k)$ 
is a given prior distribution and $\propto$ denotes ``proportional to''.
%$\mathcal{L}(y_{1:n}; \kappa_k) = \int \pi(\theta_k)\exp(-\ell(y_{1:n};\theta_k,\kappa_k))d\theta_k$ is analogous to the marginal likelihood. 
See \cite{guedj2019primer} for a review on PAC-Bayes, and \cite{grunwald2012safe} for the safe-Bayes paradigm, which can be seen as a particular case where $\ell_k$ is a tempered negative log-likelihood. %\david{Please take a look at Guedj's review, in case there are other earlier references, or more relevant, or more general (e.g. high-dimensional results).} - \jack{This is mentoned later}
At this stage we consider inference for $\theta_k$ for a given hyper-parameter $\kappa_k$, we discuss learning $\kappa_k$ later. 
\new{See also \citep{giummole2019objective} for a framework where the loss function is defined by scoring rules such as the Tsallis score and the \hyvarinen score. The latter gives rise to the \Hposterior{} discussed in Section \ref{ssec:hscore}, a critical component of our construction.}
As a key result supporting the interpretation of \eqref{Equ:GeneralBayesRule} as conditional probabilities akin to Bayes rule, \cite{bissiri2016general} showed that \eqref{Equ:GeneralBayesRule} leads to a coherent updating of beliefs, and referred to the framework as general Bayesian updating.

These results allow Bayesian inference on parameters $\theta_k$ based on loss functions.
The properties of the general Bayesian posterior have been well-studied, for example under suitable regularity conditions \cite{chernozhukov2003mcmc} and \cite{lyddon2019general} showed that it is asymptotically normal.
However, the emphasis of prior work is on inference for $\theta_k$. To our knowledge viewing $\exp \{ - \ell_k(y; \theta_k, \kappa_k) \}$ as an improper density has not been considered, which is critical for interpretation and posterior predictive inference, nor has the problem of choosing which loss best represents the data.
%However, what is currently missing is an analogue to model selection procedures to choose the loss that best represents the data.

%\jack{Should this be here or where we discuss the Fisher's divergence and Hyvarinen score}
%We remark that \cite{matsuda2019information} proposed a framework to choose between models with intractable, but finite, normalizing constants. \jack{... move this down, maybe split here} Although not considered by the authors, their framework could potentially be used for unnormalisable models. The main difference is that these authors used cross-validations and arguments similar to the AIC to propose predictive criteria, which in general do not lead to consistent model selection, whereas we focus on structural learning where one seeks guarantees on recovering the loss that best approximates the data-generating $G(y)$. We compare against this method in our experiments. 

%A particular consequence of such a principled update is it allows us to consider Bayesian inference for algorithmic problems. That is algorithms that are defined by minimising a loss function, for example between predictions and observations, without defining a model. However, what is currently missing is an analogue to model selection procedures to allow for the comparisons of how different loss functions perform the inference tasks. Of particular interest to us is the comparison between the log-likelihood of a parametric model and the extension a loss function with unnormaliseable pseudo-density.

\subsection{Motivating applications}{\label{Sec:MotivatingApplication}}

We introduce two problems which, despite being classical, cannot be tackled with standard inference.
We first consider robust regression where one contemplates a parametric model and a robust loss, and wishes to assess which represents the data best. To our knowledge there are no solutions for this problem.
We next consider learning the bandwidth in kernel density estimation, where the goal is predictive inference on future data. While there are many frequentist solutions, Bayesian methods are hampered by the associated loss defining an improper model for the observed data.

\subsubsection{Robust regression with Tukey's loss}{\label{Sec:TukeysLossIntro}}

%\jack{ Well can I optimse Tukeys loss for $\sigma^2$ I guess given the$ 1/sqrt{2\pi\sigma^2}$ maybe I can}

%\david{It's non-standard to denote by uppercase $X_i$ a vector. I changed all to $x_i$.} \jack{check later}

Consider the linear regression of $y_i$ on an 
%$n \times (p_\gamma-1)$ predictor matrix $X$
$p$-dimensional vector $x_i$,
\begin{equation}
    y_i = x_i^T \beta + \epsilon_i, \textrm{ with } E(\epsilon_i) = 0,\textrm{} V(\epsilon_i) = \sigma^2 \textrm{ for } i = 1,\ldots, n. \nonumber
\end{equation}
%
%where $\epsilon=(\epsilon_1,\ldots,\epsilon_n)$, 
%
Consider first a Gaussian model, denoted $k=1$, so that
%Traditionally, Bayesian's associate such modelling assumptions with the Gaussian distribution (call this model/loss $k = 1$), 
$f_1(y_i ; \theta_1) = \mathcal{N}(y_i;x_i^T \beta, \sigma^2)$, where $\theta_1= \{\beta,\sigma^2\}\in \mathbb{R}^p \times \mathbb{R}_{+}$, and there are no hyper-parameters ($\kappa_1 = \emptyset$).
The negative log-likelihood gives
%The maximum likelihood estimate (\MLE) of $\beta$ from this model is equivalent to minimising 
the least-squares loss
$\ell_1(y;\theta_1)= \sum_{i=1}^n \ell_1(y_i;\theta_1)$, where
\begin{align}
    \ell_1(y_i;\theta_1) &= - \log f_1(y_i ; \beta, \sigma^2)
    = \frac{1}{2}\log\left(2\pi\sigma^2\right) + \frac{(y_i - x_i^T \beta)^2}{2\sigma^2}.
    \label{Equ:ols_loss}
\end{align}
%where the first term ensures the value of $\sigma^2$ obtained by minimising $\ell_1(y_i;\theta_1)$ agrees with the \MLE under the Gaussian model. 
%
Since the least-squares loss is non-robust to outliers, %\citep{tukey1960survey, huber1981robust,rousseeuw2005robust,hampel2011robust}, \color{red} [DR. Do we need so many references?] \color{black} 
one may consider alternatives.
A classical choice is Tukey's loss \citep{beaton1974fitting}, which we denote $k = 2$, given by
\begin{align}
    \ell_2(y_i;\theta_2, \kappa_2) = \begin{cases}
%\frac{1}{2}\log\left(2\pi\sigma^2\right)  + \frac{\kappa_2^2}{6}\left[1-\left(1- \frac{(y_i -  x_i^T \beta)^2}{\sigma^2\kappa_2^2} \right)^3\right],&\textrm{ if }|y_i - X_i\beta|\leq \kappa_2\sigma\\
        \frac{1}{2}\log\left(2\pi\sigma^2\right)  + \frac{(y_i -  x_i^T \beta)^2}{2\sigma^2} - \frac{(y_i -  x_i^T \beta)^4}{2\sigma^4\kappa_2^2} + \frac{(y_i -  x_i^T \beta)^6}{6\sigma^6\kappa_2^4},&\textrm{ if }|y_i - x_i^T\beta|\leq \kappa_2\sigma\\
        \frac{1}{2}\log\left(2\pi\sigma^2\right) + \frac{\kappa_2^2}{6}, &\textrm{ otherwise }
\end{cases}\label{Equ:TukeysLoss}
\end{align}
where $\theta_2 = \{\beta, \sigma^2\}\in \mathbb{R}^p \times \mathbb{R}_{+}$ and $\kappa_2 \in \mathbb{R}_{+}$ is a cut-off hyper-parameter. Note that \eqref{Equ:TukeysLoss} is parametrised such that the units of the cut-off parameter $\kappa_2$ are standard deviations away from the mean.
%\footnote{Here we consider $\sigma^2$ to be a parameter of interest for Tukey's loss rather than a hyperparameter to maintain the notion of nestedness with the Gaussian model, however we note it may be beneficial to consider it otherwise for the two stage procedure introduced in Section \ref{ssec:TwoStep}}.

The density $f_2(y;\theta_2, \kappa_2) = \exp \{- \sum_{i=1}^n \ell_2(y_i; \theta_2, \kappa_2) \}$ integrates to infinity, defining an improper model.
Figure \ref{Fig:TukeysLoss} plots Tukey's loss for several $\kappa_2$ and their corresponding densities. 
 \eqref{Equ:TukeysLoss} is similar to  \eqref{Equ:ols_loss} when $|y_i - x_i^T \beta|$ is close to 0, while for large $|y_i - x_i^T\beta|$ it becomes flat, bounding the influence of outliers. 
%not being integrable. That is to say $\int_{-\infty}^{\infty} exp(- \ell_2(y; \theta_2, \kappa_2) dy = \infty$.
The Gaussian model is recovered when $\kappa_2 = \infty$. 
As we shall see, such nested comparisons pose methodological challenges that motivated our developments.
As a technical remark, in robust statistics $\sigma^2$ is typically estimated separately from $\beta$, either as part of a two-stage procedure \citep[see e.g.][]{chang2018robust}  %\citep[see e.g.][]{zamar1989robust, croux2014robust, chang2018robust} 
or using S-estimation \citep{rousseeuw1984robust}%, rousseeuw2005robust}
. Instead, our framework allows one to jointly estimate $(\beta,\sigma^2)$.

%\jack{Necessary??}
%Tukey's loss is just one possibility, we make no claim that it should be preferred over other robust losses, but it is a popular choice that is useful for our illustrations. One can easily add more losses into our framework, as it is not limited to two losses.
%See \cite{basu1998robust, black1996unification} for alternative losses, and \cite{wang2020tuning} for a recent high-dimensional proposal, for example.
%We further note that our framework is not limited to robust linear regression. One may replace $x_i^T \theta$ in $l_1$ and $l_2$ by a complex non-linear function, for example from a deep learning or Gaussian process regression.%, or replace the whole loss function for choices that are more suitable for other data types (e.g. classification loss for binary y, divergences between functions in functional regression, etc), for example Tukeys loss has also been shown to be useful for Robust PCA \citep{salibian2006principal}.
We note that one can add more losses into our framework, for example those in \cite{black1996unification, basu1998robust} or \cite{wang2020tuning}. Also, our framework is not limited to linear regression. One may replace $x_i^T \beta$ by a non-linear function, for example from a deep learning or Gaussian process regression. 

%\david{Can we use less references in the paragraph above? And more generally we wanna have the necessary references but there are a number of places where they disrupt the flow a bit. Please keep in mind when updating the draft. :-)}\jack{done here}

%\jack{Appendix??}
%Note, that we have embedded the standard form of Tukey's loss within a Gaussian likelihood and 
%As a technical remark, in robust statistics the scale parameter $\sigma^2$ is typically estimated separately from $\beta$.
%The main strategies either use two-step methods where one first estimates the scale robustly and then conducts M-estimation on the standardised residuals \citep[see e.g.][]{zamar1989robust, croux2014robust, chang2018robust} or to use S-estimation \citep{rousseeuw1984robust, rousseeuw2005robust}, which matches the observed loss to some fixed constant. %, often the expected loss under some assumed \DGP% - \jack{again this becomes an arbitrary plug-in also}, rather than directly minimising the loss.
%Instead, our framework allows jointly estimating $(\beta,\sigma^2)$, see below.
%We also consider estimating the parameter $\sigma^2$ via loss minimisation simultaneously to the regression coefficients, which is made possible by our framework. 

%selecting_loss_jointHyvarinen_properExperiments.Rmd
%selecting_loss_jointHyvarinen_LaplaceApprox_convergence_n_Tukeys
\begin{figure}
\begin{center}
\includegraphics[trim= {0.0cm 0.0cm 0.0cm 0.0cm}, clip,  
width=0.49\columnwidth]{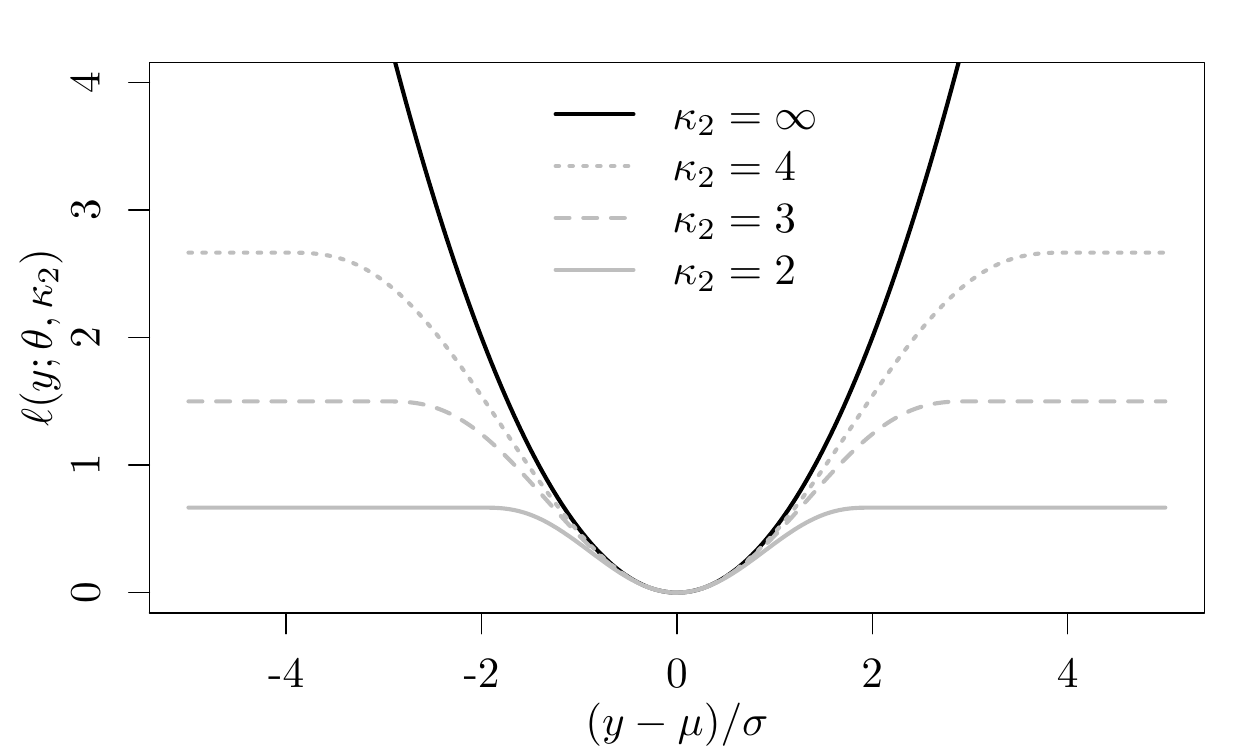}
\includegraphics[trim= {0.0cm 0.0cm 0.0cm 0.0cm}, clip,  
width=0.49\columnwidth]{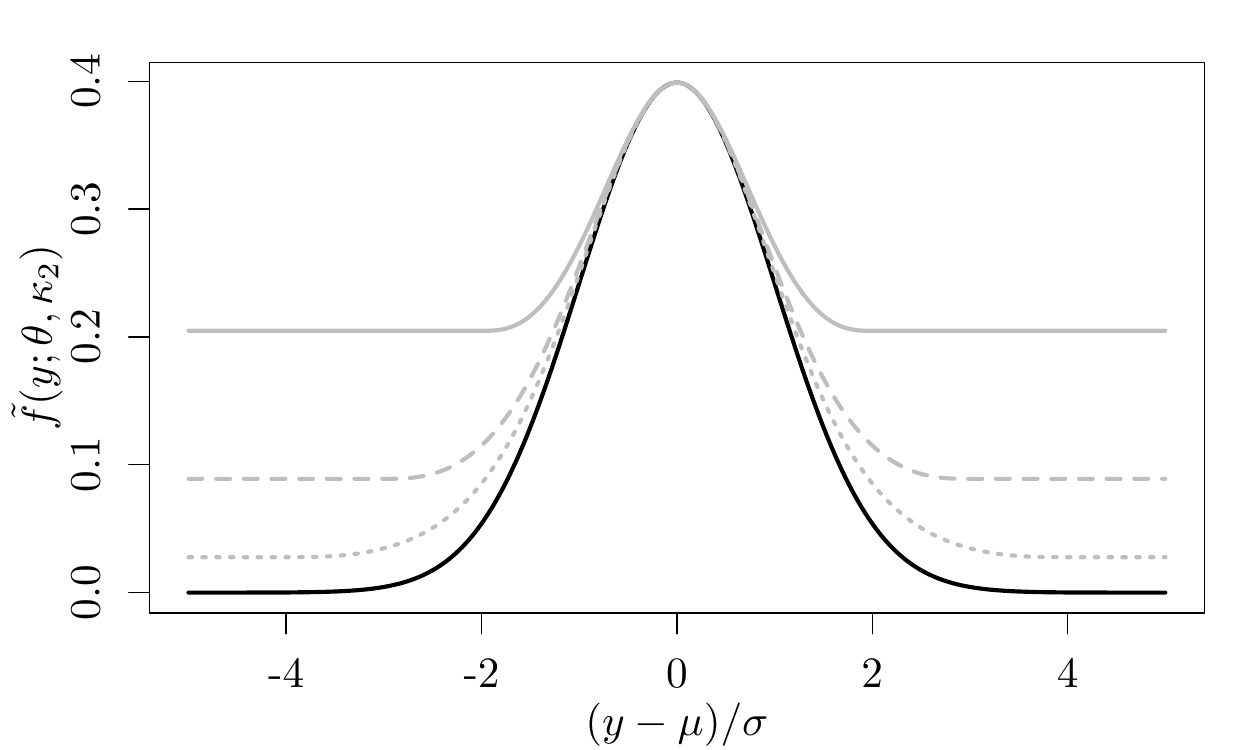}
%trim={<left> <lower> <right> <upper>}
%\caption{Squared-error loss ($\kappa_2 = \infty$) and Tukey's loss ( \eqref{Equ:TukeysLoss}) for $\kappa_2 = 2$, $3$ and $4$ (\textbf{left}) and corresponding (pseudo-)probability densities (\textbf{right}). The pseudo-densities for Tukey's loss were scaled to match the mode of the Gaussian density}
\caption{Squared-error loss and Tukey's loss (\textbf{left}) and corresponding (improper) densities (\textbf{right}). The improper densities for Tukey's loss are scaled to match the mode of the Gaussian density.}
\label{Fig:TukeysLoss}
\end{center}
\end{figure}

Setting $\kappa_2$ in Tukey's loss is related to the so-called robustness-efficiency trade-off, an issue that has been unsettled for at least 60 years; see \cite{box1953non} and \cite{tukey1960survey}. While $\kappa_2=\infty$ gives the most efficient parameter estimates if the data are Gaussian, they are least robust outliers. %to deviations from Gaussianity.
Decreasing $\kappa_2$ increases robustness, but can significantly reduce estimation efficiency if $\kappa_2$ is set too small. % This issue is particularly important when $n$ is not large, relative to $p$.
%
%To conclude this section we discuss an important issue in robust regression, that of addressing the robustness-efficiency trade-off, an at least 60-year old unsettled issue ; see \cite{box1953non} and \cite{tukey1960survey}. 
%For Tukey's loss, if the data are Gaussian (or near so), then setting $\kappa_2=\infty$ gives the most efficient (but less robust) estimator.
%Decreasing $\kappa_2$ increases robustness but can significantly affect efficiency, if one sets $\kappa_2$ to too small a value. This issue is particularly important when $n$ is not large, relative to $p$.
%this trade-off is represented by the cut-off value $\kappa_2$. As $\kappa_2\rightarrow\infty$, $\ell_2(y_i;\theta_2, \kappa_2) \rightarrow \ell_1(y_i;\theta_1)$ whose minimisation is known to be the most efficient and least robust method for estimating $\beta$. 
%Decreasing $\kappa_2$ will buy robustness at the cost of efficiency. Setting $\kappa_2$ too low can result in the missing of important features in the data set, particularly if the data is sparse, e.g. high-dim predictor space. 
%
%
We address this issue by learning from the data whether or not to estimate parameters in a robust fashion, i.e. selecting between $\ell_1$ and $\ell_2$, %and if the decision is to be robust, exactly how robust to be (i.e. estimating $\kappa_2$). 
and estimating $\kappa_2$ when $\ell_2$ is chosen.
\color{black} We remark that the estimated $\kappa_2$ does not attempt to provide an optimal robustness-efficiency trade-off in the sense of minimising estimation error under an arbitrary data-generating $g(y)$. Rather, it seeks to define a predictive distribution on $y$ that portrays its nature accurately, in terms of infinitesimal relative probabilities (Section \ref{Sec:Methodology}).
This said, our method attains a good robustness-efficiency trade-off in our examples. Also, per our model selection consistency results, if data are truly Gaussian (or nearly so) then our framework collapses to inference under the Gaussian model, leading to efficient estimates.
\color{black}

Despite the importance of the robustness-efficiency trade-off, we are aware of limited work setting $\kappa_2$ in a principled, data-driven manner. 
%methods to optimise the robustness-efficiency trade-off for a particular data set are primitive. 
Rule-of-thumb methods are popular, e.g. setting ``$\kappa_2 = 4.6851$, gives approximately 95\% asymptotic efficiency of L2 minimization on the standard normal distribution of residuals'' \citep{belagiannis2015robust}%\citep{zamar1989robust, belagiannis2015robust, chang2018robust} \david{give only 1 ref, whoever we're quoting}
, setting $\kappa_2 = 1.547$ to obtain a breakdown of 1/2 \citep{rousseeuw1984robust}%,rousseeuw2005robust}
, or a balance of breakdown and efficiency \cite[e.g.][]{riani2014consistency}. 
More formal approaches rely on 
estimating quantiles of the data (e.g. \cite{sinova2016tukey}),
%who apply the method of \cite{kim2012robust} to Tukey's loss), 
minimising an estimate of parameter mean squared error (see \cite{li2021robust}, \new{who applied the method of \cite{warwick2005choosing} %\cite{warwick2005choosing, basak2021optimal} 
to Tukey's loss),} or minimising the maximum change in parameter estimates resulting from perturbing one observation \citep{li2021robust}. %who apply the method of \cite{kang2014minimum} to Tukey's loss). 

%\cite{fujisawa2006robust} - leave one out betaD

%Probably need the Student's t in here, as we can't do the laplace.

\subsubsection{Non-parametric Kernel Density Estimation}{\label{Sec:KDEintro}}

%\david{Do we need $K_h$ and $K$? Can we just use the RHS in \eqref{Equ:kde}?}

Suppose that $y_i \sim g$ independently for $i=1,\ldots,n$ and one wishes to estimate $g$. The kernel density estimate at a given value $x$ is given by
\begin{equation}
    \hat{g}_h(x) = \frac{1}{nh}\sum_{i=1}^nK\left(\frac{x - y_i}{h}\right)
    %\hat{g}_h(x) = \frac{1}{n}\sum_{i=1}^n K_h(x - y_i) = \frac{1}{nh}\sum_{i=1}^nK\left(\frac{x - y_i}{h}\right)
%    \tilde{f}_h(x) = \frac{1}{n}\sum_{i=1}^n K_h(x - y_i) = \frac{1}{nh}\sum_{i=1}^nK\left(\frac{x - y_i}{h}\right)
\label{Equ:kde}
\end{equation}
where the kernel $K(\cdot)$ is a symmetric, finite variance probability density, 
%such that $K(x) \geq 0$,  $\int K(x) dx = 1$, $\int x K(x) dx = 0$ and $\int x^2K(x) dx < \infty$. 
%$K_h(x) = \frac{1}{h}K\left(\frac{x}{h}\right)$ 
and $h$ is the bandwidth parameter. For simplicity we focus on the Gaussian kernel $K(x) = N(x;0,1)$.
%\frac{1}{\sqrt{2\pi}h}\exp\left(- x^2 / 2h^2\right)$. 
%$K_h(x) = \frac{1}{\sqrt{2\pi}h}\exp\left(- x^2 / 2h^2\right)$. 

The bandwidth $h$ is an important parameter controlling the smoothness and accuracy of $\hat{g}_h$.  
%large values indicate a high degree of smoothing while small values will produce \KDE's that are more spiked around each observation. Setting this parameter is imperative to produce \KDE's that well capture the underlying \DGP. 
Popular strategies to set the bandwidth are rule-of-thumb and plug-in methods \cite[e.g.][]{silverman1986density}, %(e.g. \cite{silverman1986density, scott2015multivariate}) %proposed rule-of-thumbs based on the standard deviation and inter-quartile range of the data),
cross-validation \citep{habbema1974stepwise, robert1976choice} and minimising an estimate of integrated square error
\citep{rudemo1982empirical,bowman1984alternative}.%,scott1987biased}.

%\david{Use less references above. Also, sort chronologically when $>1$ are given here and throughout.}\jack{done}

Unfortunately, standard Bayesian inference cannot be used to learn $h$ from data. The reason is that although
\eqref{Equ:kde} defines a proper probability distribution for a future observation $x$, 
a Bayesian framework requires a proper model for the observed data given the parameter.
In our notation, the model likelihood
$f(y; \theta, \kappa) \propto \prod_{i=1}^n \hat{g}_h(y_i)$ with $\theta = \emptyset$ and hyper-parameter $\kappa = \{h\}$
%$f(y_{1:n}; \theta) \propto \exp \{ \sum_{i=1}^n\log \hat{g}_h(x_i) \}$, 
has an infinite normalising constant.
%the terms contribution go the in-sample log-(pseudo-)likelihood 
%$\mathcal{L}(x_{1:n}, h) = \sum_{i=1}^n\log \tilde{f}_h(x_i)$ 
To see this, note that
    \begin{align}
        \hat{g}_h(y_i) =& \frac{1}{nh}\sum_{j=1}^nK\left(\frac{y_i - y_j}{h}\right)
        =\frac{1}{n h\sqrt{2\pi}} + \frac{1}{n h\sqrt{2\pi}}\sum_{j\neq i}^n\exp\left\{-\frac{(y_i - y_j)^2}{2h^2}\right\},\label{Equ:KDEPseudoDensity}
    \end{align}
%    \begin{align}
%        \tilde{f}_h(x_i) =& \frac{1}{nh}\sum_{j=1}^nK\left(\frac{x_i - x_j}{h}\right)\nonumber\\
%        =&\frac{1}{n\sqrt{2\pi\sigma^2}} + \frac{1}{n\sqrt{2\pi\sigma^2}}\sum_{j\neq i}^n\exp\left(-\frac{(x_i - x_j)^2}{2\sigma^2}\right).\label{Equ:KDEPseudoDensity}
%    \end{align}
where the first term integrates to infinity with respect to $y_i$.
%The constant term $\frac{1}{n\sqrt{2\pi\sigma^2}}$ resulting from the kernel evaluated between the observation and itself added to the loss function causes $\int \tilde{f}_h(x_i) dx_i = \infty$. 
%Further to this, the KDE clearly offers no interpretation as a generating likelihood for the observed data set owing to their dependence upon each other. 
%In this regard, \eqref{Equ:kde} illustrates a situation where one has a loss function $\log \hat{g}_h(x)$ for future observations, but defines an improper probability model for the observed data.
Hence, \eqref{Equ:kde} illustrates a situation where one has an `algorithm' for producing a density estimate for future observations that defines an improper probability model for the observed data.

\subsection{The failure of standard technology}{\label{Sec:FailureStandardMethods}}

%\david{In the next paragraph we consider $h$ as a hyper-parameter, but when discussing KDE earlier we said it was a parameter of interest. To avoid confusions, can we say it's a hyper-parameter throughout?} \jack{done}

As we discussed a main challenge is that standard tools are not, in general, applicable to compare improper models. 
Another challenge occurs when one wishes to estimate a hyperparameter $\kappa_k$ of a given improper model $k$, e.g. Tukey's cut-off or the kernel bandwidth. 
For example, the general Bayesian might consider mimicking standard Bayes or marginal likelihood estimation by defining
%($\ell(y_{1:n};\theta,\kappa) = - \log f(y_{1:n};\theta,\kappa)$ and $\int f(y;\theta,\kappa)dy = 1$) and consider maximising the marginal pseudo-likelihood 
\begin{align}
    \hat{\kappa}_k := 
    %\arg\max_{\kappa\in\mathcal{K}} \mathcal{L}_{\kappa}(y_{1:n}) = 
    \arg\max_{\kappa_k\in\mathcal{K}} \int \pi_k(\theta_k)\exp\{-\ell_k(y;\theta_k,\kappa_k)\}d\theta_k.\label{Equ:GBIHyperparameterSelection}
\end{align}
%Or if hyperparameter estimation is necessary in order to select between two models then a prior can be elicited for $\kappa$ and a joint posterior over parameters $\theta$ and hyperparameter $\kappa$ could be considered
%\begin{align}
%    \pi(\theta, \kappa |y_{1:n}) = \frac{\pi(\theta, \kappa)\exp(-\ell(y_{1:n};\theta,\kappa))}{\mathcal{L}(y_{1:n})}\label{Equ:JointGBIPosterior}
%\end{align}
%and its normaliser $\mathcal{L}(y_{1:n}) = \int \pi(\theta, \kappa)\exp(-\ell(y_{1:n};\theta,\kappa))d\theta d\kappa$ becomes a candidate selection criterion. However analgously to the above, 
%While these are standard procedures for the selection of hyperparameters for standard Bayesian models, i.e. when $\ell(y_{1:n};\theta,\kappa) = - \log f(y_{1:n};\theta,\kappa)$ and $\int f(y;\theta,\kappa)dy = 1$ \jack{references}, 
Unfortunately, such procedure often produces degenerate estimates. %does not account for the arbitrary scale of the loss function, which may result in non-sensical estimates.   
%In particular, such procedures typically recover an optimal parameter value
%General Bayesian procedures ( \eqref{Equ:GBIHyperparameterSelection} 
%and \eqref{Equ:JointGBIPosterior}) 
%are only valid provided the tuple 
%\begin{equation}
%(\theta^{\ast},\kappa^{\ast}) = \argmin_{(\theta,\kappa)\in\Theta\times\Phi} \int \ell(x;\theta, \kappa)dG(x). \label{Equ:GeneralBayesTargetParameters}
%\end{equation}
%This value will not be the case if $\kappa$ is able to control the scale of the loss function. In fact, this is why we deliberately distinguish $\theta$ as the loss minimising parameters of interest and $\kappa$ as nuisance hyperparameters that index the losses. 
%For example, in Tukey's loss $\theta_2 = \{\beta, \sigma^2\}$ are the parameters of interest and the nuisance parameters $\kappa_2$ controls the outlier cut-off point. 
For example, from  \eqref{Equ:TukeysLoss}, it is clear that for fixed $\theta_2 = \{\beta, \sigma^2\}$ Tukey's loss is increasing in $\kappa_2$ and therefore 
 \eqref{Equ:GBIHyperparameterSelection}
selects $\hat{\kappa}_2 = 0$ independently of the data. Similarly, in the kernel density estimation example 
 \eqref{Equ:GBIHyperparameterSelection} 
%The same can be said if we consider $\ell_{KDE}(x, h) = - \log \tilde{f}_h(x_i)$ of the KDE, where $\theta = \emptyset$ and $\kappa = \{h\}$. The above procedures would again 
selects $h = 0$ \citep{habbema1974stepwise, robert1976choice}.

\subsection{Fisher's divergence and the \hyvarinen score}{\label{Sec:HyvarinenScoreIntro}}

%\begin{itemize}
%    \item MCMC means that Bayes rule doesn't need normaliser depending on $x$ but does need any part of the noralsier depending on $\theta$.
%    \item Motivates \Hscore, first \cite{hyvarinen2005estimation}, then \cite{dawid2016minimum}
%    \item Then the uniformative priors as a Bayesian contribution \cite{dawid2015bayesian, shao2019bayesian}
%    \item Then the Fisher's divergence 
%    \item \cite{holmes2017assigning} also the VB paper. 
%\end{itemize}

%The rise of Monte-Carlo methods and particularly Markov-Chain Monte Carlo \citep{gelfand1990sampling} allow Bayesian analyses to be conducted without the calculation of multiplicative terms that depend on the data and not the model parameters of interest. However, 
%Likelihood based inference, including frequentist maximum likelihood techniques, require that any normalsising constant of the likelihood itself, i.e. functions of the parameters, are known. 
%Such a requirement can become prohibitive for likelihood inference when the this normalising constant is computationally expensive to calculate or estimate \jack{FX's stuff, pseudo-marginal methods}.
%In cases where such a normalising constant is not known, 
\cite{hyvarinen2005estimation} proposed a score matching approach for models with intractable normalising constants. % \jack{our extension is to infinite}% or infinite. 
%This is clearly invariant to multiplicative feature of $f_k(x;\theta_k)$ that do not depend on $x$. In particular, inference can be achieved through such a quantity by a process called score matching. 
Score matching minimises
%expected squared distance of the model score from 
Fisher's divergence to the data-generating $g(\cdot)$ in \eqref{Equ:FishersDivergence}, that is
%
%\begin{align}
%   \theta^{\ast}_k :&= \argmin_{\theta_k\in\Theta_k} D_F(g(x)||f_k(x;\theta_k)). %:= \int \left(\nabla_x \log g(x) - \nabla_x \log f_k(x;\theta_k)\right)dG(x).
%\end{align}
%
%Divergences are by definition always positive and only 0 when the two distributions are equal and therefore by construction if $g(x) = f_k(x;\theta_k^0)$ then $\theta^{\ast}_k = \theta_k^0$.
%$\nabla_x \log f_k(x;\theta_k)$, where $\nabla_x$ denotes the first derivative in $x$. 
\begin{align}
    \theta_k^{\ast} :&= \argmin_{\theta_k \in\Theta_k} D_F(g(\cdot)||f_k(\cdot;\theta_k)) = \argmin_{\theta_k\in\Theta_k} \mathbb{E}_{z\sim g}\left[H(z;f_k(\cdot;\theta_k))\right],\label{Equ:MinFisherDivergenceHscore}
\end{align}
where the right-hand side follows from integration by parts under minimal tail conditions, and 
\begin{align}
    H(z;f_k(\cdot;\theta_k)) := 2 \frac{\partial^2}{\partial z^2} \log f_k(z;\theta_k) + \left( \frac{\partial}{\partial z} \log f_k(z;\theta_k) \right)^2,
    \label{Equ:Hscore}
\end{align}
is the \hyvarinen score (\Hscore).
\new{
As discussed we focus on univariate $z \in \mathbb{R}$, but \eqref{Equ:Hscore} can be extended to multivariate $z=(z_1,\ldots,z_d)$  \citep{hyvarinen2005estimation} via
\begin{align}
  H(z;f_k(\cdot;\theta_k)) := 2 \sum_{j=1}^d \frac{\partial^2 }{\partial z_j^2} \log f_k(z;\theta_k) + \|\nabla_z \log f_k(z;\theta_k)\|^2.
  %  H(z;f_k(\cdot;\theta_k)) := 2 \Delta_z \log f_k(z;\theta_k) + \|\nabla_z \log f_k(z;\theta_k)\|^2,
\label{Equ:Hscore_multiv}
\end{align}
%where $\Delta_z$ is the Laplacian operator $\sum_{j=1}^d \frac{\partial^2 }{\left(\partial z^{(j)}\right)^2}$, and $z^{(j)}$ the $j^{th}$ dimension of $x\in\mathbb{R}^d$.
See also \cite{lyu2009interpretation} for further options on extending score-matching to multivariate settings .
}

Given $y_i \sim g$ independently across $i=1,\ldots,n$
one can estimate $\theta_k^*$ by minimising
%The result of above is that the Fisher's divergence minimising parameter can be estimated given data  as
%
\begin{align}
    \hat{\theta}_k := \argmin_{\theta_k\in\Theta_k} \frac{1}{n}\sum_{i=1}^n H(y_i;f_k(\cdot;\theta_k)) %\rightarrow
    %\mathbb{E}_{x\sim g}\left[H(x;f_k(\cdot;\theta_k))\right], \textrm{ as }n\rightarrow\infty \textrm{ for } x_{1:n}\sim g(\cdot)
\end{align}

A critical feature for our purposes is that the \Hscore{} depends only on the first and second derivatives of $\log f_k$ and hence the normalising constant does not play a role, independently of whether it is finite or not.
\color{black} Hence, unlike methods designed for intractable but finite normalisation constants, the \Hscore{} is applicable to improper models.  \color{black}
The \Hscore{} enjoys desirable properties. For example, 
\cite{dawid2016minimum} proved that $\hat{\theta}_k$ is consistent and asymptotically normal,
%The \Hscore{} has been used to conduct Bayesian model selection with improper priors \citep{dawid2015bayesian,shao2019bayesian}. 
%In such cases the marginal likelihood used for standard Bayesian inference, $\mathcal{M}_k(x_{1:n})$ ( \eqref{Equ:MarginalLikelihood}) is no longer scale invariant as the unnormaliseable prior has arbitrary normalsising constant.
and \cite{dawid2015bayesian,shao2019bayesian} showed that its prequential application leads to consistent Bayesian model selection under improper priors. %by applying the \Hscore{} prequentially. %(predictive sequentially \citep{dawid1984present}) to the one step ahead predictive distributions for each model. 
%
%\begin{equation}
%\mathcal{H}^{preq}_k(x_{1:n}) := \sum_{i=1}^n \mathcal{H}(x_i,p_k(\cdot|x_{1:i-1})) \textrm{ where } p_k(x_i|x_{1:i-1}) = \int f_k(x_i;\theta_k)\pi(\theta_k|x_{1:i-1})d\theta_k, \label{Equ:HscoreUninformativePriors} 
%\end{equation}
%
%where $\pi(\theta|x_{1:i-1}) \propto \pi(\theta)\prod_{i=1}^n f_k(x_i;\theta)$ for likelihood model $f_k(\cdot;\theta)$ and unormaliseable prior $\propto \pi(\theta)$. 
%Fisher's divergence has found further applications in variational inference \citep{yang2019variational} and measuring the information in an experiment \cite{holmes2017assigning}.

Closest to our work, \cite{matsuda2019information} proposed the score matching information criteria (\SMIC) to select between models with intractable normalising constants. This criterion estimates Fisher's divergence by correcting the in-sample \hyvarinen score by an estimate of its asymptotic bias. 
%Such a information criteria is akin to AIC predictive inference criteria for normalised models and the authors further consider using cross validation as an alternative. 
We emphasise two main distinctions with our work. The first is the extension to improper models.
%from hard to normalise models to unnormaliseable models. Rather than simply being computationally convenient, in the next sections we motivate the \hyvarinen as a principled method for capturing how well an unnormaliseable pseudo-density represents the underlying process that generated the data. 
Second, these authors used cross-validations and predictive criteria similar to the AIC, which do not lead to consistent model selection, whereas we focus on structural learning where one seeks guarantees on recovering the loss that best approximates the data-generating $g$.

\section{Inference for improper models}{\label{Sec:Methodology}}

We now present our framework.
Section \ref{Sec:RelativeProbabilities} interprets an improper model in terms of relative probabilities and motivates Fisher's divergence as a criterion to fit such a model. 
Section \ref{ssec:hscore} proposes using the \Hscore{} to define a general Bayesian posterior to learn hyper-parameters and choose among a collection of models, some or all of which may be improper. % and defines quantities analogous to standard Bayes factors and posterior model probabilities. %, and discuss its computational advantages over a potential alternative based on prequential inference.
%Despite these advantages, the criterion requires evaluating integral that can be costly to evaluate.
%To address this issue 
Section \ref{ssec:laplaceapprox} proposes a Laplace approximation to the \HBayes{} factors, which we use both in our theoretical treatment and examples. 
Finally, Section \ref{ssec:TwoStep} argues for using the \Hscore{} within a two-step procedure: first selecting a model and estimating hyper-parameters using the \Hscore{}, then reverting to standard general Bayes to learn the parameters of interest.

%\david{Jack, regarding your comment above on whether $\sigma^2$ should be a hyper-parameter, I guess that from this description it would be more natural to have $\sigma^2$ as a parameter of interest. This way one recovers standard inference on model parameters, given the hyper-parameters and the model.}

\subsection{Inference through relative probabilities}{\label{Sec:RelativeProbabilities}}

Our goal is to select which model $f_k$ describes the data best, in terms of helping interpret the data-generating $g$. The main difficulty is that, since $f_k$ may be improper, it is unclear how to define ``best''. %, and thus we cannot use standard tools. %goodness of fit metrics like the \KLD used in standard Bayesian updating and model selection.
Our strategy is to view $f_k$ as expressing relative probabilities, in contrast to the usual absolute probabilities. For example $f_k(y_0, \theta_k)/f_k(y_1, \theta_k)$ describes how much more likely is it to observe $y$ near $y_0$ than near $y_1$. As an illustration, consider Tukey's loss in \eqref{Equ:TukeysLoss}. 
For any pair $(y_0,y_1)$ such that $|y_0 - x^T \beta|, |y_1 - x^T \beta| < \kappa_2\sigma$ are small, Tukey's loss is approximately equal to the squared loss, hence
\begin{align}
  \frac{f_2(y_0; \theta_2, \kappa_2)}{f_2(y_1; \theta_2, \kappa_2)} \approx \frac{\mathcal{N}\left(y_0; x^T \beta, \sigma^2\right)}{\mathcal{N}\left(y_1; x^T \beta, \sigma^2\right)}.
\end{align}
%
%indicating a belief that these observations behave like Gaussian random variables. 
%We note this relationship is approximate rather than exact because the shape of Tukey's loss can be seen to depart from a Gaussian with the same mean and variance before the threshold $\kappa_2$, see Fig. \ref{Fig:TukeysLoss}. % This is a result of the requirement that Tukey's loss has continuous first and second derivatives, not shared for example by Huber's loss. We propose a `sharper' version of Tukey's loss that tracks the Gaussian more closely within the threshold in Section ... 
%
In contrast, for any pair such that $|y_0 - x^T \beta|, |y_1 - x^T \beta| > \kappa_2\sigma$ we have $\nicefrac{f_2(y_0, \theta_2, \kappa_2)}{f_2(y_1, \theta_2, \kappa_2)} = 1$. 
%This indicates all observations $|y - x^T \beta| > c$ are equally likely. 
%In fact, for sets $A_0, A_1 \subset\Omega\subset \mathbb{R}$ %with $A_i \bigcap \{y : |y - x^T \beta|\leq \kappa_2\sigma\} = \emptyset$ for $i = 0, 1$, 
%that do not intersect $\{y : |y - x^T \beta|\leq \kappa_2\sigma\}$
%the pseudo-probability model induced by Tukey's loss has the property 
%
%\begin{equation}
%    \frac{F_2(A_0;\theta_2, \kappa_2)}{F_2(A_1;\theta_2, \kappa_2)} = \frac{\int_{A_0}f_2(y;\theta_2, \kappa_2)dy}{\int_{A_1}f_2(y;\theta_2, \kappa_2)dy} = \frac{\int_{A_0}dy}{\int_{A_1}dy} 
%\end{equation}
%\jack{No this should just be the size of the set given the lebesque measure no?? so only to do with the size of the set, but we could just fix $A_0$ and $A_1$ to have the same lebesque measure}
%
That is, Tukey's loss induces relative beliefs that observations near the mode behave like Gaussian variables, while all faraway observations are equally likely. This encodes the notion that one does not know much about the tails beyond their being thick, which is difficult to express using a proper probability distribution. 

We argue that Fisher's divergence is well-suited to evaluate how closely the relative probabilities of any such $f_k$ approximate those from $g$. 
%As well as being convenient to minimise from an inference perspective as explained in Section \ref{Sec:HyvarinenScoreIntro}, 
Assuming that the gradients of $g(y)$ and $f_k(y)$ are finite for all $y$,
Fisher's divergence in  \eqref{Equ:FishersDivergence} can be expressed as %\david{Changed the integration argument from $x$ to $y$, to avoid confusion with $x$ denoting covariates in Tukey's loss} \jack{I need to go through and check we do not make his mistake elsewhere}
\begin{align}
D_F(g||f_{\theta}) %&=\int \left|\left|\lim_{\epsilon\rightarrow0} \frac{\log g(x + \epsilon) - \log g(x)}{\epsilon} - \lim_{\epsilon\rightarrow0} \frac{\log f(x + \epsilon;\theta) - \log f(x;\theta)}{\epsilon}\right|\right|^2g(x)dx\nonumber\\
&=\int \left|\left|\lim_{\epsilon\rightarrow0} \frac{\log \frac{g(y + \epsilon)}{g(y)} - \log \frac{f_k(y + \epsilon; \theta)}{f_k(y; \theta)}}{\epsilon}\right|\right|_2^2 g(y)dy.\label{Equ:FishersDivergenceRatios}
\end{align} 
Therefore, minimising Fisher's divergence (equivalently, the \hyvarinen score) targets a $f_k$ that approximates the relative probabilities of $g$ in an infinitesimal neighbourhood around $y$, in the quadratic error sense, on the average with respect to $g(y)$. This observation extends the usual motivation for the \hyvarinen score as a replacement of likelihood inference when the normalising constant is intractable to being a justifiable criteria to score improper models. % - \jack{make it clear we extend from intracable to unnormalsieable}
%, specifically targets accurately capturing the \DGP's relative probabilities within a infinitesimal neighbourhood. 

%\color{red} [DR. Detail. I believe you need the assumption that gradients are finite. Otherwise, from \eqref{Equ:FishersDivergence} you should really write the gradient of $\log g$ as a limit where $\epsilon \rightarrow 0$, and that of $\log f_k$ as another limit where $\epsilon' \rightarrow 0$. If these two limits are $\infty$ for a given $x$, then you'd have $\infty - \infty$, which is not necessarily equal to your expression above (e.g. if $f_k=g$ then your expression gives 0).] \color{black}

\subsection{The \Hscore{}}
\label{ssec:hscore}

We consider a general Bayesian framework where the loss is defined by applying the \Hscore{} to the density $f_k(y;\theta_k ,\kappa_k) = \exp \{-\ell_k(y;\theta_k, \kappa_k) \}$, which gives the general posterior
\begin{equation}
    \pi^{H}(\theta_k,\kappa_k|y) \propto \pi_k(\theta_k, \kappa_k) \exp \left\{ -\sum_{i=1}^n H(y_i; f_k(\cdot;\theta_k, \kappa_k)) \right\}.\label{Equ:GeneralBayesRuleHScore}
\end{equation}
We refer to \eqref{Equ:GeneralBayesRuleHScore} as the \Hposterior,
\new{which is a particular case of the scoring rule posteriors of \cite{giummole2019objective}.}
%Clearly $\sum_{i=1}^nH(x; f_}k(\cdot;\theta_k, \kappa_k))$ is invariant to the ordering of the data and \jack{...}.
Note that \eqref{Equ:GeneralBayesRuleHScore} %results from defining the general Bayes loss as the \hyvarinen{} score associated to $\ell_k$, which 
is different from the general Bayesian posterior directly associated to  $\ell_k$ in \eqref{Equ:GeneralBayesRule}. 
An important property of \eqref{Equ:GeneralBayesRuleHScore} is that it provides a consistent estimator for parameters $\theta_k$ and hyper-parameters $\kappa_k$. Specifically, Proposition \ref{Lemma:HscoreParamConsistency} shows that, under regularity conditions, $\tilde{\eta}_k=(\tilde{\theta}_k,\tilde{\kappa}_k)$ maximising \eqref{Equ:GeneralBayesRuleHScore} recovers the optimal $\eta_k^*=(\theta_k^*,\kappa_k^*)$ according to  Fisher's divergence. 

Proposition \ref{Lemma:HscoreParamConsistency} requires mild regularity conditions A1-A3, discussed in Section \ref{App:TechnicalConditions}.
Briefly, A1 requires continuous second derivatives of the \hyvarinen score, that it has a unique minimiser, and that its first derivative has finite variance. 
A2 requires that the \hyvarinen score is dominated by an integrable function, which can be easily seen to hold for Tukey's loss, for example.
Finally, A3 requires that the Hessian of the \Hscore{} is positive and finite around $\eta_k^{\ast}$. 
%A4 imposes Lipschitz bounds on the \hyvarinen score applied to $\ell_j$ and its Hessian.
%

Proposition \ref{Lemma:HscoreParamConsistency} extends \cite{hyvarinen2005estimation} (Corollary 3), 
\new{who stated that $\tilde{\eta}_k$ converges to $\eta_k^*$ in probability, to also give a $1/\sqrt{n}$ convergence rate. Another difference is that \cite{hyvarinen2005estimation}}
considered the well-specified case where $g(y) = f_k(y; \theta_k^{\ast}, \kappa_k^{\ast})$, which in particular requires $f_k$ to be a proper model.
Proposition \ref{Lemma:HscoreParamConsistency} also extends \cite{dawid2016minimum} (Theorem 2), who proved asymptotic normality for \Hscore{} based estimators.
Said asymptotic normality does in general not hold when $\eta_k^{\ast}$ lies on the boundary of the parameter space, e.g. for Tukey's loss if the data are truly Gaussian then $\kappa^{\ast}_2 = \infty$.  
 By Proposition \ref{Lemma:HscoreParamConsistency}, even if normality does not hold, one still attains $\sqrt{n}$ consistency to estimate $\kappa_2^*$.
\new{In terms of technical conditions, \cite{dawid2016minimum} do not list them but refer to standard M-estimation theory assumptions, see Theorem 5.23 in \cite{van2000asymptotic}. These are similar to our assumptions and include differentiability and Lipschitz conditions in $\eta_k$, the existence of a second order Taylor expansion around $\eta_k^*$, and a non-singular Hessian at $\eta_k^*$. The main differences are that we require twice differentiability in $y$ and $\eta_k$, and that we consider a compact parameter space to allow for boundary parameter values (e.g. $1/\kappa_2^*=0$ for Tukey's loss, after a re-parameterisation discussed in Section \ref{Sec:HscoreModelConsistency}).
} Further, Proposition \ref{Lemma:HscoreParamConsistency} extends previous results by explicitly considering improper models and the learning of their hyperparameters.
 
 \new{As is standard under model misspecification, the shape of the \Hposterior{} does not match the frequentist distribution of the posterior mode $\tilde{\eta}_k=(\tilde{\theta}_k,\tilde{\kappa}_k)$ \citep{giummole2019objective}, i.e. it does not have valid frequentist coverage. 
 Per Proposition \ref{Lemma:HscoreParamConsistency} and Theorem \ref{Thm:HBayesFactorConsistency} this is not a major issue for selecting a loss $k$ and hyper-parameters $\kappa_k$, which is our main focus.
 However, posterior inference on $\theta_k$ under the selected $(k,\kappa_k)$ should be properly calibrated, as we discuss in Section \ref{ssec:TwoStep}.}

%In fact we dovals do not ha.ve  not expect normality to hold in nested cases where optimal hyper-parameter occurs at the boundary, e.g. Tukey's loss is equivalent to the Gaussian for $\kappa_2= \infty$.

%\jack{The story here is that Dawid proved asymptotic normality, but we are interested in cases where the hyperparameter may be on the boarder and then our proof shows we can learn these hyperparameters consistently which has not been shown/done before, part of our contribution, and then you can also do selection...}

\begin{proposition}%[\hyvarinen score parameter consistency at rate $o_p(\nicefrac{1}{\sqrt{n}})$]
Let $y=(y_1,\ldots,y_n) \sim g$, $\tilde{\eta}_k=(\tilde{\theta}_k,\tilde{\kappa}_k)$ maximise \eqref{Equ:GeneralBayesRuleHScore}, and $\eta_k^*=(\theta_k^*,\kappa_k^*)$ minimise Fisher's divergence from $f_k(\theta_k,\kappa_k)$ to $g$.
Assume Conditions A1-A2 in Section \ref{App:TechnicalConditions}.
Then, as $n\rightarrow\infty$, % the maximum \Hposterior{} parameters of model $k$ with parameter $\eta_k = (\theta_k, \kappa_k)$ have the following asymptotic behaviour
\begin{equation}
  \left|\left|\tilde{\eta}_k - \eta_k^{\ast}\right|\right|_2 = o_p(1),  \nonumber
\end{equation}
where $||\cdot||_2$ is the $L_2$-norm. 
Further, if Condition A3 also holds, then 
\begin{equation}
  \left|\left|\tilde{\eta}_k - \eta_k^{\ast}\right|\right|_2 = O_p(\nicefrac{1}{\sqrt{n}}).  \nonumber
\end{equation}
\label{Lemma:HscoreParamConsistency}
\end{proposition}
%
%The \Hposterior{} provides a way to compare two nested models where the support of $(\theta_1,\kappa_1)$ is a subset of that of $(\theta_2,\kappa_2)$, as in the Gaussian vs. Tukey's loss example. Under mild conditions the \Hposterior{} for $(\theta_2,\kappa_2)$ concentrate on the values minimising Fisher's divergence, see Proposition \ref{Lemma:HscoreParamConsistency}, hence one could in principle assess if they are consistent with the null model ($1/\kappa_2^2 = 0$ in Tukey's loss). Below we outline a more general strategy that applies both to the nested and non-nested cases.
%
%\david{Should Lemma \ref{Lemma:HscoreParamConsistency} be upgraded to a proposition? From the paragraph above, it shows that one can use the H-posterior to consistently choose between nested models.} \jack{done}
%
A consequence of Proposition \ref{Lemma:HscoreParamConsistency} is that one can use \eqref{Equ:GeneralBayesRuleHScore} to learn tempering hyper-parameters. 
Specifically, suppose that one considers a family of losses $w_k \ell_k()$, where $w_k > 0$ is a tempering parameter. 
While $w_k$ does not affect the point estimate of $\theta_k$ given by \eqref{Equ:GeneralBayesRule}, it plays an important role in driving the posterior uncertainty on $\theta_k$.
Within our framework, one may define
\begin{equation}
  f_k(y; \theta_k, \kappa_k')= \exp\{-\ell_k'(y;\theta_k,\kappa_k') \}= \exp\{- w_k \ell_k(y; \theta_k, \kappa_k)\} \nonumber 
\end{equation}
where $\kappa_k' = (\kappa_k,w_k)$ and $\ell_k'() = w_k \ell_k()$. 
By Proposition \ref{Lemma:HscoreParamConsistency}, one can consistently learn the Fisher-divergence optimal $\kappa_k'$, and in particular $w_k$. 
%inference is dependent to the scale of the loss, 
%That is, inference is not invariant to replacing $\ell'_k$ by a tempered $w_k \ell'_k$ for some $w_k > 0$ (since the derivative of $w_k \ell'_k$ with respect to $y$ depends on $w_k$). 
In contrast, in the general Bayes posterior \eqref{Equ:GeneralBayesRule} it is challenging to estimate such $w_k$. 
%Often, $w$ is considered to be a hyper-parameter. 
Current strategies to set $w_k$ are optimising an upper-bound on generalisation error in PAC-Bayes \citep{catoni2007pac}, %, estimated via cross-validation. 
estimating $w_k$ via marginalisation similar to the ``Safe Bayesian'' fractional likelihood approach of \cite{grunwald2012safe}, %which is motivated by the world of model misspecification \cite[see also][]{miller2018robust}. 
or using information theoretic arguments to calibrate $w_k$ to match certain limiting sampling distributions \citep{holmes2017assigning, lyddon2019general}. 
%However, following the interpretation of an unnormaliseable pseudo-model provided by \ref{Sec:RelativeProbabilities} such a $w$ controls how the pseudo-density is distributed, by acting to power up or down the relative probabilities. As a result, the \Hposterior{} naturally learns $w$ by incorporating it into the hyper-parameter $\kappa_k$, as illustrated in our examples (see Section \ref{Sec:KDE}).
These strategies essentially view $w_k$ as a tuning parameter. In contrast, in our framework $w_k$ is viewed as a parameter of interest that controls the dispersion of the improper model and affects its interpretation. 
%Specifically, our framework allows one to consider 
%\begin{equation}
%  f_k(y; \theta_k, \kappa_k)= \exp\{-\ell_k(y;\theta_k,\kappa_k) \}= \exp\{- w_k \ell_k'(y; \theta_k, \kappa_k')\} \nonumber 
%\end{equation}
%where $\kappa_k = (\kappa_k',w_k)$ and $\ell_k = w_k \ell_k'$. 
%By learning $\kappa_k$, in particular one can also learn the tempering parameter $w_k$ minimising Fisher's divergence to the data-generating $g$. 
See Section \ref{Sec:KDE} for an illustration in kernel density estimation.

%\color{red} [DR. Important: below I changed $\pi(\theta_k,\kappa_k)$ by $\pi_k(\theta_k,\kappa_k)$, the prior is model-dependent. Please make sure that this is fixed throughout, I fixed some further instances below, but didn't look earlier in the document nor in the appendix / proofs etc.] \color{black}

%\jack{Need to be careful later to suggest the H-score is an further method to this when not nested. I guess could provide the lemma here but then we wouldn't have space }

Recall that our main goal is model comparison. To this end, 
%To allow more general comparisons of any two given models, 
in analogy to the marginal likelihood in standard Bayesian model selection,
we define the integrated \Hscore
\begin{equation}
    \mathcal{H}_k(y) = \int\pi_k(\theta_k, \kappa_k) \exp \left\{ -\sum_{i=1}^nH(y_i; f_k(\cdot;\theta_k, \kappa_k)) \right\} d\theta_kd\kappa_k.\label{Equ:HScore}
\end{equation}
%
%Such a criteria is analogous to the marginal-likelihood $\mathcal{M}_k(y_{1:n})$ ( \eqref{Equ:MarginalLikelihood}) or it loss function extension $\mathcal{L}_k(y_{1:n})$ ( \eqref{Equ:GeneralBayesRule}), however unlike these methods applying the \Hscore{} directly to the loss function allows losses to be compared in a scale-invariant manner.
Also, analogously to Bayes factors and posterior model probabilities, we define the $\mathcal{H}$-Bayes factor as $B^{(\mathcal{H})}_{kl} := \mathcal{H}_k(y) / \mathcal{H}_l(y)$ and
\begin{align}
    \pi(k \mid y)= \frac{\mathcal{H}_k(y) \pi(k) }{\sum_l \mathcal{H}_l(y) \pi(l)}=
    \left( 1 + \sum_{l \neq k} B^{(\mathcal{H})}_{lk} \frac{\pi(l)}{\pi(k)} \right)^{-1},
    \label{eq:postprob_hscore}
\end{align}
where $\pi(k)$ is a given prior probability for each model $k$. In our examples we use uniform $\pi(k)$, since we focus on the comparison of a few models, but in high-dimensional settings it may be desirable to set $\pi(k)$ to favour simpler models. % (see \cite{rossell:2021} and references therein for recent results).

%In Section \ref{Sec:HscoreModelConsistency} we show that a computationally convenient approximation to $B^{(\mathcal{H})}_{kl}$ consistently selects the smallest model/loss minimising Fisher's divergence to the data-generating truth. In particular, this includes challenging cases such as the squared loss versus Tukey's loss example, where the models are nested for certain hyper-parameter values.

We note that an interesting alternative strategy for model comparison, also based on the \Hscore{}, is to extend the prequential framework of \cite{dawid2015bayesian} and \cite{shao2019bayesian} designed for improper priors.
Therein one could adopt a general Bayesian framework, replacing the likelihood by $f_k(x_i;\theta_k, \kappa_k) = \exp \{ -\ell_k(x_i;\theta_k \kappa_k) \}$.
Prequential approaches enjoy desirable properties, such as consistency and leading to joint coherent inference on the model and parameter values.
Unfortunately, prequential inference is computationally hard, particularly when considering several models. First, for each model inference needs to be updated $n$ times to calculate the one-step-ahead predictive distribution. Second, said updates are not permutation invariant, so one should in principle consider the $n!$ orderings of the data.
Thus, while interesting, we leave such line of research for future work.
\new{
We remark however that a salient feature in \cite{dawid2015bayesian} is that one may use improper priors. In contrast, our framework requires one to use proper priors. The reason is that otherwise one may suffer the usual Jeffreys-Lindley paradox in Bayesian model selection \citep{lindley1957statistical}. Specifically, if one considers two nested models and sets an improper prior under the larger one, then our \HBayes{} factor in favor of the smaller model is infinite, i.e. one selects the smaller model regardless of the data.
}
%Instead, we consider minimising the \Hscore{} jointly. This must be done with care: Sec. 8.1 \cite{dawid2015bayesian} provided a counter-example where applying the \Hscore{} jointly to the marginal likelihood did not provide consistent model selection. 

\subsection{Laplace approximation and BIC-type criterion}
\label{ssec:laplaceapprox}

%\color{red} [DR. To alleviate notation I suggest dropping the $(n)$ super-scripts below from $\tilde{\eta}$ and $A_j$. I think we could also safely replace $y_{1:n}$ by $y$ throughout, actually, but your call.] \color{black}

%It is well known that calculating the marginal likelihood for a Bayesian model is often plagued by intractability. Sadly, such concerns are not alleviated here and are if anything made slightly worse by the fact that using the \hyvarinen{} score removes any notion of conjugacy. 
%While there are many strategies for computing integrals such has $\mathcal{H}_k(y)$ in \eqref{Equ:HScore} exactly (see \cite{llorente2020marginal} for a review), it is convenient to have faster options. %\cite[e.g.][]{meng2002warp, perrakis2014use} in an unbiased fashion,
\new{There are many available strategies to compute or estimate integrals such as $\mathcal{H}_k(y)$ in \eqref{Equ:HScore} (see \cite{llorente2020marginal} for a review). %We focus on deterministic methods for their speed and analytic tractability.}
For its speed and analytic tractability, we consider the Laplace approximation}
\begin{align}
        \tilde{\mathcal{H}}_k(y) :&= (2\pi)^{\frac{d_k}{2}}
        \pi_k \left(\tilde{\eta}_k\right)\exp\left\{ - \sum_{i=1}^nH\left(y_i; f_k\left(\cdot;\tilde{\eta}_k\right)\right)\right\}
        |A_k\left(\tilde{\eta}_j\right)|^{-\frac{1}{2}},
        \label{Equ:LaplaceApproxHScore}\\
        \textrm{ with }\tilde{\eta}_k :&= \argmin_{\eta_k} \sum_{i=1}^n H(y_i; f_k(\cdot; \eta_k)) - \log \pi_k(\eta_k),\label{Equ:PenalisedHScoreMin}
\end{align}
being the mode of the log \Hposterior{}, $A_k\left(\eta_k\right)$ its Hessian at $\eta_k$, and $\eta_k = (\theta_k, \kappa_k)$. \new{We denote the corresponding Laplace approximate \HBayes-factor by $\tilde{B}^{(\mathcal{H})}_{kl} := \mathcal{H}_k(y) / \mathcal{H}_l(y)$.} % the whole parameter vector. 
%For notational simplicity we suppress the dependence of these on $n$.
%\begin{align}
%    A^{(n)}_j\left(\eta_j\right):= \nabla_{\eta_j}^2\left(\sum_{i=1}^nH(y_i; f_j(\cdot;\eta_j)) - \log \pi(\eta_j)\right).\nonumber
%\end{align}

%\eqref{Equ:LaplaceApproxHScore} requires optimising the log-posterior and evaluating the Hessian, which is computationally cheaper than Monte Carlo approximations. 
%tical standpoints \citep[e.g.][]{rossell2018tractable}. The Laplace approximation to $\mathcal{H}_k(y_{1:n})$ from  \eqref{Equ:HScore} is
%\eqref{Equ:LaplaceApproxHScore} only requires optimising the log-posterior and evaluating the Hessian. 
%Further,  \eqref{Equ:LaplaceApproxHScore} can also be used to produce the following approximation to the $\mathcal{H}$-Bayes factor $\tilde{B}^{(\mathcal{H})}_{kl} := \nicefrac{\tilde{\mathcal{H}}_k(y_{1:n})}{\tilde{\mathcal{H}}_l(y_{1:n})}$. 
Computational tractability is important when one considers many models or the integrand is expensive to evaluate, e.g. in our kernel density examples it requires $\mathcal{O}(n^2)$ operations. %, or in high-dimensional regression problems \citep[e.g.][]{rossell2018tractable}.
Further, the availability of a closed-form expression facilitates its theoretical study (Section \ref{Sec:HscoreModelConsistency}).
%Theorem \ref{Thm:HBayesFactorConsistency} in Section \ref{Sec:HscoreModelConsistency} shows $\tilde{B}^{(\mathcal{H})}_{kl}$ provides consistent model selection, even amongst nested models. 
%Lastly, we comment very briefly and broadly on the accuracy of such a Laplace approximations to the \Hscore{} as defined in  \eqref{Equ:Hscore}. 
See \cite{kass1990validity} for results on the validity of Laplace approximations. We do not undertake such a study, instead we prove our results directly for the approximation \eqref{Equ:LaplaceApproxHScore} that we actually use for inference.

Although we motivate our methodology from a Bayesian standpoint, we note that $\tilde{\mathcal{H}}_k(y)$ can be viewed as Bayesian-inspired information criteria analogous to the BIC \citep{schwarz1978estimating}. % for probability model selection allowing for Frequentist selection between unnormalsieable loss functions. 
Specifically, in regular settings where $A_k$ is of order $1/n$, one could take the leading terms in $\log \tilde{\mathcal{H}}_k(y)$ to obtain
\begin{align}
- \sum_{i=1}^n H\left(y_i; f_k (\cdot;\tilde{\eta}_k )\right) - \frac{d_k}{2} \log(n) + \log \pi_k \left(\tilde{\eta}_j \right) 
        \nonumber
\end{align}
as a model selection criterion. This expression is analogous to the BIC, except for the log prior density term, which converges to a constant for any $\pi_k$ bounded away from 0 and infinity. The log prior term can play a relevant role however when considering non-local priors where $\pi_k(\eta_k)$ can be equal to 0 for certain $\eta_k$, see Section \ref{Sec:HscoreModelConsistency}.

\subsection{Two-step inference and model averaging}
\label{ssec:TwoStep}

%\color{red} [DR. The notation $\theta_k^{\mathcal{H}^*}$ is cumbersome. In Sec 4 you denote the asymptotic value of the posterior mode $\tilde{\eta}_k$ by $\eta_k^*$. How about we use $\theta_k^*$ instead of $\theta_k^{\mathcal{H}^*}$, and replace the current $\theta_k^*$ by $\tilde{\theta}_k^*$? Or something along these lines.] \color{black}

%Applying the \Hscore{} to the loss function within the general Bayesian \Hposterior{} effectively changes the loss function used for inference. That is to say that the \Hposterior{} in  \eqref{Equ:GeneralBayesRuleHScore} and the original general Bayesian posterior of interest  \eqref{Equ:GeneralBayesRule}, while both depending on los $\ell_k(y; \theta_k, \kappa_k)$, have different forms. A consequence of this is that the parameters that the general Bayesian update is coherently learning about \citep{bissiri2016general} also change. Rather than learning about the parameters $\theta_k^{\ast} = \argmin_{\theta_k\in\Theta_k} \int \ell_k(x;\theta_k, \kappa_k)dG(x)$ the \Hposterior{} in  \eqref{Equ:GeneralBayesRuleHScore} is learning about 
%

As discussed the \Hposterior{} \eqref{Equ:GeneralBayesRuleHScore} asymptotically recovers the parameters $(\theta_k^{\ast},\kappa_k^{\ast})$ minimising Fisher's divergence,
%\begin{equation}
%(\theta_k^{\ast},\kappa_k^{\ast}) = \argmin_{(\theta_k,\kappa_k)\in\Theta_k\times\Phi_k} \int H(y; f_k(\cdot;\theta_k, %\kappa_k))g(y)dy.   \label{Equ:GeneralBayesHTargetParameters}
%\end{equation}
whereas the general Bayesian posterior \eqref{Equ:GeneralBayesRule} recovers 
the parameters minimising the expected loss $\tilde{\theta}_k^{\ast} = \argmin_{\theta_k\in\Theta_k} \int \ell_k(y;\theta_k, \kappa_k)g(y)dy$, for a given $\kappa_k$.

We adopt the pragmatic view that, while one may consider the \Hposterior{} to choose a model $k$ and learn the associated hyper-parameter $\kappa_k$, after said choice one may want to obtain standard inference under the selected model.
That is, one desires to learn
\begin{equation}
   \arg \min_{\theta_k\in \Theta_k} \int \ell_k(y;\theta_k, \kappa_k^{\ast}) g(y)dy.
\end{equation}
For example, suppose that the \Hscore{} selects a proper probability model (e.g. the Gaussian model). %, i.e. $\ell_{k^{\ast}} = - \log f_{k^{\ast}}$ where $f_{k^{\ast}}$ is an proper probability density. 
One may then wish that inference collapses to standard Bayesian inference under that model, rather than being based on the \Hposterior{} in \eqref{Equ:GeneralBayesRuleHScore}. 
This is easily achieved with a two-step procedure. First, one uses \eqref{eq:postprob_hscore} to select $\hat{k}$ and \eqref{Equ:GeneralBayesRuleHScore} to estimate $\hat{\kappa}_{\hat{k}}$.
%$\tilde{\mathcal{H}}_k(\*y)$ in  \eqref{Equ:LaplaceApproxHScore}) to select a loss $\hat{k}$ and estimate loss hyperparameter $\tilde{\kappa}_{\hat{k}} \in \tilde{\eta}_k$ maximising the Hposterior as defined in  \eqref{Equ:PenalisedHScoreMin}\footnote{such an estimate is convenient as it requires no further computations that those required to calculate $\tilde{\mathcal{H}}_k(\*y)$}. %, however in practise one may wish to estimate $\kappa_k$ as \arg\min_{\kappa_k \in \Phi_k} \int by marginalising over $\theta$}. 
%
Second, given $(\hat{k}, \hat{\kappa}_{\hat{k}})$ one uses the general Bayesian posterior \eqref{Equ:GeneralBayesRule} for $\theta_{\hat{k}}$. 
A further alternative to selecting a single model is to mimic Bayesian model averaging \citep{hoeting1999bayesian}, where the estimates under each model are weighted according to the posterior probabilities in  \eqref{eq:postprob_hscore}. 
%The next section investigates the formal consistency results associated with the parameter estimation and model/loss selection procedures described above.

\new{As an important point for quantifying uncertainty on $\theta_k$, if the data-generating $g(y)$ is not contained in any of the considered models, the generalized posterior $\pi_k(\theta_k \mid y, \kappa_k) \propto \exp\{ - \ell_k(y; \theta_k, \kappa_k) \} \pi_k(\theta_k \mid \kappa_k)$ in \eqref{Equ:GeneralBayesRule} is miss-calibrated relative to the sampling distribution of its posterior mode $\tilde{\theta}_k$. This issue is not specific to generalized posteriors, it also affects any standard Bayesian posterior when the model is misspecified. When the posterior is asymptotically normal, it is possible to define a calibrated version that provides valid frequentist uncertainty quantification. Briefly, following \cite{ribatet2012bayesian} and \cite{giummole2019objective}, define the calibrated posterior as
\begin{align}
\pi_k^{C}(\theta_k \mid y, \kappa_k)=
    \pi_k(\tilde{\theta}_k + C(\theta_k - \tilde{\theta}_k) \mid y, \kappa_k) 
\nonumber,
\end{align}
%where $C$ is any matrix satisfying $C^T H C= H J^{-1} H$,
%$H$ is the expected hessian of $-\ell_k$ and $J$ the covariance of its gradient, evaluated at $\tilde{\theta}_k^*$ (in practice, at its estimate $\tilde{\theta}_k$).
%It is easy to check that $\pi_k^C$ has the same mode as $\pi_k$.
%Specifically, one may take $C= H^{-1/2} \Sigma^{1/2}$, where $\Sigma= H J^{-1} H$ (\cite{ribatet2012bayesian}, Section 2.2). 
%Then, the shape of the calibrated posterior $\pi_k^C(\cdot)$ at $\tilde{\theta}_k^*$ asymptotically matches the sampling distribution of $\tilde{\theta}_k$, and hence provides valid uncertainty quantification (\cite{giummole2019objective}, Theorem 3.1).
where $C$ is any matrix satisfying $C^T J_{k} C= J_{k} K_{k}^{-1} J_{k}$,
$J_{k}$ is the expected Hessian of $-\ell_k$ and $K_{k}$ the covariance of its gradient, evaluated at $\tilde{\theta}^*_k$ (in practice, at its estimate $\tilde{\theta}_k$).
It is easy to check that $\pi_k^C$ has the same mode as $\pi_k$.
Specifically, one may take $C= J_{k}^{-1/2} \Sigma^{1/2}$, where $\Sigma= J_{k} K_{k}^{-1} J_{k}$ (\cite{ribatet2012bayesian}, Section 2.2). 
Then, the shape of the calibrated posterior $\pi^C_k$ at $\theta^*_k$ asymptotically matches the sampling distribution of $\tilde{\theta}_k$, and hence provides valid uncertainty quantification (\cite{giummole2019objective}, Theorem 3.1).
}
%\jack{I changed the notation here to be consistent wit the asymptotic MSE notation in Section \ref{App:asympt_MSE}}
%\david{I changed the notation back, at the start of this section we define $\tilde{\theta}^*$ as the loss-minimising parameter value. We should probably also use this notation in Section \ref{App:asympt_MSE}}

\section{Consistency of \Hscore{} model selection}{\label{Sec:HscoreModelConsistency}}

%\section{Theoretical guarantees}{\label{Sec:Theory}}

%\subsection{Notation and Technical Conditional}

%\jack{Do I need to define any further notation}

%\subsection{Parameter Consistency}

%\jack{Okay so I need to look at \cite{dawid2016minimum}. 
%\begin{itemize}
%    \item DO they prove $O_p(\nicefrac{1}{\sqrt{n}})$?
%    \item What are their conditions? Good to compare anyway
%    \item If for example our only addition is the inclusion of the prior then all is good
%\end{itemize}
%Basically I need a very quick stoyr as to why we have a separate theorme
%}

%\begin{lemma}[\hyvarinen-score parameter consistency]
%Assume conditions A1-A2. Given a sample $y_{1:n}\sim G$, then as $n\rightarrow\infty$ the parameters of loss $j$ have the following asymptotic behaviour
%\begin{equation}
%  \left|\left|\tilde{\eta}_j^{(n)} - \eta_j^{\ast}\right|\right|_2 = O_p(\nicefrac{1}{\sqrt{n}})
%\end{equation}
%where $||\cdot||_2$ is the $L_2$-norm.
%\label{Lemma:HscoreParamConsistency}
%\end{lemma}

%\subsection{Consistency of \Hscore{} model selection}{\label{Sec:HscoreModelConsistency}}

We now state Theorem \ref{Thm:HBayesFactorConsistency}, our main result that \eqref{Equ:LaplaceApproxHScore} consistently selects the model closest in Fisher's divergence to the data-generating $g$.
When several models attain the same minimum, as may happen when considering nested models, then  \eqref{Equ:LaplaceApproxHScore} selects that of smallest dimension. 
The proof does not require that the \hyvarinen score is asymptotically normal, 
%The first required result is the consistency of the in-sample \hyvarinen score minimising estimate to the Fisher's divergence minimising parameters. 
which holds under the conditions in Theorem 2 of \cite{dawid2016minimum},
%estimate to be asymptotically normal with mean as the expected loss minimiser and provide a form for its variance. They further consider the special case of the \hyvarinen score. 
but simply the $\sqrt{n}$-consistency proven in Proposition \ref{Lemma:HscoreParamConsistency}. %Our results requires $\left|\left|\tilde{\eta}_j - \eta_j^{\ast}\right|\right|_2 = O_p(\nicefrac{1}{\sqrt{n}})$, 
%In fact we do not expect normality to hold in nested cases where optimal hyper-parameter occurs at the boundary, e.g. Tukey's loss is equivalent to the Gaussian for $\kappa_2= \infty$.
 
%For completeness  provides a self contained proof that . %Theorem \ref{Thm:HBayesFactorConsistency} uses the $\sqrt{n}$-parameter-consistency to prove consistency of the $\mathcal{H}$-Bayes Factor.
%Further, we penalise the in-sample minimisation with the prior density. \color{red} [DR. This last sentence feels incomplete, can you elaborate, or delete if unnecessary? 
%\david{btw, shouldn't we mention here some of the non-local prior being needed for certain nested cases? In that sense, then our results wouldn't be completely analogous to those of standard BMS, which would be worth pointing out and would earn us extra points in terms of novelty.} 

%\color{red} [DR. Reference the proposition number in the appendix. In Theorem 1, provide the appendix section number where the conditions are stated. Also, I moved your earlier explanation of how to interpret the theorem outside the actual statement. btw, for some reason the Theorem number is off when referencing, it pops up as Theorem 2. btw, it would be worth adding 1-2 sentences on the conditions being standard and what they require at a high-level, then in the appendix provide a slightly longer discussion of the conditions, e.g. take what we did in any of the papers with Javier as an example.] \color{black}

Theorem \ref{Thm:HBayesFactorConsistency} mirrors standard results for Bayes factors (Theorem 1 in \new{\cite{dawid2011posterior}}), the main difference being that it involves Fisher's rather than Kullback-Leibler divergence.
As discussed after the theorem, when the optimal hyper-parameter occurs at the boundary the model selection consistency provided by Theorem \ref{Thm:HBayesFactorConsistency} may not hold, unless one uses a suitable adjustment.
Before stating the theorem, we interpret its implications.
Part (i) considers a situation where one compares two models $l$ and $k$ such that the former is closer to $g$ in Fisher's divergence. Then $\tilde{B}^{(\mathcal{H})}_{kl}$ converges to 0 at an exponential rate in $n$. 
Part (ii) considers that both models attain the same Fisher's divergence, e.g. nested models such as Tukey's and the Gaussian model when the data are truly Normal. 
Then,  $\tilde{B}^{(\mathcal{H})}_{kl}$ favours the smaller model at a polynomial rate in $n$.

%allowing for parameters at the boundary, model misspecification, and the models to be either nested or non-nested. 
\new{\cite{dawid2015bayesian} also proved model selection consistency for a prequential application of the \Hscore{} (which differs from our methodology), restricted to well-specified linear regression models. \cite{shao2019bayesian} extended the results to other models but only considered non-nested comparisons, for nested cases only a conjecture is given.
Similarly to us they require twice differentiability, Lipschitz conditions on the score and H-score functions, but other assumptions are stronger. For example, the posterior is assumed to concentrate on the optimal $\eta_k^*$ (we prove such result) and must have certain bounded posterior expectations in supremum norm.}

\begin{theorem}
%[$\mathcal{H}$-Bayes Factor Consistency]
Assume Conditions A1-A4, and let $\eta_k^*=(\theta_k^*,\kappa_k^*)$ be as in Proposition \ref{Lemma:HscoreParamConsistency}.
%Given a sample $y_{1:n}\sim G$ the Laplace approximation to the $\mathcal{H}$-Bayes factor (given by $\tilde{B}^{(\mathcal{H})}_{kl} := \nicefrac{\tilde{\mathcal{H}}_k(y_{1:n})}{\tilde{\mathcal{H}}_l(y_{1:n})}$ for $\tilde{\mathcal{H}}_j(y_{1:n})$ defined in  \eqref{Equ:LaplaceApproxHScore}) of loss $k$ over loss $l$ has the following asymptotic behaviour as $n\rightarrow\infty$
\begin{enumerate}[label=(\roman*)]
    \item Suppose that $\mathbb{E}_g[H(y; f_l(\cdot; \eta^{\ast}_l))] < \mathbb{E}_g[H(y; f_k(\cdot; \eta^{\ast}_k))]$. Then 
    \begin{align}
        \frac{1}{n}\log \tilde{B}^{(\mathcal{H})}_{kl} = \mathbb{E}_g[H(y; f_l(\cdot; \eta^{\ast}_l))] - \mathbb{E}_g[H(y; f_k(\cdot; \eta^{\ast}_k))] + o_p(1).
        \nonumber
    \end{align}

    \item Suppose that $\mathbb{E}_g[H(y; f_l(\cdot; \eta^{\ast}_l))] = \mathbb{E}_g[H(y; f_k(\cdot; \eta^{\ast}_k))]$. Then  \begin{align}
        \log \tilde{B}^{(\mathcal{H})}_{kl} = \frac{d_l - d_k}{2}\log(n) + O_p(1).
        \nonumber
    \end{align}
\end{enumerate}
\label{Thm:HBayesFactorConsistency}
\end{theorem}
%\jack{Need to copy the exact form of this into the appendix - No this is stated in regards to the text above, the supplmentary gives a full specification}

The result requires Conditions A1-A4 given and discussed in Section \ref{App:TechnicalConditions}. Conditions A1-A3 are mild and also summarised before Proposition \ref{Lemma:HscoreParamConsistency}, whereas A4 imposes a Lipschitz condition on the \Hscore{} and its Hessian.
As an important remark, the slower rate in Part (ii) is due to Condition A1 that the prior density $\pi_l(\eta_l^*)>0$ at the optimal $\eta_l^*$, where $l$ is the larger model. 
%Note that there is an asymmetry to Theorem \ref{Thm:HBayesFactorConsistency}. If the more complicated model/loss minimises Fisher's divergence then this is selected at an exponential rate (i) while in the nested case where the smaller model is sufficient, this is selected at a much slower polynomial rate ii). This can greatly affect finite sample model selection performance. 
This defines a so-called local prior, in contrast to non-local priors \citep{johnson2012bayesian, rossell2017nonlocal} which place zero density at the value $\kappa_{l\setminus k}^*$ where the more complicated model $l$ recovers the simpler model $k$. 
Since non-local priors violate A1, Corollary \ref{Thm:HscoreConsistencyNLPs} extends Theorem \ref{Thm:HBayesFactorConsistency} to show that non-local priors attain faster rates in Part (ii), while maintaining the exponential rates in Part (i). %\david{Note that we have to restrict the speed at which the log-prior goes to 0, if it where faster than exponential in $\kappa^2$, that is or order $n$, then it could kill the leading asymptotic term given by $E_G$.}
%Corollary \ref{Thm:HscoreConsistencyNLPs} proves that a modification to the results and method of Theorem \ref{Thm:HBayesFactorConsistency} including a non-local penalty in the evaluation of the Laplace approximation, but not the optimisation of the \MAP's required to evaluate this, maintains the exponential rate of selection in the case i) but allows for the incorporation of a faster rate of selection foe case ii). 
We will demonstrate that this improvement can have non-negligible practical implications in Section \ref{Sec:LPvsNLP}.

%A further, violation of the conditions of Theorem \ref{Thm:HBayesFactorConsistency} is when the matrix $A^{\ast}_j(\eta_j^{\ast})$ is not finite. In particular, in our experiments we found that this was the case for Tukey's loss as the hyperparameter $\kappa_2\rightarrow\infty$, the value recovering the nested Gaussian likelihood model, and prevented this being consistently selected. For Tukey's loss this was solved by reparametrising the model as $\nu_2 = \frac{1}{\kappa_2}$,which ensured that the expected Hessian was finite at $\nu_2 = 0$. While convenient here, reparametrising the model/loss to get a finite Hessian may not always be possible. These cases provide a further use for non-local priors, however unconventionally here they are deployed in order to ensure consistency. This can be done by combining the analysis of the determinant of the Hessian with that of the prior evaluation in the Laplace approximation  \eqref{Equ:LaplaceApproxHScore} (Steps 1 and 2 in the proof of Theorem \ref{Thm:HBayesFactorConsistency} in Section \ref{App:HscoreConsistency}) and choosing the rate of the non-local prior to shrink faster than the rate at which the Hessian term explodes. 

Another practically-relevant remark is that A3 requires a finite expected Hessian near the optimal $(\theta_k^*,\kappa_k^*)$, which can be problematic in certain settings. For example, in Tukey's loss if data are truly Gaussian then $\kappa_2^*=\infty$, which leads to an infinite Hessian. For Tukey's loss the problem can be avoided by reparameterising $\nu_2=1/\kappa_2^2$, for which the Hessian is finite, but more generally such a reparameterisation may not be obvious or not exist. 
These cases provide a further use for non-local priors. 
By Corollary \ref{Thm:HscoreConsistencyNLPs}, one may set a non-local prior that vanishes sufficiently fast at the boundary (basically, a faster rate than that at which the Hessian diverges) to attain model selection consistency. %\jack{Note, for Tukey's loss such a reparametrisation also us to consider the support of $\nu_2$ compact while also containing the Gaussian distribution ($\nu_2 = 0 \Leftrightarrow \kappa_2 = \infty$)} \jack{This was also mentioned before Proposition 1}

As a final remark, we also prove that under extended Conditions A1-A7 the score matching information criterion of \cite{matsuda2019information} does not lead to consistent model selection for the nested setting in Part (ii), see Corollary \ref{Thm:SMICConsistency}. 
This result is analogous to predictive criteria such as cross-validation or Akaike's information criterion not leading to consistent model selection, see \cite{shao1997asymptotic}.

%\subsection{Asymptotic normality of the \Hposterior}

%\begin{theorem}{(\textbf{Asymptotic normality of the \Hposterior{} })}
%\jack{I have coped the wording from Arnaud's adaptation of the UAI proof}
%Let the regularity conditions of \citep{lyddon2019general, chernozhukov2003mcmc} hold. There exists a non-singular matrix $K$ such that under the \Hposterior{} in  \eqref{Equ:GeneralBayesRuleHScore} $\pi^H(\eta_k|y_{1:n})$
%\begin{equation}
%    \sqrt{n}\left(\eta_k - \hat{\eta}^{(n)}_k\right) \stackrel{P}{\longrightarrow} \mathcal{N}\left(0, K^{-1}\right)
%\end{equation}
%almost surely as $n\rightarrow\infty$ with respect to $x_{1:\infty}\sim g(\cdot)$\footnote{$\pi^H(\eta_k|y_{1:n})$ is here interpreted as a random probability measures and a function of the random observations $y_{1:n}$} and where  $K = K\left(\theta^{\ast}_k\right)$ is given in  \eqref{Equ:HessianHscore}
%\end{theorem}

%Similarly to the consistency results of Section \ref{} these results show that the \Hscore{} posterios has similar asymptotic properties to that of the standard Bayesian posterior and also the general Bayesian posterior.

%\section{Simulated examples}

%In this next sections we show how the devloped methodology can be used to solve the problems associated with robust regression and non-paramertic density estimation described in Section \ref{}

%Next, we investigate the performance of our \Hscore{} model selection and estimation in practice on examples from robust regression and non-parametric density estimation. 

\section{Robust Regression with Tukey's loss}{\label{Sec:RobustRegression}}

%\david{When discussing what results we observed, we should always use the past tense, also for what priors we set, what simulation setting we used etc. I fixed most of these instances, but please double-check.}

%\subsection{Experimental Details}

We revisit the robust regression in Section \ref{Sec:MotivatingApplication}, where one considers a Gaussian model and the improper model defined by Tukey's loss.
Section \ref{ssec:robustness_efficiency_tradeoff} illustrates that when the data contain outliers, the \Hscore{} chooses Tukey's model and learns its cut-off hyper-parameter in a manner that leads to robust estimation.
Section \ref{Sec:LPvsNLP} shows the opposite situation, where data are truly Gaussian, and the benefits of setting a non-local prior on Tukey's cut-off hyper-parameter to improve the model selection consistency rate. Finally, Section \ref{ssec:geneexpression_datasets} shows two gene expression datasets, one exhibiting Gaussian behavior and the other thicker tails.
We compare our results to the \SMIC \citep{matsuda2019information} which, despite not being designed to compare improper models, to our knowledge is the only existing criterion that can be used for this task.

The \Hscore s for the squared loss ($\ell_1$) and Tukey's loss ($\ell_2$) (see Section \ref{App:HscoreGaussianTukeys}) are
\begin{align}
    H_{1}(y; f(\cdot;x, \theta_1)) &= \sum_{i=1}^n -\frac{2}{\sigma^2} + \frac{(y_i - x_i^T \beta)^2}{\sigma^4}\label{Equ:HscoreGaussian}\\
    H_{2}(y; f(\cdot;x, \theta_2, \kappa_2)) &=  \sum_{|y_i - x_i^T \beta|\leq \kappa_2\sigma}\left\lbrace \left(\frac{(y_i - x_i^T \beta)}{\sigma^2}-\frac{2(y_i - x_i^T \beta)^3}{\kappa_2^2\sigma^4}+\frac{(y_i - x_i^T \beta)^5}{\kappa_2^4\sigma^6}\right)^2 - \right.\nonumber\\
    &\qquad\qquad\qquad\left. 2\left(\frac{1}{\sigma^2}-\frac{6(y_i - x_i^T \beta)^2}{\kappa_2^2\sigma^4}+\frac{5(y_i - x_i^T \beta)^4}{\kappa_2^4\sigma^6}\right)\right\rbrace.\label{Equ:HscoreTukeysLoss}
\end{align}
\color{black}
Given our interest in learning $\kappa_2$, it is important to remark that its role is to define the proportion of observations that should be viewed as outliers. Thus, too small $\kappa_2$ leads to problems.
\color{black}
First, minimising $H_2$ in \eqref{Equ:HscoreTukeysLoss} has a trivial degenerate solution by setting $(\beta,\sigma,\kappa_2)$ where \color{black} $\kappa_2$  is so small that \color{black} only one observation satisfies $|y_i-x_i^T \beta| \leq \kappa_2\sigma$, and $y_i=x_i^T \beta$ for that observation. 
\color{black} Such solution is undesirable, as it views all observations but one as outliers. Fortunately, it is possible to define a range of reasonable $\kappa_2$ using \color{black}
%To avoid such solutions, we define a constraint on the parameters related to 
the notion of the breakdown point, i.e. the number of observations that can be perturbed without causing arbitrary changes to an estimator. Following \cite{rousseeuw1984robust}, this leads to a constraint on $(\beta,\sigma,\kappa_2)$ such that
%\begin{align}
%\sum_{i=1}^n \mbox{I}(y_i-x_i^T \beta \geq \sigma \kappa_2) \leq \frac{n}{2} - p,
%\label{Equ:BreakdownConstraint4}
%\end{align} \jack{Think about this - this constraint is horribly discontinuous - worked fine for correctly specified guy, not so well for DLD, i think becuase we have 17 parameters from 192 observations. This condition is only sufficient if there is no error in the other observations, else they can just be pushed close to the boundary, so you }
%where the left-hand side is the number of observations below Tukey's cutoff threshold.
\begin{align}
    \frac{1}{\rho(\kappa_2, \kappa_2)}\sum_{i=1}^n \rho\left(\frac{y_i - x_i^T \beta}{\sigma}, \kappa_2\right) \leq \frac{n}{2} - p,\label{Equ:BreakdownConstraint}
\end{align}
where $\rho\left(\frac{y_i - x_i^T \beta}{\sigma}, \kappa_2\right) = \ell_2(y_i; x_i, \theta_2, \kappa_2) - \frac{1}{2}\log(2\pi\sigma^2)$, $\ell_2$ is as in  \eqref{Equ:TukeysLoss}, and $p= \mbox{dim}(x_i)$. 
See Section \ref{App:TukeysBreakdown} for the derivation and further discussion. 
%\david{In Sec 5.2 and 5.3 we used $\sigma^2 \sim IG(2.01,0.5)$, but in Sec 5.1 we used an IG(0.01,0.01)? The editor expressed concerns about adhoc choosing of prior parameters, so it's better to use the same prior thoughout and explain it upfront.}

\new{

Regarding priors, throughout we set $\sigma^2 \sim \mathcal{IG}(2, 0.5)$ and $\beta |\sigma^2 \sim \mathcal{N}\left(0, 5\sigma^2 I\right)$ for the Gaussian and Tukey models, in the latter case truncated to satisfy \eqref{Equ:BreakdownConstraint}.
The idea is that these are mildly informative priors, e.g. the prior variance of $\sigma^2$ is infinite, but avoid degenerate solutions by truncating using \eqref{Equ:BreakdownConstraint}.
For $\kappa_2$ we set an inverse gamma prior in terms of $\nu_2=1/\kappa_2^2$, $\pi_2^{\NLP}(\nu_2) = \mathcal{IG}(\nu_2; a_0, b_0)$.
The inverse gamma is a non-local prior (i.e. has vanishing density at $\nu_2=0$), which as discussed in Section \ref{Sec:LPvsNLP} carries important benefits for model selection.
Its prior parameters $(a_0,b_0)=(4.35, 1.56)$ are set such that $P(\kappa_2 \in (1, 3))=0.95$. The reasoning is that one assigns high prior probability to the cutoff being in a reasonable default range, between 1-$\sigma$ and 3-$\sigma$. If data were truly Gaussian, 0.3\% would lie outside the 3-$\sigma$ region and 68.3\% within 1-$\sigma$. Hence, $\kappa_2=3$ would exclude clear outliers under the Gaussian, and $\kappa_2=1$ would keep most of the Gaussian data. We remark that this is a mildly informative prior, when warranted by the data the posterior can concentrate on $\kappa_2 \not\in [1,3]$ outside this interval (e.g. for the DLD data in Section \ref{ssec:geneexpression_datasets} the posterior mode was $\hat{\kappa}_2=4.12$).
See Section \ref{App:NLPSpecification} for further discussion.
}

Finally, to satisfy the differentiability conditions of Theorem \ref{Thm:HBayesFactorConsistency} %(see Section \ref{App:TechnicalConditions}) 
and enable the use of standard second order optimisation software, we implemented a differentiable approximation to the indicator function in Tukey's loss (see Section \ref{App:AbsApprox}).

\subsection{The marginal \Hscore{} in $\kappa$}
\label{ssec:robustness_efficiency_tradeoff}

Our first example illustrates the properties of the \Hscore{} for calibrating the robustness-efficiency trade-off, in a setting where the data contain outliers. 
We simulated $n=500$ observations from the data-generating $g(y) = 0.9\mathcal{N}(y; 0, 1) + 0.1\mathcal{N}(y; 5, 3)$. %, where 90\% of the observations coming from a standard normal are contaminated with 10\% from a normal with outlying mean and larger variance. 
%Related to our breakdown point discussion above, in such a scenario one wishes 
The goal is to estimate the parameters of the larger component (uncontaminated data), in a manner that is robust to the presence of data from the smaller component (outliers).
We compare the estimation from the Gaussian model $f_1$, which is correctly specified for 90\% of the data, with that of the robust improper model arising from Tukey's loss $f_2$ (where $x_i = 1$ contains only the intercept term).

A first question of interest is studying the ability of the \Hscore{} to learn the cutoff hyper-parameter $\kappa_2$. To this end, we measured the evidence for different $\kappa_2$ provided by the marginal \Hscore{}
\begin{equation}
    \mathcal{H}_2(y; \kappa_2) = \int\pi_2(\theta_2) \exp\left(-\sum_{i=1}^nH(y_i; f_2(\cdot;\theta_2, \kappa_2))\right)d\theta_2.\label{Equ:MarginalHScoreTukeys}
\end{equation}
\begin{figure}[h]
\begin{center}
%\includegraphics[trim= {0.0cm 0.0cm 0.0cm 0.0cm}, clip,  
%width=0.49\columnwidth]{figures/MarginalKappa/Hyvarinen_TukeysBayesnorm_nu_marginal_absapprox_LaplaceApprox_cont_n500_bootstrap_diag_tikz-4.pdf}
\includegraphics[trim= {0.0cm 0.0cm 0.0cm 0.0cm}, clip,  
width=0.49\columnwidth]{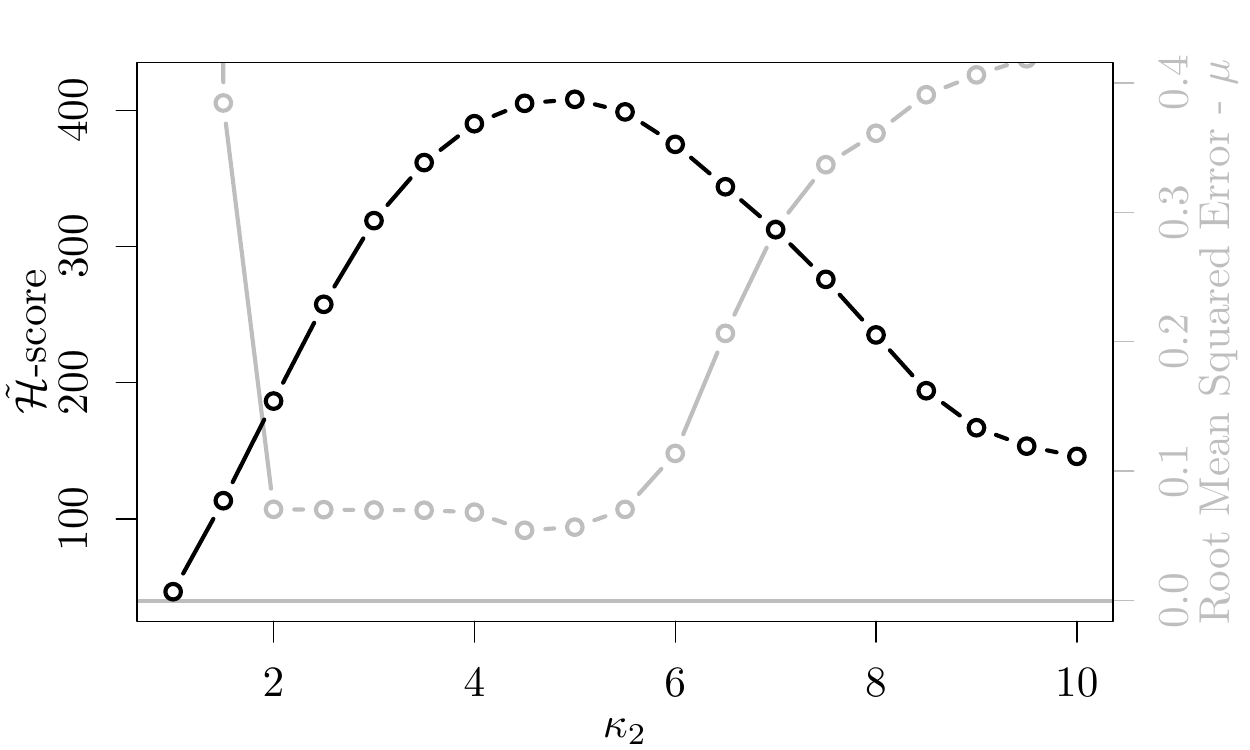}
\includegraphics[trim= {0.0cm 0.0cm 0.0cm 0.0cm}, clip,  
width=0.49\columnwidth]{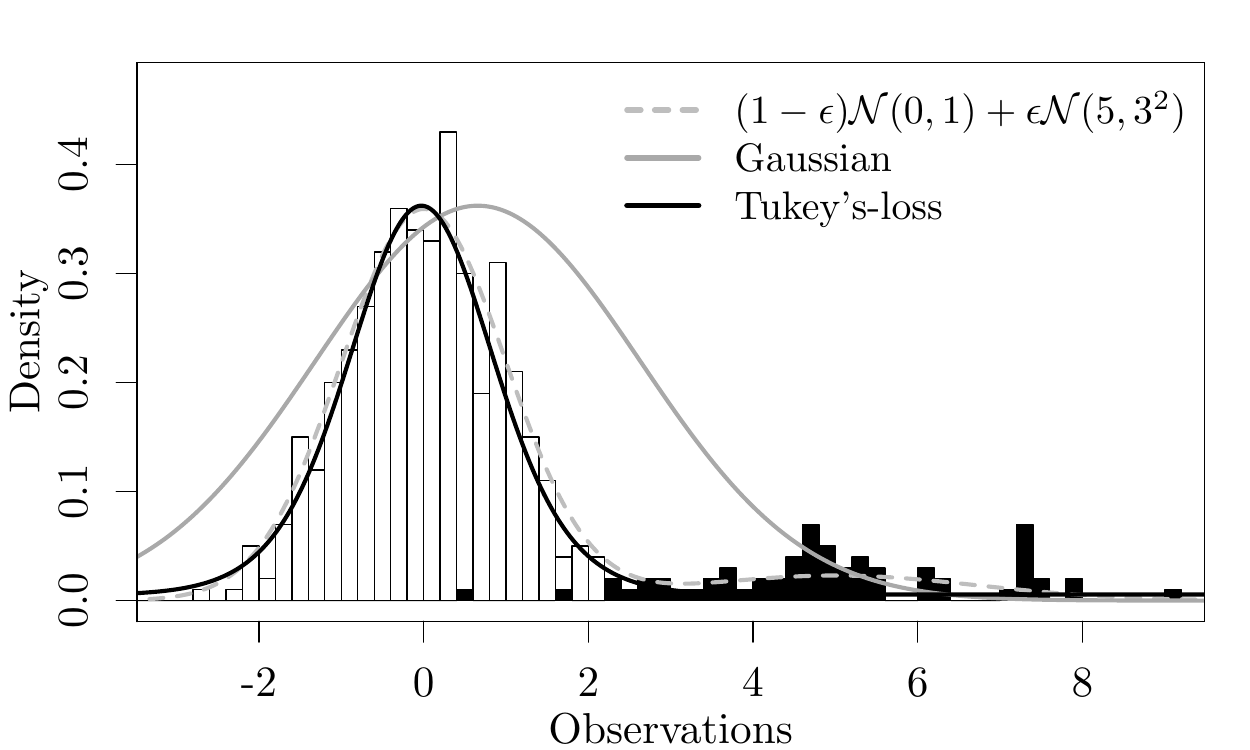}
%trim={<left> <lower> <right> <upper>}
\caption{\textbf{Left:} Marginal \Hscore{} $\mathcal{H}_2(y; \kappa_2)$ (black) %, $\mathcal{H}_1(y) = \mathcal{H}_2(y; \infty)$ (dotted) 
and asymptotic approximation to the RMSE of $\hat{\beta}(\kappa_2)$ (grey) for several $\kappa_2$. 
\textbf{Right:} Histogram of $n = 500$ observations from $g(y) = 0.9\mathcal{N}(y;0, 1) + 0.1\mathcal{N}(y;5, 3)$ and fitted Gaussian and Tukey ($\hat{\kappa}_2 = 5$) models (height set to match the mode of $g(y)$).
}
\label{Fig:MarginalHscore_epsContamination}
\end{center}
\end{figure}

The left panel in Figure \ref{Fig:MarginalHscore_epsContamination} shows $\mathcal{H}_2(y; \kappa_2)$ (black line) for a grid $\kappa_2 \in \{1, 1.5, 2, \ldots, 10\}$, along with
%We overlay the Fisher's divergence between the posterior mean fitted model and $g(y)$ to illustrate the correspondence between the empirical \hyvarinen score and the underlying target Fisher's divergence. 
%$\hat{\mu}(\kappa_2)$, the posterior mean of $\mu$ for each fixed $\kappa_2$ (obtained from $1,000$ repeat samples from $g(y)$).
\new{
an asymptotic approximation to the root mean squared error (RMSE) of $\hat{\beta}(\kappa_2)$ motivated by \cite{warwick2005choosing} (see Section \ref{App:asympt_MSE}).
The marginal \Hscore{} is highest for $\hat{\kappa}_2 = 5$, i.e. this value is chosen as best approximating the data-generating $g$.}
The right panel of Figure \ref{Fig:MarginalHscore_epsContamination} shows that Tukey's model for $\hat{\kappa}_2=5$ provides a good description of the uncontaminated component of the data, in the sense of capturing the log-gradient of $g(y)$ around the mode, and excludes most outliers.
%a histogram of the data, the data-generating $g(y)$ and the estimated Gaussian and . The latter  % of the data while the squared-error loss provides a poor correspondence everywhere. 

\new{Interestingly, the RMSE associated to $\hat{\beta}(\kappa_2)$ at $\hat{\kappa}_2=5$ was close to optimal (Figure \ref{Fig:MarginalHscore_epsContamination}, grey line). %, indicating that the estimator is on average close to the uncontaminated data mean, and exhibits small variance.
A convenient feature of the asymptotic RMSE is that we can examine the decomposed effect of the bias (due to contamination) and the variance, see Section \ref{App:asympt_MSE} for details.
For too small $\kappa_2$ the RMSE increases, since then there are more observations beyond the cutoff (viewed as outliers), increasing the variance of $\hat{\beta}(\kappa_2)$. 
%This suggests that such estimates are inefficient. 
If $\kappa_2$ is too large, then the contaminated observations lie within the cutoff, which increases the bias
%as a result of the $\epsilon$-contamination, suggesting the estimate is no longer robust. Initially it decreases away from the outliers, which we believe to be the result of the \Hscore{} still trying to exclude the outliers for increasing $\kappa_2$. 
(recall that at $\kappa_2=\infty$, $\hat{\beta}(\kappa_2)$ is the sample mean).
}

%\jack{Different objective functions that give us similar optimums}

\subsection{Non-local priors and model selection consistency}{\label{Sec:LPvsNLP}}

We demonstrate the selection consistency (Theorem \ref{Thm:HBayesFactorConsistency}) when data are truly Gaussian,
%applied to the Gaussian model and Tukey's loss, 
and that setting a non-local prior on the cutoff hyper-parameter $\kappa_2$ speeds up this selection (Corollary \ref{Thm:HscoreConsistencyNLPs}). We simulated $100$ independent data sets of sizes $n = 100$, $1,000$, $10,000$ and $100,000$ from the data-generating $g(y_i)= \mathcal{N}\left(y_i; x_i^T\beta, \sigma^2\right)$, where the first entry in $x_i \in \mathbb{R}^6$ corresponds to the intercept and the remaining entries are Gaussian with unit variances and 0.5 pairwise covariances, %$\beta$ is a 6 dimensional vector of coefficients 
$\beta = (0,0.5,1,1.5,0,0)$ and $\sigma^2 = 1$.

Recall that Tukey's loss collapses to the Gaussian model for $\kappa_2=\infty$, and otherwise adds certain flexibility by allowing one to consider an improper model.
%Recall that the Gaussian model is nested in Tukey's loss and recovered at $\kappa_2 = \infty$. %Further, Theorem \ref{Thm:HBayesFactorConsistency} proved that while the \Hscore{} can detect when the more complicated model is preferable at an exponential rate, the rate is only polynomial at detecting the simpler model.  
%Using $\kappa_2 < \infty$ includes and extra parameter in the model and causes it to be improper.
If this extra flexibility is not needed, following Occam's razor one wants to choose the Gaussian model.
%As a result, we seek to impose strong `shrinkage' towards the simpler probability model and the following 
While Theorem \ref{Thm:HBayesFactorConsistency} guarantees this to occur asymptotically, our experiments show that setting a local prior on $\kappa_2$ leads to poor performance, even for $n=100,000$. 
Specifically, we compare our default non-local prior $\pi_2^{\NLP}(\nu_2) = \mathcal{IG}(\nu_2; 4.35, 1.56)$ to a (local) half-Gaussian prior $\pi_2^{\LP}(\nu_2) \propto 1_{\nu_2\geq 0}\mathcal{N}(\nu_2; 0, 1)$, where $\nu_2 = 1/\kappa_2^2$. 
By Corollary \ref{Thm:HscoreConsistencyNLPsIG}, the log-\HBayes{} factor under the non-local prior should favor the Gaussian model at least at a $\sqrt{n}$ rate, in contrast to the local prior's $\log(n)$ rate. %provided by Theorem \ref{Thm:HBayesFactorConsistency}.

Figure \ref{Fig:Consistency_LP_vs_NLP} compares the \SMIC with our integrated \Hscore{} under the local and non-local priors. Score differences are plotted such that negative values indicate correctly selecting the Gaussian model. Firstly, there is no evidence of \SMIC being consistent as $n$ grows,
%The \SMIC also appears to be noisier than \Hscore, particularly when the sample size is small, as a result of estimating the bias correction term. 
even for $n=100,000$ the wrong model was selected $11\%$ of the time. 
%This is down to the heavy tailed nature of the asymptotic distribution of log-\hyvarinen score difference as proved in \cite{dawid2016minimum} %(which requires an adjustment to the standard $\chi^2$ distribution associated with correctly specified parametric models) which with non-negligible probability can overcome the complexity penalty imposed by the \SMIC (and also \HBayes{} factors under the \LP) even for large finite $n$. Unlike the \SMIC, 
The \Hscore{} under the local prior has a decreasing median in $n$, 
%driven by the $\frac{1}{2}\log n$ Bayes factors complexity penalty (see proof of Theorem \ref{Thm:HBayesFactorConsistency},  \eqref{eq:HBF_asymptotic}). 
but exhibits heavy tails and even for $n=100,000$ it also failed to select the Gaussian model $11\%$ of the time. Under the non-local prior, already for $n = 1,000$ the correct decision was made 99\% of the time. These experiments illustrate the benefits of non-local priors to penalise parameter values near the boundary ($1/\kappa_2^2 = 0$, in this example). 
Recall also that, as discussed in Section \ref{Sec:HscoreModelConsistency}, in general in such situations a local prior need not even attain consistency.

\begin{figure}[h!]
\begin{center}
    \includegraphics[trim= {0.0cm 0.0cm 0.0cm 0.0cm}, clip,  
    width=0.32\columnwidth]{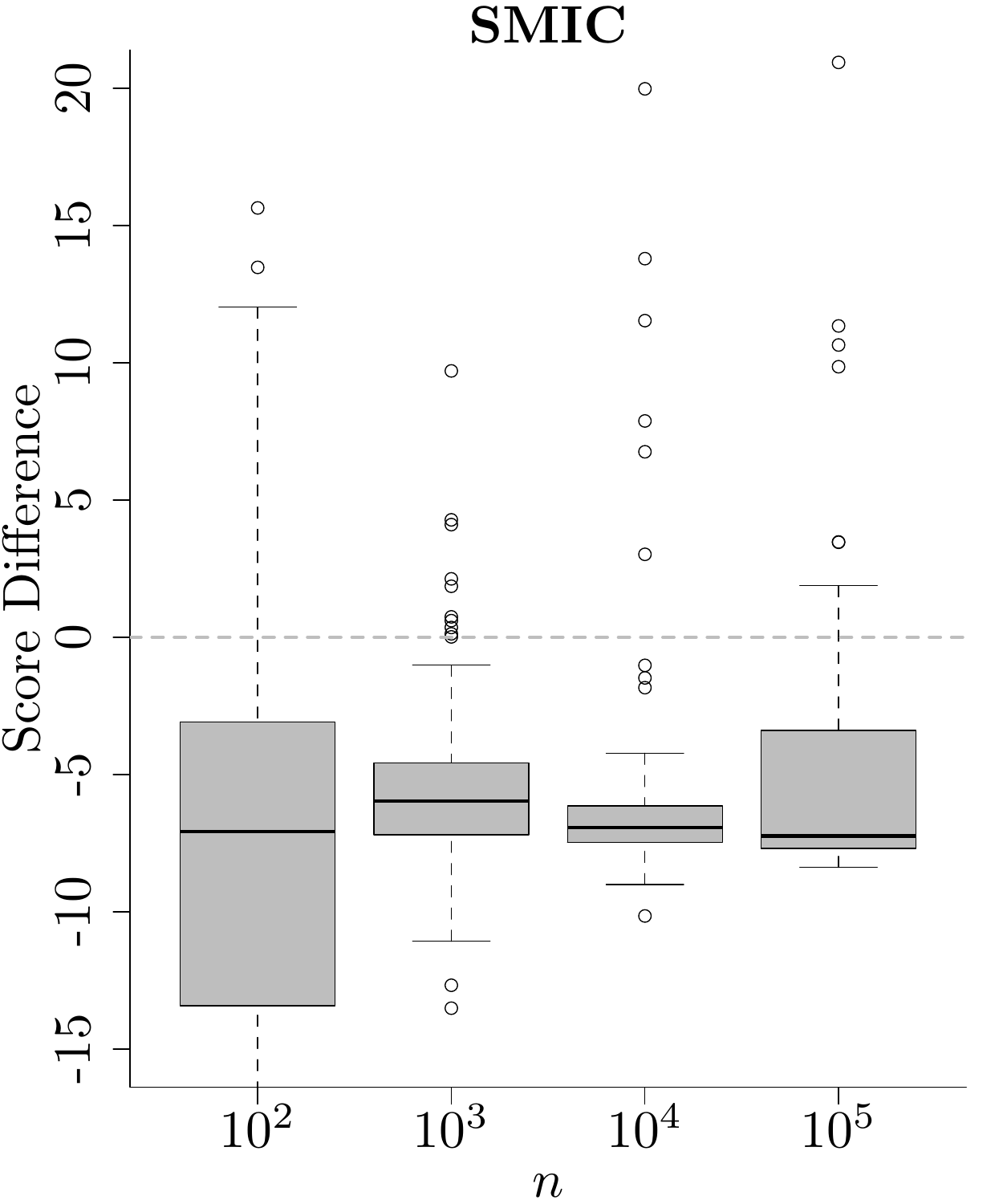}
    \includegraphics[trim= {0.0cm 0.0cm 0.0cm 0.0cm}, clip,  
    width=0.32\columnwidth]{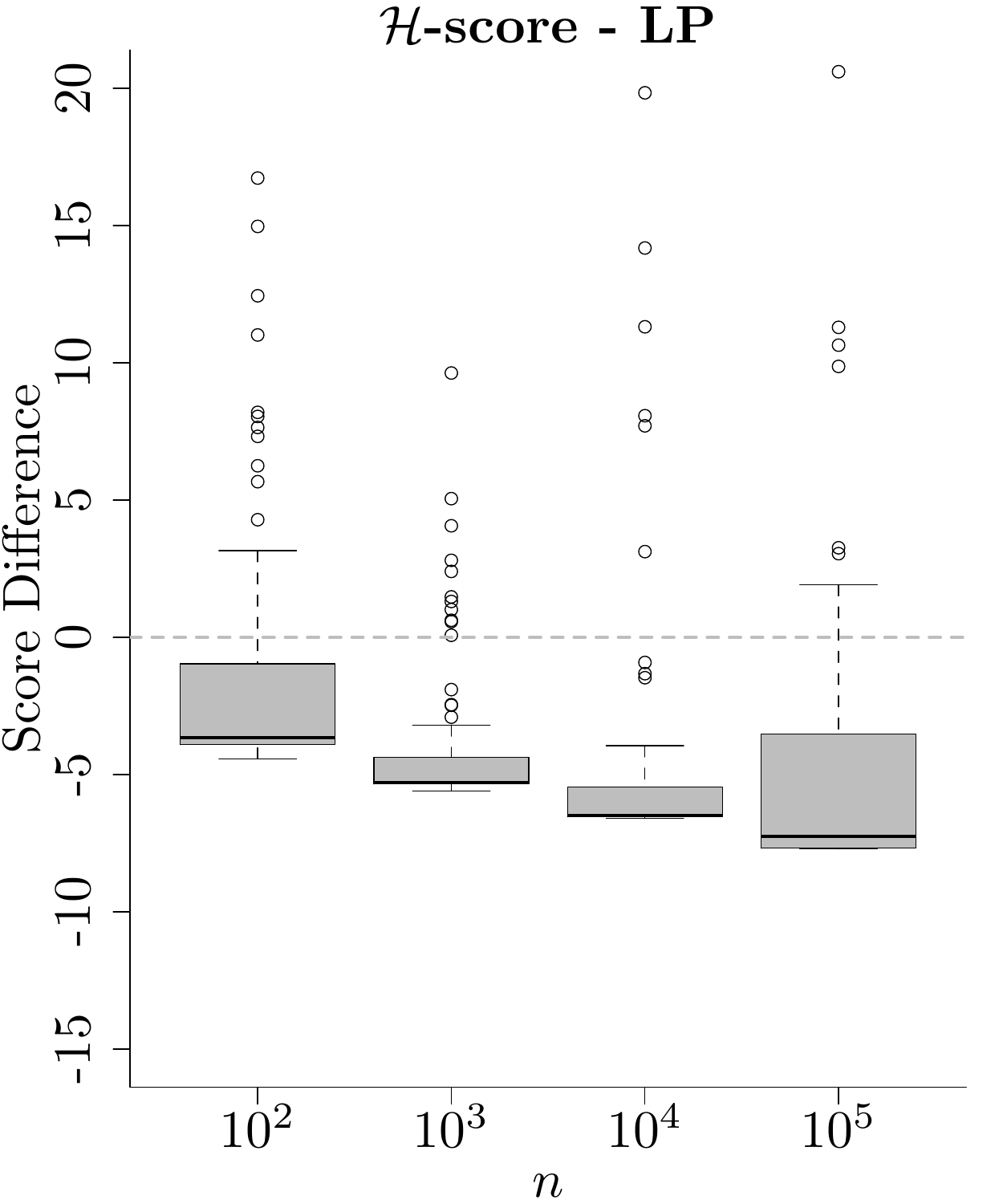}%\\
    \includegraphics[trim= {0.0cm 0.0cm 0.0cm 0.0cm}, clip,  
    width=0.32\columnwidth]{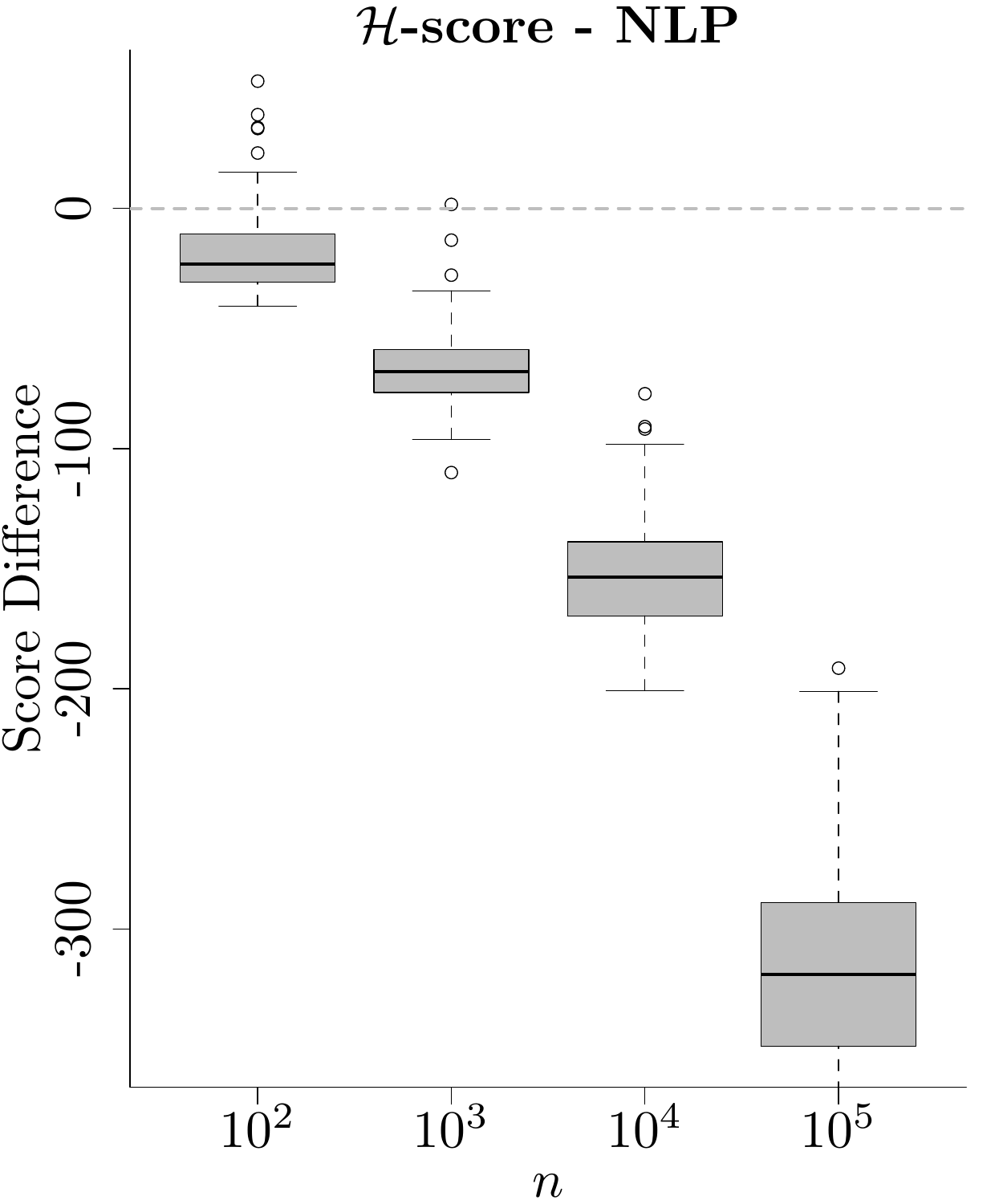}
%trim={<left> <lower> <right> <upper>}
    \caption{Selecting between the Gaussian and Tukey models in 100 truly Gaussian datasets. \SMIC ($\SMIC_1(y) - \SMIC_2(y)$) vs \Hscore{} with local and non-local priors ($\log\tilde{\mathcal{H}}_2(y) - \log\tilde{\mathcal{H}}_1(y)$). Negative values correctly select the Gaussian model.}
    \label{Fig:Consistency_LP_vs_NLP}
\end{center}
\end{figure}

\subsection{Real datasets}{\label{ssec:geneexpression_datasets}}

We considered two gene expression data sets from \cite{rossell2018tractable}. %Our focus here will not be on high-dimension variable selection, but on comparing the robustness results of our new loss function technology. 
In the first, %the consistency results of Theorem \ref{Thm:HBayesFactorConsistency} when 
the data are well-approximated by a Gaussian distribution, whereas the second exhibits thicker tails.
In both examples we used the \Hscore{} to compare the Gaussian model ($\ell_1$) and Tukey's loss ($\ell_2$).
%using the priors discussed in Section \ref{Sec:LPvsNLP}: $\sigma^2 \sim \mathcal{IG}(2.01, 0.5)$ and $\beta |\sigma^2 \sim \mathcal{N}\left(0, 5\sigma^2 I\right)$ for both models, and the non-local inverse-gamma prior $\pi_2^{\NLP}(\nu_2) = \mathcal{IG}(\nu_2; 4.35, 1.56)$.

\subsubsection{\TGFB data}{\label{TGFBanalysis}}

The dataset from \cite{calon2012dependency} concerns gene expression data for $n=262$ colon cancer patients. 
%The original data set contained $n = 262$ observations of 10,172 genes and 
Previous work \citep{rossell2017nonlocal, rossell2018tractable} focused on selecting genes that have an effect on the expression levels of \TGFB, a gene known to play an important role in colon cancer progression. %For further information see \cite{rossell2017nonlocal}.
Instead, we study the relation between \TGFB and the 7 genes (listed in Section \ref{App:SelectedVariables}) that appear in the
%The focus here is not on high-dimensional predictor selection, as a result we reduce the original 10,172 predictors down to only consider only those genes present in the data that appear in the 
`\TGFB1 pathway' according to the \textit{KEGGREST} package in \R{} \citep{tenenbaum2016keggrest}, so that $p=8$ after including the intercept. %This left 7 genes which with the addition of an intercept becomes $p = 8$, 

%selecting_loss_jointHyvarinen_LaplaceApprox_TGFB.Rmd
%\begin{figure}[h!]
%\begin{center}
%    \includegraphics[trim= {0.0cm 0.0cm 0.0cm 0.0cm}, clip,  
%    width=0.49\columnwidth]{figures/RealData/TGFB_Gaussian_Tukeys_Comparison-1.pdf}
%    \includegraphics[trim= {0.0cm 0.0cm 0.0cm 0.0cm}, clip,  
%    width=0.49\columnwidth]{figures/RealData/TGFB_Gaussian_Tukeys_Comparison-3.pdf}\\
%    \includegraphics[trim= {0.0cm 0.0cm 0.0cm 0.0cm}, clip,  
%    width=0.49\columnwidth]{figures/RealData/TGFB_Gaussian_Tukeys_Comparison-2.pdf}
%    \includegraphics[trim= {0.0cm 0.0cm 0.0cm 0.0cm}, clip,  
%    width=0.49\columnwidth]{figures/RealData/TGFB_Gaussian_Tukeys_Comparison-4.pdf}
%trim={<left> <lower> <right> <upper>}
%    \caption{\TGFB data: The \Hscore{} selected the Gaussian model.
%    \textbf{Left:} Gaussian density and Tukey's loss pseudo-density approximations to the residuals 
%    %(normalising constant set to match the mode of a \KDE of the residuals). 
%    \textbf{Top right:} The \textit{QQ}-normal plot of the fitted residuals according to the Gaussian model looks approximately normal. \textbf{Bottom right:} \MSE analysis comparing the squared difference between the \Hposterior{} mean estimates under the Gaussian model and Tukey's loss, with the increase in variance of the Tukey's loss parameter estimates (measured using $B = 500$ bootstrap re-samples)}
%    \label{Fig:TGFBresults}
%\end{center}
%
%\end{figure}

\begin{figure}[h!]
\begin{center}
%    \includegraphics[trim= {0.0cm 0.0cm 0.0cm 0.0cm}, clip,  
%    width=0.49\columnwidth]{figures/RealData/TGFB_Gaussian_Tukeys_Comparison-1.pdf}
%    \includegraphics[trim= {0.0cm 0.0cm 0.0cm 0.0cm}, clip,  
%    width=0.49\columnwidth]{figures/RealData/TGFB_Gaussian_Tukeys_Comparison-3.pdf}\\
    \includegraphics[trim= {0.0cm 0.0cm 0.0cm 0.0cm}, clip,  
    width=0.49\columnwidth]{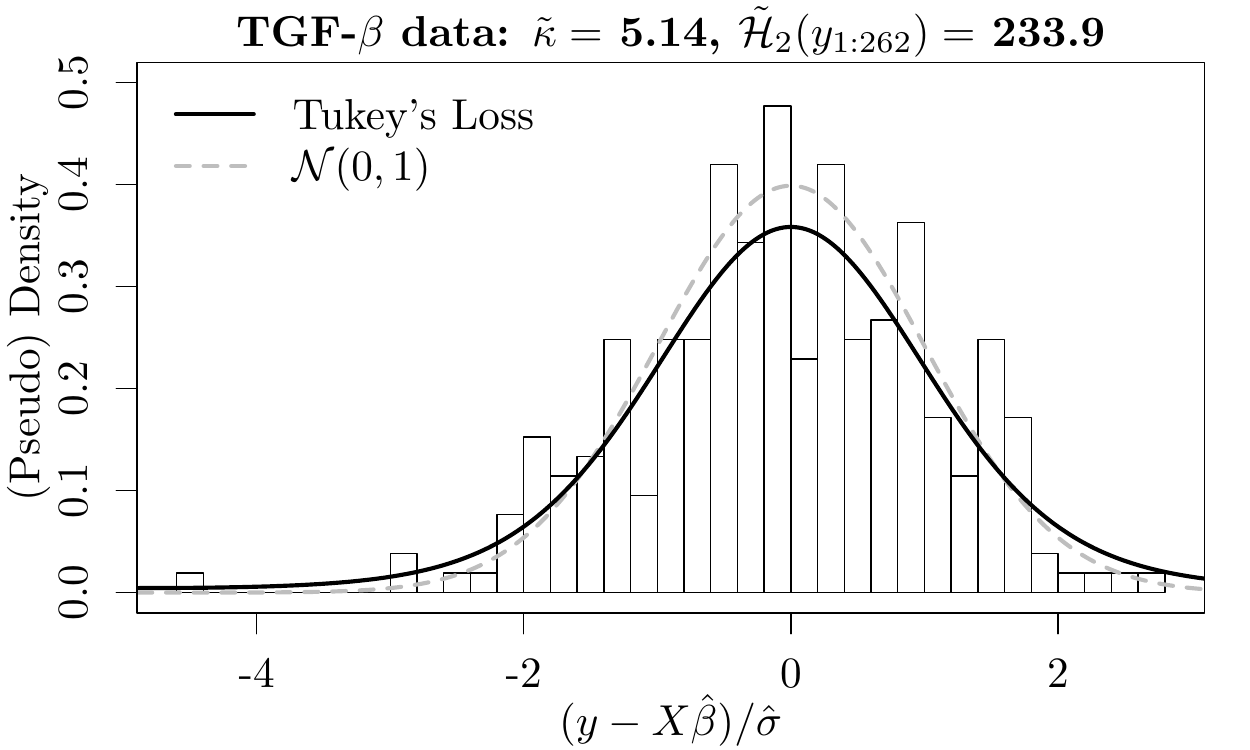}
    \includegraphics[trim= {0.0cm 0.0cm 0.0cm 0.0cm}, clip,  
    width=0.49\columnwidth]{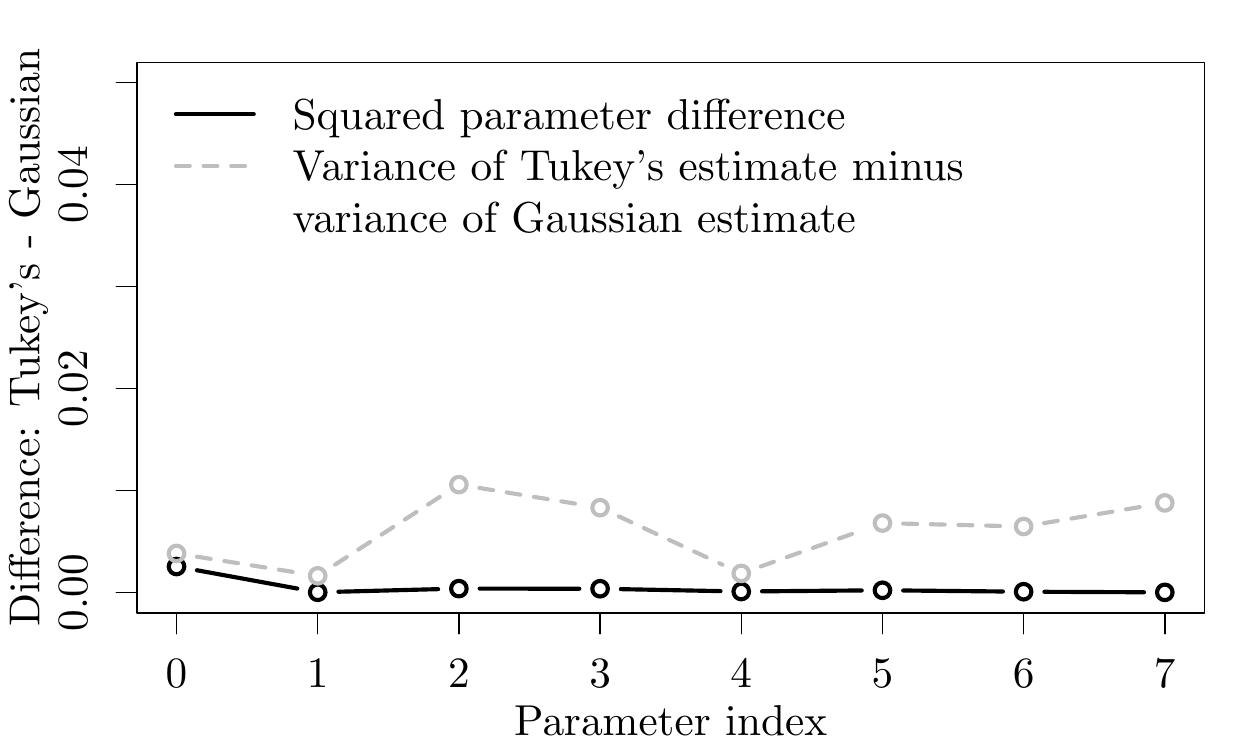}
    \includegraphics[trim= {0.0cm 0.0cm 0.0cm 0.0cm}, clip,  
    width=0.49\columnwidth]{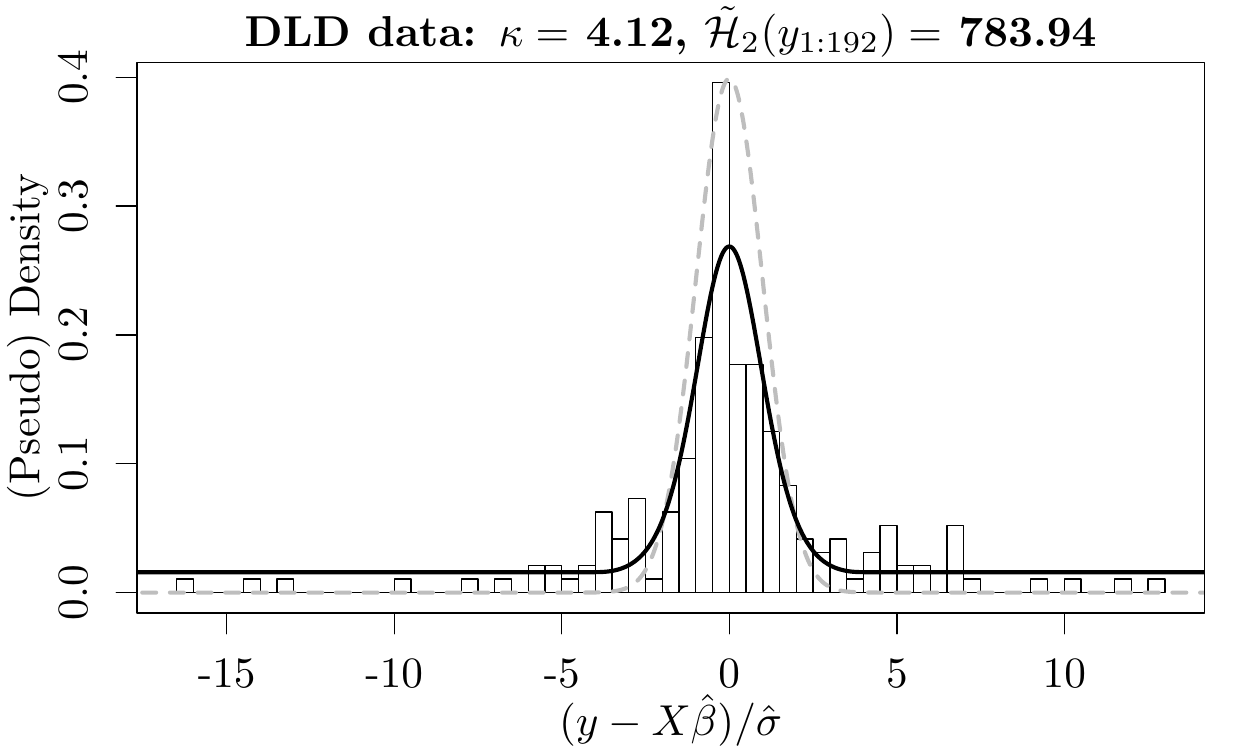}
    \includegraphics[trim= {0.0cm 0.0cm 0.0cm 0.0cm}, clip,  
    width=0.49\columnwidth]{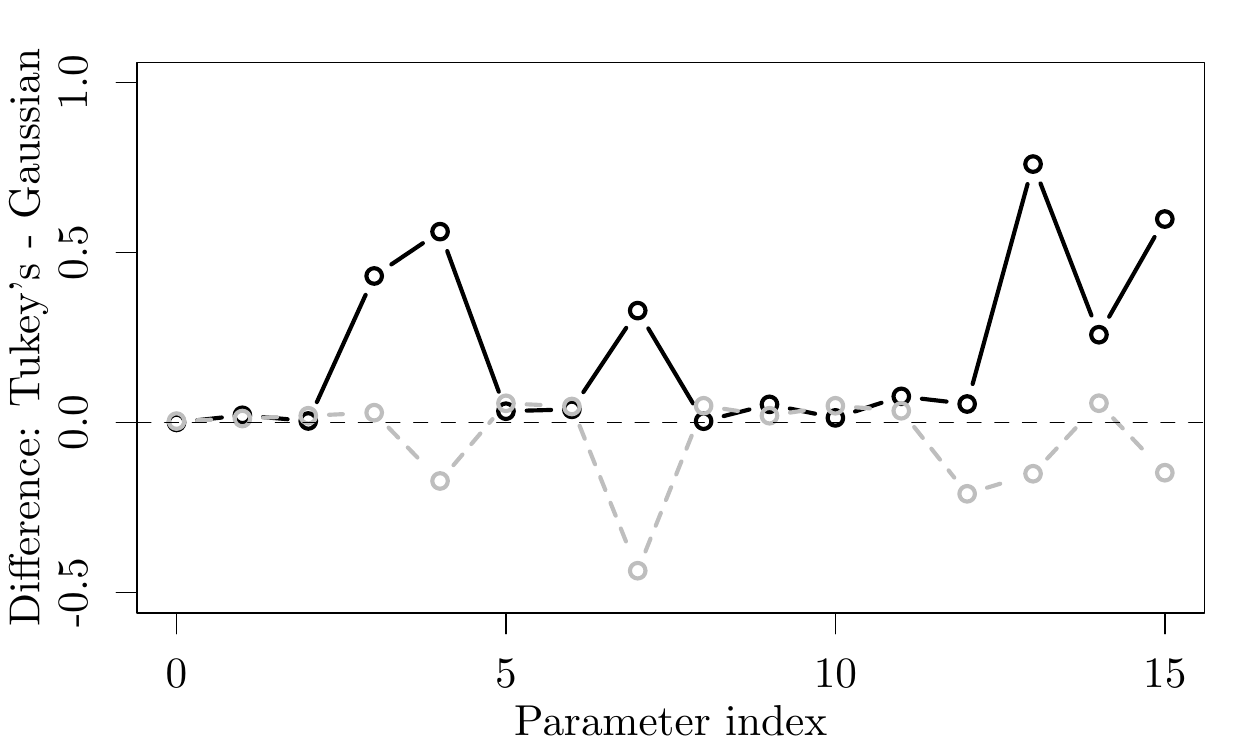}
%trim={<left> <lower> <right> <upper>}
    \caption{\textbf{Top:} \TGFB data, where the \Hscore{} selected the Gaussian model. \textbf{Bottom:} \DLD data, where the \Hscore{} selected Tukey's loss. \textbf{Left:} fitted Tukey-based density to the residuals. 
    \textbf{Right:} Squared difference between the \Hposterior{} mean estimates of each $\beta_j$ under Tukey's model minus that under the Gaussian (solid black line), and difference between their variances (estimated with $B = 500$ bootstrap re-samples).
    %\textbf{Right:} Comparison of the squared difference between the \Hposterior{} mean estimates under the Gaussian model and Tukey's loss, with the variance of the parameter estimates (measured using $B = 500$ bootstrap re-samples).
    }
    \label{Fig:TGFB_DLD_Tukeys_results}
\end{center}
\end{figure}

%\david{In Fig 5 the title "Mean squared error" doesn't make sense, please remove. The legend should read "Variance of Tukey's estimate minus variance of Gaussian estimate"} \jack{Done: is the legend ok now?}

The top panels in Figure \ref{Fig:TGFB_DLD_Tukeys_results} summarises the results.
The integrated \Hscore{} for the Gaussian model was $\tilde{\mathcal{H}}_1(y) = 272.88$ and that for Tukey's loss $\tilde{\mathcal{H}}_2(y) = 233.90$, providing strong evidence for the Gaussian model. 
%In fact, the \Hposterior{} mean estimate for the cut-off parameter was $\hat{\kappa}_2 = 5.14$. 
This is in agreement with the \SMIC ($\SMIC_1(y) = -283.93$ and $\SMIC_2(y) = -277.55$, where minimisation is desired) and results in \cite{rossell2018tractable}, who found evidence for Gaussian over (thicker) Laplace tails.  
%The non-local prior placing most prior density on values of $\kappa_2 < 3$ limits the posterior density that can be placed on large $\kappa_2$ and therefore the approximation of Tukey's loss to the residuals is inferior to that of the Gaussian. 
The left panel shows the fitted densities which, in conjunction with the \textit{Q-Q} Normal plots in Figure \ref{Fig:TGFB_DLD_Gaussian_results}, show that the residual distribution is well-approximated by a Gaussian.
%show that both models provide comparable approximations to the residuals and the \textit{QQ}-normal plot in Figure \ref{Fig:TGFB_DLD_Gaussian_results} shows that the residuals are very well approximated by a Gaussian. 
The right panel shows the squared difference between the \Hposterior{} mean parameter estimates of each $\beta_j$ under Tukey's model minus that under the Gaussian, and the differences between their sampling variances (estimated via bootstrap, dashed gray line).
Both models returned very similar point estimates, but the Gaussian had smaller variance for all parameters.
Altogether, these results strongly support that the Gaussian model should be selected over Tukey's.

%\david{In the next paragraph walk the reader over what's shown and why. Careful: do not refer to MSE since we're not really showing that.} \jack{Figure updated, need to update text}

%The bottom panel indicates that the Gaussian model would achieve better \MSE. This shows that the squared  is largely , where we measure parameter variance using bootstrap resmaples of the observed data.

%\jack{$\SMIC_1(y_{1:262}) = -283.93$ and $\SMIC_2(y_{1:262}) = -277.55$ and therefore the \SMIC agree with the \Hscore{} in selecting the Gaussian model, as information criteria are minimised}

\subsubsection{\DLD dataset}{\label{DLDanalysis}}

We consider an RNA-sequencing data set from \cite{yuan2016plasma} %. In this study RNA-sequencing %, a newer and more precise technology than microarrays,
measuring gene expression for $n = 192$ patients with different types of cancer. \cite{rossell2018tractable} studied the impact of 57 predictors on the expression of \DLD, a gene that can perform several functions such as metabolism regulation. %: it can regulate the metabolism and is associated with dehydrogenase and leukocyte adhesion deficiencies. 
To illustrate our methodology, we selected
%Similarly to the analysis in Section \ref{TGFBanalysis}, our focus here is not on high-dimensional variable selection and therefore we seek to consider only a subset of the predictors. In order to do so, we prepossess the predictors by selecting 
the 15 variables with the 5 highest loadings in the first 3 principal components, and used the integrated \Hscore{} to choose between the Gaussian and Tukey's loss. Section \ref{App:SelectedVariables} lists the selected variables.

The bottom panels of Figure \ref{Fig:TGFB_DLD_Tukeys_results} summarise the results.
The \Hscore{} strongly supported Tukey's loss ($\tilde{\mathcal{H}}_1(y) = 155.57$ for the Gaussian model, $\tilde{\mathcal{H}}_2(y) = 783.94$ for Tukey's), with \Hposterior{} mean estimate $\hat{\kappa}_2 = 4.12$.  
%The reasons for this can be seen in Figure \ref{Fig:DLDresults}. On the left hand side we can see that the fitted standardised residuals under the Gaussian model have been over-shrunk, indicating that the estimate of $\sigma^2$ was over inflated in order to account for some outliers. 
%
Indeed, the bottom left panel indicates that the residuals have thicker-than-Gaussian tails, see also the \textit{Q-Q} Normal residual plot in Figure \ref{Fig:TGFB_DLD_Gaussian_results}. %In contrast, Tukey's loss tolerates very large residuals (some with absolute value $>15$). %, to better capture the mode of the interaction between \DLD gene expression and the predictors. 
%supports the hypothesis that the fitted residuals under the Gaussian model are not normal. 
The bottom right panel illustrates two things. Firstly, the estimated coefficients of 6 of the 16 predictors differ quite considerably between the Gaussian model and Tukey's loss (solid line). Second, the latter often have smaller variance (estimated via bootstrap). Both observations align with the presence of thicker-than-normal tails, which can cause parameter estimation biases and inflated variance. 
%hypothesis that there are outliers corrupting the data. 
Notably, the \Hscore{} agrees with \cite{rossell2018tractable}, who selected Laplace over Gaussian tails, but disagrees with the \SMIC ($\SMIC_1(y_{1:192}) = -145.72$ and $\SMIC_2(y_{1:192}) = -141.67$, where minimisation is desired). 
%Further investigation shows this is a result of the bias correction term for Tukey's loss being much larger than under the Gaussian model and overwhelming the in-sample loss favouring Tukey's loss. 
We speculate that the \SMIC results may have been affected by outliers, which could cause instability in the \SMIC asymptotic bias estimation. % of estimating the many terms of the asymptotic bias  relatively few observations.

%\jack{$\SMIC_1(y_{1:192}) = -145.72$ and $\SMIC_2(y_{1:192}) = -141.67$ and therefore the \SMIC does not agree with the \Hscore{} in selecting Tukey's loss over the Gaussian model, as information criteria are minimised. Decomposing $\SMIC_1$ and $\SMIC_2$ we see that the in-sample \hyvarinen score was much lower for Tukey's loss however the bias correction favoured the Gaussian model by a large margin. Given the string evidence that outliers are present we suspect this to be cased by the instability of estimating the bias term from relatively few observations to parameters.}

\section{Kernel density estimation}{\label{Sec:KDE}}

We revisit the kernel density estimate $\hat{g}_h(\cdot)$ from Section \ref{Sec:KDEintro}. The associated loss is
\begin{align}
    \ell(y_i; y, h, w) = 
    %-\log \hat{g}_h(y_i)^{w} = 
    -w\log \hat{g}_h(y_i).\label{Equ:KDEpower_w}
\end{align}
This loss has two hyper-parameters $\kappa= (h,w)$, where $w > 0$. As discussed in Section \ref{ssec:hscore}, 
standard Bayesian inference on the bandwidth $h$ is not possible, since the loss does not define a proper probability model for the observed data $y$.
%has been included as part of the loss function, %considered previous, firstly to correct for misspecification in parametric models \cite[e.g.][]{holmes2017assigning, grunwald2017inconsistency, miller2018robust}, and for `calibrating' the a loss function against the prior in general Bayesian updating \citep{bissiri2016general, lyddon2019general}.
%Here, we simply consider $w$ to be part of the loss function 
%which, given that the density does not integrate, provides 
We also included a tempering hyper-parameter $w$ to illustrate how it can provide added flexibility and improve performance. 
%
%In Section \ref{Sec:KDEintro} we discussed that while, the estimated density $\hat{g}_{h, 1}(x) = \exp\{ - \ell(x; y, h, 1)\}$ for a future observation $x$ is a proper density, the model for the observations given the bandwidth as requires by standard Bayesian inference was improper, necessitating the use of the \hyvarinen score. 
%A downside of including $w$ however is that, when $w \neq 1$, $\hat{g}_{h, w}(x) = \exp\{ - \ell(x; y, h, w)\}$ is no longer a proper model and its normalising constant is intractable.
%A downside of including $w$ however is that, when $w \neq 1$, the estimated density $\hat{g}_{h, w}(x) = \exp\{ - \ell(x; y, h, w)\}$ for a future observation $x$ is no longer a proper model, as it was when $w = 1$ (see Section \ref{Sec:KDEintro}), and its normalising constant is intractable. 
%However, such one-dimensional integrals are easily approximated by numerical integration. 
%
%\david{A point that I asked a while back: the previous paragraph is quite confusing. Which is it: is the density improper, or does it have an intractable normalising constant? I actually think it's the latter, e.g. for $w=2$ you're taking the square of normal kernels, which should still be proper, and in fact I think you can derive its normalising constant. It's just that it's a pain to do so, and computationally costly. My suggested wording below, to keep things simple.}
%
A downside of allowing $w \neq 1$ is that the estimated density $\hat{g}_{h, w}(x) = \exp\{ - \ell(x; y, h, w)\}$ for a future observation $x$ has a convoluted normalisation constant. In our examples we found it more convenient to approximate the required univariate integrals with numerical integration.

We note that while kernel density estimation is a well-studied area, to the best of our knowledge the incorporation of a tempering parameter $w$ has not been considered, and is hence a further innovation enabled by the \Hscore{} associated to \eqref{Equ:KDEpower_w} (derived in Section \ref{App:HscoreKDE}).
\new{
%For the \Hscore{} kernel density estimate (\KDE), we set the prior $\pi(h^2) = \mathcal{IG}(h^2; a_0, b_0)$ when $w = 1$. The shape parameter was set to $a_0 = 2$ resulting in $\pi(h^2)$ with infinite variance and encoding a degree of non-informity about $h$, the scale parameter was elicited a $b_0 = 0.061$ in order to minimise the expected integrated squared error of the KDE approximation to the Gaussian density under data that was turly Gaussian, given $a_0 = 2$ (see Section \ref{App:KDEPrior}). When we estimate $w$ along side $h$ we follow the same arguments and elicit $\pi(h^2, w) = \mathcal{IG}\left(h^2; 1.09, 1.09\right)\mathcal{IG}\left(w; 1.04, \frac{0.83}{h^2}\right)$. \jack{These numbers will change}
%
%Regarding the prior on $(h,w)$, we propose a default choice that is minimally informative but is calibrated to provide accurate $\hat{g}_{h,w}$ if data were truly Gaussian. Specifically, we used the prior family $\pi(h^2, w) = \mathcal{IG}(h^2; 2, b_0)\textup{Exp}(w; \lambda_0)$, so that $h^2$ has infinite prior variance (encoding a degree of non-informativeness).
%Then we set $(b_0,\lambda_0)$ to minimise the mean integrated squared error of $\hat{g}_{h,w}$, under truly Gaussian data. This gave $b_0 = 0.024$ and $\lambda_0 = 0.725$, see Section \ref{App:KDEPrior} for further details.
Regarding the prior on $(h,w)$, we used the prior family $\pi(h^2, w) = \mathcal{IG}(h^2; 2, b_0)\textup{Exp}(w; \lambda_0)$, setting $(b_0,\lambda_0)$ to minimise the mean integrated squared error of $\hat{g}_{h,w}$ under truly Gaussian data. This results in a default $\pi(h^2, w)$ centered on $(h,w)$ values such that, if the data were truly Gaussian, the estimated density would be accurate, subject to the prior being minimally informative (e.g. $h^2$ has infinite prior variance). The elicited values were $b_0 = 0.024$ and $\lambda_0 = 0.725$. See Section \ref{App:KDEPrior} for full details. 
}

%\subsection{Estimation}
\subsection{Gaussian mixture implementations}{\label{Sec:KDEGaussianMixtures}}

We simulated data from four %estimated the underlying density of a series of 
Gaussian mixture models considered by \cite{marron1992exact}
%\footnote{\url{https://marronwebfiles.sites.oasis.unc.edu/OldResearch/parameters/nmpar.m}}. We consider four datasets here 
(with two further examples provided in Section \ref{App:KDEGaussianMictures}). % simulated from the following Gaussian mixtures
All scenarios consider a data-generating
\begin{align}
g(y)= \sum_{j=1}^J m_j \mathcal{N}(y; \mu_j, \sigma_j)
    \nonumber
\end{align}
with parameters chosen as follows.
%\david{Is it $\mu=(-1,1)$ for bimodal? Also, better write $\mu_1=-1.5, \mu_2=1.5$, $\sigma_1=\sigma_2=1/4$ etc. Also, the notation $\omega$ for the weights may lead to confusion with the tempering parameter $w$}
%\jack{no the problem here is we standardise before estimating any of the densities, maybe we could undo this when we plot?}
%\jack{check I haven't used $m$ anywehre else}
\begin{itemize}[leftmargin=*]
    \item\textit{Bimodal}: $J=2$ components,  $\mu_1 = -1.5$, $\mu_2 = 1.5$, $\sigma_1 = \sigma_2 = 1/2$ and $m_1 = m_2 = 0.5$.
    %\item[] \textit{asymmetric}: 2-components,  $\mu = (0, 1.5)$, $\sigma^2 = (1, 1/9)$ and $\omega = (0.75, 0.25)$.
    \item\textit{Trimodal}: $J=3$ components,  $\mu_1 = -1.2$, $\mu_2 = 0$, $\mu_3 = 1.2$, $\sigma_1 = \sigma_3 = 3/5$, $\sigma_2 = 1/4$, $m_1 = m_3 = 0.45$, and $m_2 = 0.1$.
    %\item\textit{Trimodal}: $J=3$ components,  $\mu = (- 1.2, 0, 1.2)$, $\sigma^2 = (9/25, 1/16, 9/25)$  and $\omega = (0.45, 0.1, 0.45)$.
    \item\textit{Claw}: $J=6$ components, $\mu_1 = 0$, $\sigma_1 = 1$ and $m_1 = 0.5$ then for $j = 2, \ldots, 6$: $\mu_j = - 2 + j/2$, $\sigma_j = 0.1$ and $m_j = 0.1$.
    %\item\textit{Claw}: $J=6$ components, $\mu = (0, -1, -0.5, 0, 0.5, 1)$, $\sigma^2_1 = 1$ and $\sigma^2_i = 0.01$ for $i = 2, \ldots, 6$ and $\omega_1 = 0.5$ and $\omega_i = 0.1$ for $i = 2, \ldots, 6$.
    %\item[] \textit{Kurtotic}: 2-components,  $\mu = (0, 0)$, $\sigma^2 = (1, 0.01)$ and $\omega = (2/3, 1/3)$.
    \item \textit{Skewed}: $J=8$ components, for $j = 1, \ldots, 8$: $\sigma_j = \left(\nicefrac{2}{3}\right)^{j - 1}$,   $\mu_j = 3(\sigma_j - 1)$ and $m_j = 1/8$.
    %\item \textit{Skewed}: $J=8$ components, $s = \{\nicefrac{2}{3}^i\}_{i=0:7}$  $\mu = 3(s - 1)$, $\sigma^2 = s^2$ and $\omega_i = 1/8$, $i = 1,\ldots, 8$.
%    \item[outlier] 2-component:  $\mu = (0, 0)$, $\sigma^2 = (1, 0.01)$ and $\omega = (0.1, 0.9)$.
\end{itemize}
%\david{Change name from Separated to Bimodal? btw, I added bullet points above} \jack{done}

%As well as implementing the $\mathcal{H}$-Bayes \KDE, we will also compare with two default competitors. The first is 
We compared the \Hscore{} estimated density when learning $h$ and $w$ from the data against the default kernel estimate in \R, which sets the bandwidth using Silverman's rule-of-thumb \citep{silverman1986density}, \new{using unbiased cross-validation \citep{rudemo1982empirical, bowman1984alternative} to set the bandwidth hyperparameter, implemented in \R's \textit{kedd} package \citep{guidoum2015kernel}}, a Bayesian non-parametric density estimate using a Dirichlet Process mixture of Gaussians (\DPMM), implemented in \R's \dirichletprocess package \citep{ross2018dirichletprocess} using their default parameters and placing a prior on the Dirichlet Process concentration parameter, \new{and a finite Bayesian Mixture model where we use the marginal likelihood to select the number of Gaussian mixture components, implemented in \R's \textit{mombf} package \citep{rossell2021package} (full details are provided in Section \ref{App:FiniteGaussianMixDetails}). 
Since the data are truly generated by a finite Gaussian mixture, the last competitor serves as a benchmark for all methods.
Additional comparisons against the \Hscore{} estimated density with fixed $w = 1$ are in Section \ref{App:KDEGaussianMictures}.}
We sampled $n = 1,000$ observations from each simulation setting above and standardised the data to zero mean and unit variance, as recommended for the \dirichletprocess package. 
%For the \DPMM we consider Gaussian mixture components and the default `uninformative' prior setting provided by the package while 

%Hyvarinen_KernelDensityEstimation_ParametricNonParametricConsistency.Rmd
\begin{figure}[h]
\begin{center}
\includegraphics[trim= {0.0cm 0.0cm 0.0cm 0.0cm}, clip,  
width=0.49\columnwidth]{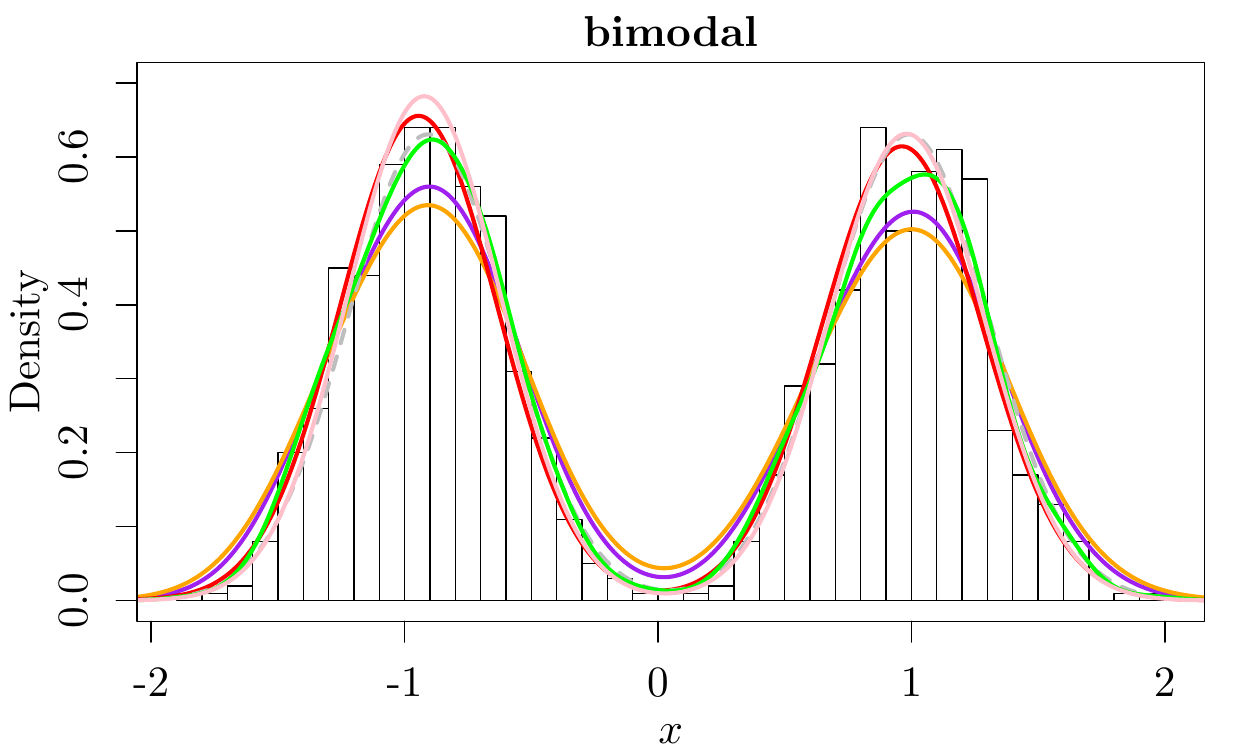}
%\includegraphics[trim= {0.0cm 0.0cm 0.0cm 0.0cm}, clip,  
%width=0.49\columnwidth]{figures/KDEplots/data_asymetric_bimodal_comparison_tikz-1.pdf}
%\includegraphics[trim= {0.0cm 0.0cm 0.0cm 0.0cm}, clip,  
%width=0.49\columnwidth]{figures/KDEplots/data_bimodal_comparison_tikz-1.pdf}
%\includegraphics[trim= {0.0cm 0.0cm 0.0cm 0.0cm}, clip,  
%width=0.49\columnwidth]{figures/KDEplots/data_outlier_comparison_tikz-1.pdf}\\
\includegraphics[trim= {0.0cm 0.0cm 0.0cm 0.0cm}, clip,  
width=0.49\columnwidth]{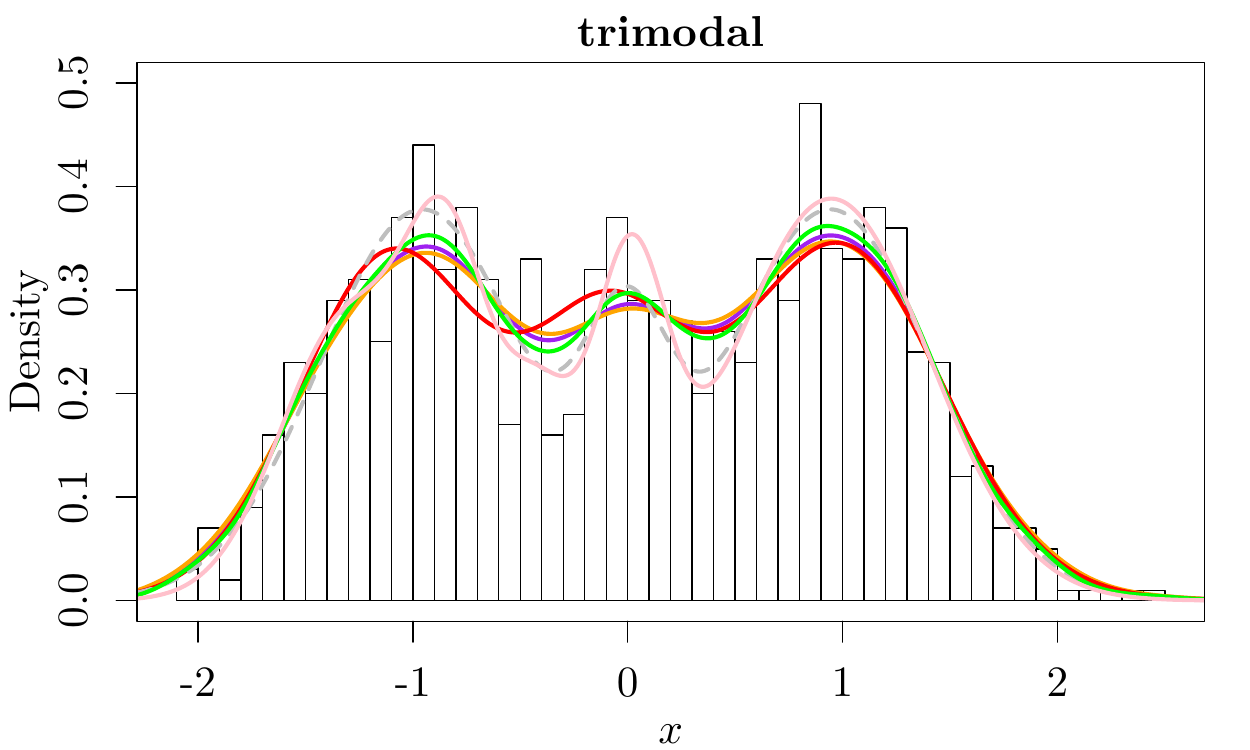}
\includegraphics[trim= {0.0cm 0.0cm 0.0cm 0.0cm}, clip,  
width=0.49\columnwidth]{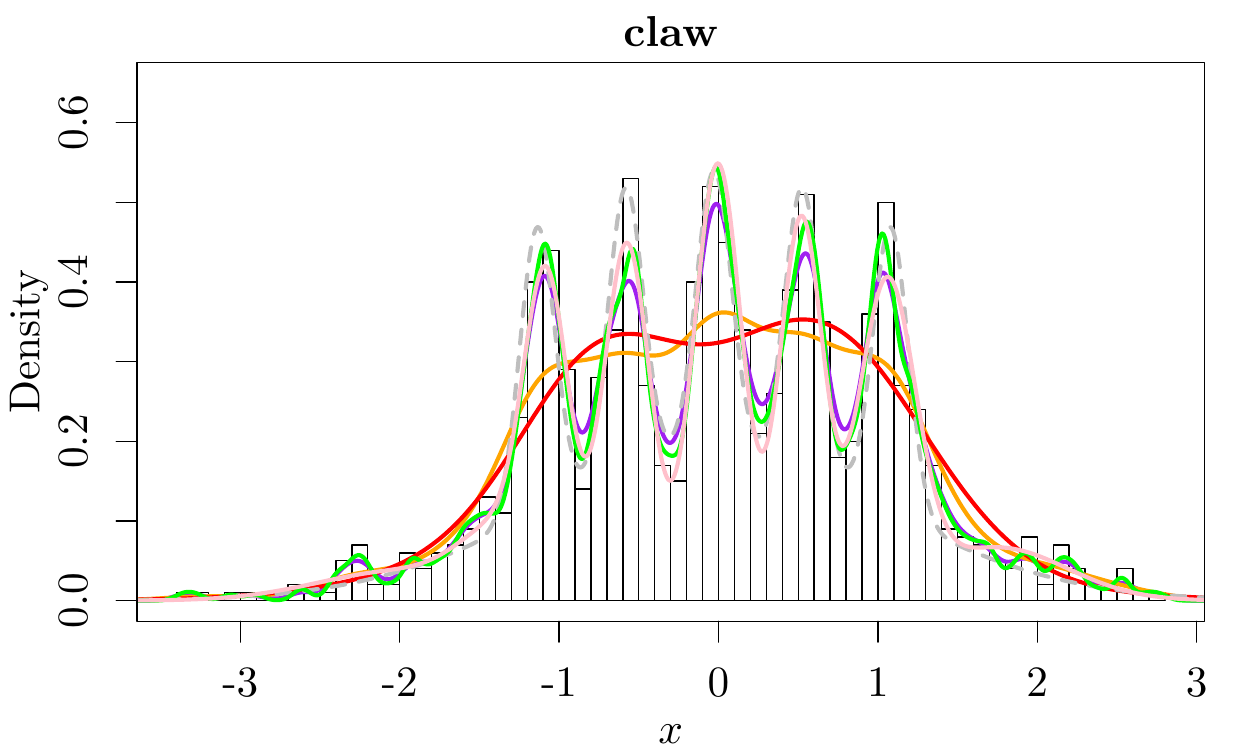}
%\includegraphics[trim= {0.0cm 0.0cm 0.0cm 0.0cm}, clip,  
%width=0.49\columnwidth]{figures/KDEplots/data_kurtotic_comparison_tikz-1.pdf}
\includegraphics[trim= {0.0cm 0.0cm 0.0cm 0.0cm}, clip,  
width=0.49\columnwidth]{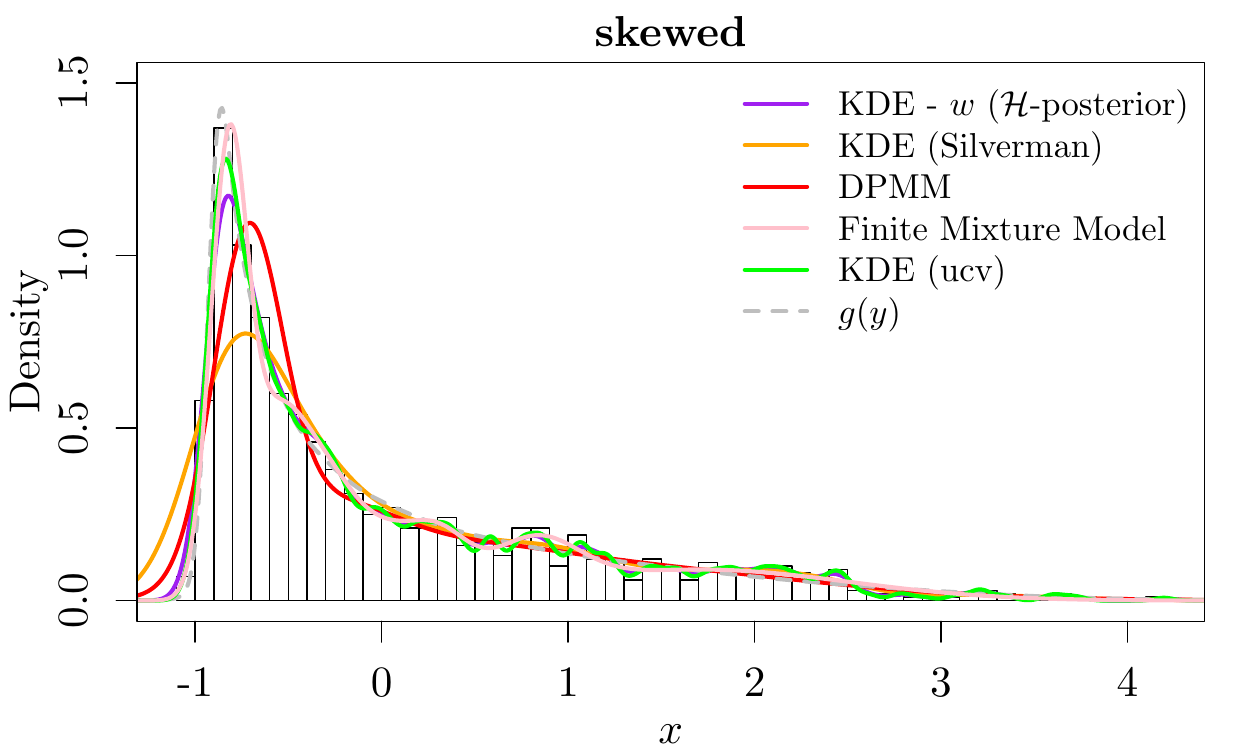}
%trim={<left> <lower> <right> <upper>}
\caption{Density estimation for Gaussian mixture data. Histograms of observed data, standardised to 0 mean and unit variance, and estimated density by  %({\color{red}{\textbf{red}}}), 
\R's density function with the bandwidth rule of \cite{silverman1986density}, unbiased cross validated bandwidth estimation, %({\color{orange}{\textbf{orange}}}), 
\DPMM, Bayesian mixtures with marginal likelihood selection of components, %, the \Hposterior{} estimate with no tempering ($w = 1$), %({\color{blue}{\textbf{blue}}}) 
and the \Hposterior{} estimate with tempering  ($w \neq 1$). For Bayesian methods the density associated with MAP estimates is plotted.%({\color{purple}{\textbf{purple}}}).
}
\label{Fig:KDEEstimation}
\end{center}
\end{figure}

Figure \ref{Fig:KDEEstimation} provides histograms of the standardised data and the estimated densities, while Table \ref{Tab:KDEEstimation} shows Fisher's divergence of the fitted models %(posterior means where appropriate) 
to the data-generating $g(y)$. 
\new{In Table \ref{Tab:KDEEstimation}, the \Hscore{} estimate was always the best or second best, with significant improvements over other methods. Its performance overall was comparable to the finite mixture model, which assumes the correct data-generating mechanism. Among the other competitors the cross-validation estimate performed best. This is also apparent in Figure \ref{Fig:KDEEstimation},}
%As is evidenced by both Figure \ref{Fig:KDEEstimation} and Table \ref{Tab:KDEEstimation} The \HBayes{} estimate while requiring the estimation of many fewer parameters performs comparably, or better than than the \DPMM in all datasets and with the exception of the \textit{trimodal} data set improves upon the standard implementation. In particular, 
%The \Hscore{} density estimate generally captured the modes better than the competing non-parametric methods. This was particularly evident 
where in the \textit{Claw} data the \DPMM and \R's kernel estimate missed the 5 modes, and to a lesser extent in the \textit{Trimodal} %, \textit{kutotic} 
and \textit{Skewed} data sets. In the \textit{Claw} data however, the cross-validation \KDE exhibited a wiggly behavior in the tails.
%\new{The \Hscore{} density estimates were further competitive with the finite mixture model, which was expected to perform well here as it contained the data generating truth. 
%However, so doing comes at the slight cost of under-smoothing the tails. This is possibly a result of the uninformative prior placed on the bandwidth, with greater generalisation and thus closeness to $g(y)$ could well be achieved under a prior that penalised under-smoothing. 
\new{Section \ref{App:KDEGaussianMictures} shows that in most examples the \Hscore{} estimate improved by learning the hyper-parameter $w$, relative to $w=1$, suggesting that the tempering parameter can add useful flexibility.} %this is something that we believe should be considered further to improve kernel density estimation. 

\begin{table}[ht]
\centering
\caption{Fisher's divergence between the density estimates and the data-generating Gaussian mixtures in four simulation scenarios. The best performing method in each is highlighted in bold face.}
\begin{tabular}{rrrrrr}
  \hline
 & \KDE & \KDE (ucv) & \DPMM & Finite Mixture Model & $\mathcal{H}$-KDE ($w \neq 1$)\\ 
   \hline
   \textit{Bimodal} & 1.03 & 0.37 & 0.13 & 0.10 & \textbf{0.09} \\ 
   \textit{Claw} & 13.77 & 6.09 & 15.25 & \textbf{2.17} & 2.51 \\ 
   \textit{Trimodal} & 0.26 & \textbf{0.12} & 0.32 & 0.33 & 0.18 \\ 
   \textit{Skewed} & 21.34 & 16.15 & 17.61 & \textbf{6.12} & 9.51 \\ 
   %\textit{Asymmetric} & 0.25 & 3.76 & 0.25 & 0.07 & 2.33 \\ 
   %\textit{Kurtotic} & 3.14 & 7.68 & 2.32 & 0.47 & 2.44 \\ 
    \hline
 \end{tabular}
\label{Tab:KDEEstimation}
\end{table}

%\subsection{Selection}

%\jack{We could consider parametric vs non-parametric selection, or selection amongst kernels as well e.g. Gaussian vs Epanechnikov}

%\jack{Might be a bit flimsy but could just take a real dataset and consider selecting between say a Gaussian and Epanechnikov function, could then point out how complicated it is to estimate the marginal likelihood of the \DPMM and how easy it is with out method.}

\section{Discussion}{\label{Sec:discussion}}

We considered a novel problem, that of selecting between probability models and possibly improper models defined via loss functions.
%provided the first principled and consistent inference procedure to select between probability models and possibly non-integrable pseudo-probability models derived from loss functions. 
Despite being non-standard, the latter can be naturally interpreted in terms of relative probabilities, and open an avenue to increase the flexibility at the disposal of the data analyst.  %rather than absolute probabilities, 
Relative probabilities motivated our use of Fisher's divergence and the \hyvarinen score, which led to an associated  \Hposterior{} and \HBayes{} factors to estimate hyper-parameters and model selection, respectively.
%to evaluate how the pseudo models perform. 
\new{Fisher's divergence satisfies two key desirata: (i) it can be applied to improper models, and (ii) the solution best approximates infinitesimal relative probabilities. Methods designed for intractable, but finite, normalization constants do not satisfy (i). The generalized Fisher divergences of \cite{lyu2009interpretation} offer an interesting basis to extend our ideas, although their proposed marginalization and the conditional expectation operators proposed do not seem directly applicable to our examples, and do not satisfy (ii). A further desideratum is being computationally tractable. The \Hscore{} leads to a closed-form posterior in exponential family models \citep{hyvarinen2007some},}
and here we used Laplace approximations for computational convenience, which we proved to recover the model closest to the data-generating truth in Fisher's divergence. 
%If a proper model is chosen, one can use a two-step procedure to ensure that inference collapses to standard Bayesian methods. % if a proper probability is the most appropriate for the data at hand. 

\new{A limitation worth mentioning of score-based methods, like Fisher's divergence, is their inability to learn the mixture weights for an assumed mixture model when the components are well-separated \citep{wenliang2020blindness}. This is not an issue in our Tukey and kernel density examples, which do not feature mixture weights to be learnt from data (in \KDE each component receives a fixed $1/n$ weight).
To illustrate this, in Section \ref{App:Hscore_multimodal} we extended the Tukey and \KDE examples to a situation where the data-generating truth features two strongly-separated components. The performance of our methodology was essentially unaffected by the larger separation.
We remark however that well-separated data would be problematic in settings such as mixture-model based regression, where one seeks to learn the mixture weights, for example.}

We illustrated how the method provides promising results in robust regression and kernel density estimation, providing a new strategy to address the robustness-efficiency trade-off for the former and enabling Bayesian inference in the latter. %The method is shown to perform excellently for robust regression and Bayesian kernel density estimation on both synthetic and real data sets.
\new{We remark however that Fisher divergence does not seek to optimize robustness nor efficiency. Rather, it defines a predictive criterion such that the chosen model best represents future observations, in terms of relative probabilities. This said, if the data-generating truth is contained in the considered models, then one selects the true model asymptotically and obtains optimal estimation. Fisher divergence was also shown by \cite{lyu2009interpretation} to possess robustness properties. It is equal to the derivative of the Kullback-Leibler divergence upon adding Gaussian noise with infinitesimal variance. Hence, Fisher divergence chooses a model that does not change much when noise is added to it.}

Our asymptotic results mimic those for Bayesian model selection in proper models. There is however an important issue when considering hyper-parameters whose optimal value lies at the boundary, where consistency may fail for standard priors and it may be necessary to use non-local priors. 
%While the current work was able to prove consistency, 
As future work, one could attempt to improve the finite sample performance of the \Hscore{} using 
%. We demonstrated the need for an non-local to overcome the long tail of the log \Hscore{} difference however 
weighting procedures to calibrate its sampling distribution, e.g. as proposed by \cite{dawid2016minimum}. %could help improve performance further. 
Another interesting direction is to extend 
%This may also prove crucial when extending 
the robust regression examples to variable selection. \cite{rossell2018tractable} showed the importance of specifying a well-fitting error distribution to avoid non-negligible drops in statistical power, particularly in high-dimensional cases. Considering robust improper models such as that stemming from Tukey's loss may improve the selection performance there.
%In such examples when $p > n$ optimising the efficiency of inference becomes very important also. 

We also provided a novel strategy to learn likelihood tempering hyper-parameters from the data, which 
%To the best of our knowledge, we considered for the first time learning the weight of a power (pseudo-)likelihood, and 
could be combined with various tempered likelihood approaches \cite[e.g.][]{grunwald2012safe, holmes2017assigning, miller2018robust} to estimate this value for standard likelihoods, \new{or to learn the Tsallis score hyper-parameter when defining a robust version of any parametric model}. 
Yet another interesting extension is to consider multivariate responses, for example in graphical or factor models where outliers can be hard to detect \cite[e.g.][]{hall2005geometric, filzmoser2008outlier} and %therefore the \Hscore{} loss selection could prove a valuable tool there. 
result in significantly worse inference. 
In summary, we believe there are many settings where the consideration of improper models can lead to fruitful research avenues.

\section*{Acknowledgements}

The authors would like to thank Piotr Zwiernik for insightful discussions about Tukey's loss minimisation and the breakdown point. JJ and DR were partially funded by the Ayudas Fundación BBVA a Equipos de Investigación Cientifica 2017 and Government of Spain's Plan Nacional PGC2018-101643-B-I00 grants. 
DR was also partially funded by the Europa Excelencia grant EUR2020-112096 and Ram\'on y Cajal Fellowship RYC-2015-18544.

% \begin{supplement}
% % \sname{Supplement A}\label{suppA} 
% % \stitle{\benni{Title of the Supplement A}}

% We provide supplemental material to the manuscript. \slink[url]{https://github.com/Beniamino92/BayesianApproxHSMM}

% \end{supplement}

\bibliography{biblio}

\newpage

\appendix

\section{Supplementary Material} 

Section \ref{ssec:proofs} contains the proofs of the theoretical results associated with Sections \ref{Sec:Methodology} and \ref{Sec:HscoreModelConsistency}  of the main paper.
Section \ref{App:derivation_hscore} provides the \Hscore{} for the squared, Tukey's and the kernel density estimation losses.
Section \ref{App:implementarion_details} describes implementation details for Tukey's loss and kernel density examples, including smooth approximations that enable the use of second order optimisation and sampling methods, and how to set a non-local prior on Tukey's cut-off hyper-parameter $\kappa_2$.
Finally, Section \ref{App:suppl_results} provides additional empirical results that complement those presented in the main paper.

\subsection{Proofs}
\label{ssec:proofs}

%\jack{In the paper I think I should add the Theorem of Lyddon ect. (don't need consistency as we have that - maybe just go straight for the asymptotic normality) saying we get asymptotic normality, caveating that obviously we do not if the parameters are on the bondary, and then we can talk about the two-step approach and our consistent model selection reducing to normal Bayes when we are on the boundary anyway.}

 The proofs are structured as follows. First we establish our notation (Section \ref{App:Notation}) in full, before stating and discussing the conditions for the theorems in Section \ref{App:TechnicalConditions}. Sections \ref{App:ParameterConsistency_op1} and \ref{App:ParameterConsistency} provide the proof of Proposition \ref{Lemma:HscoreParamConsistency}, for ease of presentation we split the proof between two separate results. Section \ref{App:ParameterConsistency_op1} provides Lemma \ref{Lemma:HscoreParamConsistency_op1} which proves consistency, $o_p(1)$, of the maximum \Hposterior{} estimates, while Section \ref{App:ParameterConsistency} provides Proposition \ref{Lemma:HscoreParamConsistency2} extending this to $O_p(\nicefrac{1}{\sqrt{n}})$ convergence. This is then used in Section \ref{App:HscoreConsistency} to prove our main result that the Laplace approximation to the $\mathcal{H}$-Bayes factor provides consistent selection amongst models and losses to that which minimises Fisher's divergence to the data generating distribution $g(y)$. Further, Section \ref{App:SMICInconsistency} contains Corollary \ref{Thm:SMICConsistency} which proves that the \SMIC of \cite{matsuda2019information} does not provide consistent model/loss selection between nested models, Corollary \ref{Thm:HscoreConsistencyNLPs} in Section \ref{App:ConsistencyNLPs} extends the \Hscore{} consistency results to situations where non-local priors are placed on the additional parameters of nested models,  and lastly Corollary \ref{Thm:HscoreConsistencyNLPsIG} extends this to the special case of the inverse-gamma non-local prior applied when selecting between a Gaussian model and Tukey's loss in Section \ref{Sec:RobustRegression}.

\subsubsection{Notation}{\label{App:Notation}}

%Consider pseudo-models $k$ and $l$ with pseudo likelihoods $f_j(y; \theta_j, \kappa_j, \omega_j) \propto \exp\left\lbrace- \ell_j(y; \theta_j, \kappa_j, \omega_j)\right\rbrace$, for $j\in\{k, l\}$. To ease notation let $\eta_j = (\theta_j, \kappa_j, \omega_j)$ and call it's dimension $d_j$. 
Consider models $k$ and $l$ with possibly improper densities $f_j(y; \theta_j, \kappa_j) \propto \exp\left\lbrace- \ell_j(y; \theta_j, \kappa_j)\right\rbrace$, for $j\in\{k, l\}$, where $\kappa_j$ already includes $w_j$ where applicable. To ease notation let $\eta_j = (\theta_j, \kappa_j)$ and call its dimension $d_j$. 
Without loss of generality assume $l$ is at least as complicated as $k$, i.e. that $d_l \geq d_k$. We obtain $\mathcal{H}$-Bayes Factor rates under the assumption that $y \sim g$. In so doing
\begin{equation}
   \eta^{\ast}_j := \argmin_{\eta_j} \mathbb{E}_g\left[H(y; f_j(\cdot; \eta_j)\right].\nonumber
\end{equation}
are the parameter minimising Fisher's divergence from model $f_j$ to $g(y)$.

Given sample $y \sim g$, define the in-sample \Hscore{} minimising estimate, $\hat{\eta}_j$, and given and parameter prior $\pi(\eta_j)$ the penalised maximum \Hposterior{} estimate, $\tilde{\eta}_j$, as 
\begin{align}
    \hat{\eta}_j := \argmin_{\eta_j} \sum_{i=1}^n H(y_i; f_j(\cdot; \eta_j), \quad \tilde{\eta}_j := \argmin_{\eta_j} \sum_{i=1}^n H(y_i; f_j(\cdot; \eta_j) - \log \pi_j(\eta_j),\label{Equ:Hscore_minimisers}
\end{align}
where throughout we suppress the dependence on $n$ to avoid cumbersome notation.
The empirical and expected Hessian matrix for model $j$ are $A_j(\eta_j)$ and $A^{\ast}_j(\eta_j)$ respectively defined as
\begin{align}
A_j\left(\eta_j\right):&= \nabla_{\eta_j}^2\left(\sum_{i=1}^nH(y_i; f_j(\cdot;\eta_j)) - \log \pi_j(\eta_j)\right),\nonumber\\
A^{\ast}_j\left(\eta_j\right) :&= \nabla_{\eta_j}^2\mathbb{E}_g\left[H(y_i; f_j(\cdot;\eta_j))\right], \label{Equ:LimitingHessian}
\end{align}
The Laplace approximation to the \HBayes{} factor is then given by
\begin{align}
\label{Equ:HBayesFactor_LaplaceApprox}
\tilde{B}^{(\mathcal{H})}_{kl} := \frac{\tilde{\mathcal{H}}_k(y)}{\tilde{\mathcal{H}}_l(y)} =
(2\pi)^{\frac{d_k-d_l}{2}}
 \frac{|A_l\left(\tilde{\eta}_l\right)|^{1/2}}{|A_k\left(\tilde{\eta}_k\right)|^{1/2}}
\frac{\pi_k\left(\tilde{\eta}_k\right)}{\pi_l\left(\tilde{\eta}_l\right)}
 e^{ \sum_{i=1}^nH\left(y_i; f_l\left(\cdot;\tilde{\eta}_l\right)\right) - \sum_{i=1}^nH\left(y_i; f_k\left(\cdot;\tilde{\eta}_k\right)\right) }
\end{align}
%\jack{Make sure all of the $B$'s are tildes}
%
For some of the results below it is sometimes easier to consider the estimating equations associated with minimising the \hyvarinen score. Define the empirical and expected vectors of first derivatives, $\Psi_j(\eta_j)$ and $\Psi_j^{\ast}(\eta_j)$ respectively, as
\begin{align}
  \Psi_j(\eta_j) :&= \nabla_{\eta_j}  \sum_{i=1}^n H(y_i; f_j(\cdot; \eta_j),\quad \phi_j(\eta_j) := \nabla_{\eta_j} \log \pi_j(\eta_j)\nonumber\\
   \tilde{\Psi}_j(\eta_j) :&= \Psi_j(\eta_j)+ \phi_j(\eta_j), \quad \Psi_j^{\ast}(\eta_j) := \nabla_{\eta_j} \mathbb{E}_g\left[H(y; f_j(\cdot; \eta_j)\right].\nonumber
\end{align}
%further note that by definition $\Psi_j^{\ast}(\eta_j^{\ast}) = 0$ and $\tilde{\Psi}_j(\tilde{\eta}_j) = 0$.

\subsubsection{Technical Conditions}{\label{App:TechnicalConditions}}

Next, we list the technical conditions required to prove the consistency of the \HBayes{} factor to the model minimising Fisher's divergence, Theorem \ref{Thm:HBayesFactorConsistency}.
\begin{itemize}
    \item[A1.] Regularity assumptions on the models $j = l, k$ and their priors

    \begin{itemize}[leftmargin=*]
        \item The parameter and hyperparameter spaces $\Theta_j \times \Phi_j$ are compact.
        \item The priors are continuous in $\eta_j$, twice continuously differentiable and $\pi_j(\eta_j) > 0$ $\forall \eta_j$.
        \item Model densities $f_j(y; \eta_j) \propto \exp\left\lbrace- \ell_j(y; \eta_j)\right\rbrace$ are twice continuously differentiable in $y$.%, ensuring that their \hyvarinen score, $H(y; f_j(\cdot; \eta_j))$, is continuous.
        %\david{No need to mention continuity of the H-score, it is implied by the twice continuously differentiable assumption below.}\jack{already there}
        
        %\item The variance of the first derivatives of the \hyvarinen score $\mathbb{V}_G\left[\Psi_j(\eta_j)\right] = \mathbb{V}_G\left[\nabla_{\eta_j} H(y; f_j(\cdot; \eta_j)\right] < \infty$. %, $\forall p\in \{1, \dots, d_j\}$.
        \item The variance of the first derivatives of the \hyvarinen score $\mathbb{V}_G\left[\Psi_j(\eta_j)\right] = \mathbb{V}_G\left[\frac{\partial}{\partial \eta_j} H(y; f_j(\cdot; \eta_j)\right] < \infty$. %, $\forall p\in \{1, \dots, d_j\}$.
        \item The \hyvarinen score, $H(y; f_j(\cdot; \eta_j))$, is twice continuously differentiable in $\eta_j$. % and $A^{\ast}_j(\eta^{\ast}_j)$ ( \eqref{Equ:LimitingHessian}) is such that $0 < |A^{\ast}_j(\eta^{\ast}_j)| < \infty$
        \item These exists a unique $\eta_j^{\ast}$ and $\tilde{\eta}_j$ for all $n$ such that the first derivatives of the empirical and expected \hyvarinen score, $\tilde{\Psi}_j(\tilde{\eta}_j) = \Psi_j^{\ast}(\eta_j^{\ast}) = 0$. %\jack{need these to be continious - maybe not}
    \end{itemize}
    \item[A2.] For model $j = l, k$, there exists a function $b_j(\cdot; z): \mathbb{R}^{d_j} \mapsto \mathbb{R}$ such that $|H(z; f_j(\cdot; \eta_j))| < b_j(\eta_j; z)$ for all $z$ where $\int b_j(\eta_j; z) d\eta_j < \infty$. %and $\mathbb{E}_z\left[b_j(\eta_j; z)] < \infty$ for all $\eta_j$.
    %\item[A2.] The expected Hessians of the \hyvarinen scores, $A^{\ast}_j\left(\eta^{\ast}_j\right)$ ( \eqref{Equ:LimitingHessian}) $j = l, k$, are finite and strictly positive definite. Additionally, defining $\lambda_{\min}^{(n)}(\tilde{\eta}_j^{\dagger(n)}) $ as the the smallest Eigen-values of $\frac{1}{n}A_j^{(n)}(\tilde{\eta}_j^{\dagger(n)})$ assume $\exists \delta_{\lambda}, M_A$ such that for all $\eta_j$ with $||\eta_j^{\ast} - \eta_j||_2 < \delta_{\lambda}$ then 
    %\begin{align}
    %    \lambda_{\min}^{(n)}(\eta_j) > M_A.\nonumber
    %\end{align}
    %\jack{Not a prob bound, we basically say for n large enough this will happen, because of the parameters, not something about this in the discussion below}
    %\jack{we could say for all $\delta_{\lambda}$ but then }
    %
    %\jack{
    %To use this we need to be able to bound the smallest and largest Eigen-values of $\frac{1}{n}A_j(\tilde{\eta}_j^{\dagger(n)})$ (= $\lambda_{\min}(\tilde{\eta}_j^{\dagger(n)})$) away from 0 for 
    %$\eta^{\dagger(n)}_j = \alpha\tilde{\eta}_j + (1-\alpha)\eta^{\ast}_j$ for $\alpha\in[0,1]^{d_j}$.
    %
    %1) Could bound for all $\eta_j$ and $x_{1:n}$, 2) bound for all $\eta_j$ and limit bound in $n$, 3) bound for $\eta_j$ in a neighbourhood of of $\eta_j^{\ast}$ and a limit bound in $n$ - I can't see how we can not have a limiting bound on $n$
    %}
    \item[A3.] For model $j = l, k$, the expected Hessian of the \hyvarinen scores \eqref{Equ:LimitingHessian} is such that 
    %\begin{itemize}
        %\item[i)] The expected Hessian of the \hyvarinen scores $A^{\ast}_j(\eta^{\ast}_j)$ \eqref{Equ:LimitingHessian} is such that 
        \begin{equation}
            0 < |A^{\ast}_j(\eta^{\ast}_j)| < \infty,\nonumber
        \end{equation}
        %\item[ii)] There exists function $b(y): \mathbb{R}^n \mapsto \mathbb{R}$ and constant $b > 0$ with $b(y) > b$ such that  for all $n$
        %\begin{equation}
        %\left|\left|\frac{1}{n}\tilde{\Psi}_j(\eta_j^{(a)}) - \frac{1}{n}\tilde{\Psi}_j(\eta_j^{(b)})\right|\right|_2 \geq b(y)||\eta_j^{(a)} - \eta_j^{(b)}||,\quad \forall \eta_j^{(a)}, \eta_j^{(b)}, \nonumber
        %\end{equation}
        %\item[ii)] There exist constants, $M_{-}, M_{+}$ and $\delta_A > 0$ such that
        %and there exist constants, $M_{-}, M_{+}$ and $\delta_A > 0$
        and there exist constant, $M_{A}$ and $\delta_A > 0$
        \begin{align}
	        0 < M_{A} \leq |\frac{1}{n}A_j(\eta_j)|, \textrm{ for all } \eta_j \textrm{ such that} ||\eta_j - \eta_j^{\ast}||_2 < \delta,\nonumber
        \end{align}
        with probability tending to 1 as $n \rightarrow\infty$.
        %\jack{for all $n$, I guess we only actually invoke this fro large $n$ anyway when}
        %\jack{I think we could easily have a $n > N$}
        %\jack{But we need this to hold for $\tilde{\eta}_j$ for the Laplace approximation to be defined anyway}
        %\item[iii)]  There exist functions $m_A(\cdot)$ with $\mathbb{E}\left[m_A(\cdot)\right] < \infty$ such that  the following Lipschitz conditions hold
        %\begin{align}
        % \left|\left|A^{(1)}_j(\eta_j^{(a)}) - A^{(1)}_j(\eta_j^{(b)})\right|\right|_2 \leq m_A(z)\left|\left|\eta_j^{(1)} - \eta_j^{(b)}\right|\right|_2, \quad \forall{}\eta_j^{(1)}, \eta_j^{(b)},\nonumber
        %\end{align}
        %for all $z \in \mathbb{R}$, where the notation $A^{(1)}_j(\cdot)$ emphasises that only one observation is involved.
    %\end{itemize} 
    \item[A4.] For model $j = l, k$ there exist functions $m_H: \mathbb{R}\mapsto \mathbb{R}_{\geq 0}$ and $m_A(\cdot): \mathbb{R}\mapsto \mathbb{R}_{\geq 0}$ with $\mathbb{E}\left[m_H(\cdot)\right] < \infty$ and $\mathbb{E}\left[m_A(\cdot)\right] < \infty$ so that  the following Lipschitz conditions hold
    \begin{align}
         \left|H(z; f_j(\cdot; \eta_j^{(a)})) - H(z; f_j(\cdot; \eta_j^{(b)}))\right| \leq m_H(z)\left|\left|\eta_j^{(a)} - \eta_j^{(b)}\right|\right|_2, \quad \forall \eta_j^{(a)}, \eta_j^{(b)}, \nonumber\\
         \left|\left|A^{(1)}_j(\eta_j^{(a)}) - A^{(1)}_j(\eta_j^{(b)})\right|\right|_2 \leq m_A(z)\left|\left|\eta_j^{(1)} - \eta_j^{(b)}\right|\right|_2, \quad \forall \eta_j^{(1)}, \eta_j^{(b)},\nonumber
    \end{align}
    for all $z \in \mathbb{R}$, where the notation $A^{(1)}_j(\cdot)$ emphasises that only one observation, $z$, is involved.
\end{itemize}

These conditions are mild and standard. A1 constitute standard assumptions required to ensure the \hyvarinen score applied to model/loss $f_j/\ell_j$, $j = l, k$ is itself continuous and has continuous first and second derivatives (in parameters, e.g. condition S1 \cite{matsuda2019information}), as well as a ensuring the existence and uniqueness of a \hyvarinen score minimising parameters and that its first derivative has finite variance. \new{Note that we will relax the strict positivity of $\pi_j(\eta_j)$ in Section \ref{App:ConsistencyNLPs} in order to consider non-local priors}. A2 requires the existence of an function dominating the \hyvarinen score applied to loss $j = k, l$ \new{ensuring absolute integrability w.r.t $\eta_j$}. Note that e.g. for Tukey's loss the improper-model itself is not dominated by an integrable function, but both the%\jack{
%\cite{dawid2016minimum} cite \cite[9.2.2][]{molenberghs2006models} and \cite{huber1967behavior} (as well as one further paper I can't access easily)
%
%\cite{molenberghs2006models} bound the expected Hessian away from 0 and less than $\infty$, and they also bound the third derivative.
%\cite[Section 4, Assuption N]{huber1967behavior} have some lipschitz looking bounds! But no I don't think these are useful!
%
%Be good to have a reference for A1 being standard for Hyvarinnen as well! - I have put \citep[e.g. condition S1][]{matsuda2019information} or $C^3$, could also look at David's \cite{dawid:1999} reference}
%}

%\jack{Another option would be to just use what \cite[Theorem 5.21]{van2000asymptotic} has instead of A2.}

\subsubsection{Lemma \ref{Lemma:HscoreParamConsistency_op1} (\hyvarinen score parameter consistency)}{\label{App:ParameterConsistency_op1}}

Before proving Theorem \ref{Thm:HBayesFactorConsistency} it is useful to first prove Proposition \ref{Lemma:HscoreParamConsistency}, providing the consistency of the \hyvarinen score minimising parameters estimates to the Fisher's divergence minimising parameters. For clarity of presentation here we split Proposition \ref{Lemma:HscoreParamConsistency} into Lemma \ref{Lemma:HscoreParamConsistency_op1} proving $o_p(1)$ and Proposition \ref{Lemma:HscoreParamConsistency2} proving $O_p(\nicefrac{1}{\sqrt{n}})$. Here we first prove Lemma \ref{Lemma:HscoreParamConsistency_op1}. We note \cite{hyvarinen2005estimation} (Corollary 3) provided a similar result except this further required that $g = f_j(\cdot;\theta_j^{\ast})$ which we do not require.

\begin{lemma}[\hyvarinen score parameter consistency]
Assume Conditions A1-A2. Given a sample $y\sim g$, then as $n\rightarrow\infty$ the maximum \Hposterior{} parameters of model $j$ have the following asymptotic behaviour
\begin{equation}
  \left|\left|\tilde{\eta}_j - \eta_j^{\ast}\right|\right|_2 = o_p(1),
\end{equation}
where $||\cdot||_2$ is the $L_2$-norm.
\label{Lemma:HscoreParamConsistency_op1}
\end{lemma}

%\begin{proof}
%
%Firstly, $||\tilde{\eta}_j - \hat{\eta}_j||_2 = o_p(1)$ as by A1 $\pi_j(\eta_j) > 0$ for all $\eta_j$. \jack{PROVE - mean value theorem - assumptions continuity }
%
%Then, Theorem 5.7 of \cite{van2000asymptotic} provides general conditions for $||\hat{\eta}_j - \eta_j^{\ast}|| = o_p(1)$ of any \textit{M}-estimator, of which the the in-sample \hyvarinen score is an example. As is discussed, \cite[p. 46]{van2000asymptotic} sufficient conditions for $o_p(1)$ convergence are, compactness of the set $\Theta_j \times \Phi_j$ (A1), uniqueness of $\eta_j^{\ast}$ (A1), continuity of $\eta_j \mapsto H(z; f_j(\cdot; \eta_j))$ for all $z$ (A1), and the dominance of $H(z; f_j(\cdot; \eta_j))$ by an integrable function (A2). As a result, under our conditions $\left|\left|\tilde{\eta}_j - \eta_j^{\ast}\right|\right|_2 = o_p(1)$. 
%\end{proof}

\begin{proof}
\color{black}
Theorem 5.7 of \cite{van2000asymptotic} provides general conditions for the convergence in probability of any \textit{M}-estimator, of which $\tilde{\eta}_j$ minimising the the in-sample \hyvarinen score plus prior penalty is an example, to its optimum, $\eta_j^{\ast}$. The theorem requires that for every $\epsilon$ 
\begin{equation}
    \underset{\eta_j: ||\eta_j - \eta_j^{\ast}||_2 \geq \epsilon}{\sup} \mathbb{E}_g\left[-H(y_1; f_j(\cdot;\eta_j)\right] < \mathbb{E}_g\left[-H(y_1; f_j(\cdot;\eta^{\ast}_j)\right],\nonumber
\end{equation}
which is provided by the uniqueness of $\eta_j^{\ast}$ assumed in A1, that $\tilde{\eta}_j$ satisfies 
\begin{align}
	%M_n(\tilde{\eta}_j) &\geq M_n(\eta_j^{\ast})\\
	%\Rightarrow  M_n(\tilde{\eta}_j) &\geq M_n(\theta_0) - o_p(1)\\
	%\frac{1}{n}\log \pi_j(\tilde{\eta}_j) - \frac{1}{n}\sum_{i=1}^n H(y_i; f_j(\cdot; \tilde{\eta}_j) &\geq \frac{1}{n}\log \pi_j(\eta_j^{\ast}) - \frac{1}{n}\sum_{i=1}^n H(y_i; f_j(\cdot; \eta_j^{\ast})\nonumber\\
	%\Rightarrow  	
	\frac{1}{n}\log \pi_j(\tilde{\eta}_j) - \frac{1}{n}\sum_{i=1}^n H(y_i; f_j(\cdot; \tilde{\eta}_j) &\geq \frac{1}{n}\log \pi_j(\eta_j^{\ast}) - \frac{1}{n}\sum_{i=1}^n H(y_i; f_j(\cdot; \eta_j^{\ast}) - o_p(1),\nonumber
\end{align}
which is does by the definition of $\tilde{\eta}_j$ in \eqref{Equ:Hscore_minimisers} of the paper, and that $\frac{1}{n}\log \pi_j(\eta_j) - \frac{1}{n}\sum_{i=1}^n H(y_i; f_j(\cdot; \eta_j)$ converges uniformly to $\mathbb{E}\left[-H(y_1; f_j(\cdot; \eta_j)\right]$. Using the triangle inequality we can decompose the required uniform convergence as
\begin{align}
&\sup_{\eta_j}\left|\frac{1}{n}\log \pi_j(\eta_j) - \frac{1}{n}\sum_{i=1}^n H(y_i; f_j(\cdot; \eta_j) - \mathbb{E}\left[-H(y_1; f_j(\cdot; \eta_j)\right]\right|\nonumber\\%\overset{P}{\rightarrow} 0\\
\leq&\frac{1}{n}\sup_{\eta_j}\left|\log \pi_j(\eta_j)\right| + \sup_{\eta_j}\left| \mathbb{E}\left[H(y_1; f_j(\cdot; \eta_j)\right] - \frac{1}{n}\sum_{i=1}^n H(y_i; f_j(\cdot; \eta_j)\right|\nonumber.
\end{align}
Condition A1 ensures that $\pi_j(\eta_j) > 0$ which in turn ensures that $\frac{1}{n}\sup_{\eta_j}\left|\log \pi_j(\eta_j)\right| = o_p(1)$. Lastly $\sup_{\eta_j}\left| \mathbb{E}\left[H(y_1; f_j(\cdot; \eta_j)\right] - \frac{1}{n}\sum_{i=1}^n H(y_i; f_j(\cdot; \eta_j)\right| = o_P(1)$ if $\{H(y_i; f_j(\cdot; \eta_j)\}$ are Glivenko-Cantelli. As is discussed, \cite[p. 46]{van2000asymptotic} sufficient conditions for $o_p(1)$ convergence are, compactness of the set $\Theta_j \times \Phi_j$ (A1), uniqueness of $\eta_j^{\ast}$ (A1), continuity of $\eta_j \mapsto H(z; f_j(\cdot; \eta_j))$ for all $z$ (A1), and the dominance of $H(z; f_j(\cdot; \eta_j))$ by an integrable function (A2). As a result, under our conditions $\left|\left|\tilde{\eta}_j - \eta_j^{\ast}\right|\right|_2 = o_p(1)$. 
\color{black}
\end{proof}

\subsubsection{Proposition \ref{Lemma:HscoreParamConsistency} (\hyvarinen score parameter consistency at rate $O_p(\nicefrac{1}{\sqrt{n}})$)}{\label{App:ParameterConsistency}}

To complete the proof of Proposition \ref{Lemma:HscoreParamConsistency} we extend Lemma \ref{Lemma:HscoreParamConsistency_op1} to establish that under Condition A3 also, the convergence of the \hyvarinen score minimising parameters happens at rate $O_p(\nicefrac{1}{\sqrt{n}})$.

\begin{proposition}[\hyvarinen score parameter consistency at rate $O_p(\nicefrac{1}{\sqrt{n}})$]
Assume Conditions A1-A3. Given a sample $y\sim g$, then as $n\rightarrow\infty$ the maximum \Hposterior{} parameters of model $j$ have the following asymptotic behaviour
\begin{equation}
  \left|\left|\tilde{\eta}_j - \eta_j^{\ast}\right|\right|_2 = O_p(\nicefrac{1}{\sqrt{n}}),
\end{equation}
where $||\cdot||_2$ is the $L_2$-norm.
\label{Lemma:HscoreParamConsistency2}
\end{proposition}

\begin{proof}
We achieve this by first proving that the gradient of the penalised empirical \hyvarinen loss function $\frac{1}{n}\tilde{\Psi}_j(\eta_j)$ converges to the gradient of the expected \hyvarinen loss function $\Psi_j^{\ast}(\eta_j)$ at rate $O_p(\nicefrac{1}{\sqrt{n}})$ (Step 1). We then use this and A3 on the Hessian matrix of the \hyvarinen score to extend the $O_p(\nicefrac{1}{\sqrt{n}})$ convergence of the gradients to that of the parameters (Step 2).  

\noindent\textbf{Step 1: Proof that $\left|\left|\frac{1}{n}\tilde{\Psi}_j(\eta_j) - \Psi_j^{\ast}(\eta_j)\right|\right|_2 = O_p(\nicefrac{1}{\sqrt{n}})$}

To do so firstly note that by the triangle inequality of norm $||\cdot||_2$
\begin{align}
    %&\left|\left|\sqrt{n}\left(\frac{1}{n}\Psi_j(\eta_j)  - \frac{1}{n}\phi_j(\eta_j) - \Psi_j^{\ast}(\eta_j)\right)\right|\right|_2 \leq \left|\left|\sqrt{n}\left(\frac{1}{n}\Psi_j(\eta_j) - \Psi_j^{\ast}(\eta_j)\right)\right|\right|_2 + \left|\left|\sqrt{n}\left(\frac{1}{n}\phi_j(\eta_j)\right)\right|\right|_2\nonumber\\
    %\Rightarrow&\left\lbrace \left|\left|\sqrt{n}\left(\frac{1}{n}\Psi_j(\eta_j)  - \frac{1}{n}\phi_j(\eta_j) - \Psi_j^{\ast}(\eta_j)\right)\right|\right|_2 > M \right\rbrace\nonumber\\
    %&\subseteq \left\lbrace \left|\left|\sqrt{n}\left(\frac{1}{n}\Psi_j(\eta_j) - \Psi_j^{\ast}(\eta_j)\right)\right|\right|_2 + \left|\left|\sqrt{n}\left(\frac{1}{n}\phi_j(\eta_j)\right)\right|\right|_2 > M \right\rbrace\nonumber\\
    & P\left(\left|\left|\sqrt{n}\left(\frac{1}{n}\Psi_j(\eta_j)  - \frac{1}{n}\phi_j(\eta_j) - \Psi_j^{\ast}(\eta_j)\right)\right|\right|_2 > M\right)\nonumber\\
    &\leq P\left( \left|\left|\sqrt{n}\left(\frac{1}{n}\Psi_j(\eta_j) - \Psi_j^{\ast}(\eta_j)\right)\right|\right|_2 + \left|\left|\sqrt{n}\left(\frac{1}{n}\phi_j(\eta_j)\right)\right|\right|_2 > M\right)\nonumber
\end{align}
for any $M > 0$ and under the assumption that $\pi_j(\theta_j) > 0$ (A1). Now for $O_p(\nicefrac{1}{\sqrt{n}})$ convergence we are interested in the behaviour of this probability for sufficiently large $n > N$ and therefore,
\begin{align}
    &P\left( \left|\left|\sqrt{n}\left(\frac{1}{n}\Psi_j(\eta_j) - \Psi_j^{\ast}(\eta_j)\right)\right|\right|_2 + \left|\left|\sqrt{n}\left(\frac{1}{n}\phi_j(\eta_j)\right)\right|\right|_2 > M\right)\nonumber\\
    =& P\left( \left|\left|\sqrt{n}\left(\frac{1}{n}\Psi_j(\eta_j) - \Psi_j^{\ast}(\eta_j)\right)\right|\right|_2 > M - \left|\left|\frac{1}{\sqrt{n}}\phi_j(\eta_j)\right|\right|_2\right)\nonumber\\
    \leq& P\left( \left|\left|\sqrt{n}\left(\frac{1}{n}\Psi_j(\eta_j) - \Psi_j^{\ast}(\eta_j)\right)\right|\right|_2 > M^{\prime}(\eta_j)\right)\quad \forall n \geq N,\nonumber
\end{align}
where $M^{\prime}(\eta_j) := M - \left|\frac{1}{\sqrt{N}}\phi_j(\eta_j)\right|$. Next note that vector valued,
\begin{align}
    \mathbb{E}_g\left[\sqrt{n}\left(\frac{1}{n}\Psi_j(\eta_j) - \Psi_j^{\ast}(\eta_j)\right)\right] &= 0\nonumber\\
    \mathbb{V}_G\left[\sqrt{n}\left(\frac{1}{n}\Psi_j(\eta_j) - \Psi_j^{\ast}(\eta_j)\right)\right] &= \mathbb{V}_G\left[\Psi_j(\eta_j)\right],\nonumber
\end{align}
where $\sum \mathbb{V}_G\left[\Psi_j(\eta_j)\right] < \infty$ (A1). 
Then we can apply the $L_2$-norm vector extension of Chebyshev's inequality \citep{ferentios1982tcebycheff} to vector valued $X_n = \sqrt{n}\left(\frac{1}{n}\Psi_j(\eta_j) - \Psi_j^{\ast}(\eta_j)\right)$ which provides that 
\begin{align}
    P\left( \left|\left|\sqrt{n}\left(\frac{1}{n}\Psi_j(\eta_j) - \Psi_j^{\ast}(\eta_j)\right)\right|\right|_2 > M^{\prime}(\eta_j)\right)  < \frac{\sum \mathbb{V}_G\left[\Psi_j(\eta_j)\right]}{\left(M^{\prime}(\eta_j)\right)^2}\nonumber
\end{align}
which proves that 
\begin{equation}
P\left(\left|\left|\sqrt{n}\left(\frac{1}{n}\tilde{\Psi}_j(\eta_j) - \Psi_j^{\ast}(\eta_j)\right)\right|\right|_2 > M\right) < \frac{\sum\mathbb{V}_G\left[\Psi_j(\eta_j)\right]}{\left(M^{\prime}(\eta_j)\right)^2}, \quad \forall n > N.\nonumber
\end{equation}
Lastly, we note that $\frac{\sum \mathbb{V}_G\left[\Psi_j(\eta_j)\right]}{\left(M^{\prime}(\eta_j)\right)^2}$ can be made arbitrarily small by making $M$ arbitrarily big  and under the assumption that $\pi_j(\theta_j) > 0$ (A1) and therefore $\left|\left|\frac{1}{n}\tilde{\Psi}_j(\eta_j) - \Psi_j^{\ast}(\eta_j)\right|\right|_2 = O_p(\nicefrac{1}{\sqrt{n}})$ for all $\eta_j$.\\

\noindent\textbf{Step 2: Proof that $||\tilde{\eta}_j - \eta_j^{\ast}||_2 = O_p(\nicefrac{1}{\sqrt{n}})$}

Next, we use A3 on the behaviour of the empirical Hessians around $\eta_j^{\ast}$ to extend the fact that $\left|\left|\frac{1}{n}\tilde{\Psi}_j(\eta_j) - \Psi^{\ast}_j(\eta_j)\right|\right|_2 = O_p(\nicefrac{1}{\sqrt{n}})$ to prove that $||\tilde{\eta}_j - \eta_j^{\ast}||_2 = O_p(\nicefrac{1}{\sqrt{n}})$.

Plugging in $\eta_j^{\ast}$ for $\eta_j$ and using $\tilde{\Psi}_j(\tilde{\eta}_j) = \Psi_j^{\ast}(\eta_j^{\ast}) = 0$ (by A1) gives, 
\begin{align}
\left|\left|\frac{1}{n}\tilde{\Psi}_j(\eta_j^{\ast}) - \Psi_j^{\ast}(\eta_j^{\ast})\right|\right|_2 &= \left|\left|\frac{1}{n}\tilde{\Psi}_j(\eta_j^{\ast}) - \frac{1}{n}\tilde{\Psi}_j(\tilde{\eta}_j)\right|\right|_2 = O_p(\nicefrac{1}{\sqrt{n}})\nonumber
\end{align}
Next, consider a first order Taylor expansion of $\frac{1}{n}\tilde{\Psi}_j(\cdot)$ about $\eta_j^{\ast}$ evaluated at $\tilde{\eta}_j$. Note that $\frac{1}{n}\Psi_j$ is a vector valued function and as a result we consider a Taylor expansion of each entry of the vector. %\jack{better as MV mean values theorem?}
\begin{align}
    &\frac{1}{n}\tilde{\Psi}_j(\tilde{\eta}_j) = \frac{1}{n}\tilde{\Psi}_j(\eta_j^{\ast}) + \nabla_{\eta_j} \frac{1}{n}\tilde{\Psi}_j(\tilde{\eta}_j^{\dagger})(\tilde{\eta}_j - \eta_j^{\ast})\nonumber \\
    \Rightarrow &\left|\left|\frac{1}{n}\tilde{\Psi}_j(\tilde{\eta}_j) - \frac{1}{n}\tilde{\Psi}_j(\eta_j^{\ast})\right|\right|_2 = \left|\left| \frac{1}{n}A_j(\tilde{\eta}_j^{\dagger})(\tilde{\eta}_j - \eta_j^{\ast})\right|\right|_2 \nonumber%\label{Equ:taylor_gradients}
\end{align}
for some $\eta^{\dagger}_j = \alpha\tilde{\eta}_j + (1-\alpha)\eta^{\ast}_j$ for $\alpha\in[0,1]^{d_j}$. %\jack{vector values, see Taylor's theorem}. 
The fact that $\tilde{\Psi}_j(\cdot)$ is already the vector valued first derivative of the \hyvarinen score meant that $\nabla_{\eta_j} \frac{1}{n}\tilde{\Psi}_j(\tilde{\eta}_j^{\dagger}) = \frac{1}{n}A_j(\tilde{\eta}_j^{\dagger})$. 
Now the right hand side can be bounded as follows
\begin{align}
\left|\left|\frac{1}{n}A_j(\tilde{\eta}_j^{\dagger})(\tilde{\eta}_j - \eta_j^{\ast})\right|\right|_2
&= \sqrt{(\tilde{\eta}_j - \eta_j^{\ast})^T\left(\frac{1}{n}A_j(\tilde{\eta}_j^{\dagger})\right)^T\frac{1}{n}A_j(\tilde{\eta}_j^{\dagger})(\tilde{\eta}_j - \eta_j^{\ast})}\nonumber\\
&\in
\left[ \lambda_{\min}(\tilde{\eta}_j^{\dagger}) ||\tilde{\eta}_j - \eta_j^{\ast} ||_2,
\lambda_{\max}(\tilde{\eta}_j^{\dagger}) || \tilde{\eta}_j - \eta_j^{\ast} ||_2,
\right]\nonumber%\label{Equ:eigen_bound}
\end{align}
where $\lambda_{\min}(\tilde{\eta}_j^{\dagger}) $ and $\lambda_{\max}(\tilde{\eta}_j^{\dagger}) $ are the smallest and largest Eigen-values of 
$\frac{1}{n}A_j(\tilde{\eta}_j^{\dagger})$  %\jack{still want to check this bound}
respectively. % which from A2, are finite and bounded away from 0 within a neighbourhood of $\eta^{\ast}_j$. 
As a result we can bound
\begin{align}
 %\left|\left|\frac{1}{n}\tilde{\Psi}_j(\tilde{\eta}_j) - \frac{1}{n}\tilde{\Psi}_j(\eta_j^{\ast})\right|\right|_2 = \left|\left| \frac{1}{n}A_j(\tilde{\eta}_j^{\dagger})(\tilde{\eta}_j - \eta_j^{\ast})\right|\right|_2 \geq  \lambda_{\min}(\tilde{\eta}_j^{\dagger}) ||\tilde{\eta}_j - \eta_j^{\ast} ||_2
  \left|\left|\frac{1}{n}\tilde{\Psi}_j(\tilde{\eta}_j) - \frac{1}{n}\tilde{\Psi}_j(\eta_j^{\ast})\right|\right|_2 \geq  \lambda_{\min}(\tilde{\eta}_j^{\dagger}) ||\tilde{\eta}_j - \eta_j^{\ast} ||_2.\nonumber
\end{align}

Further, $||\tilde{\eta}_j - \eta_j^{\ast}||_2 = o_p(1)$ from Lemma \ref{Lemma:HscoreParamConsistency_op1} and
by A3ii) there exists $\delta_A>0$ such that $\lambda_{\min}(\eta_j) > M_{A}$ with probability tending to 1 for $||\eta_j - \eta_j^{\ast}||_2 \leq \delta_A$. Therefore, with probability tending to 1
\begin{align}
 \left|\left|\frac{1}{n}\tilde{\Psi}_j(\tilde{\eta}_j) - \frac{1}{n}\tilde{\Psi}_j(\eta_j^{\ast})\right|\right|_2 \geq M_{A} ||\tilde{\eta}_j - \eta_j^{\ast} ||_2.\nonumber
\end{align}
$M_A$ is a constant and as a result $||\tilde{\eta}_j - \eta_j^{\ast} ||_2 = O_p(\nicefrac{1}{\sqrt{n}})$, as required.

\end{proof}
Combining Lemma \ref{Lemma:HscoreParamConsistency_op1} and Proposition \ref{Lemma:HscoreParamConsistency2} proves Proposition \ref{Lemma:HscoreParamConsistency}.

%\jack{Check MV Taylors is correct}

\subsubsection{Theorem \ref{Thm:HBayesFactorConsistency2} ($\mathcal{H}$-Bayes Factor Consistency)}{\label{App:HscoreConsistency}}

Given Proposition \ref{Lemma:HscoreParamConsistency} and the conditions stated in Section \ref{App:TechnicalConditions} we are now ready to prove the paper's main theorem. First we restate Theorem \ref{Thm:HBayesFactorConsistency}.

\newtheorem{thm}{Theorem}%[1]
\begin{thm}[$\mathcal{H}$-Bayes Factor Consistency]
Assume Conditions A1-A4 given in Section \ref{App:TechnicalConditions}. Given a sample $y \sim g$ the Laplace approximation to the \HBayes{} factor (given by $\tilde{B}^{(\mathcal{H})}_{kl} := \nicefrac{\tilde{\mathcal{H}}_k(y)}{\tilde{\mathcal{H}}_l(y)}$ for $\tilde{\mathcal{H}}_j(y)$ defined in  \eqref{Equ:LaplaceApproxHScore}) of loss $k$ over loss $l$ has the following asymptotic behaviour as $n\rightarrow\infty$
\begin{enumerate}[label=(\roman*)]
    \item When $\mathbb{E}_g[H(z; f_l(\cdot; \eta^{\ast}_l))] - \mathbb{E}_g[H(z; f_k(\cdot; \eta^{\ast}_k))] < 0$ then 
    \begin{align}
        \frac{1}{n}\log \tilde{B}^{(\mathcal{H})}_{kl} = \mathbb{E}_g[H(z; f_l(\cdot; \eta^{\ast}_l))] - \mathbb{E}_g[H(z; f_k(\cdot; \eta^{\ast}_k))] + o_p(1)
    \end{align}
    That is that when the more complex model $l$ decreases Fisher's divergence relative to model $k$, the \HBayes{} factor accrues evidence in favour of model $l$, $\tilde{B}^{(\mathcal{H})}_{kl}\rightarrow 0$, at an exponential rate. 
    \item When $\mathbb{E}_g[H(z; f_l(\cdot; \eta^{\ast}_l))] = \mathbb{E}_g[H(z; f_k(\cdot; \eta^{\ast}_k))]$, with $k$ being the simpler model 
    \begin{align}
        \log \tilde{B}^{(\mathcal{H})}_{kl} = \frac{d_l - d_k}{2}\log(n) + O_p(1)
    \end{align}
    That is that when  the models are equally preferable according to Fisher's divergence then $\tilde{B}^{(\mathcal{H})}_{kl}\rightarrow \infty$ at a polynomial rate in $n$.
\end{enumerate}
\label{Thm:HBayesFactorConsistency2}
\end{thm}

%Case 1) is the case when the more complicated model $l$ is preferable according the Fisher's divergence, i.e. non-nested models or nested where the biggest is true. Alternatively, Case 2) covers the possibility that both models are equally preferable according to Fisher's divergence, i.e. nested models with the simplest being sufficient, and we seek to select the simplest. \jack{This is in the paper, do we need it here}

\begin{proof}
The proof of Theorem \ref{Thm:HBayesFactorConsistency} can be broken down into 3 stages. Step 1 deals with how the prior densities of models $l$ and $k$ contribute to $\tilde{B}^{(\mathcal{H})}_{kl}$. Step 2 does the same for the Hessian's of the \hyvarinen score applied to models $l$ and $k$. Lastly, Step 3 deals with the in-sample \hyvarinen score contribution, considering the two cases of the theorem. %where one model is closer to $g$ than the other according to Fisher's divergence and the other where both are equally close.    

%\textbf{Step 1: Parameter Consistency} For the first stage of the proof parameter consistency ... was proven in Lemma,,,

\noindent\textbf{Step 1: Priors}

From Proposition \ref{Lemma:HscoreParamConsistency}, given that  $||\tilde{\eta}_j - \eta_j^{\ast}||_2 = O_p(\nicefrac{1}{\sqrt{n}})$ and under the assumption that $\pi(\eta_j)$ is continuous in $\eta_j$ (A1), the Continuous Mapping Theorem provides $\pi_j(\tilde{\eta}_j) \stackrel{P}{\longrightarrow} \pi_j(\eta_j^{\ast})$, $j \in \{l, k\}$.
Further, assuming that $\pi_l(\eta_l^{\ast}) > 0$ (A1), it follows that
\begin{align}
    \frac{\pi_k(\tilde{\eta}_k)}{\pi_l(\tilde{\eta}_l)}
\stackrel{P}{\longrightarrow}
\frac{\pi_k(\eta_k^{\ast})}{\pi_l(\eta_l^{\ast})},  \nonumber
\end{align}
a strictly positive finite constant. Note: that for non-local priors either $\pi_j(\eta_j^{\ast})$ may be zero, in that case the limit may be 0 or $\infty$ . We address this point again later Theorem \ref{Thm:HscoreConsistencyNLPs}.\\

\noindent\textbf{Step 2: Hessian's}

Now for the Hessian's firstly we rewrite $|A_j(\tilde{\eta}_j)| = n^{d_j/2}\left|\frac{1}{n}A_j(\tilde{\eta}_j)\right|$. Then we want to prove that $\frac{1}{n}A_j(\tilde{\eta}_j)\stackrel{P}{\longrightarrow} A^{\ast}_j(\eta^{\ast}_j)$. 
Firstly,
%\begin{align}
%   \frac{1}{n}A_j(\tilde{\eta}_j) -  A^{\ast}_j(\eta^{\ast}_j) &= \frac{1}{n}A_j(\tilde{\eta}_j) - A^{\ast}_j(\tilde{\eta}_j) + A^{\ast}_j(\tilde{\eta}_j) - A^{\ast}_j(\eta^{\ast}_j)\nonumber
%\end{align}
%where $\tilde{\eta}_j\stackrel{P}{\longrightarrow} \eta^{\ast}_j$ and the continuity of the map $\eta_k\mapsto A^{\ast}_j(\eta_j)$ (A1) and the uniform convergence (A3)
%\begin{equation}
% \sup_{\eta_j}\left|\frac{1}{n}A_j(\eta_j) - A^{\ast}_j(\eta_j)]\right| = o_p(1)   
%\end{equation}
%ensure that both the left two terms and the right hand terms are $o_p(1)$ providing the desired convergence. 
\begin{align}
   \frac{1}{n}A_j(\tilde{\eta}_j) -  A^{\ast}_j(\eta^{\ast}_j) &= \frac{1}{n}A_j(\tilde{\eta}_j) - \frac{1}{n}A_j(\eta^{\ast}_j) + \frac{1}{n}A_j(\eta^{\ast}_j) - A^{\ast}_j(\eta^{\ast}_j).\nonumber
\end{align}
Then, we use the weak law of large numbers and the Lipschitz condition (A4)) to show that
\begin{align}
   \left|\left|\frac{1}{n}A_j(\tilde{\eta}_j) -  A^{\ast}_j(\eta^{\ast}_j)\right|\right|_2 &= \left|\left|\frac{1}{n}A_j(\tilde{\eta}_j) - \frac{1}{n}A_j(\eta^{\ast}_j) + \frac{1}{n}A_j(\eta^{\ast}_j) - A^{\ast}_j(\eta^{\ast}_j)\right|\right|_2\nonumber\\
   &\leq \left|\left|\frac{1}{n}A_j(\tilde{\eta}_j) - \frac{1}{n}A_j(\eta^{\ast}_j)\right|\right|_2 + \left|\left|\frac{1}{n}A_j(\eta^{\ast}_j) - A^{\ast}_j(\eta^{\ast}_j)\right|\right|_2\textrm{ (tri. in)}\nonumber\\
    &\leq \frac{1}{n}\sum_{i=1}^n\left|\left|A_j(\tilde{\eta}_j) - A_j(\eta^{\ast}_j)\right|\right|_2 + o_p(1)\textrm{ (tri. in \& WLLN.)}\nonumber\\
   &\leq \frac{1}{n}\left|\left|\tilde{\eta}_j - \eta^{\ast}_j\right|\right|_2\sum_{i=1}^nm_A(y_i) + o_p(1)= o_p(1),\nonumber
\end{align}
%\jack{here I need to be very careful about what I do with the average $A$ and the Lipschitzness}

where $\frac{1}{n}\sum_{i=1}^nm_A(y_i) \stackrel{P}{\longrightarrow} \mathbb{E}\left[m_A(y)\right] < \infty$ by the weak law of large numbers and $\left|\left|\tilde{\eta}_j - \eta^{\ast}_j\right|\right|_2 = O_p(\nicefrac{1}{\sqrt{n}})$ as proved in Proposition \ref{Lemma:HscoreParamConsistency}. 
As a result we have that 
\begin{align}
    \frac{\left|A_l(\tilde{\eta}_l)\right|^{1/2}}{\left|A_k(\tilde{\eta}_k)\right|^{1/2}}\stackrel{P}{\longrightarrow}
n^{\frac{(d_l - d_k)}{2}}\frac{|A^{\ast}_l(\eta^{\ast}_l)|^{1/2}}{|A^{\ast}_k(\eta^{\ast}_k)|^{1/2}}\nonumber.
\end{align}

\noindent\textbf{Step 3: In-sample \hyvarinen score difference}

Combining Steps 1-2, the right hand terms of \eqref{Equ:HBayesFactor_LaplaceApprox} is asymptotically equivalent to
\begin{align}
 \frac{|A^{\ast}_l(\eta_l^{\ast})|^{1/2}}{|A^{\ast}_k(\eta_k^{\ast})|^{1/2}}
\frac{\pi_k(\eta_k^*)}{\pi_l(\eta_l^*)}
\left(\frac{2\pi}{n}\right)^{\frac{d_k-d_l}{2}}
 e^{ \sum_{i=1}^n H\left(y_i; f_l\left(\cdot;\tilde{\eta}_l\right)\right) - \sum_{i=1}^nH\left(y_i; f_k\left(\cdot;\tilde{\eta}_k\right)\right) },
 \label{eq:HBF_asymptotic}
\end{align}
where the first two terms are constants (A1 and A3). The only step left is to characterize the last term,  the in-sample \Hscore ratio. To do this it is necessary to distinguish two cases:
\begin{enumerate}
    \item Model $l$ decreases Fisher's divergence relative to model $k$, $\mathbb{E}_g\left[H(y; f_l(\cdot; \eta^{\ast}_{l})\right] < \mathbb{E}_g\left[H(y; f_k(\cdot; \eta^{\ast}_{k})\right]$. This encompasses the cases where the models are non-nested or nested models where the bigger of the two models is true.
    \item Model $l$ does not decrease Fisher's divergence relative to model $k$, $\mathbb{E}_g\left[H(y; f_l(\cdot; \eta^{\ast}_{l})\right] = \mathbb{E}_g\left[H(y; f_k(\cdot; \eta^{\ast}_{k})\right]$. This is the case when the models are nested and the simplest model $k$ is sufficient for minimising Fisher's divergence.
\end{enumerate}

\noindent\textbf{Case 1)}

First note that the logarithm of \eqref{eq:HBF_asymptotic} is equivalent to %\jack{Do I also mention that sometimes the expected Hessian is not finite here also }
\begin{align}
\frac{1}{n} c
- {\frac{d_k-d_l}{2n}} \log\left(\frac{2\pi}{n}\right)
+ \sum_{i=1}^n \frac{1}{n} H\left(y_i; f_l\left(\cdot;\tilde{\eta}_l\right)\right) - \sum_{i=1}^n \frac{1}{n} H\left(y_i; f_k\left(\cdot;\tilde{\eta}_k\right)\right)=
\nonumber \\
\sum_{i=1}^n \frac{1}{n} H\left(y_i; f_l\left(\cdot;\tilde{\eta}_l\right)\right) - \sum_{i=1}^n \frac{1}{n} H\left(y_i; f_k\left(\cdot;\tilde{\eta}_k\right)\right) + O(1/n), \nonumber
\end{align}
for a constant $c\in \mathbb{R}$ (A1 and A3) given by the first two terms in \eqref{eq:HBF_asymptotic}. Now 
\begin{align}
&\frac{1}{n}\sum_{i=1}^n H(y_i; f_l(\cdot; \tilde{\eta}_l)) - \frac{1}{n}\sum_{i=1}^n H(y_i; f_k(\cdot; \tilde{\eta}_k)) = \mathbb{E}_g[H(y; f_l(\cdot; \eta^{\ast}_l))] - \mathbb{E}_g[H(y; f_k(\cdot; \eta^{\ast}_k))]\nonumber\\
& + \frac{1}{n}\sum_{i=1}^n H(y_i; f_l(\cdot; \tilde{\eta}_l)) - \mathbb{E}_g[H(y; f_l(\cdot; \eta^{\ast}_l))] - \left\lbrace \frac{1}{n}\sum_{i=1}^n H(y_i; f_k(\cdot; \tilde{\eta}_k)) - \mathbb{E}_g[H(y; f_k(\cdot; \eta^{\ast}_k))]\right\rbrace,\nonumber
\end{align}
and for each of the terms on the bottom line ($j \in \{k, l\}$)
%
%\begin{align}
%&\frac{1}{n}\sum_{i=1}^n H(y_i; f_j(\cdot; \tilde{\eta}_j)) - \mathbb{E}_g[H(y; f_j(\cdot; \eta^{\ast}_j))]\nonumber\\
%=&\frac{1}{n}\sum_{i=1}^n H(y_i; f_j(\cdot; \tilde{\eta}_j)) - \mathbb{E}_g[H(y; f_j(\cdot; \tilde{\eta}_j))] + \mathbb{E}_g[H(y; f_j(\cdot; \tilde{\eta}_j))]  - \mathbb{E}_g[H(y; f_j(\cdot; \eta^{\ast}_j))]\nonumber
%\end{align}
%
%\jack{Glivenko Cantelli??}
%
%where $\tilde{\eta}_j\stackrel{P}{\longrightarrow} \eta^{\ast}_j$ and the continuity of the map $\eta_k\mapsto\mathbb{E}_g[H(y; f_j(\cdot; \eta_j))]$ and the uniform convergence (A3)
%\begin{equation}
% \sup_{\eta_j}\left|\frac{1}{n}\sum_{i=1}^n H(y_i; f_j(\cdot; \eta_j)) - \mathbb{E}_g[H(y; f_j(\cdot; \eta_j))]\right| = o_p(1)  \nonumber
%\end{equation}
%ensures that 
%\begin{align}
%    &\mathbb{E}_g[H(y; f_j(\cdot; \tilde{\eta}_j))]  - \mathbb{E}_g[H(y; f_j(\cdot; \eta^{\ast}_j))]\stackrel{P}{\longrightarrow} 0\nonumber\\
%    \textrm{and }&\frac{1}{n}\sum_{i=1}^n H(y_i; f_j(\cdot; \tilde{\eta}_j)) - \mathbb{E}_g[H(y; f_j(\cdot; \tilde{\eta}_j))]\stackrel{P}{\longrightarrow} 0\nonumber
%\end{align}
%also. 
%
%the weak law of large numbers provides that $\frac{1}{n}\sum_{i=1}^n H(y_i; f_j(\cdot; \tilde{\eta}_j)) - \mathbb{E}_g[H(y; f_j(\cdot; \tilde{\eta}_j))]\stackrel{P}{\longrightarrow} 0$. 
%
%\jack{No the LLN isn't enough as we have parameter estimates, need uniform convergence of $o_p(1)$, which must be another condition this si what they do in \cite{rossell2019additive}, not sure if it is a condition or not} 
%
\begin{align}
&\frac{1}{n}\sum_{i=1}^n H(y_i; f_j(\cdot; \tilde{\eta}_j)) - \mathbb{E}_g[H(y; f_j(\cdot; \eta^{\ast}_j))]\nonumber\\
=&\frac{1}{n}\sum_{i=1}^n H(y_i; f_j(\cdot; \tilde{\eta}_j)) - \frac{1}{n}\sum_{i=1}^n H(y_i; f_j(\cdot; \eta^{\ast}_j)) + \frac{1}{n}\sum_{i=1}^n H(y_i; f_j(\cdot; \eta^{\ast}_j))  - \mathbb{E}_g[H(y; f_j(\cdot; \eta^{\ast}_j))].\nonumber
\end{align}
We then use the weak law of large numbers and the Lipschitz condition (A4) to show that 
\begin{align}
&\left|\frac{1}{n}\sum_{i=1}^n H(y_i; f_j(\cdot; \tilde{\eta}_j)) - \mathbb{E}_g[H(y; f_j(\cdot; \eta^{\ast}_j))]\right|\nonumber\\
=&\left|\frac{1}{n}\sum_{i=1}^n H(y_i; f_j(\cdot; \tilde{\eta}_j)) - \frac{1}{n}\sum_{i=1}^n H(y_i; f_j(\cdot; \eta^{\ast}_j)) + \frac{1}{n}\sum_{i=1}^n H(y_i; f_j(\cdot; \eta^{\ast}_j))  - \mathbb{E}_g[H(y; f_j(\cdot; \eta^{\ast}_j))]\right|\nonumber\\
\leq&\left|\frac{1}{n}\sum_{i=1}^n H(y_i; f_j(\cdot; \tilde{\eta}_j)) - \frac{1}{n}\sum_{i=1}^n H(y_i; f_j(\cdot; \eta^{\ast}_j))\right| + \left|\frac{1}{n}\sum_{i=1}^n H(y_i; f_j(\cdot; \eta^{\ast}_j))  - \mathbb{E}_g[H(y; f_j(\cdot; \eta^{\ast}_j))]\right|\textrm{ (tri. in)}\nonumber\\
\leq&\frac{1}{n}\sum_{i=1}^n\left| H(y_i; f_j(\cdot; \tilde{\eta}_j)) - H(y_i; f_j(\cdot; \eta^{\ast}_j))\right| + o_p(1)\textrm{ (tri. in \& WLLN)}\nonumber\\
\leq&\frac{1}{n}\sum_{i=1}^nm_H(y_i)\left|\left| \tilde{\eta}_j - \eta^{\ast}_j\right|\right|_2 + o_p(1) = o_p(1)\nonumber
\end{align}
where $\frac{1}{n}\sum_{i=1}^nm_H(y_i) \stackrel{P}{\longrightarrow} \mathbb{E}\left[m_H(y)\right] < \infty$ by the weak law of large numbers and $\left|\left|\tilde{\eta}_j - \eta^{\ast}_j\right|\right|_2 = O_p(\nicefrac{1}{\sqrt{n}})$ by Proposition \ref{Lemma:HscoreParamConsistency}.  
As a result 
\begin{align}
\frac{1}{n}\log \frac{\mathcal{H}_k(y)}{\mathcal{H}_l(y)} = \mathbb{E}_g[H(y; f_l(\cdot; \eta^{\ast}_l))] - \mathbb{E}_g[H(y; f_k(\cdot; \eta^{\ast}_k))] + o_p(1),\nonumber
\end{align}
where $\mathbb{E}_g[H(y; f_l(\cdot; \eta^{\ast}_l))] - \mathbb{E}_g[H(y; f_k(\cdot; \eta^{\ast}_k))] < 0$ by assumption. That is, the \Hscore{} ratio accrues evidence in favour of model $l$ at an exponential rate. 

\noindent\textbf{Case 2)}

Now $\mathbb{E}_g\left[H(y; f_l(\cdot; \eta^{\ast}_{l})\right] = \mathbb{E}_g\left[H(y; f_k(\cdot; \eta^{\ast}_{k})\right]$ and we want to prove that the $\mathcal{H}$-Bayes factor convergence towards the simpler model $k$. The key is to prove that 
\begin{align}
\sum_{i=1}^n H(y_i; f_l(\cdot; \tilde{\eta}_l)) - \sum_{i=1}^n H(y_i; f_k(\cdot; \tilde{\eta}_k)) &= O_p(1). \nonumber%\\
%\Rightarrow \sum_{i=1}^n H(y_i; f_l(\cdot; \tilde{\eta}_l)) - \log \pi(\tilde{\eta}_k) - \left\lbrace \sum_{i=1}^n H(y_i; f_k(\cdot; \tilde{\eta}_k)) - \log \pi(\tilde{\eta}_k)\right\rbrace &= O_p(1)
\end{align}
Note that as model $k$ is nested in model $l$, there exists $\eta^{\ast}_{l\setminus k}$ such that
\begin{equation}
	%H(y_i; f_k(\cdot; \tilde{\eta}_k)) = H(y_i; f_l(\cdot; \left(\tilde{\eta}_k, \kappa_l^{\ast}\right))).\nonumber
	H(y_i; f_k(\cdot; \eta_k)) = H(y_i; f_l(\cdot; \left(\eta_k, \eta^{\ast}_{l\setminus k}\right)))\quad \forall \eta_k,\nonumber
\end{equation}
where $\eta^{\ast}_{l\setminus k}$ is the value of the hyperparameters of model $l$ that recover model $k$. % This used to be $\kappa^{\ast}_l$ but I think it is now clearer \jack{Might be clearer to have this more generally rather than a specific kappa, i.e. $\eta$}
Henceforth, denote $\tilde{\eta}^{\prime}_k = \left(\tilde{\eta}_k, \eta^{\ast}_{l\setminus k}\right)$. We can rewrite our objective of interest as
\begin{align}
\sum_{i=1}^n H(y_i; f_l(\cdot; \tilde{\eta}_l)) - \log \pi_l(\tilde{\eta}^{\prime}_k) - \left\lbrace \sum_{i=1}^n H(y_i; f_k(\cdot; \tilde{\eta}_k)) - \log \pi_l(\tilde{\eta}^{\prime}_k)\right\rbrace,\nonumber
\end{align}
provided $\pi_l(\eta_l) > 0$, $\forall \eta_l$ (A1). This facilitates the Taylor expansion of $\tilde{H}_l(\tilde{\eta}^{\prime}_k) := \sum_{i=1}^n H(y_i; f_l(\cdot; \tilde{\eta}^{\prime}_k)) - \log \pi_l(\tilde{\eta}^{\prime}_k)$ about $\tilde{\eta}_l$, assuming that its first and second derivatives are finite. 
\begin{align}
%	 & \tilde{H}_l(\tilde{\eta}^{\prime}_k) = \tilde{H}_l(\tilde{\eta}_l) + (\tilde{\eta}_k - \tilde{\eta}_l)\nabla_{\eta_l} \tilde{H}_l(\tilde{\eta}_l) + \frac{1}{2}(\tilde{\eta}_k - \tilde{\eta}_l)^T\nabla_{\eta_l}^2 \tilde{H}_l(\tilde{\eta}^{\dagger(n)}_l)(\tilde{\eta}_k - \tilde{\eta}_l)\nonumber
	 & \tilde{H}_l(\tilde{\eta}^{\prime}_k) = \tilde{H}_l(\tilde{\eta}_l) + (\tilde{\eta}_k - \tilde{\eta}_l)\tilde{\Psi}_l(\tilde{\eta}_l) + \frac{1}{2}(\tilde{\eta}_k - \tilde{\eta}_l)^T\nabla_{\eta_l}^2 A_l(\tilde{\eta}^{\dagger}_l)(\tilde{\eta}_k - \tilde{\eta}_l),\nonumber
\end{align}
for some $\eta^{\dagger}_l = \alpha\tilde{\eta}_k + (1-\alpha)\tilde{\eta}_l$ with $\alpha\in [0, 1]^{d_l}$.

Now by definition we have that $\tilde{\Psi}_l(\tilde{\eta}_l) = 0$ for all $n$ and therefore plugging this back into our target equation we get that 
\begin{align}
&\sum_{i=1}^n H(y_i; f_l(\cdot; \tilde{\eta}_l)) - \log \pi_l(\tilde{\eta}^{\prime}_k) - \left\lbrace\sum_{i=1}^n H(y_i; f_k(\cdot; \tilde{\eta}_k)) - \log \pi_l(\tilde{\eta}_l)\right\rbrace\nonumber\\
=& \log \pi_l(\tilde{\eta}_l) - \log \pi_l(\tilde{\eta}^{\prime}_k) + \frac{1}{2}(\tilde{\eta}_k - \tilde{\eta}_l)^TA_l(\tilde{\eta}^{\dagger}_l)(\tilde{\eta}_k - \tilde{\eta}_l)\nonumber\\
=& \log \pi_l(\tilde{\eta}_l) - \log \pi_l(\tilde{\eta}^{\prime}_k) +  \frac{1}{2}\sqrt{n}(\tilde{\eta}_k - \tilde{\eta}_l)^T\frac{1}{n}A_l(\tilde{\eta}^{\dagger}_l)\sqrt{n}(\tilde{\eta}_k - \tilde{\eta}_l).\nonumber
\end{align}
Now, $||\tilde{\eta}_j - \eta_j^{\ast}||_2 = O_p\left(1/\sqrt{n}\right)$, $j = \{k, l\}$ from Proposition \ref{Lemma:HscoreParamConsistency}, where in this nested case $\eta_l^{\ast} = \eta_k^{\ast} = \eta^{\ast}$, which further provides that $||\eta^{\dagger}_l - \eta^{\ast}||_2 = O_p(\nicefrac{1}{\sqrt{n}})$ also. From Step 1 we know that $\log \pi_j(\tilde{\eta}_j) \stackrel{P}{\longrightarrow} \log \pi_j(\eta^{\ast}_j)$, $j = k, l$, and therefore $\log \pi_l(\tilde{\eta}_l) - \log \pi_l(\tilde{\eta}^{\prime}_k) \stackrel{P}{\longrightarrow} 0$. 
%Lastly, Step 3 required the assumed uniform convergence (\jack{We no longer have uniform convergence}) $\sup_{\eta_j}|\frac{1}{n}A_j(\eta_j) - A^{\ast}_j(\eta_j)| = o_p(1)$ \jack{Norm!} of the Hessian's, and therefore 
%\begin{align}
%\frac{1}{n}\frac{\partial^2}{\partial \eta_l^2} \tilde{H}_l(\tilde{\eta}^{\dagger(n)}_l) \stackrel{P}{\longrightarrow} A^{\ast}_l(\eta_l^{\ast}).\nonumber
%\end{align}
Lastly, Step 2 proved that $\left|\left|\frac{1}{n}A_j(\tilde{\eta}_j) -  A^{\ast}_j(\eta^{\ast}_j)\right|\right|_2 = o_p(1)$
%Finally, Lemma \ref{Lemma:HscoreParamConsistency} further proved that $\left|\left|\tilde{\eta}_j - \eta^{\ast}_j\right|\right|_2 = O_p(\nicefrac{1}{\sqrt{n}})$ 
and therefore
\begin{align}
\sum_{i=1}^n H(y_i; f_l(\cdot; \tilde{\eta}_l)) - \sum_{i=1}^n H(y_i; f_k(\cdot; \tilde{\eta}_k)) &=\frac{1}{2}\sqrt{n}(O_p(\nicefrac{1}{\sqrt{n}})(A^{\ast}_l(\eta_l^{\ast}) + o_p(1))\sqrt{n}(O_p(\nicefrac{1}{\sqrt{n}})\nonumber\\
&= O_p(1) + O_p(1)o_p(1) = O_p(1)\nonumber
\end{align}
As a result, that log-\Hscore{} accrues evidence in favour of simpler model $k$ at a $\log(n)$ rate. That is, the \Hscore{} accrues evidence at polynomial rate. 

\end{proof}

\begin{comment}
\subsubsection{Asymptotic normality of the \Hposterior}

While Theorem \ref{Thm:ConsistencyParameter} established the consistency and asymptotic normality of the frequentist \MHSE parameter, the following theorem establishes the asymptotic normality of the \Hposterior{} in  \eqref{Equ:GeneralBayesRuleHScore}. 

\jack{This could be just for the whole $\eta = (\theta, \kappa)$ or given a plug in $\kappa$ and just for $\theta$, would come after the two stage inference}

\begin{theorem}{(\textbf{Asymptotic normality of the \Hposterior{} })}
%\jack{I have coped the wording from Arnaud's adaptation of the UAI proof}
Let the regularity conditions of \citep{chernozhukov2003mcmc, lyddon2019general} hold. There exists a non-singular matrix $K$ such that under the \Hposterior{} in  \eqref{Equ:GeneralBayesRuleHScore} $\pi^H(\eta_k|y)$
\begin{equation}
    \sqrt{n}\left(\eta_k - \hat{\eta}_k\right) \stackrel{P}{\longrightarrow} \mathcal{N}\left(0, K^{-1}\right)
\end{equation}
almost surely as $n\rightarrow\infty$ with respect to $x_{1:\infty}\sim g(\cdot)$\footnote{$\pi^H(\eta_k|y)$ is here interpreted as a random probability measures and a function of the random observations $y$} and where  $K = K\left(\theta^{\ast}_k\right)$ is given in  \eqref{Equ:HessianHscore}
\end{theorem}

\begin{proof}
Apply  (4) of \cite{lyddon2019general} with $\ell(\theta_k, x) = H(x;f_k(\cdot;\theta_k))$. For the regularity conditions see \cite{chernozhukov2003mcmc} and the supplementary material of \cite{lyddon2019general}.
\end{proof}
\end{comment}

\subsubsection{Corollary \ref{Thm:SMICConsistency} (\SMIC inconsistency for nested models)}{\label{App:SMICInconsistency}}

Here we investigate whether the \SMIC model selection criteria provided by \cite{matsuda2019information} can provide the same model/loss selection consistency as Theorem \ref{Thm:HBayesFactorConsistency}, in particular focusing on nested models. We do so under the same conditions as Theorem  \ref{Thm:HBayesFactorConsistency}, where the Laplace approximation of the \Hscore{} were shown to be consistent, but with minor further conditions. Firstly, define 
\begin{align}
    I_j(\eta_j) :&= \sum_{i=1}^n\left(\nabla_{\eta_j} H(y_i; f_j(\cdot; \eta_j))\right)\left(\nabla_{\eta_j} H(y_i; f_j(\cdot; \eta_j))\right)^T\nonumber\\
    I^{\ast}_j(\eta_j) :&= \mathbb{E}_g\left[\left(\nabla_{\eta_j} H(z; f_j(\cdot; \eta_j))\right)\left(\nabla_{\eta_j} H(z; f_j(\cdot; \eta_j))\right)^T\right].\nonumber
\end{align}
We then consider conditions
\begin{itemize}
    \item[A5.] For model $j = l, k$ there exist functions $m_I(\cdot)$ with $\mathbb{E}\left[m_I(\cdot)\right] < \infty$ so that the following Lipschitz condition holds
    \begin{align}
         \left|\left|I^{(1)}_j(\eta_j^{(a)}) - I^{(1)}_j(\eta_j^{(b)})\right|\right|_2 \leq m_I(z)\left|\left|\eta_j^{(a)} - \eta_j^{(b)}\right|\right|_2, \quad \forall{}\eta_j^{(a)}, \eta_j^{(b)},\nonumber
    \end{align}
    where the notation $I^{(1)}_j(\cdot)$ emphasises that only one observation, $z$, is involved.
    \item[A6] Models $k$ and $l$ are such that $||\tilde{\eta}_j - \eta^{\ast}_j||_2 > o_p(1/\sqrt{n})$.
\end{itemize}

A5 provides a standard Lipschitz condition for the matrices $I_j(\eta_j)$, while A6 says that the convergence of $\tilde{\eta}_j$ to $\eta^{\ast}_j$ is not faster than $1/\sqrt{n}$. While Proposition \ref{Lemma:HscoreParamConsistency} proved this convergence of $O_p(\nicefrac{1}{\sqrt{n}})$, it is rare to obtain parametric convergence of $o_p(1/\sqrt{n})$, for simplicity here we simply assume this to be the case. 
Further, note that while we desired maximisation of our \Hscore, the \SMIC is to be minimised.

\begin{corollary}[\SMIC inconsistency for nested models]
Assume Conditions A1-A4, A5 and parametric convergence according to A6. Given a sample $y \sim g$, the difference in the \SMIC proposed by \cite{matsuda2019information} of loss $k$ over loss $l$ has the following asymptotic behaviour as $n\rightarrow\infty$.
\begin{enumerate}[label=(\roman*)]
    \item When $\mathbb{E}_g[H(z; f_l(\cdot; \eta^{\ast}_l))] = \mathbb{E}_g[H(z; f_k(\cdot; \eta^{\ast}_k))]$, with $k$ being the simpler model
    \begin{align}
        \SMIC_k - \SMIC_l > o_p(1)
    \end{align}
    That is that when the models are equally preferable according to Fisher's divergence the \SMIC will not consistently select the simpler model.
\end{enumerate}
\label{Thm:SMICConsistency}
\end{corollary}

%\david{Note that the difference of \SMIC's being $O_p(1)$ doesn't mean that it's model selection inconsistent. For example, $W=o_p(1)$ is $O_p(1)$, but $W=o_p(1)$ means that $W$ can be arbitrarily small. Ideally what you want to say is that the probability that the difference between \SMIC's is above some constant $c$ is always positive, as $n \rightarrow \infty$ (that is, that said difference $o_p(1)$). This result should be immediate from your proof, I suppose, just a matter of stating it this way.} we need to argue that $\left\lbrace\sum_{i=1}^n H(y_i; f_k(\cdot; \tilde{\eta}_k)) - \sum_{i=1}^n H(y_i; f_l(\cdot; \tilde{\eta}_l)) \right\rbrace \neq o_p(1)$?

%\jack{we could just have a condition that $||\tilde{\eta}_j - \eta^{\ast}_j||_2 > o_p(1/\sqrt{n})}$, which then proves this using the Taylor series, then argue well we proved that it is $O_p(\nicefrac{1}{\sqrt{n}})$ under some conditions but it is much harder to prove that it is $o_p(1/\sqrt{n})$} \jack{Or of course we could look at their proof also!}

\begin{proof}
To prove this we first establish that the trace terms, $\textrm{tr}\left(\frac{1}{n}I_j(\tilde{\eta}_j)\left(\frac{1}{n}A_j(\tilde{\eta}_j)\right)^{-1}\right)$, tend to constants. Then we show that under the conditions the difference $\left\lbrace\sum_{i=1}^n H(y_i; f_k(\cdot; \tilde{\eta}_k)) - \sum_{i=1}^n H(y_i; f_l(\cdot; \tilde{\eta}_l)) \right\rbrace$ is not $o_p(1)$ and therefore for any $n$ there is non-zero probability of $SMIC_k - SMIC_l < 0$ when the model $k$ was sufficient according to Fisher's divergence. 

From \cite{matsuda2019information}  (34)
\begin{align}
    \SMIC_k - \SMIC_l =& \left\lbrace\sum_{i=1}^n H(y_i; f_k(\cdot; \tilde{\eta}_k)) - \sum_{i=1}^n H(y_i; f_l(\cdot; \tilde{\eta}_l)) \right\rbrace\nonumber \\
    & + \textrm{tr}\left(\frac{1}{n}I_k(\tilde{\eta}_k)\left(\frac{1}{n}A_k(\tilde{\eta}_k)\right)^{-1}\right) - \textrm{tr}\left(\frac{1}{n}I_l(\tilde{\eta}_l)\left(\frac{1}{n}A_l(\tilde{\eta}_l)\right)^{-1}\right).\nonumber
\end{align}
Firstly, \textbf{Step 2} of Theorem \ref{Thm:HBayesFactorConsistency} proved that 
\begin{align}
    ||\frac{1}{n}A_j(\tilde{\eta}_j) - A^{\ast}_j(\eta^{\ast})||_2 = o_p(1), \quad j = l, k.\nonumber
\end{align}
Next, we use similar arguments, and A5 to prove that $||\frac{1}{n}I_j(\tilde{\eta}_j) - I^{\ast}_j(\eta^{\ast})||_2 = o_p(1)$.
\begin{align}
   \left|\left|\frac{1}{n}I_j(\tilde{\eta}_j) -  I^{\ast}_j(\eta^{\ast}_j)\right|\right|_2 &= \left|\left|\frac{1}{n}I_j(\tilde{\eta}_j) - \frac{1}{n}I_j(\eta^{\ast}_j) + \frac{1}{n}I_j(\eta^{\ast}_j) - I^{\ast}_j(\eta^{\ast}_j)\right|\right|_2\nonumber\\
   &\leq \left|\left|\frac{1}{n}I_j(\tilde{\eta}_j) - \frac{1}{n}I_j(\eta^{\ast}_j)\right|\right|_2 + \left|\left|\frac{1}{n}I_j(\eta^{\ast}_j) - I^{\ast}_j(\eta^{\ast}_j)\right|\right|_2\textrm{ (tri. in)}\nonumber\\
    &\leq \frac{1}{n}\sum_{i=1}^n\left|\left|I_j(\tilde{\eta}_j) - I_j(\eta^{\ast}_j)\right|\right|_2 + o_p(1)\textrm{ (tri. in \& WLLN.)}\nonumber\\
   &\leq \frac{1}{n}\left|\left|\tilde{\eta}_j - \eta^{\ast}_j\right|\right|_2\sum_{i=1}^nm_I(y_i) + o_p(1)= o_p(1),\nonumber
\end{align}
%\jack{here I need to be very careful about what I do with the average $A$ and the Lipschitzness}
%
where $\frac{1}{n}\sum_{i=1}^nm_I(y_i) \stackrel{P}{\longrightarrow} \mathbb{E}\left[m_I(y)\right] < \infty$ by the weak law of large numbers and $\left|\left|\tilde{\eta}_j - \eta^{\ast}_j\right|\right|_2 = O_p(\nicefrac{1}{\sqrt{n}})$ as proved in Proposition \ref{Lemma:HscoreParamConsistency}. Therefore, for constant $C^{\ast} = \textrm{tr}\left(\frac{1}{n}I_k^{\ast}(\eta^{\ast}_k)\left(\frac{1}{n}A_k^{\ast}(\eta^{\ast}_k)\right)^{-1}\right) - \textrm{tr}\left(\frac{1}{n}I_l^{\ast}(\eta^{\ast}_l)\left(\frac{1}{n}A_l^{\ast}(\eta^{\ast}_l)\right)^{-1}\right)$ we have that 
\begin{align}
    \SMIC_k - \SMIC_l &= \left\lbrace\sum_{i=1}^n H(y_i; f_k(\cdot; \tilde{\eta}_k)) - \sum_{i=1}^n H(y_i; f_l(\cdot; \tilde{\eta}_l)) \right\rbrace + C^{\ast} + o_p(1).\nonumber
\end{align}
Lastly, \textbf{Step 3, Case 2)} of Theorem \ref{Thm:HBayesFactorConsistency} proved that in the nested case
\begin{align}
    %\left\lbrace\sum_{i=1}^n H(y_i; f_k(\cdot; \tilde{\eta}_k)) - \sum_{i=1}^n H(y_i; f_l(\cdot; \tilde{\eta}_l)) \right\rbrace = O_p(1).\nonumber
    \left\lbrace\sum_{i=1}^n H(y_i; f_k(\cdot; \tilde{\eta}_k)) - \sum_{i=1}^n H(y_i; f_l(\cdot; \tilde{\eta}_l)) \right\rbrace %=& \log \pi_l(\tilde{\eta}_l) - \log \pi_l(\tilde{\eta}^{\prime}_k) +  \frac{1}{2}\sqrt{n}(\tilde{\eta}_k - \tilde{\eta}_l)^T\frac{1}{n}A_l(\tilde{\eta}^{\dagger}_l)\sqrt{n}(\tilde{\eta}_k - \tilde{\eta}_l)\nonumber\\
    = & o_p(1) + \frac{1}{2}\sqrt{n}(\tilde{\eta}_k - \tilde{\eta}_l)^T\left(\frac{1}{n}A_l^{\ast}(\eta^{\ast}_l) + o_p(1)\right)\sqrt{n}(\tilde{\eta}_k - \tilde{\eta}_l).\nonumber
\end{align}
As a result unless $||\tilde{\eta}_k - \tilde{\eta}_l||_2 = o_p(1/\sqrt{n})$, which when $\mathbb{E}_g[H(z; f_l(\cdot; \eta^{\ast}_l))] = \mathbb{E}_g[H(z; f_k(\cdot; \eta^{\ast}_k))]$ will not happen by assumption (A6), then 
\begin{align}
    \left\lbrace\sum_{i=1}^n H(y_i; f_k(\cdot; \tilde{\eta}_k)) - \sum_{i=1}^n H(y_i; f_l(\cdot; \tilde{\eta}_l)) \right\rbrace > o_p(1) \nonumber
\end{align}
as a result there exists $\epsilon^{\ast}> 0$ and $M^{\ast} > 0$ such that for all $n$
\begin{align}
    P\left(\left\lbrace\sum_{i=1}^n H(y_i; f_k(\cdot; \tilde{\eta}_k)) - \sum_{i=1}^n H(y_i; f_l(\cdot; \tilde{\eta}_l)) \right\rbrace > M^{\ast} \right) > \epsilon. \nonumber
\end{align}

The consequences of this are as follows, in the nested modelling scenario $\{\sum_{i=1}^n H(y_i; f_k(\cdot; \tilde{\eta}_k)) - \sum_{i=1}^n H(y_i; f_l(\cdot; \tilde{\eta}_l)) \} \geq 0$, therefore if $C^{\ast} > 0$, then model $k$ will never be preferred to model $l$. For nested models however it is often the case that $C^{\ast}$ is less than 0. In this case is is possible that $M^{\ast}$ is such that $M^{\ast} + C^{\ast} > 0$ and therefore there is still non-zero probability that model $l$ is chosen over model $k$ even as $n \rightarrow \infty$. Knowing whether this is the case is impossible without knowledge of the data generating density $g$ and therefore a practitioner can not be sure the SMIC will provide consistent model selection. Therefore, we describe it as inconsistent for this task. 

\end{proof}

Note Corollary \ref{Thm:SMICConsistency} considered the \SMIC evaluated at the penalised \hyvarinen score minimiser $\tilde{\eta}_k$ in order to simplify the extension from Theorem \ref{Thm:HBayesFactorConsistency}, however it is straightforward to see how this also holds for the unpenalised $\hat{\eta}_j$ also. A further consequence of this result is that cases where $A_j^{\ast}(\eta_j^{\ast})$ is not a finite matrix, as discussed in Section \ref{Sec:HscoreModelConsistency}, will also affect the performance of the \SMIC. For the \SMIC there is no natural method to ameliorate this.%, while we argue in Section \ref{Sec:HscoreModelConsistency}

\subsubsection{Theorem \ref{Thm:HscoreConsistencyNLPs} ($\mathcal{H}$-Bayes Factor Consistency under non-local priors)}{\label{App:ConsistencyNLPs}}

%\jack{SO I have made this quite specific for clarity of notation and argument however the notation could be made more general i.e. NLPs on $\theta$ too}

The conditions of Theorem \ref{Thm:HBayesFactorConsistency} required that $\pi_j(\theta_j) > 0$ for all $\theta_j$ and therefore these conditions will be violated if a non-local prior \citep{johnson2012bayesian, rossell2017nonlocal} is placed on any model parameters. Here, we extend Theorem \ref{Thm:HBayesFactorConsistency} to allow for these, particularly focusing on nested models 
%
%\jack{Need to include the stuff below either as a condition, or as part of the proof! Or maybe a s a little discussion after, at the moment it just looks unfamiliar}
In order to do so we consider the following extended conditions defining the form of the nested models and the non-local prior.
\begin{itemize}
    \item[A7] Models $k$ and $l$ are such that $\eta_k = \{\theta_k\}$ with $\theta_k \in \Theta_k = \Theta$ (i.e. no hyperparameters $\kappa_k = \emptyset$) and $\eta_l = \{\theta_l, \kappa_l\}$ with $\kappa_l \neq \emptyset$ and $\theta_l \in \Theta_l = \Theta$. Define $\kappa_{l\setminus k}$ as the hyperparameters of model $l$ that recovers model $k$, %\jack{$w_j$ is already part of $\kappa_j$}
    \begin{equation}
	    H(y_i; f_k(\cdot; \eta_k)) = H(y_i; f_l(\cdot; \left\lbrace\eta_k, \kappa_{l\setminus k}\right\rbrace)), \quad \textrm{for all } y_i.\nonumber
    \end{equation}
    \item[A8] Models $l$ and $k$ satisfying A7 have priors
    \begin{itemize}  
        \item[i)] The prior $\pi_k(\eta_k) = \pi_k(\theta_k)$ for model $k$ satisfies A1 and the prior for model $l$ is $\pi_l(\eta_l) = \pi_l(\theta_l, \kappa_l) = \pi^{\NLP}_l(\kappa_l)\pi_l(\theta_l | \kappa_l)$ where  $\pi_l(\theta_l | \kappa_l)$ satisfy A1 also. 
        %\item[i)] Model $k$ has a local, continuous prior $\pi_k(\theta_k) > 0$ for all $\theta_k$
        %\item[ii)] Model $l$ has prior $\pi_l(\kappa_l,\theta_l) = \pi^{\NLP}_l(\kappa_l)\pi_l(\theta_l | \kappa_l)$ where $\pi_l(\theta_l | \kappa_l) > 0 $ is a local, continuous prior.
        \item[ii)] The non-local prior on $\kappa_l$ is $\pi^{\NLP}_l(\kappa_l) = d_l(\left|\left|\kappa_l - \kappa_{l \setminus k}\right|\right|_2)\pi^{\LP}_l(\kappa_l)$,  %\jack{split this by the cases in the theorem}
            where $\pi^{\LP}_l(\kappa_l)$ satisfies A1, $d_l(z) > 0$ for all $z \neq 0$, and there is a monotonically decreasing function $c(z): \mathbb{R}\mapsto \mathbb{R}$ such that $\lim_{z\rightarrow 0} c(z) = \infty$, $c(z/a) \leq c(z)/c(a)$, and 
            %$||\tilde{\kappa}_l - \kappa_{l \setminus k}||_2 = O_p(\nicefrac{1}{\sqrt{n}})$ it is the case that 
            $\frac{c(z)}{- \log d_l(z)} = O(1)$ as $z\rightarrow 0$. %, as $n \rightarrow\infty$
            % \jack{is this interesting, or just call it $z^k$}
    \end{itemize}  
\end{itemize}  
A7 defines the form of the nested models we consider. Both models share the same parameter space $\Theta$, but the simpler model $k$ has no hyperparameters. An example of A7 is where $k$ is the Gaussian model and $l$ Tukey's loss and $\kappa_{l\setminus k} = \infty$. A8 specifies that a non-local prior is placed on the additional hyperparameter of model $l$ which penalises their being too close to the value that recovered the simpler model $k$ at a rate as least as fast as $c(\cdot)$. The parameters shared by both models are given local priors.

%However, in general that we can also consider members of $\theta$ in a similar vein, for example in a variable selection exercise, and place \NLP's on these also. 
We define the non-local prior adjusted Laplace approximate Bayes factor as
\begin{equation}
    \tilde{B}_{kl}^{\mathcal{H}-\NLP} := \nicefrac{\tilde{\mathcal{H}}_k(y)}{\tilde{\mathcal{H}}^{\LP}_l(y)}\times \nicefrac{1}{d_l(||\tilde{\kappa}_l - \kappa_{l\setminus k}||_2)},\label{Equ:LaplaceApproxHScoreNLP}
\end{equation}
where $\tilde{\mathcal{H}}_k(y)$ and $\tilde{\mathcal{H}}^{\LP}_l(y)$ are Laplace approximations to the marginal likelihoods \eqref{Equ:LaplaceApproxHScore} whose respective priors $\pi_k(\eta_k) = \pi_k(\theta_k)$ and $\pi_l(\eta_l) = \pi_l(\theta_l | \kappa_l)\pi_l^{\LP}(\kappa_l)$ satisfy the conditions of A1 according to A8. This allows us to invoke Theorem \ref{Thm:HBayesFactorConsistency} for the asymptotic behaviour of $\nicefrac{\tilde{\mathcal{H}}_k(y)}{\tilde{\mathcal{H}}^{\LP}_l(y)}$. Some justification for this objective is provided in Remark \ref{Rem:JustificationNLPBF} at the end of this section.

The following Theorem proves that \eqref{Equ:LaplaceApproxHScoreNLP} 
maintains consistent model selection of the \HBayes{} factor as proved in Theorem \ref{Thm:HBayesFactorConsistency}, and can improve the rate at which the simpler of two nested models is selected when it is sufficient for minimising Fisher's divergence.

\begin{theorem}[$\mathcal{H}$-Bayes Factor Consistency under Non-Local Priors]
Assume Conditions A1-A4, models satisfying A7 and non-local priors satisfying A8. %\jack{z's in the theorem}
\begin{enumerate}[label=(\roman*)]
    \item When $\mathbb{E}_g[H(z; f_l(\cdot; \eta^{\ast}_l))] - \mathbb{E}_g[H(z; f_k(\cdot; \eta^{\ast}_k))] < 0$ then 
    \begin{align}
        \frac{1}{n}\log \tilde{B}^{\mathcal{H}-\NLP}_{kl} = \mathbb{E}_g[H(z; f_l(\cdot; \eta^{\ast}_l))] - \mathbb{E}_g[H(z; f_k(\cdot; \eta^{\ast}_k))] + o_p(1).
    \end{align}
    That is that when the more complex model $l$ decreases Fisher's divergence relative to model $k$, the non-local \HBayes{} factor accrues evidence in favour of model $l$, $\tilde{B}^{\mathcal{H}-\NLP}_{kl}\rightarrow 0$, at an exponential rate. 
    \item When $\mathbb{E}_g[H(z; f_l(\cdot; \eta^{\ast}_l))] = \mathbb{E}_g[H(z; f_k(\cdot; \eta^{\ast}_k))]$, with $k$ being the simpler model, there exists a constant $M_d$ such that with arbitrarily high probability 
    \begin{align}
    \log \tilde{B}_{kl}^{\mathcal{H}-\NLP} > c(1/\sqrt{n})\left(\frac{d_l - d_k}{2}\frac{\log(n)}{c(1/\sqrt{n})} + o_p(1) + M_d\right).\label{Equ:logBNLP_case2}
    \end{align}
    That is that when the models are equally preferable according to Fisher's divergence then $\tilde{B}^{\mathcal{H}-\NLP}_{kl}\rightarrow \infty$ at a rate depending on the \NLP specification via the function $c(\cdot)$.
\end{enumerate}
\label{Thm:HscoreConsistencyNLPs}
\end{theorem}

%\jack{Simplify this proof based on what David did for the corollary}

\begin{proof}

To prove this we decompose the log non-local \HBayes{} factor in the local \HBayes{} factor, which was dealt with by Theorem \ref{Thm:HBayesFactorConsistency}, and the non-local penalty term. Then, we consider the two cases of the theorem. When the more complicated model is true, we show that the non-local penalty term tends to a positive constant while when the simpler model is sufficient, we show that the properties of the non-local penalty function control the rate of convergence.

Firstly, from \eqref{Equ:LaplaceApproxHScoreNLP}
\begin{align}
    \log \tilde{B}_{kl}^{\mathcal{H}-\NLP} = \log \tilde{B}^{(\mathcal{H})}_{kl} - \log d_l(||\tilde{\kappa}_l - \kappa_{l\setminus k}||_2),\label{Equ:logBNLP}
\end{align}
and by Slutsky's Theorem we can investigate both terms separately. The term $\tilde{B}^{(\mathcal{H})}_{kl}$ was constructed so that it satisfied the conditions of Theorem \ref{Thm:HBayesFactorConsistency} and therefore we can invoke these results here. Further, $\tilde{\eta}_l = \{\tilde{\theta}_l, \tilde{\kappa}_l\}$ where defined as the \hyvarinen score minimisers penalised by the local prior, therefore satisfying the conditions of Proposition \ref{Lemma:HscoreParamConsistency} and ensuring that  $||\tilde{\eta}_j - \eta_j^{\ast}||_2 = O_p(\nicefrac{1}{\sqrt{n}})$. Next, we consider the two cases.

\noindent\textbf{Case 1)}

From Theorem \ref{Thm:HBayesFactorConsistency} we have that 
    \begin{align}
        \frac{1}{n}\log \tilde{B}^{(\mathcal{H})}_{kl} = \mathbb{E}_g[H(z; f_l(\cdot; \eta^{\ast}_l))] - \mathbb{E}_g[H(z; f_k(\cdot; \eta^{\ast}_k))] + o_p(1).\nonumber
    \end{align}
Further, for $\mathbb{E}_g[H(z; f_l(\cdot; \eta^{\ast}_l))] - \mathbb{E}_g[H(z; f_k(\cdot; \eta^{\ast}_k))] < 0$ it must be the case that $\kappa^{\ast}_l \neq \kappa_{l\setminus k}$ (by A7). Therefore, $||\tilde{\kappa}_l - \kappa_{l\setminus k}||_2 \nrightarrow 0$ and by the continuous mapping theorem we have that $|d(||\tilde{\kappa}_l - \kappa_{l\setminus k}||_2) - d_l(||\kappa^{\ast}_l - \kappa_{l\setminus k}||_2)| = O_p(\nicefrac{1}{\sqrt{n}})$ with  $d_l(||\kappa^{\ast}_l - \kappa_{l\setminus k}||_2) > 0$ as a strictly positive constant (A8ii)). As a result $\frac{1}{n}d(||\tilde{\kappa}_l - \kappa_{l\setminus k}||_2) = O(1/n))$ and
\begin{align}
    \frac{1}{n}\log \tilde{B}^{\mathcal{H}-\NLP}_{kl} = \mathbb{E}_g[H(z; f_l(\cdot; \eta^{\ast}_l))] - \mathbb{E}_g[H(z; f_k(\cdot; \eta^{\ast}_k))] + o_p(1).\nonumber
\end{align}

\noindent\textbf{Case 2)}

In this case, $\mathbb{E}_g[H(y; f_l(\cdot; \eta^{\ast}_l))] = \mathbb{E}_g[H(y; f_k(\cdot; \eta^{\ast}_k))] = 0$ which by A7 requires that $\kappa^{\ast}_l = \kappa_{l\setminus k}$. Therefore, $||\tilde{\kappa}_l - \kappa_{l\setminus k}||_2 = O_p(\nicefrac{1}{\sqrt{n}})$, which by the continuous mapping theorem implies $d_l(||\tilde{\kappa}_l - \kappa_{l\setminus k}||_2) \rightarrow0$. Here we establish the asymptotic rate of this convergence.

By \eqref{Equ:logBNLP}, in order to establish \eqref{Equ:logBNLP_case2} we require that
\begin{align}
\forall \epsilon,\quad \exists M_d, N_d > 0 \textrm{ such that } P\left(- \log d_l(||\tilde{\kappa}_l - \kappa_{l\setminus k}||_2) > M_dc(1/\sqrt{n})\right) > 1 - \epsilon,\quad \forall n > N_d  \nonumber  
\end{align}
From A8ii), $\frac{c(z)}{- \log d_l(z)} = O(1)$ as $z\rightarrow 0$ which guarantees the existence of $M_{\delta_{\kappa}}$ and $\delta_{\kappa}>0$ such that 
\begin{equation}
\frac{- \log d_l(z)}{c(z)} > M_{\delta_{\kappa}} \textrm{ for all } z \textrm{ such that } |z| < \delta_{\kappa}. \nonumber
\end{equation}
As a result, by the law of total probability 
\begin{align}
    &P(- \log d_l(||\tilde{\kappa}_l - \kappa_{l\setminus k}||_2) > M_dc(1/\sqrt{n})) \nonumber\\
    =& P(-\log d_l(||\tilde{\kappa}_l - \kappa_{l\setminus k}||_2) > M_dc(1/\sqrt{n}) \mid ||\tilde{\kappa}_l - \kappa_{l\setminus k}||_2\leq \delta_{\kappa})P(||\tilde{\kappa}_l - \kappa_{l\setminus k}||_2 \leq \delta_{\kappa}) \nonumber\\
    &+ P(-\log d_l(||\tilde{\kappa}_l - \kappa_{l\setminus k}||_2) > M_dc(1/\sqrt{n})\mid ||\tilde{\kappa}_l - \kappa_{l\setminus k}||_2 > \delta_{\kappa})P(||\tilde{\kappa}_l - \kappa_{l\setminus k}||_2 > \delta_{\kappa}) \nonumber\\
    \geq&  P(c(||\tilde{\kappa}_l - \kappa_{l\setminus k}||_2) > \frac{M_d}{M_{\delta_{\kappa}}}c(1/\sqrt{n}) \mid ||\tilde{\kappa}_l - \kappa_{l\setminus k}||_2\leq \delta_{\kappa})P(||\tilde{\kappa}_l - \kappa_{l\setminus k}||_2 \leq \delta_{\kappa}). \nonumber
\end{align}
Given that $||\tilde{\kappa}_l - \kappa_{l\setminus k}||_2 = O_p(\nicefrac{1}{\sqrt{n}})$ from Proposition \ref{Lemma:HscoreParamConsistency} and $\delta_{\kappa}$ is a constant, the term $P(||\tilde{\kappa}_l - \kappa_{l\setminus k}||_2 \leq \delta_{\kappa})$ is arbitrarily close to 1 as $n\rightarrow\infty$. Regarding the other terms, by the decreasing monotonicity of $c(z)$ (A8ii)) the conditioning events is such that 
\begin{align}
\left\lbrace||\tilde{\kappa}_l - \kappa_{l\setminus k}||_2 \leq \delta_{\kappa}\right\rbrace = \left\lbrace c\left(||\tilde{\kappa}_l - \kappa_{l\setminus k}||_2\right) \geq c(\delta_{\kappa})\right\rbrace,\nonumber
\end{align}
and therefore for $\frac{M_d}{M_{\delta_{\kappa}}}c(1/\sqrt{n}) \geq c(\delta_{\kappa})$ we have that 
%\begin{align}
%    &P(c(||\tilde{\kappa}_l - \kappa_{l\setminus k}||_2) > \frac{M_d}{M_{\delta_{\kappa}}}c(1/\sqrt{n}) \mid ||\tilde{\kappa}_l - \kappa_{l\setminus k}||_2\leq \delta_{\kappa})P(||\tilde{\kappa}_l - \kappa_{l\setminus k}||_2 \leq \delta_{\kappa})\nonumber\\
%    =& P(c(||\tilde{\kappa}_l - \kappa_{l\setminus k}||_2) > \frac{M_d}{M_{\delta_{\kappa}}}c(1/\sqrt{n}) \mid c(||\tilde{\kappa}_l - \kappa_{l\setminus k}||_2) > c(\delta_{\kappa}))P(c(||\tilde{\kappa}_l - \kappa_{l\setminus k}||_2) \geq c(\delta_{\kappa}))\nonumber\\
%    \geq& P(c(||\tilde{\kappa}_l - \kappa_{l\setminus k}||_2) > \frac{M_d}{M_{\delta_{\kappa}}}c(1/\sqrt{n}))P(c(||\tilde{\kappa}_l - \kappa_{l\setminus k}||_2) \geq c(\delta_{\kappa}))\nonumber
%\end{align}
\begin{align}
    P\left(c(||\tilde{\kappa}_l - \kappa_{l\setminus k}||_2) > \frac{M_d}{M_{\delta_{\kappa}}}c(1/\sqrt{n}) \mid c(||\tilde{\kappa}_l - \kappa_{l\setminus k}||_2) > c(\delta_{\kappa})\right) \geq P\left(c(||\tilde{\kappa}_l - \kappa_{l\setminus k}||_2) > \frac{M_d}{M_{\delta_{\kappa}}}c(1/\sqrt{n})\right).\nonumber
\end{align}
Since $||\tilde{\kappa}_l - \kappa_{l\setminus k}||_2 = O_p(\nicefrac{1}{\sqrt{n}})$, by definition for every $\epsilon' > 0$ there exists $M_{\epsilon'} >0$ and $N_{\epsilon'} > 0$ such that 
\begin{align}
P(\sqrt{n}||\tilde{\kappa}_l - \kappa_{l\setminus k}||_2  < M_{\epsilon'}) > 1 - \epsilon',\quad \forall n > N_{\epsilon}. \nonumber
\end{align}
Further, since $c(z)$ is strictly decreasing and $c(z/a) \leq c(z)/c(a)$ by Assumption A8ii), we have that
\begin{align}
P(c\left(||\tilde{\kappa}_l - \kappa_{l\setminus k}||_2 \right) / c(1/\sqrt{n}) > c(M_{\epsilon'}))
\geq P(c\left(||\tilde{\kappa}_l - \kappa_{l\setminus k}||_2 /(1/\sqrt{n})\right) > c(M_{\epsilon'})) > 1 - \epsilon'\quad  \forall n > N_{\epsilon'}.\nonumber
%\Rightarrow& P(c\left(||\tilde{\kappa}_l - \kappa_{l\setminus k}||_2 \right) > c(\delta_{\epsilon})c(1/\sqrt{n})) \leq \epsilon\quad  \forall n > N_{\epsilon}\nonumber\\
%\Rightarrow&  P(c\left(||\tilde{\kappa}_l - \kappa_{l\setminus k}||_2 \right) / c(1/\sqrt{n}) > c(M_{\epsilon})) > \epsilon\quad  \forall n > N_{\epsilon}\nonumber
\end{align}
Consider the particular choice $M_d = c(M_{\epsilon'})M_{\delta_{\kappa}}$ so that $\frac{M_d}{M_{\delta_{\kappa}}} = c(M_{\epsilon'})$ giving that
\begin{align}
    P(c\left(||\tilde{\kappa}_l - \kappa_{l\setminus k}||_2 \right) > \frac{M_d}{M_{\delta_{\kappa}}} c(1/\sqrt{n})) > 1 - \epsilon'\quad  \forall n > N_{\epsilon'}\nonumber
\end{align}
%Further, $\frac{M_d}{M_{\delta_{\kappa}}}c(1/\sqrt{n}) \geq c(\delta_{\kappa})$
%\begin{align}
%        &P(c\left(||\tilde{\kappa}_l - \kappa_{l\setminus k}||_2 \right) > \frac{M_d}{M_{\delta_{\kappa}}} c(1/\sqrt{n})) > 1 - \epsilon\quad  \forall n > N_{\epsilon}\nonumber\\
%        \Rightarrow& P(c(||\tilde{\kappa}_l - \kappa_{l\setminus k}||_2) > c(\delta_{\kappa})) > 1 - \epsilon \quad \forall n > \max\{N_{\epsilon}, N_c\}\nonumber
%\end{align}
%where $N_c := \left(\nicefrac{1}{c^{-1}\left(\frac{M_{\delta_{\kappa}}}{M_d}c(\delta_{\kappa})\right)}\right)^2$ and therefore 
Therefore, taking $(1 - \epsilon') > (1 - \epsilon)$
\begin{align}
    P(- \log d_l(||\tilde{\kappa}_l - \kappa_{l\setminus k}||_2) > M_dc(1/\sqrt{n})) > (1 - \epsilon')P(||\tilde{\kappa}_l - \kappa_{l\setminus k}||_2) > (1-\epsilon)\nonumber
\end{align}
for large enough $n$. Further, invoking Theorem \ref{Thm:HBayesFactorConsistency} we have that with arbitrarily high probability there exists constant $M_d$ such that for sufficiently large $n$
\begin{align}
    \log \tilde{B}_{kl}^{\mathcal{H}-NLP} &= \log \tilde{B}^{(\mathcal{H})}_{kl} - \log d_l(||\tilde{\kappa}_l - \kappa_{l\setminus k}||_2)\nonumber\\
    &> \log \tilde{B}^{(\mathcal{H})}_{kl} + M_d \times c(1/\sqrt{n})\nonumber\\
    &= \frac{d_l - d_k}{2}\log(n) + O_p(1) + M_d \times c(1/\sqrt{n})\nonumber\\
    &= c(1/\sqrt{n})\left(\frac{d_l - d_k}{2}\frac{\log(n)}{c(1/\sqrt{n})} + \frac{O_p(1)}{c(1/\sqrt{n})} + M_d\right).\nonumber
\end{align}

\end{proof}

We follow Theorem \ref{Thm:HscoreConsistencyNLPs} with the special case corollary considering the inverse-gamma non-local prior applied to $\nu_2 = \frac{1}{\kappa_2^2}$ where $\kappa_2$ was the cut-off parameter of Tukey's loss (Section \ref{Sec:LPvsNLP}). First, we generalise that prior set-up from Section \ref{Sec:LPvsNLP} as the following condition.
\begin{itemize}
    \item[A9] Under the model setup of A7
    \begin{itemize}
        \item[i)] Let $\kappa_{l} = \{\nu\}\in \mathbb{R}$ be a univariate hyperparameter with $\kappa_{l\setminus k} = \nu_0$ and $\nu \geq \nu_0$. 
        \item[ii)] $\pi_l^{\LP}(\nu) = 2\mbox{I}_{\nu \geq \nu_0}\mathcal{N}(\nu - \nu_0; 0, s_0^2)$, a half-Gaussian distribution with scale $s_0$.
        \item[iii)] $\pi_l^{\NLP}(\nu) = \mathcal{IG}(\nu - \nu_0; a_0, b_0)$, an inverse-gamma distribution with shape $a_0$ and scale $b_0$. 
    \end{itemize}
\end{itemize}
We note that the corresponding local prior to an inverse-gamma non-local prior would normally be a gamma distribution. However, this does not have finite second derivative as $\nu - \nu_0\rightarrow 0$ and would therefore violate the conditions of Theorem \ref{Thm:HBayesFactorConsistency}.

\begin{corollary}[$\mathcal{H}$-Bayes Factor Consistency under an Inverse-Gamma Non-Local Prior]
Assume Conditions A1-A4, models satisfying A7 and non-local priors satisfying A8(i) and A9.
\begin{enumerate}[label=(\roman*)]
    \item When $\mathbb{E}_g[H(z; f_l(\cdot; \eta^{\ast}_l))] - \mathbb{E}_g[H(z; f_k(\cdot; \eta^{\ast}_k))] < 0$ then 
    \begin{align}
        \frac{1}{n}\log \tilde{B}^{\mathcal{H}-\NLP}_{kl} = \mathbb{E}_g[H(z; f_l(\cdot; \eta^{\ast}_l))] - \mathbb{E}_g[H(z; f_k(\cdot; \eta^{\ast}_k))] + o_p(1).
    \end{align}
    That is that when the more complex model $l$ decreases Fisher's divergence relative to model $k$, the non-local \HBayes{} factor accrues evidence in favour of model $l$, $\tilde{B}^{\mathcal{H}-\NLP}_{kl}\rightarrow 0$, at an exponential rate. 
    \item When $\mathbb{E}_g[H(z; f_l(\cdot; \eta^{\ast}_l))] = \mathbb{E}_g[H(z; f_k(\cdot; \eta^{\ast}_k))]$, with $k$ being the simpler model, there exists a constant $M_d$ such that with arbitrarily high probability
    \begin{align}
    \log \tilde{B}_{kl}^{\mathcal{H}-\NLP} > \sqrt{n}(o_p(1) + M_d).
    \end{align}
    That is that when the models are equally preferable according to Fisher's divergence then $\log \tilde{B}^{\mathcal{H}-\NLP}_{kl}\rightarrow \infty$ at a rate at least as fast as $\sqrt{n}$ as $n \rightarrow\infty$.
\end{enumerate}
\label{Thm:HscoreConsistencyNLPsIG}
\end{corollary}

\begin{proof}

To prove this we derive the non-local penalty function associated with the prior setup of A9 and calculate its bounding function $c(\cdot)$ associated with A8ii). This allows us to invoke Theorem \ref{Thm:HscoreConsistencyNLPs} to provide the desired asymptotic behaviour.

Firstly, from A9, $\nu \geq \nu_0$ and therefore $d_l(|\nu - \nu_0|) = d_l(\nu - \nu_0)$. Then, the non-local penalty function $d_l(\nu - \nu_0) = \frac{\pi^{\NLP}_l(\nu)}{\pi^{\LP}_l(\nu)}$ given by the prior specification of A9 is
\begin{align}
    d_l(\nu - \nu_0) &= \frac{\pi_l^{\NLP}(\nu)}{\pi_l^{\LP}(\nu)}\nonumber\\
               %&= \frac{\frac{b_0^{a_0}}{\Gamma(a_0)}\left(\frac{1}{\nu_2}\right)^{a_0 + 1}\exp\left(-\frac{b_0}{\nu_2}\right)}{2\frac{1}{\sqrt{2\pi}\sqrt{s_0^2}}\exp\left(- \frac{\nu_2^2}{s_0^2}\right)}\nonumber\\
               &\propto \frac{b_0^{a_0}\sqrt{2\pi}\sqrt{s_0^2}}{2\Gamma(a_0)}\left(\frac{1}{\nu - \nu_0}\right)^{a_0 + 1}\exp\left(\frac{(\nu - \nu_0)^2}{2s_0^2} - \frac{b_0}{\nu - \nu_0}\right),\nonumber
               %&=O\left(\exp\left(-\frac{b_0}{\nu_2}\right)\right).
\end{align}
and therefore 
\begin{align}
    - \log d_l(\nu - \nu_0) &= - \log \frac{b_0^{a_0}\sqrt{2\pi}\sqrt{s_0^2}}{2\Gamma(a_0)} + (a_0 + 1)\log \left(\nu - \nu_0\right) - \frac{\left(\nu - \nu_0\right)^2}{2s_0^2} + \frac{b_0}{\left(\nu - \nu_0\right)}.\nonumber
\end{align}
Therefore, there exists $\delta_{\nu}>0$ and $M_{\delta_{\nu}} >0$ such that for all $\nu$ s.t. $|\nu - \nu_0| < \delta_{\nu}$
\begin{align}
    - \log d_l(\nu - \nu_0) > M_{\delta_{\nu}}\frac{1}{\nu - \nu_0}, \nonumber%\label{Equ:d_bound}
\end{align}
and as a result, $c(z) = \frac{1}{z}$. Now we consider the two cases.

\noindent\textbf{Case 1)} 

Follows directly from the proof of Theorem \ref{Thm:HscoreConsistencyNLPs}, noting that $d_l(z) > 0$ for all $0 < z < \infty$.

\noindent\textbf{Case 2)}

Using Theorem \ref{Thm:HscoreConsistencyNLPs} with $c(z) = \frac{1}{z}$ provides that with arbitrarily high probability there exists $M_d$ such that for sufficiently large $n$
\begin{align}
    \log \tilde{B}_{kl}^{\mathcal{H}-\NLP} &> c(1/\sqrt{n})\left(\frac{d_l - d_k}{2}\frac{\log(n)}{c(1/\sqrt{n})} + \frac{O_p(1)}{c(1/\sqrt{n})} + M_d\right)\nonumber\\
                                          &= \sqrt{n}\left(\frac{d_l - d_k}{2}\frac{\log(n)}{\sqrt{n}} + \frac{O_p(1)}{\sqrt{n}} + M_d\right)\nonumber\\
                                          &= \sqrt{n}(o_p(1) + M_d),\nonumber
\end{align}
as required.

\end{proof}

\begin{remark}[Justification for \eqref{Equ:LaplaceApproxHScoreNLP}]{\label{Rem:JustificationNLPBF}}
Lastly, we provide some justification for the objective function \eqref{Equ:LaplaceApproxHScoreNLP} to conduct model selection using non-local priors. 
Using the arguments of \cite{rossell2017nonlocal}, the integrated \Hscore{} \eqref{Equ:HScore} associated with the non-local prior specification A8 can be rewritten as     
\begin{align}
    \mathcal{H}^{\NLP}_l(y) &= \int\pi^{\NLP}_l(\kappa_l)\pi_l(\theta_l|\kappa_l) \exp \left\{ -\sum_{i=1}^nH(y_i; f_l(\cdot;\theta_l, \kappa_l)) \right\} d\theta_ld\kappa_l\nonumber\\
    &= \mathcal{H}^{\LP}_l(y)\int d_l(||\kappa_l - \kappa_{l\setminus k}||_2) \frac{\pi^{\LP}_l(\kappa_l)\pi_l(\theta_l| \kappa_l) \exp \left\{ -\sum_{i=1}^nH(y_i; f_l(\cdot;\theta_l, \kappa_l)) \right\} }{\mathcal{H}^{\LP}_l(y)}d\theta_ld\kappa_l\nonumber\\
    &= \mathcal{H}^{\LP}_l(y)\int d_l(||\kappa_l - \kappa_{l\setminus k}||_2) \pi^{\LP}_l(\theta_l, \kappa_l|y)d\theta_ld\kappa_l\nonumber
\end{align}
where
\begin{align}
    \pi^{\LP}_l(\theta_l, \kappa_l|y) &\propto \pi_l^{\LP}(\kappa_l)\pi_l(\theta_l| \kappa_l) \exp \left\{ -\sum_{i=1}^nH(y_i; f_l(\cdot;\theta_j, \kappa_j)) \right\},\nonumber\\
    \mathcal{H}^{\LP}_l(y) &= \int\pi_l^{\LP}(\kappa_l)\pi_l(\theta_l| \kappa_l) \exp \left\{ -\sum_{i=1}^nH(y_i; f_l(\cdot;\theta_l, \kappa_l)) \right\} d\theta_ld\kappa_l,\nonumber
\end{align}
As a result we can write the \HBayes{} factor of nested models satisfying A7 and non-local priors according to A8 %$k$ over $l$ given non-local priors on the extra parameters of the bigger model ($\kappa_{l\setminus k}$) 
as 
\begin{align}
   B_{kl}^{\mathcal{H}-\NLP} = \frac{\mathcal{H}_k(y)}{\mathcal{H}^{\NLP}_l(y)} = \frac{\mathcal{H}_k(y)}{\mathcal{H}^{\LP}_l(y)}\frac{1}{\int d_l(||\kappa_l - \kappa_{l\setminus k}||_2) \pi^{\LP}_l(\theta_l, \kappa_l|y)d\theta_ld\kappa_l}.\nonumber
\end{align}
Now, the terms involved in the construction of $\mathcal{H}_k(y)$ and $\mathcal{H}^{\LP}_l(y)$ satisfy the conditions of Theorem \ref{Thm:HBayesFactorConsistency} (by A8i)), and therefore the behaviour of the Laplace approximations \eqref{Equ:LaplaceApproxHScore} to the \HBayes{} factor,  $\nicefrac{\tilde{\mathcal{H}}_k(y)}{\tilde{\mathcal{H}}^{\LP}_l(y)}$, is detailed in Theorem \ref{Thm:HBayesFactorConsistency}. We further consider a somewhat looser, but convenient, approximation of $\int d_l(||\kappa_l - \kappa_{l\setminus k}||_2) \pi^{\LP}_l(\theta_l, \kappa_l|y)d\theta_ld\kappa_l$ as $d_l(||\tilde{\kappa}_l - \kappa_{l\setminus k}||_2)$ where $\tilde{\kappa}_l \in \tilde{\eta}_l$ are the parameters maximising the \Hposterior{} $\pi^{\LP}_l(\theta_l, \kappa_l|y)$. As a result we consider \eqref{Equ:LaplaceApproxHScoreNLP} to approximate the \HBayes{} factor under non-local priors 
%\begin{equation}
%    \tilde{B}_{kl}^{\mathcal{H}-\NLP} := \nicefrac{\tilde{\mathcal{H}}_k(\*y)}{\tilde{\mathcal{H}}^{\LP}_l(\*y)}\times \nicefrac{1}{d_l(\tilde{\kappa}_l}).\label{Equ:LaplaceApproxHScoreNLP}
%\end{equation}
%This is particularly convenient because $\log \tilde{B}_{kl}^{\mathcal{H}-\NLP} = \log \tilde{B}^{(\mathcal{H})}_{kl} - \log d_l(\tilde{\kappa}_l)$, where Theorem \ref{Thm:HBayesFactorConsistency} has already detailed the asymptotic behaviour of the first term. \jack{repeated}

We highlight that calculating the approximation $\tilde{\mathcal{H}}^{\NLP}_l(y)$ uses the maximum a posteriori estimates and evaluates the observed Hessian matrices according to the \Hposterior{} under the local prior, before evaluating the prior density as it appears in the Laplace approximation \eqref{Equ:LaplaceApproxHScore} at the non-local prior, as shown in  \eqref{Equ:LaplaceApproxHScoreNLP}. This is convenient as it allows for a direct extension of Theorem \ref{Thm:HBayesFactorConsistency}. %In fact, we actually expect such an approach to achieve slightly faster selection of the simpler nested model if it minimises the expected Fisher's divergence to $G$, than a direct Laplace approximation of $\mathcal{H}^{\NLP}_j(\*y)$. This is because $\tilde{\theta}_j, \tilde{\kappa}_l$ will have lower density under the \NLP when it is only optimised according to \LP \jack{Work on this}.
\end{remark}

\subsection{Derivation of \hyvarinen scores}
\label{App:derivation_hscore}

Section \ref{App:HscoreGaussianTukeys} provides the \hyvarinen score for the Gaussian log-likelihood and for Tukey's loss.

\subsubsection{The \hyvarinen score of the Gaussian Model and Tukey's Loss}{\label{App:HscoreGaussianTukeys}}

Here we provide the derivations of the \hyvarinen score applied to the Gaussian model \eqref{Equ:HscoreGaussian} and Tukey's loss \eqref{Equ:HscoreTukeysLoss}. Firstly, for the Gaussian model $f_1(y_i; x_i, \theta_1, \kappa_1) = \mathcal{N}(y_i;x_i^T \beta,\sigma^2)$ we have that 
\begin{align}
\log f_1(y_i; x_i, \theta_1, \kappa_1) &= -\frac{1}{2}\log (2\pi) - \frac{1}{2}\log \sigma^2 -\frac{(y_i - x_i^T \beta)^2}{2\sigma^2}\nonumber\\
\frac{\partial }{\partial y_i}\log f_1(y_i;x_i, \theta_1, \kappa_1) &= -\frac{(y_i - x_i^T \beta)}{\sigma^2}\nonumber\\
\frac{\partial^2 }{\partial y_i^2}\log f_1(y_i;x_i, \theta_1, \kappa_1) &= -\frac{1}{\sigma^2},\nonumber
\end{align}
which results in \hyvarinen score 
\begin{align}
    H_{1}(y_i; f(\cdot;x_i, \theta_1)) &= -\frac{2}{\sigma^2} + \frac{(y_i - x_i^T \beta)^2}{\sigma^4}.\nonumber
\end{align}
For Tukey's loss, $\log f_2(y_i; x_i, \theta_2, \kappa_2) = -\ell_2(y_i; x_i, \theta_2, \kappa_2)$ given in  \eqref{Equ:TukeysLoss} and
\begin{align}
    \frac{\partial}{\partial y_i}\ell_2(y_i; x_i, \theta_2, \kappa_2) &= \begin{cases}
    \frac{(y_i - x_i^T \beta)}{\sigma^2}-\frac{2(y_i - x_i^T \beta)^3}{\kappa_2^2\sigma^4}+\frac{(y_i - x_i^T \beta)^5}{\kappa_2^4\sigma^6} &\textrm{ if } |y_i - x_i^T \beta|\leq \kappa_2\sigma\\
    0 &\textrm{ otherwise}
    \end{cases}\nonumber\\
    \frac{\partial^2}{\partial y_i^2}\ell_2(y_i; x_i, \theta_2, \kappa_2) &= \begin{cases}
    \frac{1}{\sigma^2} - \frac{6(y_i - x_i^T \beta)^2}{\kappa_2^2\sigma^4} + \frac{5(y_i - x_i^T \beta)^4}{\kappa_2^4\sigma^6} &\textrm{ if } |y_i - x_i^T \beta|\leq \kappa_2\sigma\\
    0 &\textrm{ otherwise}
    \end{cases},\nonumber
\end{align}
which results in \hyvarinen score 
\begin{align}
    H_{2}(y_i; f(\cdot;x_i, \theta_2, \kappa_2)) &=  \mbox{I}(|y_i - x_i^T \beta|\leq \kappa_2\sigma)\left\lbrace \left(\frac{(y_i - x_i^T \beta)}{\sigma^2}-\frac{2(y_i - x_i^T \beta)^3}{\kappa_2^2\sigma^4}+\frac{(y_i - x_i^T \beta)^5}{\kappa_2^4\sigma^6}\right)^2 \right.\nonumber\\
    &\qquad\qquad\qquad\left. - 2\left(\frac{1}{\sigma^2}-\frac{6(y_i - x_i^T \beta)^2}{\kappa_2^2\sigma^4}+\frac{5(y_i - x_i^T \beta)^4}{\kappa_2^4\sigma^6}\right)\right\rbrace.\nonumber%\label{Equ:HscoreTukeysLoss}
\end{align}

\subsubsection{The \hyvarinen score of kernel density estimation}{\label{App:HscoreKDE}}

%\jack{call this $h^2$?}

We derive the \hyvarinen score for the kernel density estimation examples implemented in Section \ref{Sec:KDE}. Combining the kernel density estimate improper density \eqref{Equ:KDEPseudoDensity} with the power $w$  \eqref{Equ:KDEpower_w} results in in-sample
\begin{align}
\hat{g}_{w, h}(y_i) =& \left(\frac{1}{n\sqrt{2\pi} h}\sum_{j=1}^n\exp\left(-\frac{(y_i - y_j)^2}{2h^2}\right)\right)^{w},\nonumber
\end{align}
or equivalently defining the kernel density loss function to be the log-density
\begin{align}
\log \hat{g}_{w, h}(y_i) =& w\log\left(\frac{1}{n\sqrt{2\pi} h}\sum_{j=1}^n\exp\left(-\frac{(y_i - y_j)^2}{2h^2}\right)\right)\nonumber\\
%=& w\log\left(\frac{1}{n\sqrt{2\pi\sigma^2}}exp\left(0\right) + \frac{1}{n\sqrt{2\pi\sigma^2}}\sum_{j\neq i}^n\exp\left(-\frac{(x_i - x_j)^2}{2\sigma^2}\right)\right)\\
=& w\log\left(\frac{1}{n\sqrt{2\pi} h} + \frac{1}{n\sqrt{2\pi} h}\sum_{j\neq i}^n\exp\left(-\frac{(y_i - y_j)^2}{2h^2}\right)\right).\nonumber
\end{align}
The \hyvarinen score is then composed of the first an second derivatives of the log-density, given by 
\begin{align}
\frac{\partial}{\partial y_i} \log \hat{g}_{w,h}(y_i) =& w\frac{-\sum_{j\neq i}^n\frac{(y_i - y_j)}{n\sqrt{2\pi}h^{3}}\exp\left(-\frac{(y_i - y_j)^2}{2h^2}\right)}{\left(\frac{1}{n\sqrt{2\pi} h} + \frac{1}{n\sqrt{2\pi h}}\sum_{j\neq i}^n\exp\left(-\frac{(y_i - y_j)^2}{2h^2}\right)\right)}\nonumber\\
%\end{align}
%
%\begin{align}
\frac{\partial^2}{\partial y_i^2} \log \hat{g}_{w, h}(y_i) =& w\frac{\sum_{j\neq i}^n\frac{(y_i - y_j)^2}{n\sqrt{2\pi}h^{5}}\exp\left(-\frac{(y_i - y_j)^2}{2h^2}\right) - \sum_{j\neq i}^n\frac{1}{n\sqrt{2\pi}h^{3}}\exp\left(-\frac{(y_i - y_j)^2}{2h^2}\right)}{\left(\frac{1}{n\sqrt{2\pi} h} + \frac{1}{n\sqrt{2\pi} h}\sum_{j\neq i}^n\exp\left(-\frac{(y_i - y_j)^2}{2h^2}\right)\right)}\nonumber\\
&\quad - w\frac{\left(\sum_{j\neq i}^n\frac{(y_i - y_j)}{n\sqrt{2\pi}h^{3}}\exp\left(-\frac{(y_i - y_j)^2}{2h^2}\right)\right)^2}{\left(\frac{1}{n\sqrt{2\pi} h} + \frac{1}{n\sqrt{2\pi} h}\sum_{j\neq i}^n\exp\left(-\frac{(y_i - y_j)^2}{2h^2}\right)\right)^2}.\nonumber
\end{align}
As a result, the \hyvarinen score of the kernel density estimate is 
\begin{align}
H(y_i, \log \hat{g}_{w,h}(\cdot))
&= 2w\frac{\sum_{j\neq i}^n\frac{(y_i - y_j)^2}{h^{4}}\exp\left(-\frac{(y_i - y_j)^2}{2h^2}\right) - \sum_{j\neq i}^n\frac{1}{h^{2}}\exp\left(-\frac{(y_i - y_j)^2}{2h^2}\right)}{\left(1 + \sum_{j\neq i}^n\exp\left(-\frac{(y_i - y_j)^2}{2h^2}\right)\right)}\nonumber\\
&\quad -(2w - w^2)\frac{\left(\sum_{j\neq i}^n\frac{(y_i - y_j)}{h^{2}}\exp\left(-\frac{(y_i - y_j)^2}{2h^2}\right)\right)^2}{\left(1 + \sum_{j\neq i}^n\exp\left(-\frac{(y_i - y_j)^2}{2h^2}\right)\right)^2}.
\end{align}
%It is illustrative to note that the numerators of both terms in $H(y_i, \log \hat{g}_{w,h}(\cdot))$ sum over $j \neq i$. These are conceptually similar to a `leave-one-out' likelihood \citep{habbema1974stepwise, robert1976choice} and provide some idea as to why the \hyvarinen score is effective here.

\color{black}
\subsection{The Tsallis score}{\label{App:Tsallis}}

The robust regression example introduced in Section \ref{Sec:TukeysLossIntro} focused on Tukey's loss as this provided a traditional alternative to the Gaussian likelihood that had one parameter, was nested in the Gaussian loss and corresponded to a improper model. A loss function satisfying similar properties is the Tsallis score \citep{tsallis1988possible} (often referred to as the density power divergence \citep{basu1998robust} or $\beta$-divergence) which applied to proper probability model $f(y_i; \theta)$ is given by
\begin{align}
   \ell_{\beta}(y_i, \theta) := - \frac{1}{\beta} f(y_i; \theta)^{\beta} + \frac{1}{\beta + 1}\int  f(z; \theta)^{\beta+1}dz.
\end{align}
As $\beta \rightarrow 0$ then $\ell_{\beta}(y_i, \theta) \rightarrow -\log f(y_i, \theta)$ and therefore $f(y_i; \theta)$ is nested within $\exp\left\{-\ell_{\beta}(y_i, \theta)\right\}$. Generally, $\exp\left\{-\ell_{\beta}(y_i, \theta)\right\}$ is also improper for $\beta > 0$. For $y_i$ such that $f(y_i;\theta)\rightarrow 0$, it can be seen that $\exp\left\{-\ell_{\beta}(y_i, \theta)\right\} \rightarrow \exp\left\{-\frac{1}{\beta + 1}\int  f(z; \theta)^{\beta+1}dz\right\}$ which is a constant in $y_i$ and strictly greater than 0.

Figure \ref{Fig:TsallisTukeysLoss} plots $\ell_{\beta}(y_i, \theta)$ where $f(y_i; \theta) = \mathcal{N}(y_i; \mu, \sigma^2)$ and its corresponding pseudo density. Comparing this with the Tukey's loss in Figure \ref{Fig:TukeysLoss}, we can see that the Tsallis score behaves very similarly. In fact since the first version of this paper appeared, \cite{yonekura2021adaptation} provided methodology analogous to that developed in this paper to set $\beta$ in a data-driven manner.

%selecting_loss_jointHyvarinen_properExperiments.Rmd
%selecting_loss_jointHyvarinen_LaplaceApprox_convergence_n_Tukeys
\begin{figure}
\begin{center}
%\includegraphics[trim= {0.0cm 0.0cm 0.0cm 0.0cm}, clip,  
%width=0.49\columnwidth]{figures/Tukeys_loss_function-1.pdf}
%\includegraphics[trim= {0.0cm 0.0cm 0.0cm 0.0cm}, clip,  
%width=0.49\columnwidth]{figures/Tukeys_pseudo_density-1.pdf}
\includegraphics[trim= {0.0cm 0.0cm 0.0cm 0.0cm}, clip,  
width=0.49\columnwidth]{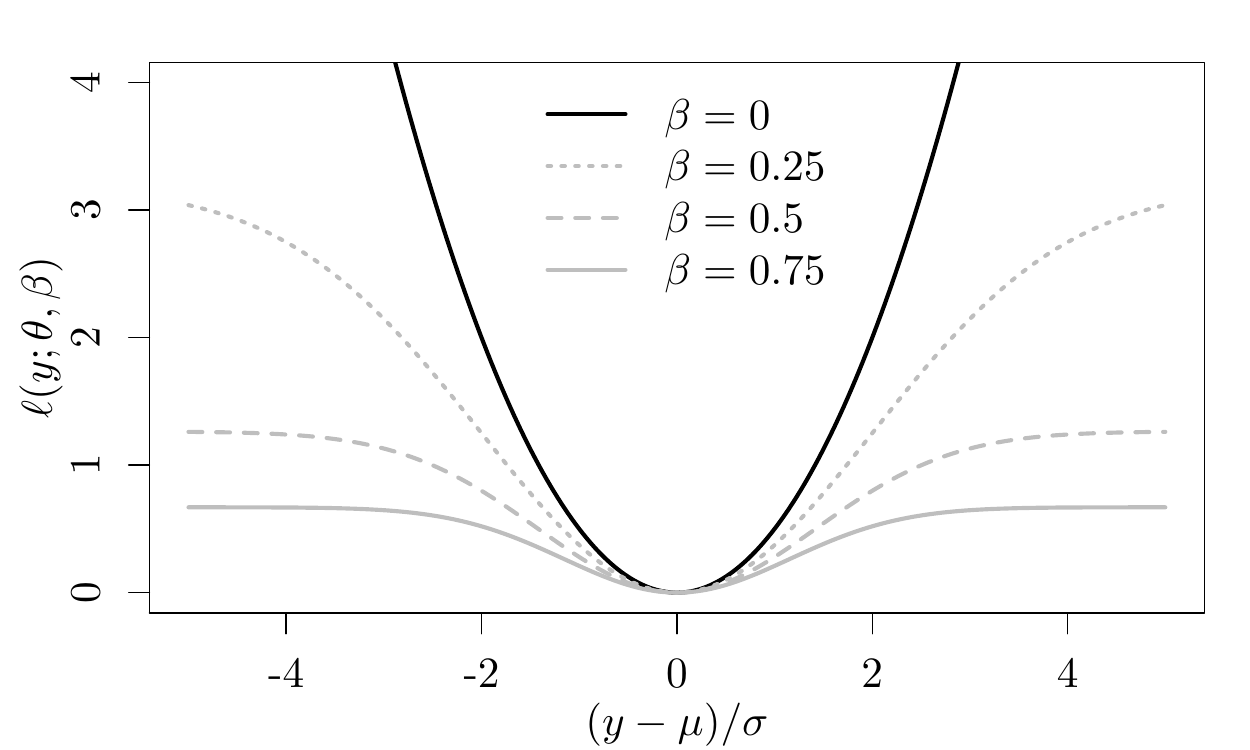}
\includegraphics[trim= {0.0cm 0.0cm 0.0cm 0.0cm}, clip,  
width=0.49\columnwidth]{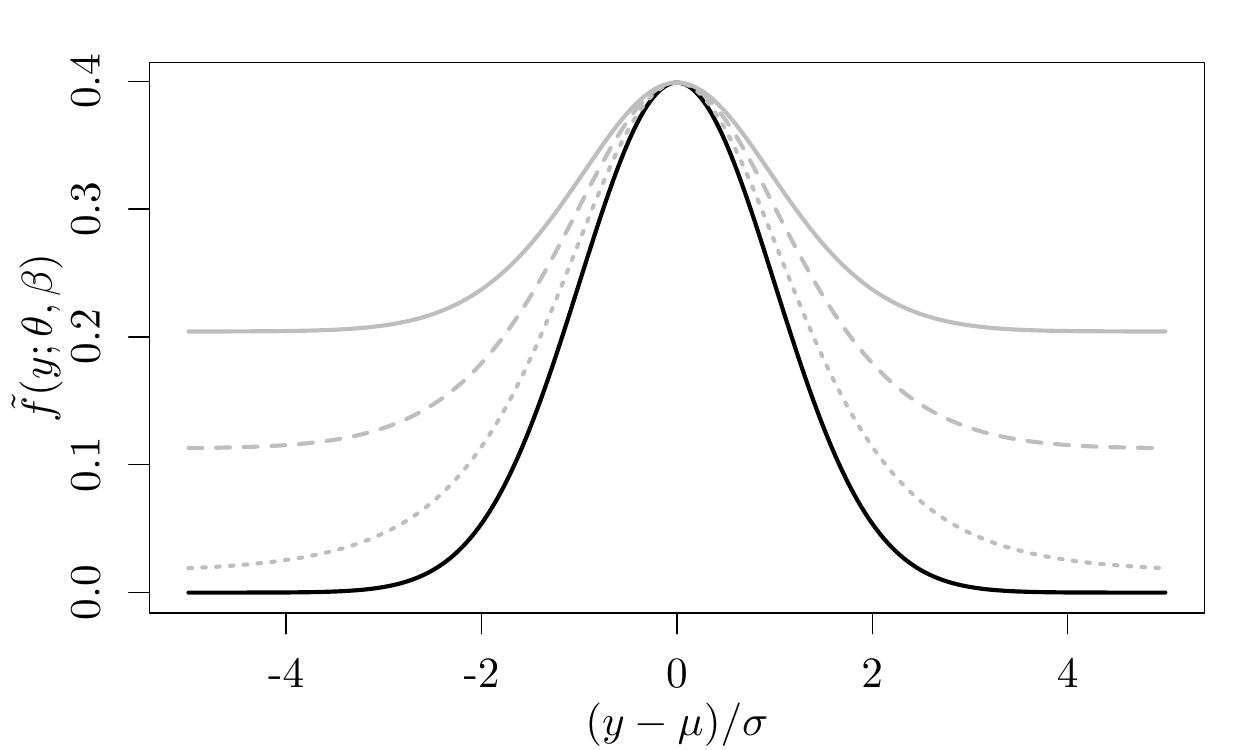}
%trim={<left> <lower> <right> <upper>}
%\caption{Squared-error loss ($\kappa_2 = \infty$) and Tukey's loss ( \eqref{Equ:TukeysLoss}) for $\kappa_2 = 2$, $3$ and $4$ (\textbf{left}) and corresponding (pseudo-)probability densities (\textbf{right}). The pseudo-densities for Tukey's loss were scaled to match the mode of the Gaussian density}
\caption{Squared-error loss and Tsallis score applied to a Gaussian model (\textbf{left}) and corresponding (improper) densities (\textbf{right}). The improper densities for Tsallis score are scaled to match the mode of the Gaussian density.}
\label{Fig:TsallisTukeysLoss}
\end{center}
\end{figure}

\color{black}
\subsection{Implementation Details}
\label{App:implementarion_details}

This sections provides a thorough description of the experiments conducted in Section \ref{Sec:RobustRegression} and \ref{Sec:KDE} of the main paper.
Section \ref{App:computation} outlines our implementation of the \Hscore{} in Stan.
Section \ref{App:TukeysBreakdown} discusses a restriction placed on Tukey's loss parameters to prevent degenerate parameter estimates based on including a single observation in the fit. 
Given that Tukey's loss is not twice continuously differentiable, Section \ref{App:AbsApprox} discusses a differentiable approximation to Tukey's loss that allows using second order sampling and optimisation methods, as well as satisfying the conditions of Theorem \ref{Thm:HBayesFactorConsistency}.
Finally, Section \ref{App:NLPSpecification} discusses how to set a non-local prior on Tukey's loss cut-off hyper-parameter $\kappa_2$.

Code to reproduce the examples of Sections \ref{Sec:RobustRegression} and \ref{Sec:KDE} can be found in the repository %\href{https://github.com/jejewson/HyvarinenImproperModels}{https://github.com/jejewson/HyvarinenImproperModels}.
\url{https://github.com/jejewson/HyvarinenImproperModels}.

\subsubsection{Computation}
\label{App:computation}

%\david{I summarized the exposition of the breakdown point. The computation section can be moved to the supplement. In fact, part of what you discuss (LBFGS, stan) applies to other examples beyond Tukey, right? When discussing computation clarify what's only for Tukey, and what's our more general strategy.}
%Using unnormaliseable probability models do not in general lead to unnormaliseable posteriors, although care should be taken to only consider priors where this is not the case. As a result, standard computational techniques continue to be useful for the \Hposterior \eqref{Equ:GeneralBayesRuleHScore}. In order 
To calculate the Laplace approximations of the \Hscore{} we used the second order LBFGS optimisation routine implemented in the probabilistic programming language Stan \citep{carpenter2016stan}. %Which as well as performing excellently, allowed for easy implementation of the breakdown constraint and can also provide estimates for the required Hessian matrices. We note more generally, that
Further, \stan's implementation of the No U-Turn Sampler \citep{hoffman2014no} can further be used to sample from the General Bayesian posteriors. %A fruther convenient consequence of this approach is it is straight-forward  to then use bridge-sampling \citep{meng1996simulating,meng2002warp} and in particular the \R package \texttt{bridgesampling} \citep{gronau2017bridgesampling} to calculate the \Hscore,  \eqref{Equ:HScore}. 
The code associated with all of our experiments is available as part of the supplementary material. 

\subsubsection{The breakdown of Tukey's loss}{\label{App:TukeysBreakdown}}

%\textbf{The one argument I think we are missing!!}
%\jack{I really need to get the idea across that the Fisher's divergence minimiser is a well defined and useful target parameter for an unnormaliseable loss i.e. the expectation of the loss is minimised in a useful place, where the KLD minimiser is not, this in expectation is minimised in a horrible place. This is definitely true for Tukey's, what about the KDE - well I can just actually calculate this for the KDE as it only has one parameter, well Tukey's loss if we only learn the kappa has one parameter as well, ahh no cuase we need to integrate over the dataset not the parameters!!). Now clearly we are not able to minimise the loss in expectation, we have to minimise a finite sample approximation, and this is where we run into some problems, like in the case of the mixture model (well exactly analaous), I think this is basically what gives us the long right tails of the likelihood ratio distribution (can cite dawid here). And therefore we may need to apply some things to stop this happening, i.e. breakdown - which is an incredibly natural finite sample thing to do. But we still need to argue for there being a `useful' second mode in the finite sample setting, i.e. that we don't usually hit the breakdown threshol, which is true actually in experiments!!!}

%However, a broader, but not unrelated problem is that of the breakdown of robust estimators. 
As discussed in Section \ref{Sec:RobustRegression}, we define a restriction of the parameter space of Tukey's loss $(\beta,\sigma,\kappa_2)$ that prevents degenerate solutions where $\kappa_2$ is set to exclude all observations but one. Intuitively, the idea is to avoid $\kappa_2$ being too small. More specifically, we motive the restriction via a related notion of the breakdown point.
The breakdown of an estimator is the number of observations that can be arbitrarily perturbed without causing arbitrary changes to the estimator. The squared loss can be seen to have breakdown 0, while clearly for Tukey's-loss any observation outside the threshold can be perturbed arbitrarily without changing the estimator. 
There exists abundant literature arguing that
%It is widely accepted throughout statistics that 
the breakdown can be no greater than $1/2$, which intuitively means that an estimator should depend on more than half of the data. %Clearly when we estimate Tukey's loss we wish the same to hold. 
See \cite{rousseeuw1984robust} and \cite{rousseeuw2005robust} for discussion and thorough investigation of Tukey's loss breakdown in the context of S-estimation. %S-estimation extends M-estimation to learn the scale of the data also. \cite{rousseeuw1984robust, rousseeuw2005robust} in order 
The authors showed that for Tukey's breakdown to be less than $1/2$ the condition in \eqref{Equ:BreakdownConstraint} must hold. 
Therefore, we restrict the parameter space to $(\beta, \sigma, \kappa_2)$ satisfying \eqref{Equ:BreakdownConstraint}.
%Although seeking the minimise the \Hscore{} is a different approach to S-estimation we argue such a criterion must still hold and therefore we place such a constraint on our Bayesian learning.
%We note that \eqref{Equ:BreakdownConstraint3} is necessary, but in general not sufficient, to ensure that the breakdown point is $\leq 1/2$, and also that \eqref{Equ:BreakdownConstraint} applies to S-estimation rather than to our H-scores. 
We note that \eqref{Equ:BreakdownConstraint} applies to S-estimation rather than to our \Hscore s however, we found that \eqref{Equ:BreakdownConstraint} provides a simple rule that showed good performance in our examples, hence we recommend it as a default.

\subsubsection{A differentiable approximation to the indicator}{\label{App:AbsApprox}}

%Since these methods require require twice continuously differentiable log-posteriors we consider a continuously differentiable approximation to the indicator $\mathbb{I}_{|y - x^T \beta|\leq \kappa_2\sigma}$ . This also ensured that the Hessians in the Laplace approximation satisfy the 

%\jack{THIS IS NEEDED TO SATISFY THE CONDITIONS OF THE THEOREM}

In order to satisfy the conditions of Theorem \ref{Thm:HBayesFactorConsistency} as well as facilitate second order optimisation (or sampling) methods, it is necessary to have a log-density that is twice continuously differentiable. While Tukey's loss itself is defined to ensure it that this is the case, the \hyvarinen score applied to Tukey's loss, which is a function of its first and second derivatives is not. This is because the indicator $\mbox{I}(|y_i - x_i^T \beta|\leq \kappa_2\sigma)$ %$\mathbb{I}_{|y_i - x_i^T \beta|\leq \kappa_2\sigma}$ 
is not differentiable at $|y_i - x_i^T \beta| = \kappa_2\sigma$. We ensured the required conditions by the approximating
\begin{align}
    |x| &\approx \sqrt{x^2 + \frac{1}{k_1}} \textrm{ for large } k_1\nonumber\\
    %\mathbb{I}_{x\geq 0} &\approx \sigma_{k_2}(x) = \frac{1}{1 + \exp( - k_2x)} \textrm{ for large } k_2\nonumber
    \mbox{I}(x\geq 0) &\approx \sigma_{k_2}(x) = \frac{1}{1 + \exp( - k_2x)} \textrm{ for large } k_2\nonumber
\end{align}
in \eqref{Equ:HscoreTukeysLoss}. % after applying the \Hscore{} to Tukey's loss
For our results we set $k_1 = k_2 = 100$.

\subsubsection{Non-Local Prior Specification for Tukey's cut-off hyper-parameter}{\label{App:NLPSpecification}}

To place a non-local prior on $\kappa_2$, it is convenient to consider the reparametrisation  $\nu_2 = \frac{1}{\kappa_2^2}$ (as discussed in Section \ref{Sec:HscoreModelConsistency}, this is also needed to ensure that the Hessian of Tukey's loss is finite as $\kappa_2\rightarrow\infty$ ($\nu_2\rightarrow 0$)). 
Given that the Gaussian model is recovered at $\nu_2 = 0$, any prior setting positive density at $\nu_2=0$ satisfies the definition of being a local prior. In our examples we considered a half-Gaussian prior $\pi_2^{\LP}(\nu_2) \propto \mathcal{N}(\nu_2; 0, s_0^2) I(\nu_2\geq 0)$ where $I()$ is the indicator function and $s_0 = 1$. 
A possible alternative that also assigns non-zero density at $\nu_2=0$ is to consider a gamma prior with shape parameter $<1$, but the associated log-prior does not have a finite Hessian as $\nu_2\rightarrow 0$ and violates Condition A3 of Theorem \ref{Thm:HBayesFactorConsistency}. 

As a non-local prior we considered an inverse-gamma distribution $\pi_2^{\NLP}(\nu_2) = \mathcal{IG}(\nu_2; a_0, b_0)$, which has zero density at $\kappa_2=0$. The values of the prior hyper-parameters $(a_0, b_0)$ can have a considerable affect on the finite sample performance of the \Hscore{} model selection, and therefore setting these in a principled manner is important. Fortunately, $\kappa_2$ has a natural interpretation that allows restricting attention to a range of values that would be reasonable in most applications. 
Recall from Section \ref{Sec:TukeysLossIntro} that the role of $\kappa_2$ is to exclude observations that are more than $\kappa_2$ `standard deviations' $\sigma$ away from the mean $\mu$.
First, if the data are Gaussian then there is $>0.99$ probability of an observation being within $3\sigma$'s of the mean. We thus set $(a_0,b_0)$ such that there is high prior probability that $\kappa_2 \leq 3$.
Second, given that the Gaussian has a substantial amount of central data within 1$\sigma$ from the mean, we set high prior probability that $\kappa_2 > 1$, to ensure that this central component is captured. 
%Further, we want to make sure that Tukey's loss is correctly capturing the central component of the data also. 
Given these considerations, we set default parameter values $(a_0, b_0) = (4.35, 1.56)$, which ensures a prior probability $P(\kappa_2 \in (1, 3)) = 0.95$. Figure \ref{Fig:Priors} plots these local and non-local priors for $\nu_2$ and the corresponding priors under the original parametrisation, $\kappa_2$. %The log-density plot clearly shows that the local prior assigns much greater density to regions of larger $\kappa_2$.

%\david{I think there's a problem with Fig A.1. I thought that the local prior was a half-Gaussian, I don't see that in the figure. More problematically, the local prior density vanishes at 0, that means it's a non-local prior! I also edited the caption to indicate the exact priors that we use. btw, I suggest removing the log-density, not really needed. }

%\jack{They are local and non-local on $\nu$, I will add pictures, check David didn't make any edits to the section where we talk abotu this }

%selecting_loss_jointHyvarinen_LaplaceApprox_convergence_n_Tukeys.Rmd - updated
\begin{figure}[h!]
\centering
\includegraphics[trim= {0.0cm 0.0cm 0.0cm 0.0cm}, clip,width=0.49\columnwidth]{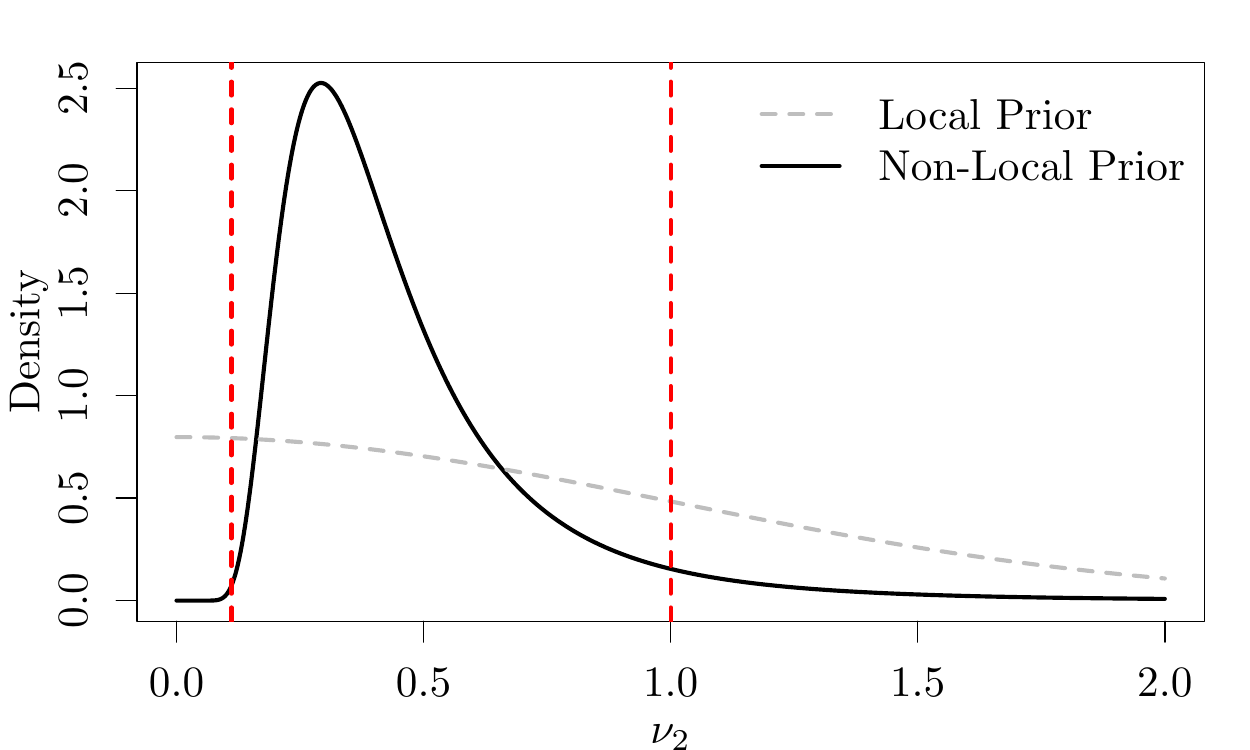}
\includegraphics[trim= {0.0cm 0.0cm 0.0cm 0.0cm}, clip,width=0.49\columnwidth]{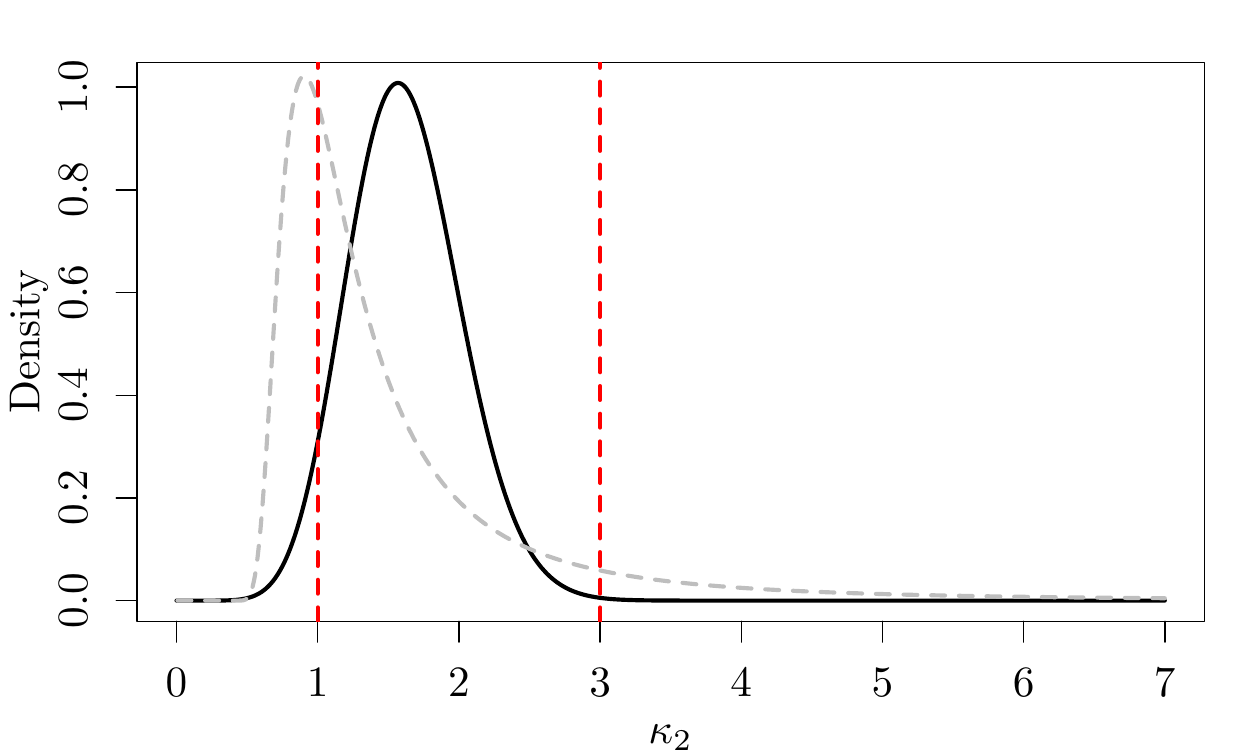}
%\includegraphics[trim= {0.0cm 0.0cm 0.0cm 0.0cm}, clip,width=0.49\columnwidth]{figures/Consistency_LP_vs_NLP/pres_Tukeys_Loss_NLP_LP_nu-2.pdf}
%\caption{Top: Density (left) and log-density (right) of the local half-standard Gaussian and non-local $ \mathcal{IG}(\nu_2; 4.35, 1.56)$ priors for $\nu_2$. Bottom: The corresponding prior density (left) and log-density (right) for $\kappa_2 = \nicefrac{1}{\sqrt{\nu_2}}$. The non-local prior sets $P(\kappa \in (1,3))=0.95$.}
\caption{Left: Density of the local half-standard Gaussian and non-local $ \mathcal{IG}(\nu_2; 4.35, 1.56)$ priors for $\nu_2$. Right: The corresponding prior density for $\kappa_2 = \nicefrac{1}{\sqrt{\nu_2}}$. The non-local prior sets $P(\kappa_2 \in (1,3))=0.95$.}
\label{Fig:Priors}
\end{figure}

\color{black}
\subsubsection{Asymptotic MSE}\label{App:asympt_MSE}

%\jack{epsilon-contaminated MSE}

In Figure \ref{Fig:MarginalHscore_epsContamination} we compare the \Hscore{} for fixed $\kappa_2$ with an asymptotic estimate of the root mean squared error for the location parameter $\hat{\beta}(\kappa_2)$ of $\hat{\theta}_2(\kappa_2)$ minimising Tukey's loss $\ell_2(y; \theta_2, \kappa_2)$. %In the context of robust estimation, 
\cite{warwick2005choosing} propose a form of mean squared error in situations where the data generating distribution $g(y)$ has a $\epsilon$-contamination form 
\begin{equation}
    g(y) = (1 - \epsilon)f(y; \theta_0) + \epsilon h(y),\nonumber
\end{equation}
like that considered in Section \ref{ssec:robustness_efficiency_tradeoff}. Here $(1-\epsilon)\times 100\%$ of the data comes from the parametric model with generating parameter $\theta_0$ and $\epsilon \times 100\%$ comes from contaminating distribution $h(\cdot)$. %argue that the parameters of interest are the uncontaminated parameters...

In this scenario, an asymptotic approximation to the MSE of $\hat{\theta}_2(\kappa_2)$ estimated using $n$ observations from $g(y)$ is given by 
\begin{align}
    \mathbb{E}_g\left[\left(\hat{\theta}_2(\kappa_2) - \theta_0\right)^T\left(\hat{\theta}_2(\kappa_2) - \theta_0\right)\right] = \left(\tilde{\theta}^*_{\kappa_2} - \theta_0\right)^T\left(\tilde{\theta}^*_{\kappa_2} - \theta_0\right) + \frac{1}{n}\textrm{tr}\left\{J^{-1}_{\kappa_2}(\tilde{\theta}^*_{\kappa_2})K_{\kappa_2}(\tilde{\theta}^*_{\kappa_2})J^{-1}_{\kappa_2}(\tilde{\theta}^*_{\kappa_2})\right\},\label{Equ:AsympMSE}
\end{align}
where $\textrm{tr}\left\{\cdot\right\}$ is the trace operation and
\begin{align}
    \tilde{\theta}^*_{\kappa_2} &:= \argmin_{\theta_2}\mathbb{E}_g\left[\ell_2(y; \theta_2, \kappa_2)\right]\nonumber\\
    K_{\kappa_2}(\theta_2) &:= \mathbb{E}_g\left[\nabla_{\theta_2}\ell_2(y_1; \theta_2, \kappa_2)\nabla_{\theta_2}\ell_2(y_1; \theta_2, \kappa_2)^T\right]\nonumber\\
    J_{\kappa_2}(\theta_2) &:= \mathbb{E}_g\left[\nabla^2_{\theta_2}\ell_2(y_1; \theta_2, \kappa_2)\right]\nonumber
\end{align}
In \eqref{Equ:AsympMSE}, the first term is the summed squared bias of the minimiser of Tukey's loss in expectation relative to the data generating parameter $\theta_0$. The second term is the asymptotic sandwich covariance matrix of $\hat{\theta}_2(\kappa_2)$ based on $n$ observations. 

In Figure \ref{Fig:MarginalHscore_epsContamination} we consider only the MSE of the location component $\beta$ of $\theta_2$ and therefore $\beta_0 = 0$ is the mean of the uncontaminated data. Although \eqref{Equ:AsympMSE} has a very different form to the \Hscore{} in \eqref{Equ:HscoreTukeysLoss} of the main paper, both are optimised for similar parameter values.

\subsubsection{Finite Gaussian Mixture Model Implementation}{\label{App:FiniteGaussianMixDetails}}

Section \ref{Sec:KDEGaussianMixtures} implemented Gaussian Mixture Models using the marginal likelihood to select the number of components. Here we provide details of the implementation used. We estimated Gaussian mixture models of the form 
\begin{align}
    f(y; m, \mu, \sigma, J) = \sum_{j=1}^J m_j N(y; \mu_j, \sigma_j)\nonumber
\end{align}
using Normal-Inverse-Gamma-Dirichlet priors 
\begin{align}
    (m_1, \ldots, m_J) &\sim \textup{Dir}(\alpha_1, \ldots, \alpha_j)\nonumber\\
    \sigma_j^2 &\sim \mathcal{IG}\left(\frac{\nu_0}{2}, \frac{S_0}{2}\right), \quad j = 1,\ldots, J\nonumber\\
    \mu_j | \sigma_j &\sim \mathcal{N}(0, \sqrt{\kappa}\sigma_j), \quad j = 1,\ldots, J\nonumber
\end{align}
with $\alpha_j = 1$ $j = 1,\ldots, J$, $\nu_0 = 5$, $S_0 = 0.2$ and $\kappa = 5.68$ following the recommendations of \cite{fuquene2019choosing}. We first estimated the number of components as 
\begin{align}
    \hat{J} = \argmax_{1\leq J \leq10} \int f(y; \mu, \sigma, J)\pi(\mu |\sigma^2)\pi(\sigma^2)\pi(m){}d\mu{} d\sigma^2 dm\nonumber
\end{align}
where the integral is estimated using the posterior probability that one cluster is empty under each possible $J$ \citep{fuquene2019choosing} and is implemented in the \textit{mombf} \citep{rossell2021package} in \R. Then given $\hat{J}$, we found MAP estimates for $\{m_j\}_{j=1}^{\hat{J}}$, $\{\mu_j\}_{j=1}^{\hat{J}}$ and $\{\sigma_j\}_{j=1}^{\hat{J}}$.

\subsubsection{The prior on the KDE bandwidth}{\label{App:KDEPrior}

%\david{In the spirit of Bayesian prior elicitation methods, a possibility is to view $h$ not as a tuning parameter that has no immediate interpretation, but as its defining prior beliefs regarding (or regularization towards) certain features of the target density.}

When using the \Hscore{} to conduct Bayesian Kernel Density Estimation in Section \ref{Sec:KDE} we considered prior distributions of the form $\pi(h^2, w) = \mathcal{IG}(h^2; a_0, b_0)\textup{Exp}(w; \lambda_0)$. The inverse-gamma shape parameter was set to $a_0 = 2$ resulting in $\pi(h^2)$ with infinite variance. Prior parameters $b_0$ and $\lambda_0$ were elicited in order to provide regularisation towards the values that perform well when the data is Gaussian. The elicitation procedure is outlined below. %In particular, we specify priors on $h$ and $w$ such that if the data were truly Gaussian, the KDE approximation of its density would be as accurate as possible. 

%\paragraph{Elicitation Procedure}
%The elicitation proceeds as follows: 
Suppose that data were truly Gaussian and consider simulating 
%$z \in \mathbb{R}^n$ from a Gaussian matching the sample mean and variance of the observed $y \in \mathbb{R}^n$, i.e. $z_i \sim N(\bar{y}, S_y)$ for $i=1,\ldots,n$. 
$z \in \mathbb{R}^n$ from a standard Gaussian, i.e. $z_i \sim N(0, 1)$ for $i=1,\ldots,n$. This is consistent with the recommendations of the \textit{dirichletprocess} package to  standardised the data. 
Then the mean integrated squared error (MISE) of \KDE $\hat{g}_{h, w}$ with bandwidth $h$ and tempering parameter $w$ from the generating distribution $N(0, 1)$ is
%and its best Gaussian approximation based on the sampled $z$'s $N(\bar{z}, S_z)$ is
%\begin{equation}
%    EISE[g_{norm}]:= \widebar{EISE}[\bar{y}, S_y, n] := \mathbb{E}_{z}\left[ \int [N(y^{\prime}; \bar{z}, S_z)  - N(y^{\prime}; \bar{y}, S_y) ]^2 dy^{\prime}\right]
%\end{equation}
%where the expectation is over $(z_1, \ldots, z_n) \overset{iid}{\sim} N(\bar{y}, S_y)$, $\bar{z} = \frac{1}{n}\sum_{i=1}^n z_i$ and $S_z$ is the empirical variance of the sampled $(z_1, \ldots, z_n)$. Given that we generate from a Gaussian distribution this provides a lower bound for the EISE of any density estimation technique. 
%We can define the same criteria for the KDE approximation to the generating distribution $N(\bar{y}, S_y)$ using sample $z$, bandwidth $h$ and tempering parameter $w$
\begin{equation}
    \textup{MISE}[\hat{g}_{h, w}] := \mathbb{E}_{z}\left[ \int [\hat{g}_{h, w}(y^{\prime})  - N(y^{\prime}; 0, 1) ]^2 dy^{\prime}\right].\nonumber
\end{equation}
where the expectation is over $(z_1, \ldots, z_n) \overset{iid}{\sim} N(0, 1)$.
We then elicit prior parameters $b_0$ and $\lambda_0$, given $a_0 = 2$ such that 
\begin{equation}
    b_0, \lambda_0 := \argmin_{b, \lambda}\mathbb{E}_{\pi(h^2, w)} \textup{MISE}[\hat{g}_{h, w}].\nonumber
\end{equation}
%This procedure elicits priors on $h$ and $w$ such that if the data were truly Gaussian, the KDE approximation of its density would be as accurate as possible, whilst being minimally informative.
This results in a default $\pi(h^2,w)$ centered on $(h,w)$ values such that, if the data were truly Gaussian, the estimated density would be accurate, subject to the prior being minimally informative (e.g. $h^2$ has infinite prior variance).

In the case without tempering when $w = 1$, we consider only prior $\pi(h^2) = \mathcal{IG}\left(h^2; 2, b_0\right)$ and $b_0 = 0.061$ was elicited. In the case where we also estimated tempering parameter $w$, we elicited $b_0 = 0.024$ and $\lambda_0 = 0.725$.

\color{black}

\subsection{Supplementary results}
\label{App:suppl_results}

\subsubsection{Variable selection for \TGFB and \DLD dataset}{\label{App:SelectedVariables}}

To illustrate our methodology for robust regression in Section \ref{Sec:RobustRegression}, we selected between the Gaussian model and Tukey's loss when regressing the \TGFB and \DLD gene expressions on a subset of the variables available in the full data sets. The procedures for which variables were selected was outlined in Sections \ref{TGFBanalysis} and \ref{DLDanalysis}. To ensure that our results are reproducible, below we indicate the selected covariates and the supplementary material contains code for these variable pre-screening steps. We also provide supplementary figures providing a visual inspection of the Gaussian fit to the residuals of the \TGFB and \DLD examples.

\paragraph{\TGFB:}

For the \TGFB analysis we focused on 7 of the 10172 genes available in the data set that appear in the ‘\TGFB 1 pathway’ according
to the KEGGREST package in R (Tenenbaum, 2016). These were the VIT, PDE4B, ATP8B1, MAGEA11, PDE6C, PDE9A and SEPTIN4 genes.

\paragraph{\DLD:}

For the \DLD analysis we selected the 15 genes with the 5 highest loadings in the first 3 principal components of the original 57 predictors. This procedure selected the following genes C15orf52, BRAT1, CYP26C1, SLC35B4, GRLF1, RXRA, RAB3GAP2, NOTCH2NL, SDC4, TTC22, PTCH2, ECH1, CSF2RA, TP53AIP1, and RRP1B. 
%which after removing the \DLD gene expression for $y$ and considering the intercept as $X_0$ correspond to predictors $X_1$, $X_4$, $X_6$, $X_{18}$, $X_{23}$, $X_{27}$, $X_{29}$, $X_{32}$, $X_{35}$, $X_{36}$, $X_{39}$, $X_{42}$, $X_{45}$, $X_{49}$, $X_{53}$. 
%Code to run this variable selection is provided in.

\paragraph{Supplementary figures:}

We provide Figure \ref{Fig:TGFB_DLD_Gaussian_results} to visually inspect the Gaussian approximation to the fitted standardised residuals of the \TGFB and \DLD datasets. The left panels overlay the fitted Gaussian models to the histogram of the residuals, while the right are \textit{Q-Q} Normal residual plots. The top panels show that the \TGFB data is well-approximated by a Gaussian model, whereas the bottom panels show that the \DLD data set exhibits heavier tails. These graphical inspections support the \Hscore{} findings from Section \ref{Sec:RobustRegression} and Figure \ref{Fig:TGFB_DLD_Tukeys_results}.

\begin{figure}[h!]
\begin{center}
    \includegraphics[trim= {0.0cm 0.0cm 0.0cm 0.0cm}, clip,  
    width=0.49\columnwidth]{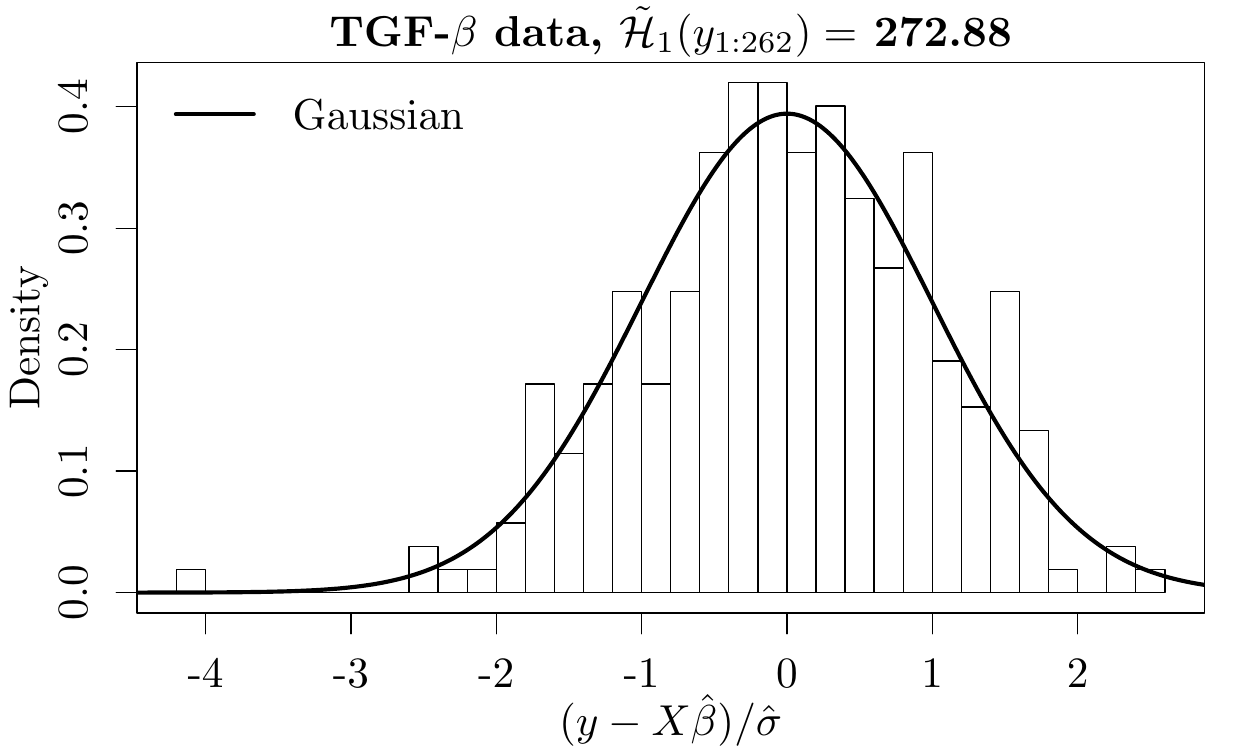}
    \includegraphics[trim= {0.0cm 0.0cm 0.0cm 0.0cm}, clip,  
    width=0.49\columnwidth]{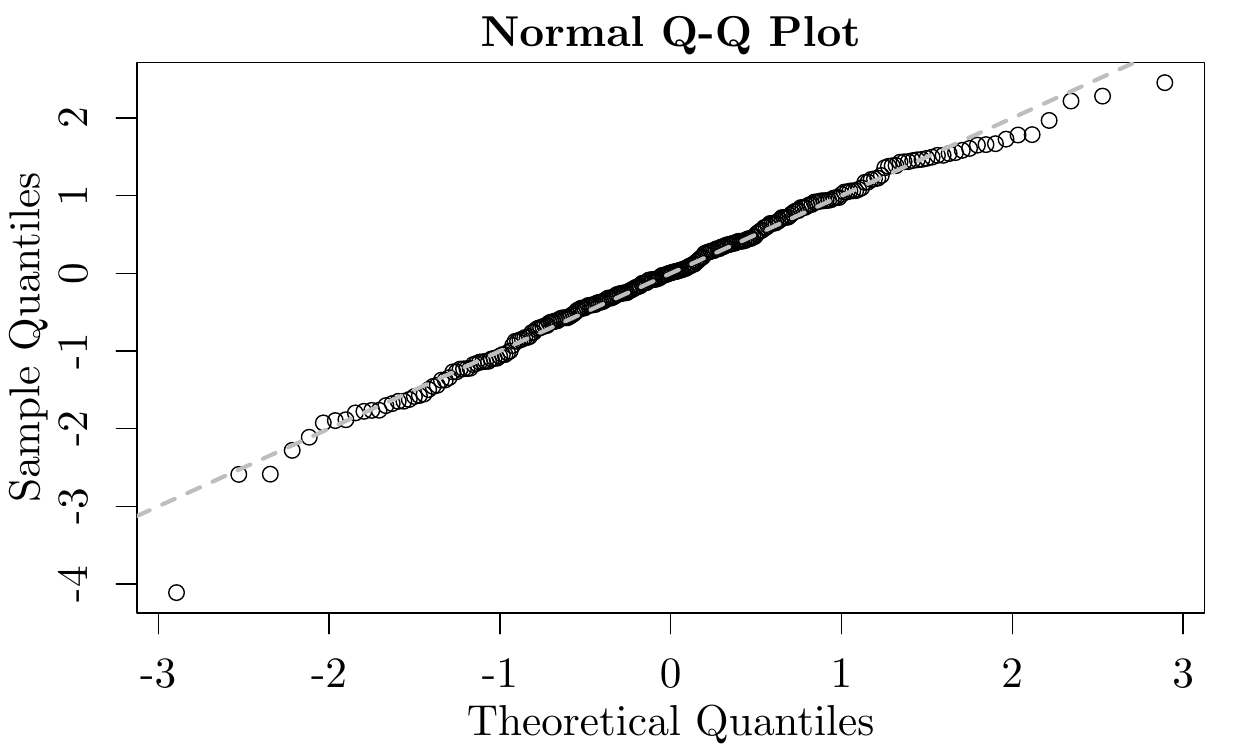}\\
    \includegraphics[trim= {0.0cm 0.0cm 0.0cm 0.0cm}, clip,  
    width=0.49\columnwidth]{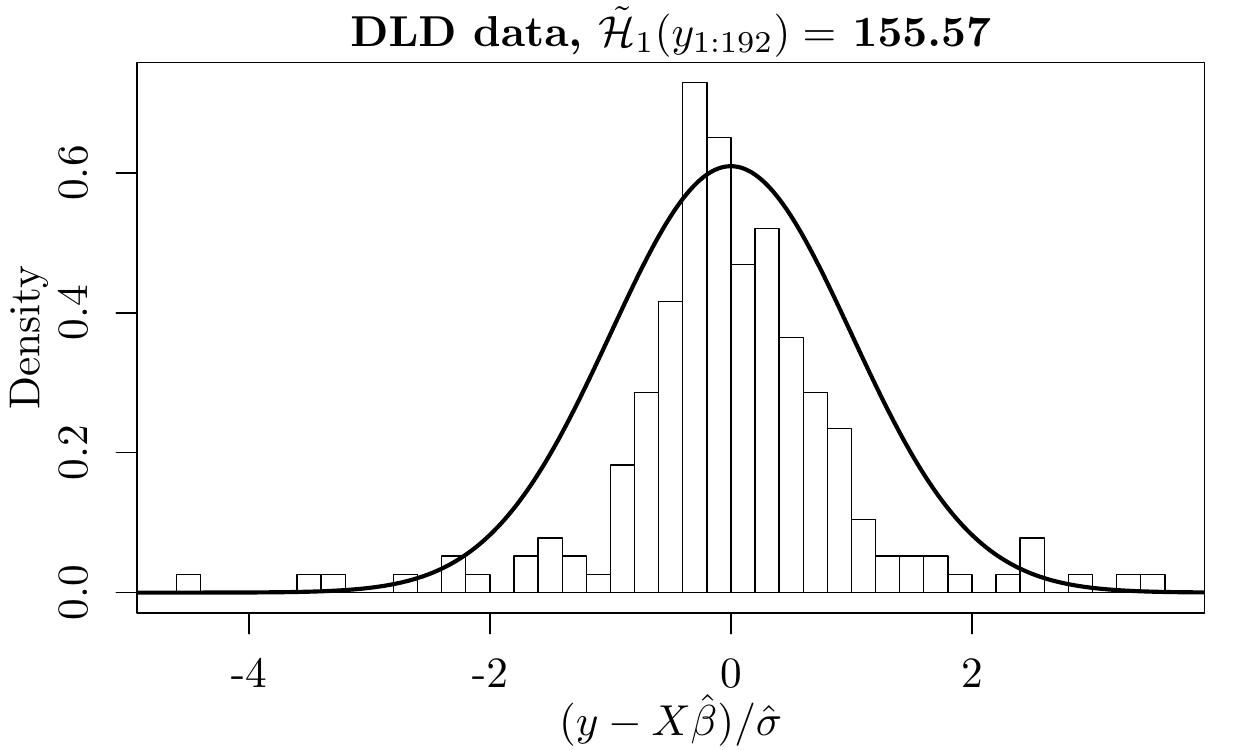}
    \includegraphics[trim= {0.0cm 0.0cm 0.0cm 0.0cm}, clip,  
    width=0.49\columnwidth]{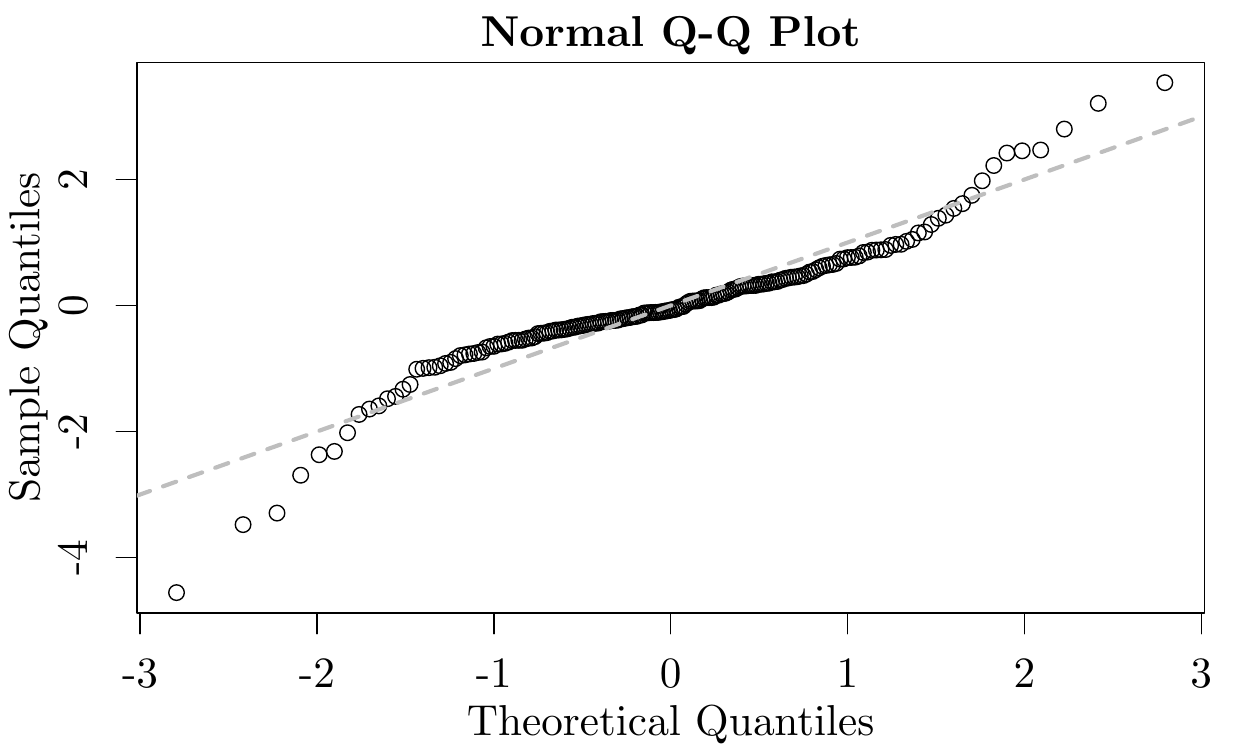}\\
    \caption{\textbf{Top:} \TGFB data, where the integrated \Hscore{} selected the Gaussian model. \textbf{Bottom:} \DLD data, where the integrated \Hscore{} selected Tukey's loss. \textbf{Left:} Gaussian approximations to the residuals. \textbf{Right:} \textit{Q-Q} normal plot of the fitted residuals according to the Gaussian model.}
    \label{Fig:TGFB_DLD_Gaussian_results}
\end{center}
\end{figure}

\subsubsection{Further Gaussian Mixture Examples}{\label{App:KDEGaussianMictures}}

\color{black}

%Section \ref{Sec:KDE} demonstrated how the \Hscore{} can be used to estimate bandwidth and tempering parameters for kernel density estimation.   However, the standard implementation for kernel density estimation considers $w = 1$. 
We provide the results of the \Hscore{} bandwidth estimation for \KDE's fixing $w = 1$, in addition to the results presented in Section \ref{Sec:KDEGaussianMixtures}.  
\color{black}
%Additionally to the experiments in Section \ref{Sec:KDEGaussianMixtures}, 
We also provide two further Gaussian mixture data sets from \cite{marron1992exact}
%
%Further to the experiments in Section \ref{Sec:KDEGaussianMixtures}, we provide implementation of our \Hscore{} inference for two further Gaussian mixture data sets from \cite{marron1992exact}
\begin{itemize}[leftmargin=*]
    %\item[] \textit{separated}: 2-components,  $\mu = (1.5, 1.5)$, $\sigma^2 = (1/4, 1/4)$ and $\omega = (0.5, 0.5)$.
    \item\textit{Asymmetric}: 2-components,  $\mu_1 = 0$, $\mu_2 = 1.5$, $\sigma_1 = 1$, $\sigma_2 = 1/3$, $m_1 = 0.75$ and $m_2 = 0.25$.
   %\item[] \textit{trimodal}: 3-components,  $\mu = (- 1.2, 0, 1.2)$, $\sigma^2 = (9/25, 1/16, 9/25)$  and $\omega = (0.45, 0.1, 0.45)$.
    %\item[] \textit{claw}: 6-components, $\mu = (0, -1, -0.5, 0, 0.5, 1)$, $\sigma^2_1 = 1$ and $\sigma^2_i = 0.01$ for $i = 2, \ldots, 6$ and $\omega_1 = 0.5$ and $\omega_i = 0.1$ for $i = 2, \ldots, 6$.
    \item\textit{Kurtotic}: 2-components,  $\mu_1 = \mu_2 = 0$, $\sigma_1 = 1$ $\sigma_2 = 0.1$, $m_1 = 2/3$ and $m_2 = 1/3$.
    %\item[] \textit{skewed}: 8-components, $s = \{\nicefrac{2}{3}^i\}_{i=0:7}$  $\mu = 3(s - 1)$, $\sigma^2 = s^2$ and $\omega_i = 1/8$, $i = 1,\ldots, 8$.
%    \item[outlier] 2-component:  $\mu = (0, 0)$, $\sigma^2 = (1, 0.01)$ and $\omega = (0.1, 0.9)$.
\end{itemize}
Figure \ref{Fig:KDEEstimation2} plots the approximation to the underlying $g(y)$ of the methods under consideration in Section \ref{Sec:KDEGaussianMixtures} as the well as the \Hscore{} \KDE with $w = 1$, while Table \ref{Tab:KDEEstimation2} estimates the Fisher's divergence of the estimates (posterior mean where appropriate) to $g(y)$. \new{For the datasets considered in Section \ref{Sec:KDEGaussianMixtures} we see that generally learning the tempering parameter $w$ improved the density estimation relative to $w = 1$, however in most cases $w = 1$ remains competitive with the other methods considered}. For both the extra examples the finite mixture model provides the best estimate according the Fisher's divergence. For the \textit{Kurtotic} data set, the \Hposterior{} \KDE's perform comparably to the \DPMM and better than standard the \KDE implementations. However, for the \textit{Asymmetric} data set the \Hposterior{} \KDE's under-smooth the data and perform considerably worse than the other methods.

\begin{table}[ht]
\centering
\caption{Fisher's divergence between the density estimates and the data-generating Gaussian mixtures in four simulation scenarios. The best performing method in each is highlighted in bold face.}
\begin{tabular}{rrrrrrr}
  \hline
 & \KDE & \KDE (ucv) & \DPMM & Finite Mixture Model & $\mathcal{H}$-KDE ($w = 1$) & $\mathcal{H}$-KDE ($w \neq 1$)\\ 
   \hline
   \textit{Bimodal} & 1.03 & 0.37 & 0.13 & 0.10 & 0.26 & \textbf{0.09} \\ 
   \textit{Claw} & 13.77 & 6.09 & 15.25 & \textbf{2.17} & 3.37 & 2.51 \\ 
   \textit{Trimodal} & 0.26 & \textbf{0.12} & 0.32 & 0.33 & 0.18 & 0.18 \\ 
   \textit{Skewed} & 21.34 & 16.15 & 17.61 & \textbf{6.12} & 10.52 & 9.51 \\ 
   \textit{Asymmetric} & 0.25 & 3.76 & 0.25 & \textbf{0.07} & 48.84 & 2.33 \\ 
   \textit{Kurtotic} & 3.14 & 7.68 & 2.32 & \textbf{0.47} & 2.56 & 2.44 \\ 
   \hline
 \end{tabular}
\label{Tab:KDEEstimation2}
\end{table}

%Hyvarinen_KernelDensityEstimation_ParametricNonParametricConsistency.Rmd
\begin{figure}[H]
\begin{center}
\includegraphics[trim= {0.0cm 0.0cm 0.0cm 0.0cm}, clip,  
width=0.49\columnwidth]{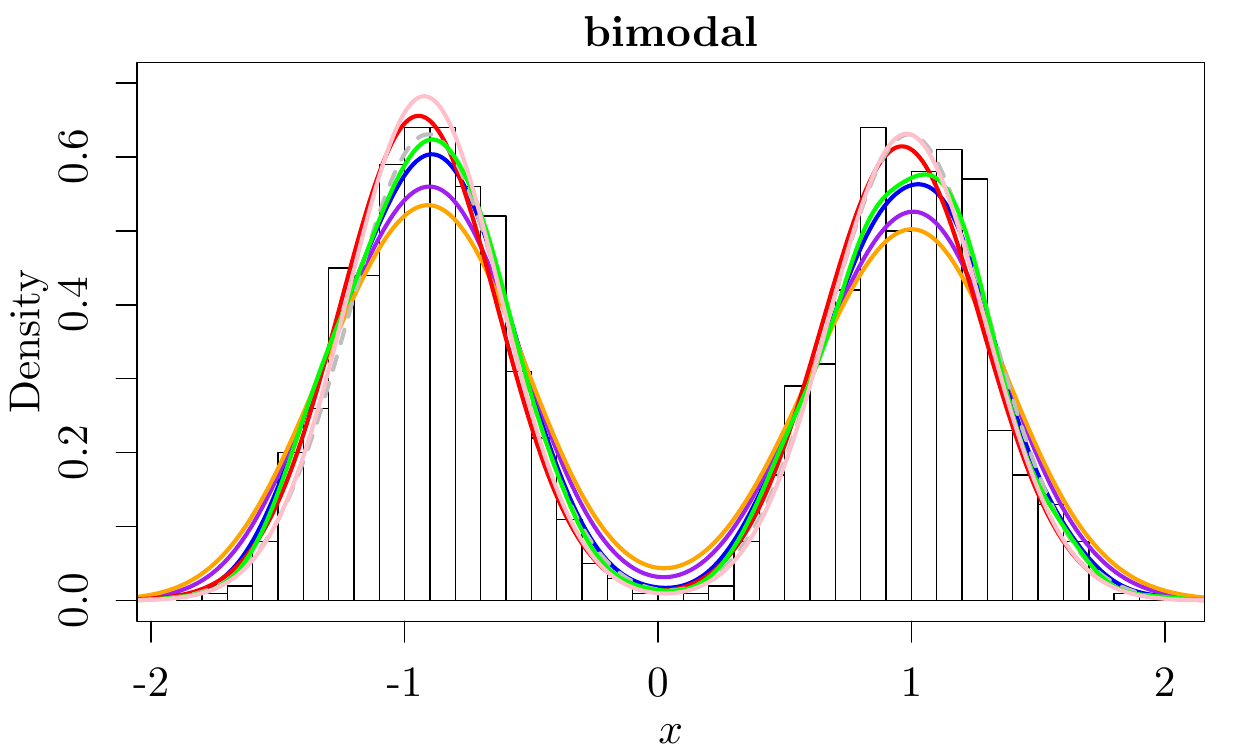}
\includegraphics[trim= {0.0cm 0.0cm 0.0cm 0.0cm}, clip,  
width=0.49\columnwidth]{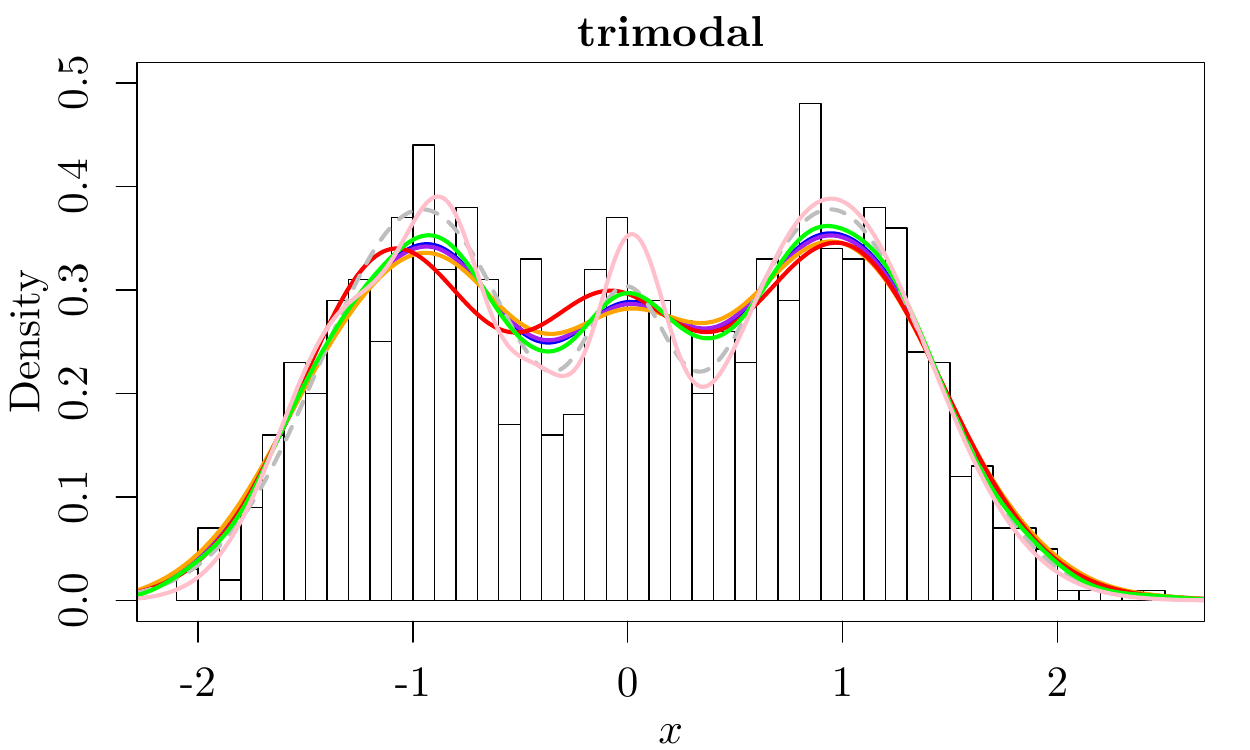}
\includegraphics[trim= {0.0cm 0.0cm 0.0cm 0.0cm}, clip,  
width=0.49\columnwidth]{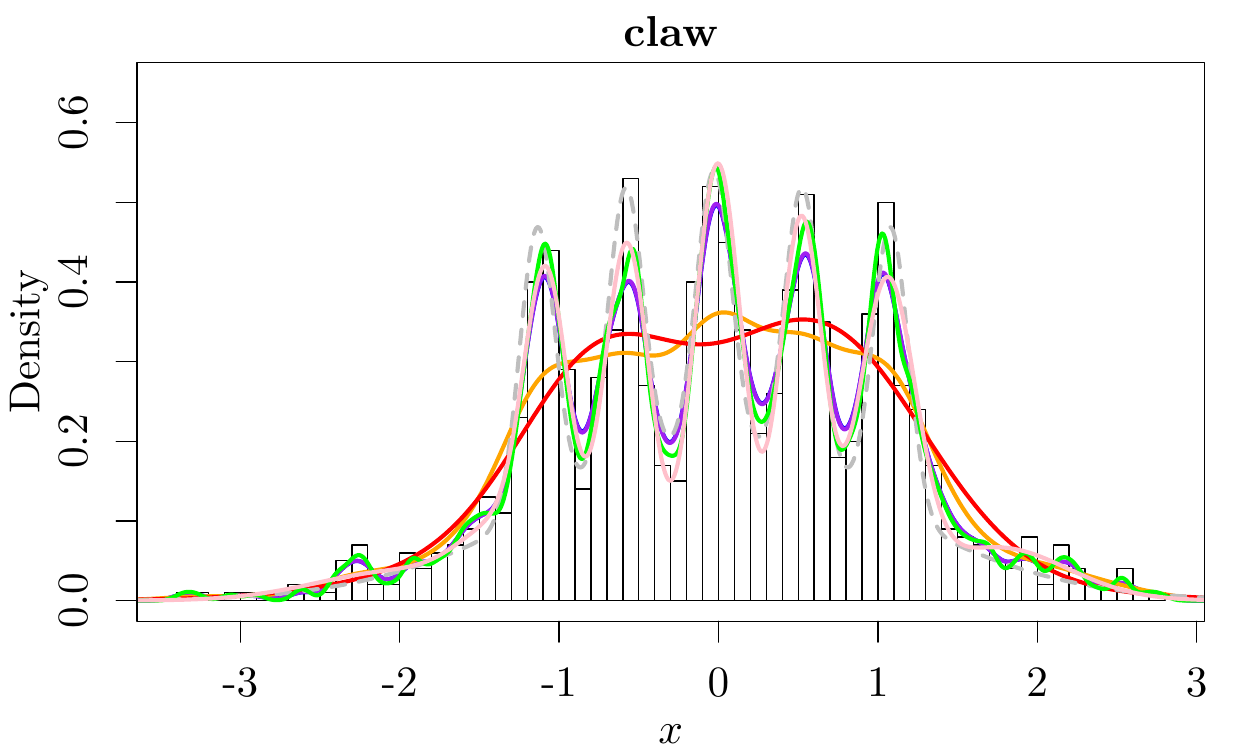}
\includegraphics[trim= {0.0cm 0.0cm 0.0cm 0.0cm}, clip,  
width=0.49\columnwidth]{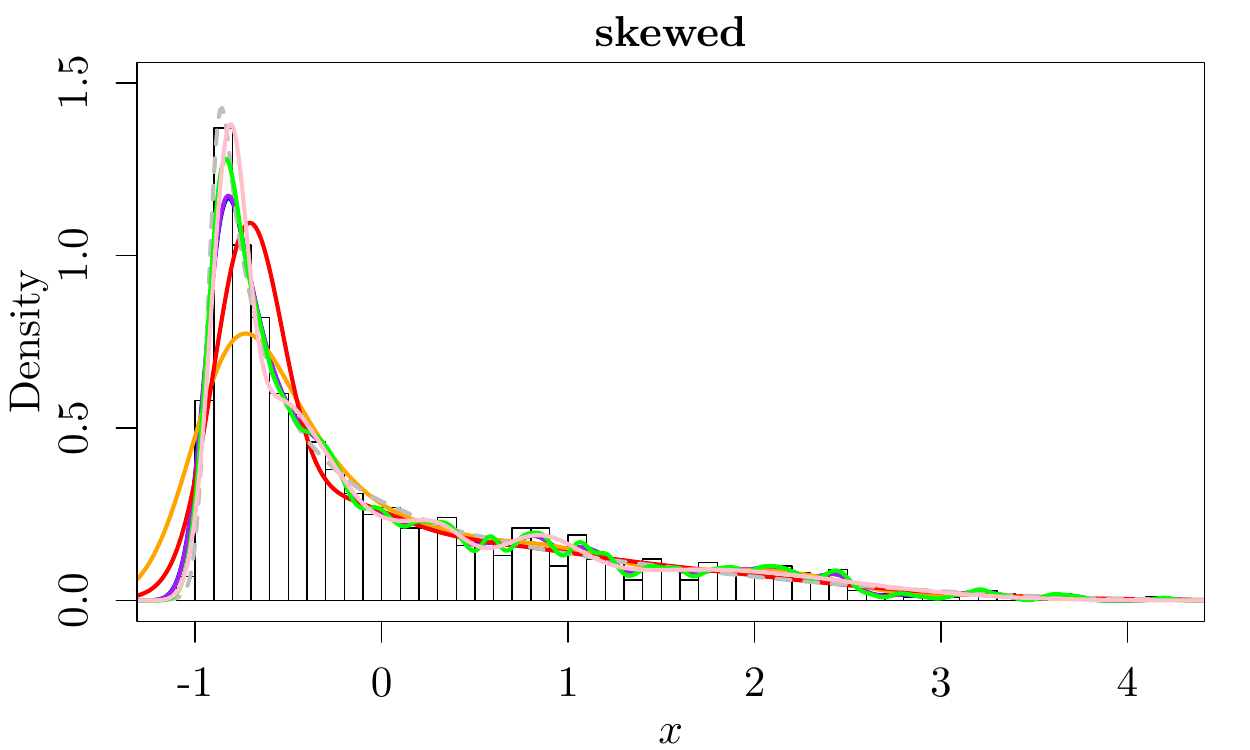}
\includegraphics[trim= {0.0cm 0.0cm 0.0cm 0.0cm}, clip,  
width=0.49\columnwidth]{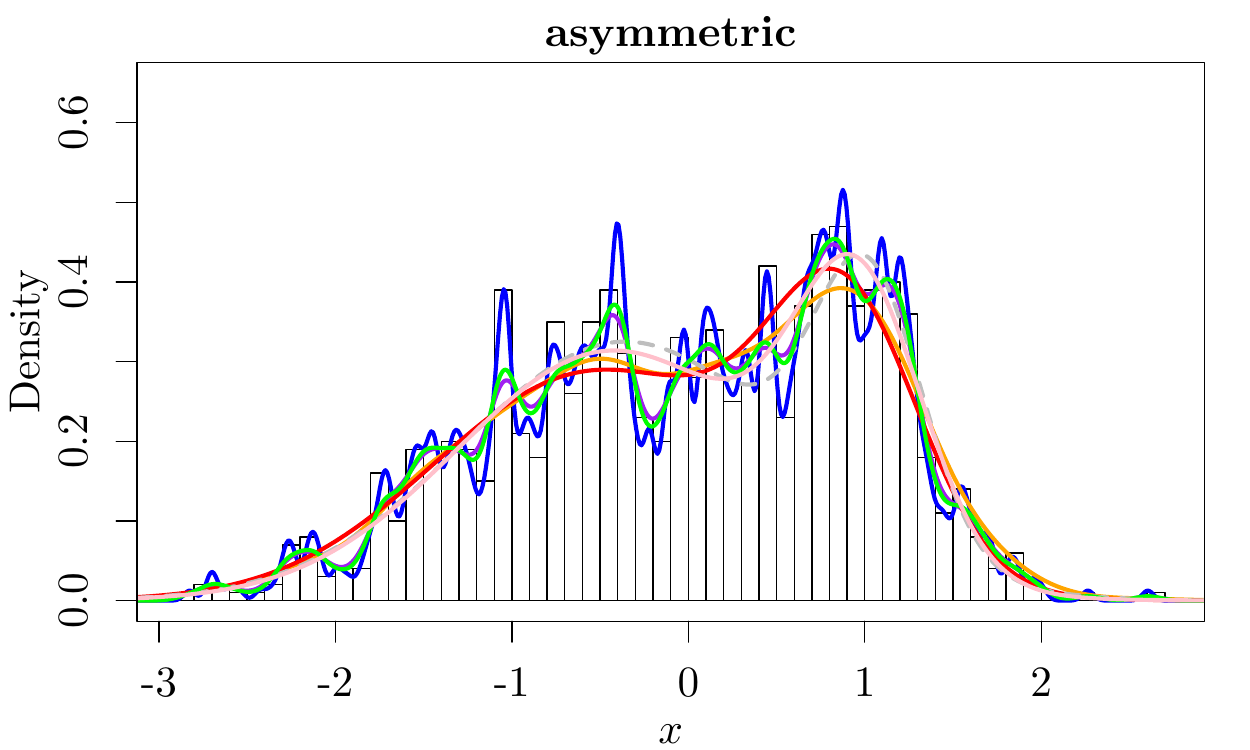}
\includegraphics[trim= {0.0cm 0.0cm 0.0cm 0.0cm}, clip,  
width=0.49\columnwidth]{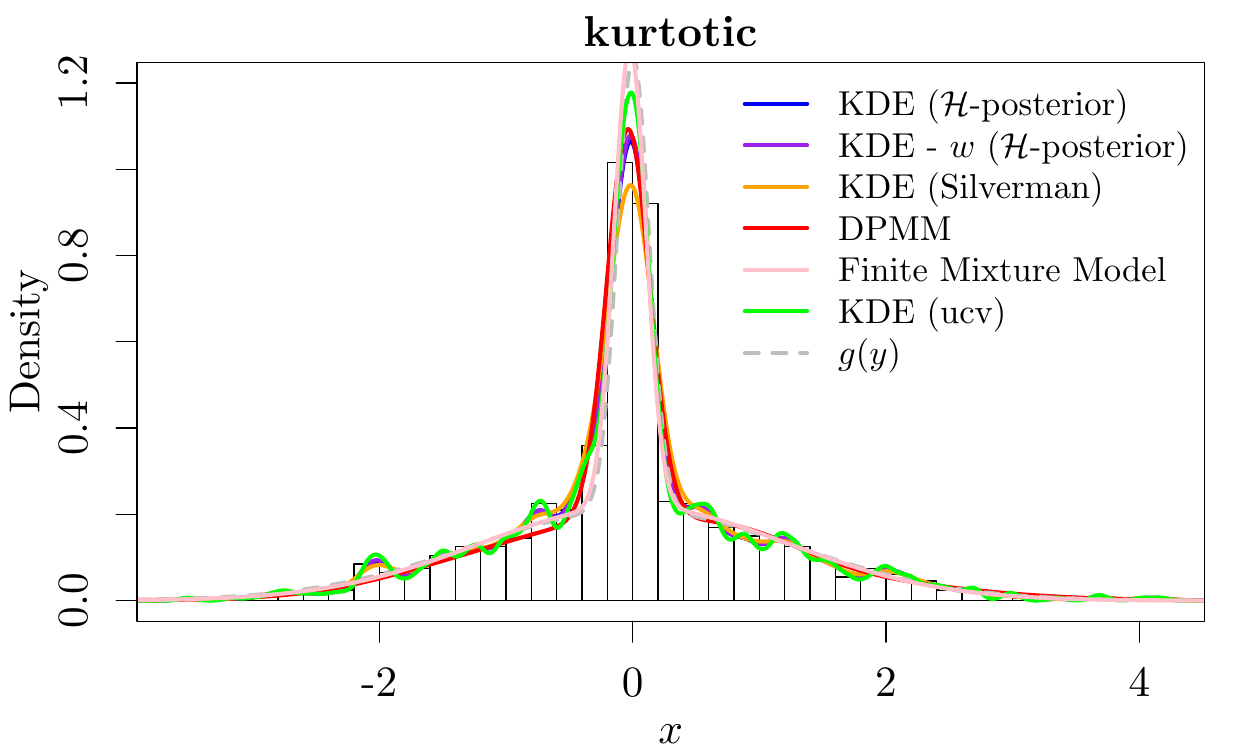}
%trim={<left> <lower> <right> <upper>}
\caption{Density estimation for Gaussian mixture data. Histograms of observed data, standardised to 0 mean and unit variance, and estimated density by   
\R's density function with the bandwidth rule of \cite{silverman1986density}, unbiased cross validated bandwidth estimation,  
\DPMM, Bayesian mixtures with marginal likelihood selection of components, the \Hposterior{} estimate with no tempering ($w = 1$),  
and with tempering  ($w \neq 1$). For Bayesian methods the density associated with MAP estimates is plotted.
}
\label{Fig:KDEEstimation2}
\end{center}
\end{figure}

%\begin{table}[H]
%\centering
%\caption{Fisher’s divergence between the density estimates and the data-generating Gaussian mixtures in two simulation scenarios.  The best performing method in each is highlighted in bold face.}
%\begin{tabular}{rrrrr}
%  \hline
%   & \KDE & \DPMM & \hyvarinen \KDE $(w=1)$ & \hyvarinen \KDE{} $(w \neq 1)$ \\ 
%   \hline
%   %separated & {\color{gray}{\textbf{1.03}}} & \textbf{0.04} & 0.26 & 0.09 \\ 
%   %outlier & 0.31 & 0.07 & 0.19 & \textbf{0.03} \\ 
%   %\textit{Asymetric} & 0.21 & \textbf{0.13} & 0.22 & 0.20 \\ 
%   %\textit{Kurtotic} & 4.03 & \textbf{3.36} & 4.03 & 3.89 \\ 
%   %claw & {\color{gray}{\textbf{13.77}}} & {\color{gray}{\textbf{15.29}}} & 3.39 & \textbf{2.51} \\ 
%   %trimodal & \textbf{0.41} & 0.64 & 0.56 & 0.57 \\ 
%   %strongly skewed & {\color{gray}{\textbf{20.23}}} & {\color{gray}{\textbf{16.46}}} & 11.62 & \textbf{10.46} \\ 
%   %bimodal & 0.15 & \textbf{0.03} & 0.24 & 0.28 \\ 
%   %
%   \textit{Asymmetric} & 0.25 & \textbf{0.04} & 2.17 & 2.24 \\ 
%   \textit{Kurtotic} & 3.14 & \textbf{2.05} & 2.55 & 2.44 \\ 
%    \hline
% \end{tabular}
% \label{Tab:KDEEstimation2}
% \end{table}
%%\footnotetext{Note: Fisher's divergences were estimated across the support of the observed data for all methods, rather than the support of the underlying $g(y)$.}

\color{black}
\subsubsection{The \Hscore{} inference for multimodal datasets}{\label{App:Hscore_multimodal}}
We provide two extended experiments to show that Fisher's divergence/\hyvarinen scores' inability to learn the mixture weights for an assumed mixture model when the components are well-separated \citep{wenliang2020blindness} is not an issue in our Tukey's loss or kernel density estimation examples, as discussed in Section \ref{Sec:discussion}.

\paragraph{Robust Regression with Tukey's loss:} We extended the experiments conducted in Section \ref{ssec:robustness_efficiency_tradeoff}, estimating Tukey's loss to outlier contaminated data. Here we considered $g(y) = 0.9\mathcal{N}(y; 0, 1) + 0.1\mathcal{N}(y; \mathbf{10}, 3)$ with the outlying mean twice as far away from the uncontaminated mean as it was in Section \ref{ssec:robustness_efficiency_tradeoff}. The resulting $g(y)$ has an area of almost 0 density in between its two modes. The left hand side of Figure \ref{Fig:MarginalHscore_epsContamination_10} plots the $\mathcal{H}_2(y, \kappa_2)$ compared with an asymptotic approximation to the RMSE of $\hat{\beta}(\kappa_2)$ \citep{warwick2005choosing} (see Section \ref{App:asympt_MSE}) and the right hand side of Figure \ref{Fig:MarginalHscore_epsContamination_10} compares the fitted Gaussian and Tukey's loss densities to a histogram of the generated data 

For this example $\mathcal{H}_2(y, \kappa_2)$ is maximised at  $\kappa_2 = 7.5$ and similarly to the results in Section \ref{ssec:robustness_efficiency_tradeoff}, we see that the resulting approximation of Tukey's loss to the data captures its central component and excludes outliers. We further see that values of $\kappa_2$ that achieved high \Hscore's appear to also correspond to values of $\kappa_2$ whose Tukey's loss minimising estimates $\hat{\beta}(\kappa_2)$ have small RMSE. The separation of the modes appears not to have adverse effects on the \Hscore's performance.

\begin{figure}[h]
\begin{center}
\includegraphics[trim= {0.0cm 0.0cm 0.0cm 0.0cm}, clip,  
width=0.49\columnwidth]{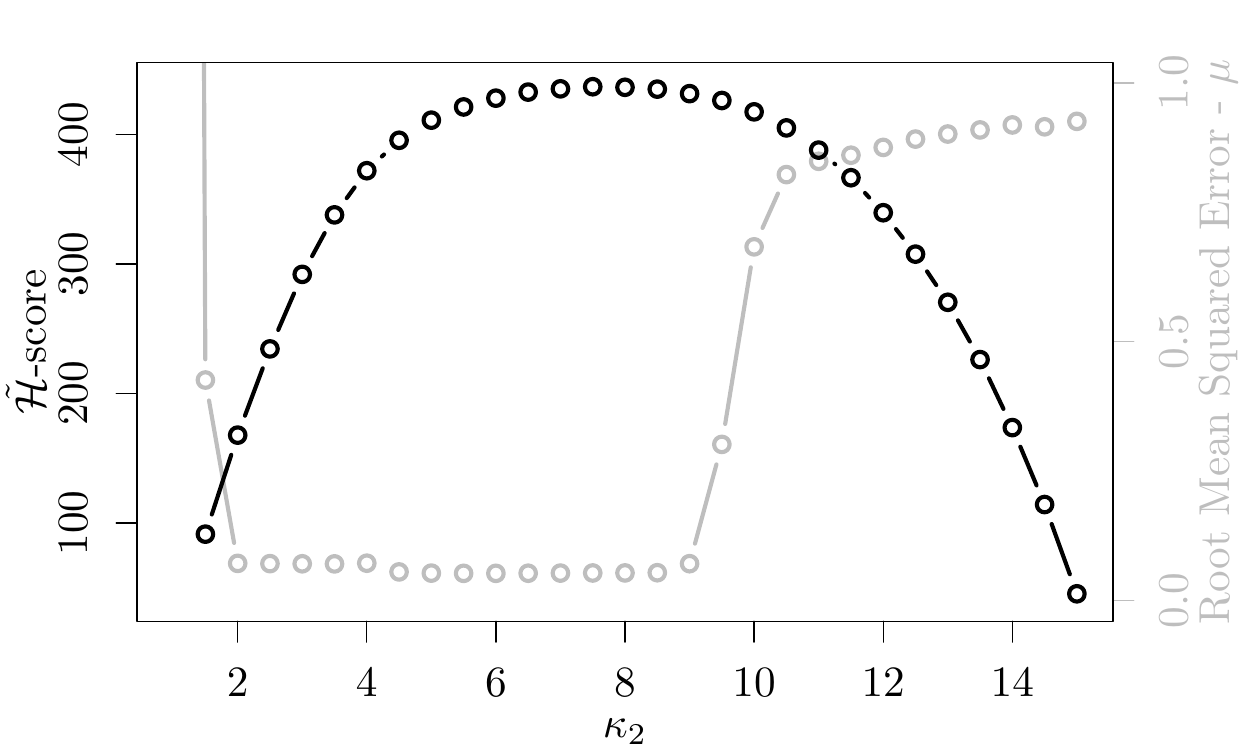}
\includegraphics[trim= {0.0cm 0.0cm 0.0cm 0.0cm}, clip,  
width=0.49\columnwidth]{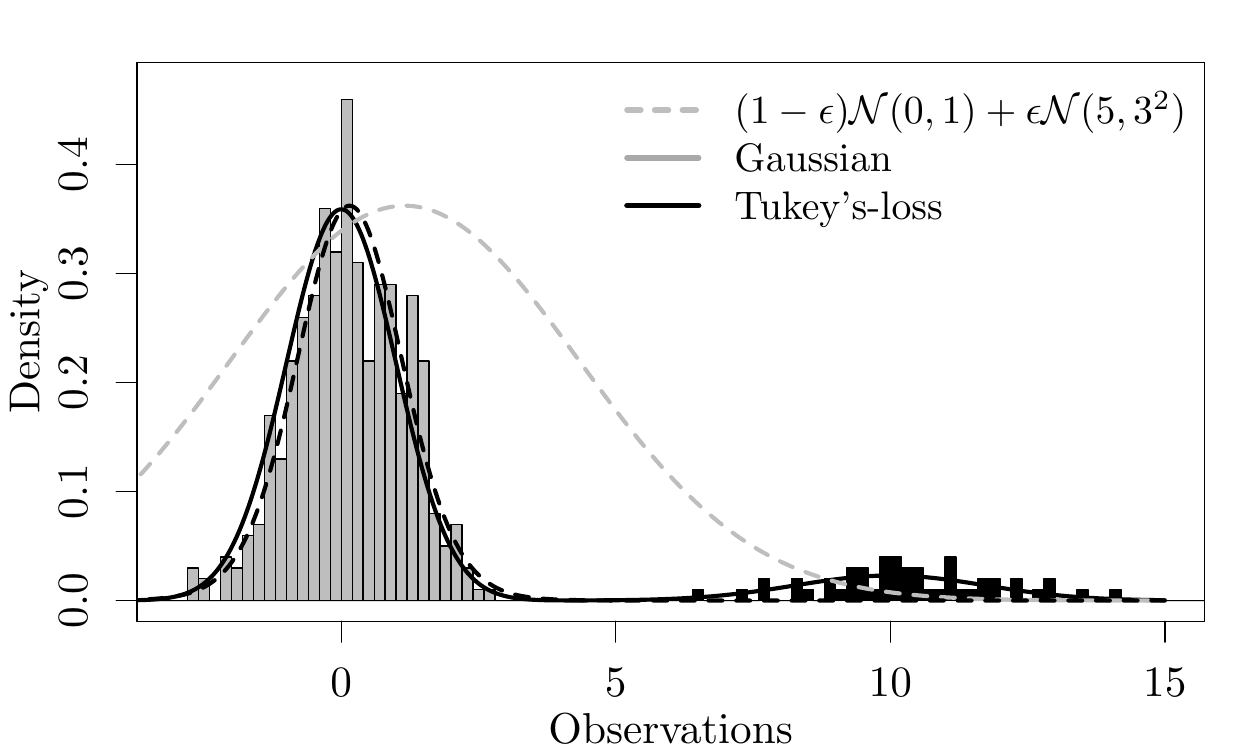}
%trim={<left> <lower> <right> <upper>}
\caption{\textbf{Left:} The integrated \Hscore{} $\mathcal{H}_2(y; \kappa_2)$ (black)
and asymptotic approximation to the RMSE of $\hat{\beta}(\kappa_2)$ (grey) for varying values of $\kappa_2$. 
\textbf{Right:} Histogram of $n = 500$ observations from $g(y) = 0.9\mathcal{N}(y;0, 1) + 0.1\mathcal{N}(y;\mathbf{10}, 3)$ and fitted densities of Tukey's loss with $\hat{\kappa}_2 = 7.5$ and the Gaussian model (the height of Tukey's loss was set to match the mode of $g(y)$). 
}
\label{Fig:MarginalHscore_epsContamination_10}
\end{center}
\end{figure}

\paragraph{Kernel density estimation} We extend the experiments in Section \ref{Sec:KDEGaussianMixtures}, estimating the bandwidth parameter for kernel density estimation. Here we consider dataset, \textit{Bimodal2}, generated from
\begin{align}
    %g(y)= \sum_{j=1}^J m_j N(y; \mu_j, \sigma_j)
    %\nonumber
    g(y)= 0.5\mathcal{N}\left(y; -2, \frac{1}{\sqrt{6}}\right) + 0.5 \mathcal{N}\left(y; 2, \frac{1}{\sqrt{6}}\right)
    \nonumber
\end{align}
%with $J=2$ components,  $\mu_1 = -2$, $\mu_2 = 2$, $\sigma_1 = \sigma_2 = 1/\sqrt{6}$ and $m_1 = m_2 = 0.5$. 
Importantly, data generated form \textit{Bimodal2} has two disjoint modes with areas of almost 0 density in between. Figure \ref{Fig:KDEEstimation3} plots a histogram of the generated data and the density estimated by the methods considered in Section \ref{Sec:KDEGaussianMixtures}. Table \ref{Tab:KDEEstimation3} presents the Fisher's divergence from each density estimate to the true generating $g(y)$. Similarly to Section \ref{Sec:KDEGaussianMixtures} we see that the \Hscore{} KDE performs comparably with standard methods when the tempering parameter $w$ is estimated alongside the bandwidth.

\begin{figure}[H]
\begin{center}
\includegraphics[trim= {0.0cm 0.0cm 0.0cm 0.0cm}, clip,  
width=0.49\columnwidth]{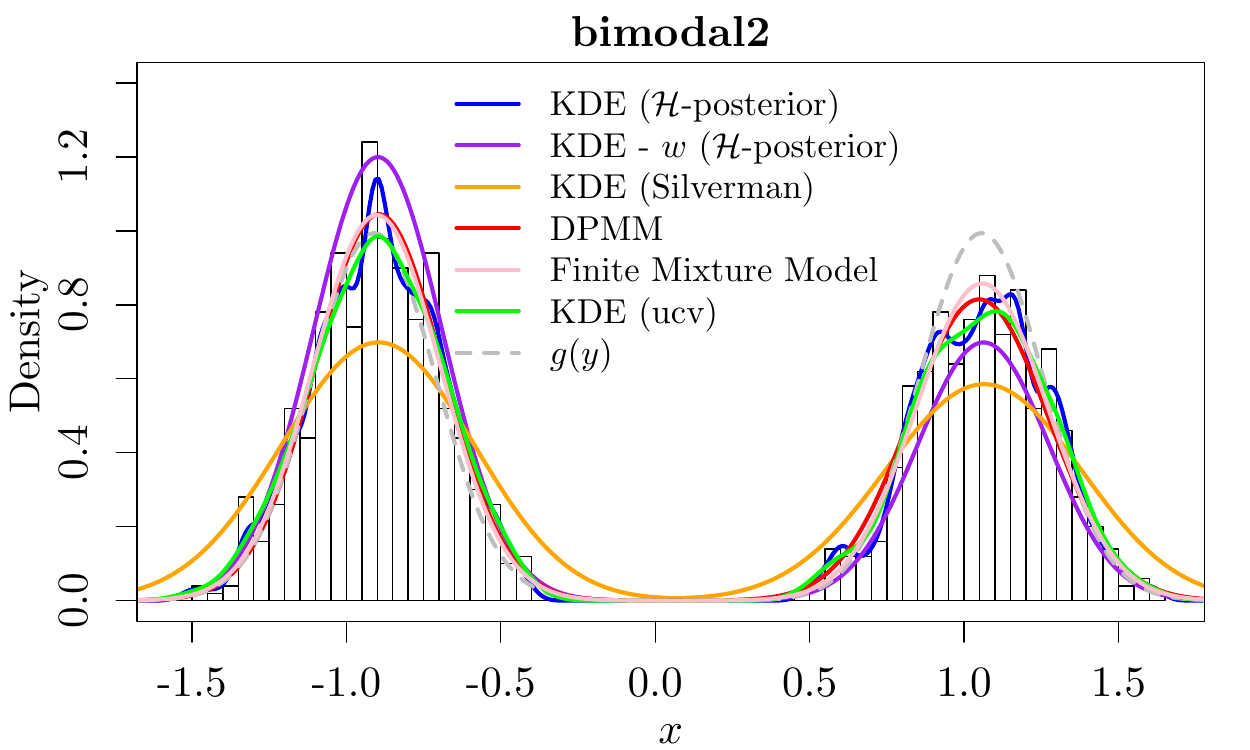}
%trim={<left> <lower> <right> <upper>}
\caption{Density estimation for Gaussian mixture data. Histograms of observed data, standardised to 0 mean and unit variance, and estimated density by \R's density function with the bandwidth rule of \cite{silverman1986density}, unbiased cross validated bandwidth estimation, \DPMM, Bayesian mixtures with marginal likelihood selection of components, the \Hposterior{} estimate with no tempering ($w = 1$), and with tempering  ($w \neq 1$). For Bayesian methods the density associated with MAP estimates is plotted.
}
\label{Fig:KDEEstimation3}
\end{center}
\end{figure}

\begin{table}[ht]
\centering
\caption{Fisher’s divergence between the density estimates and the data-generating Gaussian mixtures in two simulation scenarios.  The best performing method in each is highlighted in bold face.}
\begin{tabular}{rrrrrrr}
  \hline
 & \KDE & \KDE (ucv) & \DPMM & Finite Mixture Model & $\mathcal{H}$-KDE ($w = 1$) & $\mathcal{H}$-KDE ($w \neq 1$)\\ 
  \hline
   \textit{Bimodal2} & 8.36 & 2.40 & 0.62 & \textbf{0.23} & 29.27 & 0.24 \\ 
    \hline
\end{tabular}
\label{Tab:KDEEstimation3}
\end{table}

\color{black}

\end{document}